\documentclass{amsart}
\usepackage[utf8]{inputenc}
\usepackage[OT2,T1,B1]{fontenc}
\usepackage{epsfig}
\usepackage{graphicx}
\usepackage{xcolor}
\usepackage{eepic}
\usepackage{mathrsfs,wasysym}
\usepackage{amssymb,latexsym,ifthen}
\usepackage[all]{xy}
\usepackage[]{hyperref}
\usepackage{datetime2}
\usepackage{bbm}
\usepackage{mathbbol}
\usepackage{slashed}

\renewcommand{\mathbbm}{\mathbb}
\usepackage{esint}
\usepackage{scalerel}
\numberwithin{equation}{section}
\makeatletter
\renewcommand\subsubsection{\@secnumfont}{\bfseries}%
\renewcommand\subsubsection{\@startsection{subsubsection}{3}
  \z@{.5\linespacing\@plus.7\linespacing}{-.5em}%
  {\normalfont\bfseries}}
\makeatother
\usepackage{mathrsfs}
\usepackage{yfonts}
\usepackage[f]{esvect}  

\usepackage{carolmin}  	
\usepackage{uncial}		
\usepackage{humanist}	
\usepackage{auncial}    
\usepackage{inslrmin}   

\usepackage{calligra}   
\usepackage{aurical}    

\usepackage[draft]{media9}

\usepackage[normalem]{ulem} 	

\xyoption{matrix}
\xyoption{arrow}
\newtheorem{Def}{Definition}
\newtheorem{Fact}{Fact}
\newtheorem{Thm}{Theorem}

\newtheorem{Rem}{Remark}
\newtheorem{Prop}{Proposition}

\newtheorem{lemma}{Lemma}

\newtheorem{Conjecture}{Conjecture}

\newcommand{\D}{\displaystyle}

\newcommand{\p}{\scriptstyle}
\newtheorem{Cor}{Corollary}

\newcommand{\tr}{\operatorname{tr}}
\newcommand{\Tr}{\operatorname{Tr}}

\newcommand{\supp}{\operatorname{supp}}

\newcommand\textcyr[1]{{\fontencoding{OT2}\fontfamily{wncyr}\selectfont #1}}
\newcommand{\cyrLL}{\text{\textcyr{L}}}

\newcommand{\cyrdd}{\text{\textcyr{d}}}
\allowdisplaybreaks

\newcommand{\Obj}{\mathtt{Obj}}

\newcommand{\thedelta}{{\text{\iminfamily{d}}}}

\newcommand{\deltae}{\delta_{\scaleto{\!\epsilon}{3.14pt}}}
\renewcommand{\deltae}{\text{\iminfamily{d}}_{\scaleto{\!\epsilon}{3.14pt}}}
\newcommand{\deltaee}{\delta_{\scaleto{\!\epsilon\epsilon}{3.14pt}}}
\renewcommand{\deltaee}{\thedelta_{\scaleto{\!\epsilon\epsilon}{3.14pt}}}
\newcommand{\deltaeo}{\delta_{\scaleto{\!\epsilon_1}{3.14pt}}\!}
\renewcommand{\deltaeo}{\thedelta_{\scaleto{\!\epsilon_1}{3.14pt}}\!}
\newcommand{\deltaet}{\delta_{\scaleto{\!\epsilon_2}{3.14pt}}\!}
\renewcommand{\deltaet}{\thedelta_{\scaleto{\!\epsilon_2}{3.14pt}}\!}
\newcommand{\deltaeot}{\delta_{\scaleto{\!\epsilon_1\epsilon_2}{3.14pt}}\!}
\renewcommand{\deltaeot}{\thedelta_{\scaleto{\!\epsilon_1\epsilon_2}{3.14pt}}\!}

\newcommand{\thetan}{\theta_\mathrm{n}}
\newcommand{\thetas}{\theta_\mathrm{s}}

\newcommand{\Deltaconf}{{\underline{\boldsymbol{\Delta}}}}
\usepackage[alpine]{ifsym}

\newcommand{\uG}{\text{\hminfamily{G}}}
\newcommand{\uT}{\text{\hminfamily{T}}}
\newcommand{\uV}{\text{\hminfamily{V}}}
\newcommand{\uS}{\text{\hminfamily{S}}}
\newcommand{\ue}{\text{\hminfamily{e}}}

\newcommand{\um}{\text{\hminfamily{m}}}
\renewcommand{\ue}{\mathtt{e}}
\newcommand{\uu}{\mathtt{u}}
\newcommand{\uv}{\mathtt{v}}
\newcommand{\uw}{\mathtt{w}}
\newcommand{\un}{\mathtt{n}}
\newcommand{\us}{\mathtt{s}}
\newcommand{\ut}{\mathtt{t}}
\newcommand{\uo}{\mathtt{o}}
\newcommand{\up}{\mathtt{p}}
\newcommand{\uq}{\mathtt{q}}
\newcommand{\ux}{\mathtt{x}}
\newcommand{\uy}{\mathtt{y}}
\newcommand{\ug}{\mathtt{g}}

\newcommand{\uf}{\text{\hminfamily{f}}}
\renewcommand{\uf}{\mathtt{f}}

\newif\ifcomment


\begin{document}
\commenttrue

\title{Perturbing Isoradial Triangulations}

\author{Fran\c cois  David}
\address{François David : Université Paris-Saclay, CNRS, CEA, Institut de physique théorique, 91191 Gif-sur-Yvette Cedex, France}
\email{francois.david@ipht.fr}
\author{Jeanne Scott}
\address{Jeanne Scott : Department of Mathematics, Brandeis University, 415 South Street, Waltham, MA 02453, United States}
\email{jeanne@imsc.res.in}

\begin{abstract}

We consider infinite, planar, Delaunay graphs $\uG_\epsilon$ obtained
by locally deforming the coordinate embedding of a general, isoradial graph $\uG_\mathrm{cr}$, with respect to a real deformation parameter $\epsilon$.
We study three operators 
on $\uG_\epsilon$: the Beltrami-Laplace operator $\Delta(\epsilon)$, a conformal Laplacian $\Deltaconf(\epsilon)$, and the David-Eynard  K\"ahler operator $\mathcal{D}(\epsilon)$. 
The determinant of the later appears is the model of random Delaunay triangulations proposed by \cite{DavidEynard2014}, where it was conjectured  to be a discrete version of the Faddeev-Popov determinant in Polyakov's theory of 2D gravity (Liouville gravity).

All three operators coincide in the $\epsilon\to 0$ limit  with R. Kenyon's critical Laplacian $\Delta_\mathrm{cr}$ on $\uG_\mathrm{cr}$, whose Green's function $\Delta_\mathrm{cr}^{-1}$ 
is known.
Using Kenyon's exact and asymptotic results for $\Delta_\mathrm{cr}^{-1}$ we calculate 
the first and second order terms in the $\epsilon$-expansion of the log-determinant of $\Delta(\epsilon)$, $\mathcal{D}(\epsilon)$ and $\Deltaconf(\epsilon)$, and then study the large distance asymptotics of the \textit{bi-local} part in the second order term.
This entails a careful analysis of edge flips induced by the deformation and the Delaunay constraints.

We show that scaling limits of the second order bi-local term for both $\Delta$ and $\mathcal{D}$ exist and are independent 
 of the choice of initial isoradial graph $\uG_\mathrm{cr}$, while such a scaling limit does not exist in general for $\Deltaconf$ (due to the formation of curvature dipoles).
Moreover we find that the scaling limits for $\Delta$ and $\mathcal{D}$ coincide.
 
We define a discrete analogue of the stress-energy tensor for each of the three operators, and interpret the scaling limits of the bi-local term
for $\Delta$ and $\mathcal{D}$ as an operator product expansion (OPE) in a corresponding conformal field theory (CFT).

This OPE, for $\Delta$, is consistent with a CFT of central charge $c = - 2$ (an expected result), but, for $\mathcal{D}$, is at odds with the value of $c = - 26$ expected by Polyakov's 2D gravity.

Furthermore, while the stress-energy tensor associated to $\Delta$ has a well-defined scaling limit, the stress-energy tensor associated to $\mathcal{D}$ involves some terms (depending on the local geometry of the graph) which are not meaningful in the scaling limit (although its OPE makes sense).

Connections of our work  with some discrete statistical models at criticality are also explored. 

\end{abstract}
\maketitle

\tableofcontents

\section{Introduction}
\label{sIntro}
\subsection{Purpose and motivation}
\label{ssIntPurpMot}
This paper studies deformations of infinite isoradial planar graphs and triangulations, 
and the effect of these deformations on
three discrete Laplace-like operators defined on these graphs, and on their determinants.

Our initial motivation for this study is to better understand  the relation between a  model of Random Delaunay Triangulations (RDT) in the plane, proposed in \cite{DavidEynard2014} by B. Eynard and the first author (FD), and the field theory model of two-dimensional gravity (Liouville gravity) proposed first by A. Polyakov in \cite{Polyakov81}.
The random triangulation model of \cite{DavidEynard2014} is a model where any (finite) distribution of points $\mathbf{z}=\{z_i\}$ in the plane (or the Riemann sphere) is weighted by the determinant $\det[\mathcal{D}]$ of a discrete operator $\mathcal{D}$ defined from the Delaunay triangulation $\mathcal{T}$ associated to $\mathbf{z}$ (hence the terminology Random Delaunay Triangulation Model). This discrete model has very interesting properties.
Thanks to \cite{Rivin1994} and \cite{DavidEynard2014}, it can be viewed as a model of $\mathrm{PSL}(2,\mathbb{C})$ invariant embeddings of 
{random abstract discrete rhombic surfaces}  into the Riemann sphere, and thus is related to random planar maps models and discrete 2d gravity. 
As shown in \cite{CharbonnierDavidEynard2019}, it also provides an alternative description of  the moduli space of the punctured sphere $\mathcal{M}_{0,N}$ equipped with the Weil-Petersson metric, since $\mathcal{D}$ defines a K\"ahler form on the space of Delaunay triangulations.  

The authors of  \cite{DavidEynard2014} pointed out a similarity between the discrete operator $\mathcal{D}$ of the RDT model and the continuous gauge fixing Faddeev-Popov operator $\mathbf{J}$ in Polyakov's model \cite{Polyakov81} (see Appendix \ref{AppBCsystem} and references therein for details), whose determinant gives the famous Liouville action for the conformal factor (the Liouville field, see Appendix \ref{AppBCsystem}). 
Like the discrete RDT model, Liouville gravity is a conformal field theory (CFT), invariant under PSL(2,$\mathbb{C}$) transformations on the sphere. We therefore want to probe this analogy further, and find discrete analogs of CFT structures in the RDT model, such as a stress-energy tensor $T$, some short distance operator product expansion (OPE), and whether such an OPE has a discrete central charge $c$ which can be compared to the central charge of the ghost sector of Liouville gravity, 
famously known to be $c_{\scriptscriptstyle{\mathrm{ghost}}}=-26$. 

The operator $\mathcal{D}$ on a general Delaunay triangulation is a special case of discretized Laplace-like operator (elliptic operator) defined on graphs. These operators can be viewed as discretizations of differential operators defined on Riemannian spaces, with respect to some metric, 
and are related to {some} quantum field theories (QFT), in particular some CFT's. 
They are interesting objects in their own right, both in mathematics (index theorems, Seeley-DeWitt heat kernel expansions, trace formulas) and in physics (conformal field theories, quantum gravity, string theory, statistical mechanics, etc.).
The simplest and perhaps most notable example is the scalar Laplace-Beltrami operator $\Delta$ 
acting on functions $\phi$ over a Riemannian manifold $M$ with metric ${\bf g} = (g_{\mu\nu})$ 
and given by
\begin{equation}
\label{ }
\Delta= -{1\over \sqrt{g}}\partial_\mu\, \sqrt{g}\,g^{\mu\nu}\,\partial_\nu
\end{equation}
with $\partial_\mu$ the standard derivative w.r.t. the local coordinate $x^\mu$ acting on scalar functions. 
The operator $\Delta$ is related to the massless scalar quantum free field theory (i.e. the Gaussian free field, or GFF) 
on a manifold $M$, see Appendix~\ref{AppGFF} for details.
Its functional determinant (properly defined), is related to the GFF partition function $Z$ through
\begin{equation}
\label{ZGFF}
\displaystyle
\det(\Delta)\ =\ Z^{-2}\quad \text{with}\qquad  Z\  =\  \int \frak{D} [ \phi ] \, e^{- S[\phi]}
\end{equation}
with $\phi$ the free field (a random scalar real function) and $S[\phi]$ the GFF action (see ~\ref{BosonAction}).
Both the action $S[\phi]$ and the partition function $Z$ depend 
explicitly on the metric $\mathbf{g}$ on $M$ and
the effect of varying the metric in the action $S[\phi]$ is encoded in the so-called stress-energy tensor $\mathbf{T}=(T^{\mu\nu})$. For a CFT such as the GFF, one has to consider the holomorphic and antiholomorphic components $T=-2\pi\,T_{zz}$ and $\overline T=-2\pi\,T_{\bar z\bar z}$ of $\mathbf{T}$ which encode specifically the effect {of changing the metric by an infinitesimal anti-analytic diffeomorphism} 
\begin{equation}
\label{zepsilonbar}
z\to z+\epsilon\, F(\bar z)
\end{equation} 
with $F$ an anti-analytic function of the complex coordinates $(z,\bar z)$ on $M$ (see Appendix~\ref{ssT2DCFT} for details).
The OPE for $T$
\begin{equation}
\label{OPETone}
 T(z)T(z')={c\over 2} {1\over (z-z')^4}+\cdots
\end{equation}
with $c$ the central charge of the CFT, is of special importance. It implies (see \ref{ssT2DCFT}) that the second variation of the logarithm of the partition function of the CFT, $\log Z$, under \ref{zepsilonbar} is
\begin{equation}
\label{{zepsilonbar}}
{c\over 4\,\pi^2} \iint d^2u\,d^2v \  
{\bar\partial  F(u)\,\bar \partial  F(v)
\over (u-v)^4}+
{\partial  \bar F(u)\, \partial  \bar F(v)\over (\bar u-\bar v)^4}+\text{contact terms}
\end{equation}
{In the cases we are interested in, the central charge $c$ is real, and this can of course be rewritten as the double integral of the real part of ${\bar\partial  F(u)\,\bar \partial  F(v)\over (u-v)^4}$.}

Accordingly, we shall try to define: (i) a discrete analog of the diffeomorphisms \ref{zepsilonbar} for Delaunay triangulations, (ii) a discrete stress-energy tensor $\mathbf{T}$ associated to the operator $\mathcal{D}$ of the RDT model, (iii) an analog for $\mathcal{D}$ of the OPE \ref{OPETone} and of formula \ref{{zepsilonbar}}.
This requires us to introduce and study an appropriate ``scaling limit'' (in the QFT sense)  of the RDT model. 

\medskip

This program turns out to be very difficult for general random Delaunay triangulations. As a first step we shall study deformations of a very specific subclass of Delaunay triangulations, namely \emph{isoradial Delaunay triangulations}.
One reason for this restriction is technical. The analysis is much simpler and explicit calculations can be done, thanks to the fact that on isoradial triangulations, the $\mathcal{D}$ operator is proportional to the critical Laplacian $\Delta_\mathrm{cr}$ considered by Kenyon in \cite{Kenyon2002} (to be defined later). 
Thanks to the methods of discrete analyticity, both the determinant $\det[\Delta_\mathrm{cr}]$ and the Green's function $\Delta_\mathrm{cr}^{-1}$ (the propagator) take simple explicit forms in terms of the geometry of the isoradial triangulation. 
A second reason is that isoradial triangulations can be viewed as  analogs of ``discrete flat metric'' (see section \ref{ssDelGraph}). 
It is therefore natural to study the operator $\mathcal{D}$ and the associated measure for triangulations which are close to but not exactly isoradial, as a means of understanding the relationship between the RDT model and 2d gravity, as suggested in \cite{DavidEynard2014}.

\medskip
This work is rather technical, limited in scope, but it represents a first step in this general program. 
In addition to studying the operator $\mathcal{D}$, we carry out a similar analysis for two related operators, also defined for Delaunay graphs: 
(i) the discrete Laplace-Beltrami operator $\Delta$, and (ii) a conformal Laplacian $\Deltaconf$ with PSL(2,$\mathbb{C}$) invariance properties.

Several issues require a lot of attention: 
(1) Under deformations,
the Delaunay constraints (defined precisely in \ref{sssNotations}) cause the incidence relations of the graph
to change (by edge flips). 
These flips are a potential source of discontinuities and  singularities for the operators and {the} determinants we are interested in.
(2) We want uniform estimates for the variation of operators and determinants, independent of the initial isoradial Delaunay graph. 
This is not always possible. 
(3) We also look for the existence of ``scaling limits'' (in the usual sense of statistical mechanics and quantum field theory), in particular for the discrete analogs of the OPE \ref{OPETone} and of formula \ref{zepsilonbar}, in order to recover a continuous QFT interpretation of our results.
 
Point (1) is treated thoroughly. Whitehead flips are under control for the operators $\mathcal{D}$ and $\Delta$, but are shown to induce discontinuities for $\Deltaconf$.
 
For point (2), uniform estimates are obtained for all three operators, for which
we get discrete analogs of the stress-energy tensor $\mathbf{T}$.
For $\Delta$ the OPE \ref{OPETone} \ref{zepsilonbar} holds with central charge $c=-2$, as expected. For $\mathcal{D}$ {an OPE holds} as well, and unexpectedly we obtain a central charge $c=-2$ too. 
We do not, however, recover formulas \ref{OPETone} and \ref{zepsilonbar} (specific to CFT's) for the conformal Laplacian $\Deltaconf$.

For point (3), good scaling limit results are obtained for $\Delta$ (this was to be expected), we prove similar results  $\mathcal{D}$, but they are valid under some restrictions. 
This program fails for $\Deltaconf$.

\medskip
Let us now be more specific, and summarize: (1) the main concepts and tools used in this paper, (2) the main results, and (3) the detailed plan and content of the paper.

\subsection{The concepts}
\label{ssConcepts}

\subsubsection{Delaunay graphs}
\label{ssDelGraph}
Delaunay triangulations in the plane are models of discrete space which has been studied by many authors, in particular in high energy physics \cite{CHRIST198289} and well as in statistical physics, condensed matter and soft matter physics.
Anticipating the precise definitions and details given in section 
\ref{sCritLapl}, we highlight some notions which are important:

A {\textbf{polyhedral graph}} $\uG$ is a planar graph (with finitely or infinitely many vertices) equipped
with an embedding $z: \mathrm{V}(\uG) \longrightarrow \Bbb{C}$ of its vertex set $\mathrm{V}(\uG)$ such 
that edges are mapped to straight line segments and faces are mapped to 
convex, cyclic polygons.
Accordingly we can associate with each face 
$\uf$ of $\uG$ the circumcircle $C_\uf$, the circumdisk $D_\uf$, and the corresponding circumradius $R(\uf)$
of its cyclic polygon with respect to the embedding.

A \textbf{Delaunay graph}  is a polyhedral graph $\uG$ such that, under the embedding,
(1) the interior of the circumdisk of each face of $\uG$ contains no vertices, and 
(2) no two faces share the same circumdisk. 
Equivalently the dual of Delaunay graph $\uG$ is the 
Voronoi complex $\uV$ associated to the set of (embedded) vertices of $\uG$. 

A \textbf{weak Delaunay graph} is a polyhedral graph $\uG$ such that condition (1) is satisfied.

A \textbf{(weak) Delaunay triangulation} $\uT$ is a (weak) Delaunay graph whose faces are all triangles.

\begin{figure}[h]
\begin{center}
\raisebox{-.8in}{\includegraphics[width=2in]{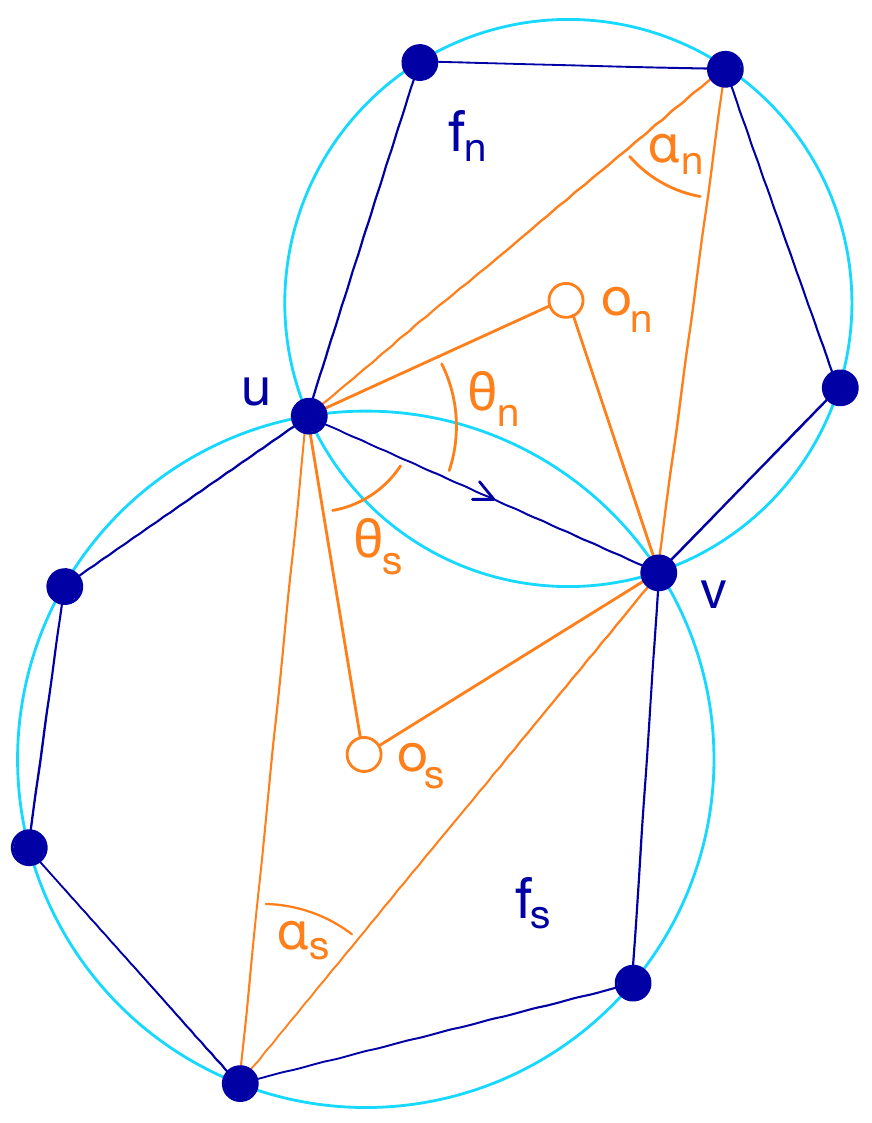}}
\caption{The north and south faces $\uf_\mathrm{n}$ and $\uf_\mathrm{s}$ of 
an (oriented) edge $\vec{\ue}=(\uu,\uv)$ are drawn in dark blue while the corresponding
circumcircles are outlined in light blue. Their
respective circumcenters $\uo_\mathrm{n}$ and $\uo_\mathrm{s}$, as well
the associated north and south angles $\theta_\mathrm{n}(\vec {\ue} \, )$ and $\theta_\mathrm{s}(\vec {\ue} \, )$,
are highlighted in orange.}
\label{triangles}
\end{center}
\end{figure}

Following \cite{DavidEynard2014}, we associate 
to an oriented edge $\vec{\ue}=(\uu,\uv)$ 
``North'' and ``South'' faces 
$\uf_\mathrm{n}$ and $\uf_\mathrm{s}$
along with angles $\theta_\mathrm{n}(\vec{\ue} \, )=\angle{\uv\uu\uo_\mathrm{n}}$ and 
$\theta_\mathrm{s}(\vec {\ue} \, )=\angle{\uo_\mathrm{s}\uu\uv}$ where $\uo_\mathrm{n}$ and $\uo_\mathrm{s}$ 
are the respective circumcenters of 
$\uf_\mathrm{n}$ and $\uf_\mathrm{s}$, as depicted in fig.~\ref{triangles}. 
Reversing the orientation of $\vec{\ue}$ interchanges the roles of North and South.
The \textbf{conformal angle} $\theta(\ue)$  associated to the unoriented edge $\ue$ is defined as %
\footnote{Note that the conformal angle $\boldsymbol{\theta}(\ue)$ considered in \cite{DavidEynard2014} is twice the conformal angle $\theta(\ue)$ defined here by \ref{thetaDSdef}  (i.e. $\boldsymbol{\theta}(\ue)= 2\,\theta(\ue)$). We choose  definition  \ref{thetaDSdef}  for compatibility with Kenyon's notations in \cite{Kenyon2002}.} 
\begin{equation}
\label{thetaDSdef}\theta(\ue)=(\theta_\mathrm{n}(\vec{\ue} \, )+\theta_\mathrm{s}(\vec{\ue} \, ))/2\ .
\end{equation}
The c.w. orientability of the north and south triangles enforces
\begin{equation}
\label{thetaorient} 
-\pi/2<\theta_\mathrm{n}(\vec{\ue})\ \text{and}\ \theta_\mathrm{s}(\vec{\ue})<\pi/2\ .
\end{equation}
The {Delaunay} condition ensures that 
\begin{equation}
\label{DelThetCond}
0 < \theta(\ue)<\pi/2
\end{equation}
while the weak Delaunay condition ensures that $0 \leq \theta(\ue)<\pi/2$.

Finally, as explained in \cite{DavidEynard2014} and in section~\ref{sCritLapl}, to each plane Delaunay graph $\uG$ we can associate an abstract \textbf{rhombic surface}
$\uS_\uG^\lozenge$ obtained by gluing rhombi $\lozenge(\ue)$ associated to the edges $\ue$ of $\uG$ according to the incidence relations of $\uG$. Each rhomb 
$\lozenge(\ue)$ has unit edge length and has a corresponding rhombus angle $2\theta(\ue)$. We view $\uS_\uG^\lozenge$ as a discretized Riemann surface with curvature concentrated at {the} vertices. This rhombic surface $\uS_\uG^\lozenge$ will be ``flat'', i.e. can be isometrically embedded in the plane, \emph{if and only if } for each face $\uf$ of $\uG$, the sum of the conformal angles of the edges $\ue$ which form the boundary of $\uf$ equals $\pi/2$
\begin{equation}
\label{ }
\sum\limits_{\ue \in\partial\uf} \theta(\ue)=\pi/2
\end{equation}

Equivalently, the Delaunay graph $\uG$ is \textbf{isoradial}, i.e. the circumradii $R(\uf)$ are all equal.
Alternatively a Delaunay graph $\uG$ is isoradial if and only if $\uS_\uG^\lozenge$ coincides with the planar 
bipartite \textbf{kite graph} $\uG^\lozenge$ discussed in section~\ref{ssBasicPlanar}.
Isoradial Delaunay graphs are also referred to as \textbf{flat graphs} or \textbf{critical graphs}.

\subsubsection{The random Delaunay triangulation model}
\label{sssRDTm}
The David-Eynard model of \cite{DavidEynard2014} is a theory of random (finite) Delaunay 
graphs which are sampled (with Le\-bes\-gue measure) 
according to the conformal angle values of the corresponding edges. 
By the Voronoi construction, 
a configuration of $N \geq 3$ distinct marked points $\{ z_1,  \dots , z_N \}$ 
in the extended plane $\Bbb{CP}^1$
is equivalent to a Delaunay graph $\uG$ with vertex set $\mathrm{V}(\uG) = \{1, \dots, N \}$ and
embedding $k \mapsto z_k$. 
This correspondence between point configurations and graphs is
$\mathrm{PSL}(2,\Bbb{C})$ equivariant
in the sense that the incidence relations which
define the Delaunay graph are invariant
under the action of $\mathrm{PSL}(2,\Bbb{C})$
by M\"obius transformations.
In this formulation
the relevant measure on the space of configurations of marked points is, fixing three points $(z_1, z_2, z_3)$ thanks to PSL(2,$\mathbb{C}$) invariance

\begin{equation}
\label{theDEmeasure}
\prod_{k = 4}^\mathrm{N}  dz_k^2  \,  { {\det}'\mathcal{D}\over { |z_1 - z_2|^2 |z_2 -z_3|^2 |z_1 - z_3|^2} }
\end{equation}
where $\mathcal{D}$ is the David-Eynard discrete K\"ahler operator 
of the graph $\uG$ as defined in \ref{KaelDelta} below, and ${\det}'\mathcal{D}$ 
is the  $(N-3) \times (N-3)$ principal minor of $\mathcal{D}$ with row and column
set $\{4, \dots, N \}$, see \cite{DavidEynard2014}.
We view ${\det}'\mathcal{D}$ as a {\it reduced determinant}
which suppresses the effect of the zero modes of $\mathcal{D}$.
As shown in \cite{CharbonnierDavidEynard2019}, the measure in \ref{theDEmeasure} 
is $\mathrm{PSL}(2,\Bbb{C})$ invariant and coincides with the Weil-Petersson measure on $\mathcal{M}_{0,N}$.
Points configurations whose corresponding Delaunay graph is a triangulation form a 
Zariski open subset and consequently the subspace of non-triangulations has measure zero.
For this reason we speak of the David-Eynard model as a theory of
random triangulations.

\subsubsection{The operators $\Delta$, $\Deltaconf$ and $\mathcal{D}$}
\label{sss3operators}
In this paper, we are interested in the three discrete operators defined on generic polyhedral graphs $\uG$: 
the Beltrami-Laplace operator $\Delta$, the conformal Laplacian $\Deltaconf$, and the David-Eynard K\"ahler operator $\mathcal{D}$.
All three operators act on the space $\Bbb{C}^{\mathrm{V}(\uG)}$ consisting of complex valued functions supported on the vertices $\mathrm{V}(\uG)$ 
of the graph  $\uG$.
\medskip
\par\noindent$\bullet$ 
The discrete \textbf{Beltrami-Laplace operator} $\Delta$ is defined for $\phi \in \Bbb{C}^{\mathrm{V}(\uG)}$ by 
\begin{equation}
\label{BLDelta}
\Delta \phi(\uu)=\sum_{\mathrm{edges}\ \vec{\ue} \, =(\uu, \uv)} c(\vec{\ue}\,)(\phi(\uu)-\phi(\uv))
\quad,\quad c(\vec{\ue} \,)={1\over 2} \big(\tan \theta_\mathrm{n}(\vec{\ue} \,)+\tan \theta_\mathrm{s}(\vec{\ue}\,) \big)
\end{equation}
This is a standard discretization of the Laplacian in the plane, both in physics  (see e.g. \cite{CHRIST198289}) and in mathematics.
It is a symmetric real operator.
\medskip
\par\noindent$\bullet$ 
The \textbf{conformal Laplacian} $\Deltaconf$, that we introduce here, is defined as
\begin{equation}
\label{ConfDelta}
\Deltaconf \phi(\uu)=\sum_{\mathrm{edges}\ \vec{\ue} \, =(\uu, \uv)} \tan \theta(\ue) \big(\phi(\uu)-\phi( \uv) \big)
\end{equation}
It is invariant under global conformal transformations  $z \stackrel{g}{\mapsto} {az + b \over {cz + d}}$
of the graph embedding $z: \mathrm{V}(\uG) \longrightarrow \Bbb{C}$ for $g \in \mathrm{PSL}_2(\mathbb{C})$.
It's worth noting that $\Deltaconf$ can be viewed as the discrete Laplace-Beltrami operator defined, not on the planar graph $\uG$, but rather on the 
image of $\uG$ inside the rhombic surface 
$\uS_\uG^\lozenge$ (i.e. the black vertices of $\uS_\uG^\lozenge$
where two black vertices are joined by an edge 
if and only if they lie on a common rhomb).
We point the reader \textcolor{blue}{to} a related construction 
in \cite{Mercat2000}.
As such, $\Deltaconf$ is a discretization of the Beltrami-Laplace operator 
on a Riemann surface with respect to a non-flat metric.
It is also a symmetric real operator.
\medskip
\par\noindent$\bullet$ 
The \textbf{K\"ahler operator} $\mathcal{D}$, which we are interested in, has been introduced in \cite{DavidEynard2014}. It is defined in terms of the geometry of the graph $\uG$ as
\begin{equation}
\label{KaelDelta}
\mathcal{D} \phi(\uu)=\sum_{\mathrm{edges}\ \vec{\ue} \, =(\uu, \uv)} 
{1\over 2}\left({\tan\theta_\mathrm{n}(\vec{\ue} \,)+ \mathrm{i} \over R^2_\mathrm{n}(\vec{\ue} \,)}+{\tan\theta_\mathrm{s}(\vec{\ue} \,)
- \mathrm{i} \over R^2_\mathrm{s}(\vec{\ue}\,)}\right)\big( \phi(\uu)-\phi(\uv) \big)
\end{equation}
where $ R_\mathrm{n}(\vec{\ue} \,)$ and $ R_\mathrm{s}(\vec{\ue} \,)$ are the circumradii of the north and south {faces $\uf_\mathrm{n}$ and 
$\uf_\mathrm{s}$} adjacent to $\vec{\ue}$ respectively.
It is a Hermitian complex operator.
Although not obvious from this definition \ref{KaelDelta}, the operator $\mathcal{D}$ transforms covariantly under global conformal PSL(2,$\mathbb{C}$) transformations of the graph embedding, and (even less obviously) it defines a K\"ahler metric $dz_\uu\mathcal{D}_{\uu\uv}d\bar z_\uv$ on the space of 
Delaunay graphs in the plane.

\medskip
These three operators can be defined for any {polyhedral graph} $\uG$.
The {weak} Delaunay condition on $\uG$ ensures that the three operators are positive semi-definite. 
If $\uG$ is isoradial (with common circumradius $R>0$) 
then the operators $\Delta$, $\Deltaconf$ and $R^2\mathcal{D}$ all coincide, and agree with the \textbf{critical Laplacian} $\Delta_\mathrm{cr}$
considered in \cite{Kenyon2002} and 
defined by
\begin{equation}
\label{crDeltaDef}
\Delta_\mathrm{cr} \, \phi(\uu) :=\sum_{\mathrm{edges}\ \vec{\ue} \, =(\uu, \uv)}  \tan \theta(\ue)(\phi(\uu)-\phi(\uv))
\end{equation}
\noindent
This coincidence occurs because $\theta_\mathrm{n}(\vec{\ue} \, )  = \theta_\mathrm{s}(\vec{\ue} \, ) = \theta (\ue)$ for any
(oriented) edge $\vec{\ue}$ in the isoradial case.
The Green's function $\Delta_\mathrm{cr}^{-1}$ 
of the {critical Laplacian} $\Delta_\mathrm{cr}$ (see section \ref{sLaplacians}) 
turns out to be accessible and can be written explicitly in terms of the graph's local structure; furthermore, the log-determinant of the critical Laplacian can be computed as a finite sum of local contributions if in addition one assumes the graph is periodic.

\subsection{The main results of this paper}
\label{ssResults}

\subsubsection{Asymptotics of the critical Green's function}
\label{sssIntroAsymptG}
The asymptotic behaviour of $\Delta_\mathrm{cr}^{-1}$ is essential to us. The leading asymptotics has been worked out by Kenyon in \cite{Kenyon2002}. 
Our first result is a refinement of Kenyon's estimate, which entails isolating subleading terms.
We shall  make use of this series development later in the proof of  Theorem~\ref{discrete-theorem}.

\begin{Prop} For any vertices $\uu$ and $\uv$ in an isoradial Delaunay graph $\uG_\mathrm{cr}$ 
\label{Prop-SharpAsymptotics1}
\begin{equation}
\label{SharpAsymptotics1}
\Big[ \Delta^{-1}_\mathrm{cr} \Big]_{\uu, \uv}  = 
-{1 \over {2 \pi}}  \left( 
\log \Big( 2  \big| p_1( \uu, \uv) \big| \Big)  +  \gamma_\mathrm{euler}  
+  {\frak{Re} \big[ p_3 (\uu,\uv) \big] \over {6 \big| p_1(\uu, \uv) \big|^3} }  
+  \mathrm{O} \Bigg( {1 \over {\big| p_1(\uu, \uv) \big|^4}} \Bigg) \right)
\end{equation}
where $\gamma_\mathrm{euler}$ is the Euler-Mascheroni constant and 
$$p_1(\uu,\uv) = z_\mathrm{cr}(\uv) - z_\mathrm{cr}(\uu)\ .$$
The term $p_3(\uu,\uv)$ is introduced in Def.~\ref{def-p_n} (and written explicitly in \ref{stuff2}) in subsection \ref{ssExpDiscExp}.
It depends on the local geometry of the graph
$\uG_\mathrm{cr}$ between $\uu$ and $\uv$, but is bounded uniformly and linearly by  
$$\big| p_3(\uu, \uv) \big| \ \le \ 3\, 
\big| p_1(\uu,\uv) \big|
$$
\end{Prop}

\begin{Rem}
Proposition  \ref{Prop-SharpAsymptotics1}
sharpens Kenyon's Theorem 7.3  in \cite{Kenyon2002} by identifying and obtaining a uniform bound on the first non-constant  
subdominant term 
$$ { 1 \over 6} \, \big| p_1(\uu, \uv) \big| ^{-3} \,
\frak{Re} \big[ p_3( \uu, \uv) \big]  
\ \leq \ {1 \over 2} \, 
\big| p_1(\uu,\uv) \big|^{-2} 
$$
Proposition \ref{pGasymptotics} in section ~\ref{sAsymptotics} extends these asymptotics to all orders of the large distance asymptotic series expansion of the Green's function, and gives uniform bounds for those terms.
Proposition \ref{Prop-SharpAsymptotics1} follows from Proposition \ref{pGasymptotics}.
\end{Rem}

\subsubsection{Deformations of critical graphs}
\label{sssDeformIntro}
We introduce a scheme for deforming Delaunay graphs (and general polyhedral graphs) and study the response of the corresponding operators supported on the deformed graph.
\begin{Def}
\label{DefOne}
A \textbf{Delaunay deformation} $\uG_\epsilon$ of an initial Delaunay graph $\uG_0$ is defined as follows.
We start by deforming the initial vertex embedding  $\uv\mapsto z_0(\uv)$ for $\uv\in\mathrm{V}(\uG_0)$ by 
\begin{equation}
\label{thezdeform}
z_\epsilon(\uv) \, := \  z_0(\uv)  \, + \, \epsilon \, F(\uv)
\end{equation}
where $\epsilon \geq 0$ is a real parameter, and where $F: \mathrm{V}(\uG_0) \rightarrow \Bbb{C}$ is a displacement function with finite support $\Omega_F\subset  \mathrm{V}(\uG_0) $.
Provided the mapping $\uv \mapsto z_\epsilon(\uv)$ is one-to-one, 
the corresponding Delaunay deformation $\uG_\epsilon$ of $\uG_0$ is defined to be the unique Delaunay graph with vertex set $\mathrm{V}(\uG_\epsilon)=\mathrm{V}(\uG_0)$ and planar graph embedding $\uv\mapsto z_\epsilon(\uv)$. 
For a generic polyhedral graph $\uG$, the \textbf{lattice closure } $\overline\Omega_{F}$ of $\Omega_F$  is
\begin{equation}
\label{omegaclosure}
\overline\Omega_{F}=\{\uv\in \mathrm{V}(\uG):\, \uv\ \text{shares a face $\uf\in F(\uG)$ with a vertex}\ \uu\in\Omega_F \}
\end{equation}
\end{Def}

We first need to control the geometry  of the deformed graph $\uG_\epsilon$ in term of the deformation parameter $\epsilon$.
This is ensured by Lemmas \ref{MFbound},  \ref{LemRigidDef}, \ref{Lemma11} and Prop.~\ref{epsilontildeF}, which we summarize in the following Proposition.

\begin{Prop}
\label{LemmaOne}
Let $\uG_{0}$ be an isoradial Delaunay graph, and $F$ a displacement function as above.
There is a threshold $\tilde{\epsilon}_F>0$ such that whenever $0 \leq \epsilon < \tilde{\epsilon}_F$ 
\begin{enumerate}
  \item $z_\epsilon: \mathrm{V}(\uG_0) \longrightarrow \Bbb{C}$ is an embedding
  \item there is an inclusion of edge sets $\mathrm{E}(\uG_0) \subseteq  \mathrm{E}(\uG_\epsilon)$
  \item the edge sets are stable, i.e. $\mathrm{E}(\uG_{\epsilon_1}) = \mathrm{E}(\uG_{\epsilon_2})$ whenever $0 < \epsilon_1, \epsilon_2 < \tilde{\epsilon}_F$
\end{enumerate}
\end{Prop}

Prop.~\ref{LemmaOne} ensures the existence of a right-sided limit graph when $\epsilon\to 0$ which is  weakly Delaunay, namely

\begin{Def}
\label{defIsoRef}
The \textbf{refinement} $\uG_{0^+}$ of $\uG_0$ 
determined by $F$ is the  weak Delaunay graph
with vertex set $\mathrm{V}(\uG_{0^+}) := \mathrm{V}(\uG_0) $ 
and embedding $z_{0^+} := z_0$ whose edge set is
given by 
\[ \mathrm{E}(\uG_{0^+}) :=  \lim_{\ \ \epsilon \rightarrow 0^+} \mathrm{E}(\uG_\epsilon) \]
\end{Def}
Note that $\uG_{0^+}$ will be a {\bf weak} Delaunay graph precisely when
the inclusion of edge sets is strict, otherwise $\uG_{0^+}$ and $\uG_0$ will coincide.
It will be convenient to \textbf{complete} $\uG_{0^+}$ to a (weak) Delaunay triangulation $\widehat{\uG}_{0^+}$ 
by maximally saturating $\mathrm{E}(\uG_{0^+})$ with additional non-crossing edges  (see Def.~\ref{CompletionTofG}). 
The choice of these additional edges (referred to as {\bf chords}, introduced in Def.~\ref{ChordEdge}) will not
affect our calculations because the weights assigned to these edges 
by the operators $\Delta$, $\mathcal{D}$, and $\Deltaconf$ always vanish.
We want to emphasize that  $\widehat{\uG}_{0^+} = \uG_{0^+} = \uG_0$ 
whenever $\uG_0$ is a triangulation.

\medskip
Our chief interest is when the initial graph $\uG_0$ is a \textbf{critical graph}, i.e. an isoradial Delaunay graph $\uG_\mathrm{cr}$ with isoradius $R_\mathrm{cr}$.
A Delaunay deformation $\uG_\mathrm{cr}\to\uG_\epsilon$ (corresponding to some $F$) supports a K\"ahler operator $\mathcal{D}(\epsilon)$, as well as a Laplace-Beltrami operator $\Delta(\epsilon)$ and a conformal Laplacian $\Deltaconf(\epsilon)$.
All three of these operators degenerate on the critical graph when $\epsilon\to 0$
$$\lim_{\epsilon\to 0}\mathcal{D}(\epsilon)=\Delta_\mathrm{cr}/R_\mathrm{cr}^2\quad,\qquad
\lim_{\epsilon\to 0}\Delta(\epsilon)=\lim_{\epsilon\to 0}\Deltaconf(\epsilon)=\Delta_\mathrm{cr}$$
where $\Delta_\mathrm{cr}$ is the critical Laplacian of Kenyon on $\uG_\mathrm{cr}$.
 
Let $\mathcal{O}$ denote either $\Delta$, $\Deltaconf$, or $\mathcal{D}$.
Accordingly $\mathcal{O}(\epsilon)$ will denote 
the corresponding operator on the perturbed Delaunay graph $\uG_\epsilon$
while $\mathcal{O}_\mathrm{cr}$ will denote the operator
on the critical graph $\uG_\mathrm{cr}$.
We introduce the variation of operators 
\begin{equation}
\label{deltaOdef}
\delta \mathcal{O}(\epsilon) := \mathcal{O}(\epsilon) - \mathcal{O}_\mathrm{cr}
\end{equation}
and formally expand the
log-determinant $\log \det \mathcal{O}(\epsilon)$ using the Green's function 
$\mathcal{O}^{-1}_\mathrm{cr}$ of the critical operator as 
\begin{equation}
\label{log-expansion}
\begin{array}{ll}
\D \log \det \mathcal{O}(\epsilon)
&\D= \  \log \mathrm{det} \mathcal{O}_\mathrm{cr} \ + \ \tr \Big[ \delta \mathcal{O}(\epsilon) \cdot  \mathcal{O}_\mathrm{cr}^{-1}   \Big] 
\ - \ {1 \over 2} \, \tr \Big[ \big(\delta \mathcal{O}(\epsilon) \cdot \mathcal{O}_\mathrm{cr}^{-1} \big)^2 \Big] \ + \ \cdots
\end{array}
\end{equation}
The trace terms occurring on the right-hand side of equation \ref{log-expansion} are well defined 
owing to the fact that the support of the perturbation is compact; consequently the difference
$\log \det \mathcal{O}(\epsilon) \, - \, \log \det \mathcal{O}_\mathrm{cr}$ is well defined and takes a finite real value.

Our most significant results concern the second-order  term 
$\tr \big[ (\delta \mathcal{O}(\underline{\epsilon}) \cdot \mathcal{O}_\mathrm{cr}^{-1})^2 \big]$
arising from a bi-local version of the deformation given in (\ref{thezdeform}),
executed simultaneously at two distant sites and
controlled by a pair $\underline{\epsilon} =(\epsilon_1, \epsilon_2)$ of 
independent deformation parameters.
These results are mainly given by Prop.~\ref{ThDelta} for $\Delta$, Prop.~\ref{The3} for $\mathcal{D}$, 
while the analysis of $\Deltaconf$ is handled in Sects. \ref{ssConfLapl2nd} and \ref{subsubsection-dipole}.
The following theorem summarizes these results.

\begin{Thm}
\label{discrete-theorem}
Consider two complex functions $F_1(z)$ and $F_2(z)$ whose supports $\Omega_1=\supp F_1$ and $\Omega_2=\supp F_2$ in 
the vertex set $\mathrm{V}(\uG_\mathrm{cr})$ are \emph{finite and disjoint} (hence at finite distance),  and a bi-local deformation
of the embedding $\uv \mapsto z_\mathrm{cr}(\uv)$ given by
\[ z_{\underline{\epsilon}}(\uv) \, := \ z_\mathrm{cr}(\uv) \, + \, \epsilon_1 F_1 (\uv ) \, + \, \epsilon_2  F_2 ( \uv)  \]
where $\underline{\epsilon} = (\epsilon_1, \epsilon_2)$ is a pair if independent deformation parameters.

As functions of $\epsilon_1$ and $\epsilon_2$, $\log\det\Delta(\underline\epsilon)$, $\log\det\Deltaconf(\underline\epsilon)$, $\log\det\mathcal{D}(\underline\epsilon)$ are analytic  within the range $0\le\epsilon_1,\epsilon_2< \min(\tilde\epsilon_{F_1},\tilde\epsilon_{F_2})$.

Furthermore, the $\epsilon_1 \epsilon_2$ cross-term in the perturbative expansion of $\log\det \Delta(\underline{\epsilon})$, denoted $\deltaeot \log \det {\Delta}$,  is obtained from 
$\tr \big[{ (\delta \Delta(\underline{\epsilon}) \cdot \Delta_\mathrm{cr}^{-1})}^2 \big]$.
It takes the asymptotic form
\begin{equation}
\label{finalOPElike-discrete}
\begin{split}
&\deltaeot \log \det {\Delta}\ =\ \\
&\hskip .5cm
 - {2\over {\pi^2}} \!\!\! \sum_{\stackrel{\scriptscriptstyle \mathrm{triangles}}{\scriptscriptstyle{\ux_1, \ux_2}} }
\!\! \!A(\ux_1)  A(\ux_2)  
 \Bigg( \frak{Re} \Bigg[{\overline{\nabla} F_1( \ux_1 ) \, 
 \overline{\nabla} F_2 (\ux_2 )\over \big( z_\mathrm{cr}(\ux_1) - z_\mathrm{cr}(\ux_2) \big)^4} \Bigg] \, + \, \mathrm{O} 
 \Big( \,  \big| z_\mathrm{cr}(\ux_1) -  z_\mathrm{cr}(\ux_2) \big|^{-5} \Big) \Bigg)
\end{split}
\end{equation}
where $\ux_i \in \mathrm{F}(\widehat{\uG}_{0^+})$ 
is a triangle having at least one vertex in $\Omega_i$, whose center has coordinate $z_\mathrm{cr}(\ux_i)$, and whose area is $A (\ux_i)$
with $i=1, 2$. 
Formula \ref{finalOPElike-discrete} makes use of the discrete derivative operators $\nabla,\ \overline{\nabla}: \Bbb{C}^{\mathrm{V}(\uT)} \longrightarrow
\Bbb{C}^{\mathrm{F}(\uT)}$ introduced in \cite{DavidEynard2014} for polyhedral triangulations $\uT$; see Def.~\ref{DefNablaBarNabla}.

Likewise, the $\epsilon_1 \epsilon_2$ cross-term in the expansion of $\log\det \mathcal{D}(\underline{\epsilon})$, $\deltaeot \log \det {\mathcal{D}}$, is obtained from 
$\tr \big[{ (\delta \mathcal{D}(\underline{\epsilon}) \cdot \mathcal{D}_\mathrm{cr}^{-1})}^2 \big]$ and takes the same asymptotic form
as formula \ref{finalOPElike-discrete}. 

For the conformal Laplacian $\Deltaconf$ , the $\epsilon_1 \epsilon_2$ cross-term in the expansion of $\log\det \Deltaconf(\underline{\epsilon})$, $\deltaeot \log \det {\Deltaconf}$, does not in general have an asymptotic form given by \ref{finalOPElike-discrete}.
``Anomalous'' chord-to-edge and chord-to-chord terms have to be added to Formula \ref{finalOPElike-discrete} in order to obtain the correct asymptotics.
They can be interpreted as ``curvature defects'' arising from the deformation of the graph.
\end{Thm}

\begin{Rem}
Formula \ref{finalOPElike-discrete} is independent of the 
choice of triangulation $\widehat{\uG}_{0^+}$ used to complete $\uG_{0^+}$,
in light of a discretized 
version of Green's theorem, namely Lemma \ref{greens-theorem} and Corollary \ref{common-refinement}
as detailed in Subsection \ref{FactoLaplNabl}.
\end{Rem}

\subsubsection{Smooth deformations and scaling limits}
\label{sssSmoDefIntro}
In this paper we are interested in the existence and the form of the continuum limit of the results in Theorem \ref{discrete-theorem}. For this purpose we shall consider \textbf{smooth Delaunay deformations} implemented by test functions, defined below.
We aim for results independent of the initial critical graph $\uG_\mathrm{cr}$, and which reconstitute the continuous formula \ref{zepsilonbar} expected from CFT's.

\begin{Def} 
\label{defSmDeform}
Let $F$ be a smooth (non-holomorphic) function  $F: \Bbb{C} \longrightarrow \Bbb{C}$ with compact support $\Omega\subset\mathbb{C}$, and consider its restriction to an initial Delaunay graph $\uG_0$ by declaring
\begin{equation}
\label{smooth-perturbation-0}
F(\uv) \, := \  \, F \big( z_0(\uv) \big) 
\end{equation}
where $\uv \in \mathrm{V}(\uG_0)$ is a vertex.

The  \textbf{smooth Delaunay deformation} $\uG_\epsilon$ of $\uG_0$ corresponding to $F$  is the Delaunay deformation of $\uG_0$ given by Def.~\ref{DefOne} with the function $F(\uv)$ given by \ref{smooth-perturbation-0}.

Moreover, we shall incorporate a parameter $\ell >0$ into our deformation rubric \ref{smooth-perturbation-0} by rescaling the displacement function 
accordingly:
\begin{equation}
\label{smooth-perturbation}
F_\ell(\uv) \, = F_\ell\big(z_0(\uv)\big)\ :=\ \ell \, F \big( z_0(\uv) / \ell \big) 
\end{equation}
Using the construction above, we obtain a rescaled, deformed embedding $z_{\epsilon,\ell}$
and corresponding Delaunay graph $\uG_{\epsilon, \ell}$ together with an 
attending refinement $\uG_{0^+,\ell}$ and completion $\widehat{\uG}_{0^+,\ell}$.
We shall denote by $\Delta(\epsilon, \ell)$, $\mathcal{D}(\epsilon, \ell)$,
and $\Deltaconf(\epsilon , \ell)$ the discrete Beltrami-Laplace operator, K\"ahler operator, 
and conformal Laplacian on the graph $\uG_{\epsilon,\ell}$ respectively.
\end{Def}

The following estimate (see Appendix \ref{prooflemmabound} for proof) explains why 
$\nabla$ and $\overline{\nabla}$ 
should be considered as 
discrete analogues of the holomorphic and anti-holomorphic derivatives $\partial$ and $\overline{\partial}$.

\begin{lemma}
\label{lemmabound}
Given a smooth function $\phi: \Bbb{C} \longrightarrow \Bbb{C}$ and a 
triangle $\uf$ with vertices $z_1$, $z_2$, $z_3$ (listed in counter-clockwise order), circumcenter
$z(\uf)$, and circumradius $R(\uf)$, we have the following estimate

\begin{equation}
\Big| \nabla \phi (\uf )   - \partial \phi (z(\uf) )\Big| \leq R(\uf) \, \Bigg(  \, {3 \over 2} \, 
\sup_{z \in B_\uf}
\big| \partial^2 \phi  \big| \, + \, 
2 \,  \sup_{z \in B_\uf} \big| \partial \overline{\partial} \phi \big| \, + \,  {1 \over 2} \, \sup_{z \in B_\uf}
\big| \overline{\partial}^2  \phi \big| \, \Bigg)
\end{equation}
where $B_\uf$ is the disk bounded by the circumcircle of $\uf$
\begin{equation}
\label{ }
B_\uf=\{z\,;\, |z-z(\uf)]\le R(\uf) \}
\end{equation}
\end{lemma}

Using Lemma \ref{lemmabound}
we are able to obtain a smooth version of Theorem \ref{discrete-theorem}
involving a scaling parameter $\ell>0$ 
 whose continuum limit 
is consistent with formula \ref{zepsilonbar}

\begin{Thm}
\label{TfinalOPElike1}
Consider two smooth complex functions $F_1(z)$ and $F_2(z)$ whose supports $\Omega_1=\supp F_1$ and $\Omega_2=\supp F_2$ in the plane are \emph{compact and disjoint}(hence at finite distance),  and a bi-local deformation of the embedding given by

\[ z_{\underline{\epsilon}, \ell}(\uv) \, := \  z_\mathrm{cr}(\uv) \, + \, \epsilon_1 
F_{1;\ell} (\uv ) \, + \, \epsilon_2 
F_{2;\ell} ( \uv)  \]

\bigskip
\noindent
where $\ell>0$ is a scaling parameter and $F_{i;\ell}(\uv) := \ell F_i \big(z_\mathrm{cr}(\uv)/ \ell \big)$ 
for $i=1, 2$.
The scaling limit $\ell\to\infty$ of the $\epsilon_1\epsilon_2$ cross term in the expansion of $\log\det\Delta (\underline{\epsilon} \, , \ell)$ 
and of $\log \det \mathcal{D}(\underline{\epsilon} \, , \ell)$ (given by theorem~\ref{discrete-theorem}) exist and are given by
\begin{equation}
\label{finalOPElike1}
\begin{split}
\lim_{\ell \rightarrow \infty} \ \deltaeot \log \det \Delta(\ell) \ = \ &\lim_{\ell \rightarrow \infty} \ \deltaeot \log \det \mathcal{D}(\ell) \\
 =\ &
{1 \over {\pi^2}} \iint_{\Omega_1 \times \Omega_2} dx_1^2  \, dx_2^2 \ \frak{Re} 
 \Bigg[{\overline{\partial} F_1( x_1 ) \, 
 \overline{\partial} F_2 (x_2 )\over (x_1-x_2)^4} \Bigg]
\end{split}
\end{equation}
The limit value in formula \ref{finalOPElike1} is independent of the initial isoradial Delaunay graph $\uG_\mathrm{cr}$.
\end{Thm}

\begin{Rem}
\label{Remark3}
Whenever $\uG_{0^+ \!, \ell}$ contains finitely
many chords (see Def.~\ref{ChordEdge}), the $\ell \rightarrow \infty$ scaling limit of the bi-local formula 
for the $\epsilon_1 \epsilon_2$ cross-term in
$\tr \big[ \delta \Deltaconf (\underline{\epsilon} \, , \ell) \cdot \Delta_\mathrm{cr}^{-1} \big]^2$
of the conformal Laplacian 
(as presented in section \ref{ssConfLApDEl}) agrees
with the limit value in formula \ref{finalOPElike1} of Theorem
\ref{TfinalOPElike1}. 
In general this is not the case, and a scaling limit does not exist. If it exists, the effect of the anomalous terms may be present in the scaling limit, which is not universal. An example is given in Appendix~\ref{tiling-by-cocyclic-quad}.
\end{Rem}

Formula~\ref{finalOPElike1} in Theorem~\ref{TfinalOPElike1} implicitly involves a nested limit where the deformation parameters 
$\underline\epsilon=(\epsilon_1,\epsilon_2)$ are first taken to zero, and subsequently the scaling parameter $\ell$ is taken to $\infty$.
An interesting question is whether the $\underline\epsilon\to 0$ limit and the $\ell\to\infty$ limits can be interchanged.

To study this question, one needs uniform bounds on the variations  $\delta{\Delta}(\epsilon)$ and 
$\delta\mathcal{D}(\epsilon)$  (see \ref{deltaOdef}) with respect to the space of isoradial Delaunay graphs. 
Since the bound $\tilde \epsilon_F$ in Proposition \ref{LemmaOne} depends on the graph, we cannot hope to make a stable deformation simultaneously for all Delaunay graphs. 
This requires us to work with Delaunay deformations beyond the $\tilde \epsilon_F$ threshold, and take into account the occurrence of Whitehead flips as the graph is deformed.
This is addressed in Sect.~\ref{aVarOpFlips}.

\medskip

\noindent{\textbf{Bounding the variation of the circumradii}}\\
In order to bound the operator variations $\delta{\Delta}(\epsilon)$ and $\delta\mathcal{D}(\epsilon)$ it is necessary to track the circumradius $R(\uf_\epsilon)$ of each face  $\uf_\epsilon$ of $\uG_\epsilon$  as a function of $\epsilon$.
In Proposition~\ref{PropReps}, we  bound the radius $R(\uf)$ uniformly for all faces and all initial isoradial Delaunay graph $\uG_\mathrm{cr}$ with isoradius $R_\mathrm{cr}=R_0$.
Specifically, we show there exists $\epsilon_\mathtt{max}(R_\mathrm{cr})$, and two functions $\bar R_-(\epsilon,R_\mathrm{cr})$ and $\bar R_+(\epsilon,R_\mathrm{cr})$ such that for $0\le \epsilon < \epsilon_\mathtt{max}(R_\mathrm{cr})$
\begin{equation}
\label{RepsIneqA}
\bar R_-(\epsilon,R_\mathrm{cr}) \ \le\  R(\uf_\epsilon)\ \le  \bar R_+(\epsilon,R_\mathrm{cr})
\end{equation}
and 
\begin{equation}
\label{ }
\lim_{\epsilon\to 0} \bar R_-(\epsilon,R_\mathrm{cr}) = \lim_{\epsilon\to 0} \bar R_+(\epsilon,R_\mathrm{cr}) = R_\mathrm{cr}
\end{equation}
The quantities $\epsilon_\mathtt{max}(R_\mathrm{cr})$, $\bar R_-(\epsilon,R_\mathrm{cr})$ and $\bar R_+(\epsilon,R_\mathrm{cr})$ depend only on $R_\mathrm{cr}$ and on the smooth displacement function $F$. 
They are given explicitly in Prop.~\ref{PropReps}. 

\medskip
\noindent{\textbf{Results for interchanging the $\epsilon\to 0$ and $\ell\to\infty$ limits}}

The matrix entries of the operators $\Delta(\epsilon)$ and $\mathcal{D}(\epsilon)$ are continuous functions of $\epsilon$,
and using the bounds $\bar R_-$ and $\bar R_+$ of \ref{RepsIneqA}%
, we can show that the derivatives $\Delta'(\epsilon)$ and $\mathcal{D}'(\epsilon)$ with respect to 
$\epsilon$ are piecewise continuous functions of $\epsilon$ and obtain uniform bounds on their matrix entries. This leads us to the following conjecture for $\Delta$.

\begin{Conjecture}
\label{ConjectureLimTr2D}
Let $\uG_\mathrm{cr}$ be an isoradial, Delaunay triangulation
with embedding $\uv \mapsto z_\mathrm{cr}(\uv)$, let $F_1, F_2$ be two smooth displacement functions
with disjoint, compact supports, and let $z_{\underline{\epsilon}, \ell} = z_\mathrm{cr} + \epsilon_1 F_{1;\ell} + \epsilon_2 F_{2; \ell}$
be the corresponding scaled and deformed embedding with 
respect to a pair of independent parameters $\underline{\epsilon} = (\epsilon_1, \epsilon_2)$
and $\ell >0$. Then
\begin{equation}
\label{ }
\begin{split}
&\lim_{\underline{\epsilon}\to 0}\lim_{\ell \to \infty}\mathrm{tr}\left[
{\partial \over {\partial \epsilon_1}}  \Delta (\underline{\epsilon} \, , \ell) 
\cdot \Delta_\mathrm{cr}^{-1} 
\cdot 
{\partial  \over {\partial \epsilon_2}}  \Delta(\underline{\epsilon} \, , \ell)
\cdot \Delta_\mathrm{cr}^{-1}  
\right]
\\
&\qquad=\ \lim_{\ell\to \infty}\lim_{\underline{\epsilon}\to 0}\mathrm{tr}\left[
{\partial \over {\partial \epsilon_1}}  \Delta (\underline{\epsilon} \, , \ell) 
\cdot \Delta_\mathrm{cr}^{-1} 
\cdot 
{\partial  \over {\partial \epsilon_2}}  \Delta(\underline{\epsilon} \, , \ell)
\cdot \Delta_\mathrm{cr}^{-1}  
\right]
\\
&\qquad\qquad
=\ {2 \over \pi^2} \int_{\Omega_1} d^2x_1 \int_{\Omega_2} d^2x_2 \ 
\frak{Re} \Bigg[ {\bar\partial F_1(x_1)\ \bar\partial F_2(x_2)\over (x_1-x_2)^4} \Bigg]
\end{split}
\end{equation}
\end{Conjecture}

\medskip
Conjecture \ref{ConjectureLimTr2D} is a special case of Prop.~\ref{prvsepsComm}, which relies on the rigorous estimates obtained in Sect.~\ref{aVarOpFlips}, and also on Conjecture~\ref{Conf3bound}, the later of which stipulates a bound on $\nabla p_3$ for critical lattices (where $p_3$ is defined in \ref{def-p_n} and appears already in Prop.~\ref{Prop-SharpAsymptotics1}).

\medskip
For $\mathcal{D}$ we do no get such a strong result, but only the following {weaker} conjecture, which follows from Propositions~\ref{NoLimitDprime} and \ref{PropLimTr2D}.

\begin{Conjecture}
\label{ConjectureLimTr2DK}
Let $\uG_\mathrm{cr}$ be an isoradial, Delaunay triangulation
with embedding $\uv \mapsto z_\mathrm{cr}(\uv)$, let $F_1, F_2$ be two smooth displacement functions
with disjoint, compact supports, and let $z_{\underline{\epsilon} , \ell} = z_\mathrm{cr} + \epsilon_1 F_{1;\ell} + \epsilon_2 F_{2; \ell}$
be the corresponding scaled and deformed embedding with 
respect to a pair of independent parameters $\underline{\epsilon} = (\epsilon_1, \epsilon_2)$
and $\ell >0$. 

In general the limit
\begin{equation}
\label{ }
\lim_{\ell \to \infty}\mathrm{tr}\left[
{\partial \over {\partial \epsilon_1}}  \mathcal{D} (\underline{\epsilon} \, , \ell) 
\cdot \mathcal{D}_\mathrm{cr}^{-1} 
\cdot 
{\partial  \over {\partial \epsilon_2}}  \mathcal{D}(\underline{\epsilon} \, , \ell)
\cdot \mathcal{D}_\mathrm{cr}^{-1}  
\right]
\end{equation}
does not exist for non-zero $\underline\epsilon$.

The double ``simultaneous'' limit exists, where $\ell\to \infty$ and $\underline\epsilon\to 0$ such that $\ell\underline\epsilon=\underline{\mathtt{c}}$ {with $\underline{\mathtt{c}}>0$ staying} constant.
Its value is 
\begin{equation}
\label{ }
\begin{split}
\lim_{\substack{\ell\to \infty\\ \ell\underline{\epsilon}=\underline{\mathtt{c}}}} 
&\tr\left[
{\partial \over {\partial \epsilon_1}}  \mathcal{D} (\underline{\epsilon} \, , \ell) 
\cdot \mathcal{D}_\mathrm{cr}^{-1} 
\cdot 
{\partial  \over {\partial \epsilon_2}}  \mathcal{D}(\underline{\epsilon} \, , \ell)
\cdot \mathcal{D}_\mathrm{cr}^{-1} \right] 
\\
& \ =\ \hskip 1.em {2\over \pi^2}\int_{\Omega_1}\! d^2 x_1 \int_{\Omega_2}\! d^2 x_2\ 
\frak{Re}\left[    {\bar\partial F_1(x_1)\,\bar\partial F_2(x_2)\over (x_1-x_2)^4}   \right] 
\end{split}
\end{equation}
\end{Conjecture}

\subsubsection{Interpretation in terms of discrete stress-energy tensors and discrete central charge}
\label{sssIntroTC}
The results presented above can be formulated in the language of CFT in terms of an action, a stress-energy tensor and a central charge.
This is done in Sect.~\ref{sDiscus1stO}.
For the Laplace Beltrami operator $\Delta$ the associated discrete action is
\begin{equation}
\label{SLBDiscrIntro}
S[\Phi,\bar\Phi]=\Phi{\cdot}\Delta\bar\Phi=\sum_{\stackrel{\scriptstyle \mathrm{vertices}}{\ \, \uu, \uv \, \in \, 
\uG}} \Phi_{\uu} \Delta_{\uu \uv} \bar\Phi_{\uv}
\end{equation}
where $(\Phi,\bar\Phi)$  are Grassmann fields supported on vertices of the Delaunay graph $\uG$. The corresponding functional integral is
\begin{equation}
\label{ }
\det(\Delta)=\int \frak{D}[\Phi{,}\bar\Phi]\,\mathrm{e}^{-S[\Phi,\bar\Phi]}
\end{equation}

A general deformation $z\mapsto z+\epsilon\, F$ of the coordinate embedding induces a deformed action $S_\epsilon[\Phi,\bar\Phi]=\Phi{\cdot}\Delta(\epsilon)\bar\Phi$ , which we can develop as $S_\epsilon= S+ \epsilon\,\deltae S + \mathrm{O}(\epsilon^2)$.
Using the variation of $\Delta(\epsilon)$ given by Prop. \ref{PVarDelta}, the linear term $\deltae S$ reads explicitly as 
\begin{equation}
\label{ }
\deltae S[\Phi,\bar\Phi] = - 4\sum_{\stackrel{\scriptstyle \mathrm{faces}}{\ux\in\widehat\uG_{0^+}}} A(\ux)\left(\overline\nabla F(\ux)\, \nabla\Phi(\ux) \nabla\bar\Phi(\ux)\,+\,\mathrm{c.\,c.}\right)
\end{equation}
In analogy with the countinuous case, the components of the discrete stress-energy tensor $\mathbf{T}_{_{\!\Delta}}$ can be identified as
\begin{equation}
\label{TLBDiscrINtro} 
T_{_{\!\Delta}}(\ux)=-4\pi\, \nabla\Phi(\ux) \nabla\bar\Phi(\ux)
\quad\text{and}\quad\overline T_{_{\!\Delta}}(\ux)=-4\pi\, \overline\nabla\Phi(\ux) \overline\nabla\bar\Phi(\ux)
\end{equation}
while $\mathbf{T}_{_{\!\Delta}}$ is traceless, namely $\tr\mathbf{T}_{_{\!\Delta}}(\ux)=0$.
See \ref{TLBDiscr} for details.
Taking vacuum expectation values of the components of the stress-energy tensor 
we recover our results (Prop.~\ref{T1stVar} and Th.~\ref{discrete-theorem}) for the first and second order variations of $\log\det (\Delta)$, in the case of a critical lattice $\uG=\uG_\mathrm{cr}$. 
In  the scaling limit, the discrete $T_{_{\!\Delta}}$ given in \ref{TLBDiscrINtro}  becomes the continuum stress-energy tensor $T= -4\pi\, \partial\Phi\partial\bar\Phi$ for the CFT of a free Grassmann field (see Appendix~\ref{aCFTreminder}  for details).

Theorem~\ref{TfinalOPElike1} shows that, when perturbing a critical lattice, the scaling limit for the second order variation $\deltaeot\log\det\Delta$ of the Laplace-Beltrami operator $\Delta$ exists, and can be calculated in term of the connected vacuum expectation values
\begin{equation}
\label{ }
\langle T_{\!\scriptscriptstyle{\Delta}}(\ux)\, T_{\!\scriptscriptstyle{\Delta}}(\uy)\rangle
\qquad\text{and}\qquad \langle \overline T_{\!\scriptscriptstyle{\Delta}}(\ux)\, \overline T_{\!\scriptscriptstyle{\Delta}}(\uy)\rangle
\end{equation}
This implies that in the scaling limit, we recover a short distance operator product expansion for $T_{_{\!\Delta}}$,
\begin{equation}
\label{OPETDelta}
\langle T_{_{\!\Delta}}\!(u)\,T_{_{\!\Delta}}\!(v)\rangle \ =\ -\,{1\over (u-v)^4}\ +\ \cdots
\end{equation}
which is the OPE for a CFT with central charge $c=-2$ (through \ref{OPETone}).
This is the expected result for a complex Grassmann field, which is indeed a conformal field theory with central charge $c=-2$ (see Appendix~\ref{aCFTreminder} and e.g. \cite{francesco1997conformal}).

\medskip

The same analysis is carried out for the K\"ahler operator $\mathcal{D}$.
From the variation of $\mathcal{D}(\epsilon)$ given by Prop.~\ref{PVarKahler}, we can isolate the components of the corresponding discrete stress-energy tensor $\mathbf{T}_{\!\scriptscriptstyle{\mathcal{D}}}$ (see  \ref{TKdiscr})
\begin{equation}
\label{TKdiscrIntro}
\begin{split}
 T_{\!\scriptscriptstyle{\mathcal{D}}}(\ux) & = 
-\, 4\pi\, {1\over R(\ux)^2}\,\left( \nabla\Phi(\ux)\,\nabla\bar\Phi(\ux)+ C(\ux)\, \overline\nabla \Phi(\ux)\ \nabla\bar\Phi(\ux)\right) \\
 \overline T_{\!\scriptscriptstyle{\mathcal{D}}}(\ux) & = -\, 4\pi\, {1\over R(\ux)^2}\,\left( \overline\nabla\Phi(\ux)\,\overline\nabla\bar\Phi(\ux)+ \bar C(\ux)\, \overline\nabla \Phi(\ux)\ \nabla\bar\Phi(\ux)\right) \\
 \tr \mathbf{T}_{\!\scriptscriptstyle{\mathcal{D}}}(\ux) & =\quad\  8 \  {1\over R(\ux)^2}\,\left( \overline\nabla\Phi(\ux)\,\nabla\bar\Phi(\ux)\right)
\end{split}
\end{equation}
where $R(\ux)$ is the radius of the face $\ux$, and $C(\ux)$ is a  geometrical factor given in formula~\ref{CDef}, which depends on the shape and orientation of the face $\ux$.
Note that $\mathbf{T}_{\!\scriptscriptstyle{\mathcal{D}}}$ is no longer traceless.
$C(\ux)$ 
has no obvious scaling limit $\ell\to\infty$, independent on the details of the Delaunay lattice, so we cannot associate a stress-energy tensor for some continuum QFT to the discrete $\mathbf{T}_{\!\scriptscriptstyle{\mathcal{D}}}$, as we did for $\Delta$ by replacing $\nabla$ with $\partial$.

Surprisingly, Theorem~\ref{TfinalOPElike1} also shows that, when perturbing a critical lattice, the scaling limit for the second order variation $\deltaeot\log\det\mathcal{D}$ of the K\"ahler operator $\mathcal{D}$ still exists, and can still be calculated in term of the connected vacuum expectation values
\begin{equation}
\label{ }
\langle T_{\!\scriptscriptstyle{\mathcal{D}}}(\ux)\, T_{\!\scriptscriptstyle{\mathcal{D}}}(\uy)\rangle
\qquad\text{and}\qquad \langle \overline T_{\!\scriptscriptstyle{\mathcal{D}}}(\ux)\, \overline T_{\!\scriptscriptstyle{\mathcal{D}}}(\uy)\rangle
\end{equation}
Moreover, we recover a short distance OPE for $T_{\!\scriptscriptstyle{\mathcal{D}}}$ which is identical to the OPE for $T_{_{\!\Delta}}$ given by \ref{OPETDelta}
\begin{equation}
\label{OPETKaehler}
\langle T_{\!\scriptscriptstyle{\mathcal{D}}}(u)\,T_{\!\scriptscriptstyle{\mathcal{D}}}(v)\rangle\ =\ -\,{1\over (u-v)^4} +\ \cdots
\end{equation}
and therefore we can associate a ``central charge'' $c_{\!\scriptscriptstyle{\mathcal{D}}}=-2$ with the same value as the central charge 
$c_{\!\scriptscriptstyle{\Delta}}=-2$ for $\Delta$ (see Sect.~\ref{sDiscus1stO} for a more thorough discussion).

Finally, a similar analysis is taken in \ref{sssConfDeltaT} for the conformal Laplacian $\Deltaconf$, and leads to a discrete stress-energy tensor $\mathbf{T}_{_{\Deltaconf}}$. 
The trace term $\tr\left(\mathbf{T}_{_{\Deltaconf}}\right)$ vanishes, and the discretized holomorphic (and anti-holomorphic) component $T_{\!\scriptscriptstyle{\Deltaconf}}$ (and $\bar T_{\!\scriptscriptstyle{\Deltaconf}}$) can be written explicitly as 
\begin{equation}
\label{TconfDiscrSchem}
\text{\cminfamily{A}}(\ux)\nabla\Phi(\ux)\nabla\bar\Phi(\ux)+\text{\cminfamily{B}}(\ux)\nabla\Phi(\ux)\overline\nabla\bar\Phi(\ux)+\text{\cminfamily{C}}(\ux)\overline\nabla\Phi(\ux)\nabla\bar\Phi(\ux)+\text{\cminfamily{D}}(\ux)\overline\nabla\Phi(\ux)\overline\nabla\bar\Phi(\ux)
\end{equation}
The coefficients $\text{\cminfamily{A}}(\ux),\dots, \text{\cminfamily{D}}(\ux)$ depend in a non-trivial way on the geometry of the face $\ux$ and its three neighbouring faces in $\widehat\uG_{0^+}$, and have no meaningful continuum limit.
This is reflected in the fact that the second order variation $\deltaeot\log\det\Deltaconf(\ell)$ has no scaling limit in general, as stated in Remark~\ref{Remark3}.

\subsection{Plan of the paper}
\label{ssPlan}

This paper is organised as follows: 

\medskip 
The present section~\ref{sIntro} is the introduction.

\medskip
Section \ref{sCritLapl} presents 
basic concepts about the geometry of planar graphs which are relevant to the paper. 
Most of the material is standard,
however we introduce the notion of a {\it chord} (see Def.~\ref{ChordEdge})
which allows us to slightly broaden the definition of
an isoradial triangulation (given in \cite{Kenyon2002}) to accommodate
configurations with four or more cocyclic vertices.
Section~\ref{ssBasicPlanar} gives definitions and sets notation for polyhedral graphs, edges and chords, (weak) Delaunay graphs, isoradial graphs, etc. and makes precise the notion of the abstract rhombic surface $\uS_\uG^\lozenge$ associated to a polyhedral graph $\uG$ alluded {\color{blue} to} in 
Section~\ref{ssDelGraph}.
Section~\ref{ssRhomGraph} addresses geometrical concepts and properties of rhombic graphs, 
mainly following the presentations of \cite{Kenyon2002} and \cite{Kenyon14rhombicembeddings}.
In order to help establish the asymptotic formula in Prop.~\ref{Prop-SharpAsymptotics1} 
we undertake in Proposition \ref{semi-circle-property} 
a careful analysis of the interval of possible angles taken by any path in the rhombic graph 
of an isoradial Delaunay graph.

\medskip
In section \ref{sLaplacians} we review the $\nabla$ and $\overline{\nabla}$ operators of \cite{DavidEynard2014}
and how they are used to obtain ``local factorizations'' of the  
Laplace-Beltrami and K\"ahler operators $\Delta$ and $\mathcal{D}$
for a general polyhedral triangulation;
see remarks
\ref{DNablaForm} and \ref{DeltaNablaForm}.
We remark that the conformal Laplacian $\Deltaconf$ however does not 
admit a simple, local factorization.
Following this, we recall two approaches used to define the (normalised) log-determinant of a Laplace-like operator such as 
$\Delta$, $\mathcal{D}$, and $\Deltaconf$ for infinite polyhedral graphs 
which are either (1) doubly periodic or (2) obtained as a nested 
limit of finite graphs each with Dirichlet boundary conditions. Formulae \ref{norm-log-det} and \ref{dirichlet-log-det}
serve respectively as definitions in these two cases. 
We end the section by discussing Kenyon's local formula in \cite{Kenyon2002} for the normalised log-determinant of the critical Laplacian
for doubly periodic, isoradial (weak) Delaunay graphs, as well as its formal extension to the non-periodic case.

\medskip
In section \ref{sAsymptotics} we derive the long-range asymptotic formula for the Green's function of the critical Laplacian (associated to an isoradial Delaunay graph) as stated in Proposition~\ref{Prop-SharpAsymptotics1}. 
We rely on the methods of \cite{Kenyon2002} along with some added improvements, in particular for non-periodic graphs.
Among other things our analysis
provides uniform bounds on the coefficients of the asymptotic expansion (see \ref{lemma-p-estimate} and \ref{def-c-md}) thus sharpening 
the results and approximations in \cite{Kenyon2002}.

\medskip
Section \ref{sVarOp} addresses deformations of Delaunay graphs and corresponding operators.
In section  \ref{ssVarOp} we introduce the notions of {\it Delaunay} and {\it rigid} deformations: In both
cases the coordinate embedding of the graph is perturbed by a local displacement function
together with a deformation parameter $\epsilon  \geq 0$.
Delaunay deformations modify the incidence relations (i.e. the edge and face sets) 
so that the Delaunay constraints are maintained 
while rigid deformations always fix the incidence relations
of the initial graph (and so the resulting graph may cease to be Delaunay).
In the case of a Delaunay deformation, we explain in Lemma~\ref{Lemma11} how to regulate the 
parameter $\epsilon \geq 0$ so that the edges of {the} initial graph are {\it stable} and do not undergo Lawson ``flips''.
A generic Delaunay deformation, however, can break the cyclicity of faces {having four or more vertices}
in the initial graph
and introduce new edges which subdivide these faces.
Nevertheless these additional edges are shown to be stable for values of $\epsilon > 0$ 
which are bounded appropriately.
This follows from Prop.~\ref{epsilontildeF} which 
also proves the existence of a (weak) Delaunay limit graph $\uG_{0^+}$ whose redaction
$\uG_{0^+}^\bullet$ coincides with the initial Delaunay graph $\uG_0$.
In Sect.~\ref{ssVarOp} we study 
the first order variation of the Laplace-Beltrami and K\"ahler operators,  
when the underlying polyhedral triangulation is subject to a rigid deformation.
Results are given in Props. ~\ref{PVarDelta} and \ref{PVarKahler} respectively.
The conformal Laplacian $\Deltaconf$ does not admit a local factorization of the kind presented in Props.~\ref{PVarDelta} and \ref{PVarKahler}
and for this reason there is no analogous formula for its first order variation. Section \ref{ssGenVarNot} sets up notation.

\medskip
The calculations of the first and second order variations of the log-determinant for 
the Beltrami-Laplace operator, the K\"ahler operator, and the conformal Laplacian
are undertaken in Section \ref{sCalculations}.
The first order variation formulae are entirely local, i.e. expressed as sums 
of weights of edges. The second order variations, on the other hand, 
involve long-range effects of the critical Green's function $\Delta_\mathrm{cr}^{-1}$
associated to pairs of distant vertices and, in principle,  
register aspects of the global geometry of the initial isoradial Delaunay graph $\uG_\mathrm{cr}$.

In Propositions \ref{T1stVar} and \ref{T1stVarD} of Section \ref{subsection-first-order} 
we present formulae for the first order variations of the Beltrami-Laplace and K\"ahler operators 
which are valid uniformly for all isoradial Delaunay graphs.
The first-order formula for the conformal Laplacian 
incorporates an additional term which accounts for the effect made by chords in $\uG_{0^+}$
and is given in Proposition \ref{T1stVarC}.
The second order formulae for the variation of the log-determinant of the Beltrami-Laplace and K\"ahler operators
are calculated separately in Propositions \ref{ThDelta} and \ref{The3} of Section 
\ref{subsection-second-order} respectively;
this is the content of Theorem \ref{discrete-theorem}.
In both cases, our approach
relies on the asymptotics of the Green's function
in Proposition \ref{Prop-SharpAsymptotics1} 
and Lemma \ref{lNGN}; the latter makes use of the operator
factorisations in Propositions \ref{PVarDelta} and \ref{PVarKahler}
as well as a novel estimate presented in Lemma \ref{lnablabound}.

Formula \ref{finalOPElike-discrete} of Theorem \ref{discrete-theorem}
is not valid for the conformal Laplacian 
and it must be modified by defect terms which take into account the effect of
chords in $\uG_{0^+}$. See formulae \ref{chord2chord} and \ref{chord2edge}.
We propose that these defects are indicative of a discrete curvature anomaly arising 
from the perturbation. This is examined in Section \ref{subsubsection-dipole}.

\medskip
Section \ref{sScalingLimit} handles the proof of Theorem~\ref{TfinalOPElike1}, which deals with the existence and value of the scaling limit given by formula 
\ref{finalOPElike1} for the Beltrami-Laplace and K\"ahler operators.
Sections~\ref{ssRescaling} and \ref{ssRsBilDef} address 
some technical points about bi-local deformations, scaling limits, and re-summation. 
In Section~\ref{sScalLimit} we prove the existence of the scaling limit of 
\ref{finalOPElike-discrete} in the case of a continuous bi-local deformation and settle Theorem~\ref{TfinalOPElike1}. 
The basic idea is to interpret \ref{finalOPElike-discrete} as a Riemann sum with a mesh 
controlled by the scaling parameter. 
The scaling limit considered in Section~\ref{sScalLimit}  
is taken with respect to an isoradial refinement $\widehat{\uG}_{0^+,\ell}$
associated to a (scaled) deformation of our initial, isoradial Delaunay graph $\uG_\mathrm{cr}$. 
In effect the result is a calculation of a nested limit: First we 
take the deformation parameter limit $\epsilon_1, \epsilon_2 \rightarrow 0$
(bringing us to $\widehat{\uG}_{0^+,\ell}$) and then
we subsequently take the $\ell \rightarrow \infty$ scaling limit.

In section~\ref{ssBeyondCheck} 
we ask whether these two limits can be 
interchanged. This question is
related to whether the scaling limit in Theorem \ref{TfinalOPElike1}
exists for a Delaunay graph (not necessarily isoradial) 
which is obtained as a small deformation of an isoradial Delaunay graph. We 
return to this issue in Section \ref{aVarOpFlips}.

Section~\ref{brutal-approach} addresses the issue of the uniform convergence in the ``flip problem'' for smooth, scaled deformations. 
A first attempt is offered in Lemma~\ref{epsFbound2}, where we introduce a lower bound on the range of conformal angles for an isoradial Delaunay triangulation.
This constraint ensures that no flips occur whenever the deformation parameter $\epsilon$ is bounded above by a threshold $\check\epsilon_F$ 
which is uniform both with respect to the scaling parameter and this proper subclass of isoradial Delaunay triangulations.
\medskip

In Section \ref{aVarOpFlips} we return to the general case of deformations $\uG_\epsilon$ of Delaunay graphs $\uG_\mathrm{cr}$ which may incur edge flips.
We look for uniform bounds on the variation of the corresponding
operators $\Delta(\epsilon)$ and $\mathcal{D}(\epsilon)$
for small but non-zero values of the deformation parameters.
In order to get {uniform bounds} with respect to the choice of the initial graph  $\uG_0$
we obtain in Prop.~\ref{PropReps} estimate for the variation of the radius $R(\uf_\epsilon)$ of an arbitrary triangle $\uf_\epsilon$ of $\uG_\epsilon$ 
as the deformation parameter $\epsilon$ varies.
We deduce strong results (summarized in Conj.~\ref{ConjectureLimTr2D}) on the uniform convergence of 
the scaling limit for $\Delta$ (Prop. \ref{propboundDelta})
and of the scaling limit of the second order bi-local
term (leading to the OPE) (Prop. \ref{PropLimTr2});
the later result depends on a conjectural, uniform
estimate (Conj. \ref{Conf3bound})
on $\nabla p_3 (\uf)$ and $\bar{\nabla} p_3(\uf)$ 
in terms of the radius
$R(\uf)$ of a face $\uf$
and the scaling parameter.
We finish the section by showing that there is a qualitative difference
between $\Delta$ and $\mathcal{D}$, and we obtain a weaker
but interesting ``simultaneous convergence'' result for the scaling limit of the second order
bi-local term for $\mathcal{D}$ (Prop. \ref{PropLimTr2D}), summarized in Conj.~\ref{ConjectureLimTr2D}.

\medskip

Section~\ref{sDiscussion} summarizes our results, and presents them from a more statistical physics point of view.
After reviewing the aims of the paper in \ref{sDiscusAim} we discuss in \ref{sDiscus1stO} the first order variation of the log-determinant for the three operators $\Delta$, $\Deltaconf$ and $\mathcal{D}$
vis-\`a-vis the Gaussian Free Field.
We show that formula \ref{varTrLlLB2} for the Laplace-Beltrami operator $\Delta$ 
can be re-expressed in terms of the vacuum expectation value of 
a discrete stress-energy tensor $\mathrm{T}_{\! \scriptscriptstyle{\Delta}}$ for a Grassmann free field 
theory (for convenience we opt for a fermionic analogue of the Massless Free Field (GFF)) supported on $\uG_\mathrm{cr}$
and whose scaling limit coincides with the standard continuous free field. 
This is not a surprise.
Our results for $\mathcal{D}$ and $\Deltaconf$ are similarly expressed using 
discrete stress-energy tensors $\mathrm{T}_{\! \scriptscriptstyle{\mathcal{D}}}$ and $\mathrm{T}_{\! \scriptscriptstyle{\Deltaconf}}$ 
however neither formula \ref{varTrLlLK1} nor formula \ref{varTrLlC1} have a {obvious} continuous limit relating it to the continuous free field.

In \ref{ssDi2ndOr} we discuss the bi-local second order variation formula and the universal form of 
its scaling limit for $\Delta$ and $\mathcal{D}$ in terms of their respective discrete stress-energy tensors.
Furthermore we address the (in general) non-existence of a scaling limit for $\Deltaconf$.

In \ref{ssRelDiff} we discuss the relation and differences between: (i) the model and the questions addressed for Delaunay graphs in our work, and (ii)  previous studies made by Chelkak et al. on the O($n$) model and by Hongler et al. on the GFF and the Ising model on the hexagonal and square lattices respectively.

Finally, in \ref{ssOQandPE} we briefly list some open questions and some possible extensions of this work.

\color{black}
\medskip

Some standard material, technical derivations of results and matters not central to this work are relegated to appendices.

\medskip

Appendix~\ref{aCFTreminder} presents some standard notations and reminders about QFT, CFT and the stress-energy tensor,  in particular for the free boson and the $\mathbf{b}{-}\mathbf{c}$ ghost theory. 

\medskip

Appendix~\ref{prooflemmabound} gives the derivation of Lemma~\ref{lemmabound}, which is instrumental for Theorem~ \ref{TfinalOPElike1} and the derivation of the scaling limit.

\medskip

Appendix~\ref{tiling-by-cocyclic-quad} examines 
the conformal Laplacian $\Deltaconf$ on a particular 
critical Delaunay graph $\uG$ 
as well as the anomalous terms associated with chords in $\uG_{0^+}$
which arise in the second order variation of 
the log-determinant formula for
$\Deltaconf$ addressed in \ref{ssConfLApDEl}.
The graph $\uG$ is sufficiently regular and $\uG_{0^+}$ has a sufficient density of 
chords to ensure that these anomalous terms have a convergent
scaling limit, which is computed explicitly in Proposition~\ref{continuum-limits-anomalies}.

\newpage
\renewcommand{\imath}{\mathrm{i}}

\section{Planar graphs and rhombic graphs}
\label{sCritLapl}

\subsection{Definitions and properties of the basic objects}
\label{ssBasicPlanar}
Let us first introduce the basic geometrical objects that we shall consider: plane triangulations, plane polyhedral graphs (whose faces are cyclic polygons), Delaunay graphs, rhombic graphs, etc. Most of the notations and properties are standard, and can be found in for instance \cite{deBergetalBook2008},\cite{Gallier2013} or \cite{Aurenhammer2013}. Some notations and concepts on isoradial graphs come from 
\cite{Kenyon2002} and \cite{Kenyon14rhombicembeddings}.

\subsubsection{Plane graphs and Delaunay graphs}
\label{sssNotations}

\begin{Def}\label{EmbPlanGr}\end{Def}An \textbf{{embedded planar graph}} will be -- for the purpose of this article -- a graph $\uG$ given  
by a set of vertices $\mathrm{V}(\uG)$ and a set of edges $\mathrm{E}(\uG)$, together with
an injective map $z: \mathrm{V}(\uG) \longrightarrow \Bbb{C}$. 
For a vertex $\uv \in \mathrm{V}(\uG)$ we shall denote 
its complex coordinate by $z(\uv)$;
if there is no risk of confusion we shall sometimes denote the complex coordinate by the 
vertex label $\uv$ itself. 
Each edge $\ue= \overline{\mathtt{uv}}$ is embedded as a \emph{straight line segment} joining its end-points $z(\uu)$ and $z(\uv)$
while the oriented edge $\vec{\ue} = (\uu, \uv)$ corresponds to the displacement vector
$z(\uv) - z(\uu)$. We require that for any pair of edges the corresponding line segments
are non-crossing (i.e. do not share any interior points).
The embedding determines an abstract set of faces $\mathrm{F}(\uG)$ and we
require that each face
$\uf \in \mathrm{F}(\uG)$ is embedded as a \emph{convex polygon}
endowed with a counter-clockwise orientation (so that no face is folded onto an adjacent face). Furthermore 
the set of faces must cover the plane and they must not 
accumulate in any finite region of the plane (i.e. each open 
disk must contain only finitely many faces).
We shall sometimes elide between the description of $\uG$ as a abstract combinatorial entity (i.e. vertices, edges, faces and their incidence relations) 
and its description as an embedded object in the plane (points, segments, and polygons with the geometrical restrictions described above).

\begin{Def}\label{PolyhGr}\end{Def}A \textbf{polyhedral graph} will be an embedded planar graph such that each face is a \emph{cyclic polygon}, i.e. all the vertices of the face lie on a circle (the circumcircle $C_\uf$ of the face $\uf$), in \emph{cyclic order}. Two faces may have the same circumcircle.

\begin{Def}\label{ChordEdge}\end{Def}An edge $\ue \in \mathrm{E}(\uG)$ of a polyhedral graph $\uG$ is a \textbf{chord} if the two faces $\uf$ and 
$\mathtt{g}$ of $\uG$ adjacent to $\ue$ share the same circumcircle (i.e. the circumcenters of $\uf$ and $\mathtt{g}$ coincide).
An edge which is not a chord is said to be a \textbf{regular edge} of $\uG$. If no ambiguity arises, we shall use the term edge for regular edges only, and chords for the others.

\begin{Def}\label{ChLessPolyGr}\end{Def}A \textbf{chordless polyhedral graph} is a  polyhedral graph without chords, i.e. no pair of faces share the same circumcircle. Obviously chordless polyhedral graphs correspond to a special class of circle patterns in the plane.
In a general polyhedral graph, a face which does not share its circumcircle with another face will be said to be a \textbf{chordless face}. 

\begin{Def}[\textbf{Redacted graph}]
\label{defRegGraph}
\end{Def}
Given a polyhedral graph $\uG$, let $\uG^\bullet$ be the graph with the same vertex set $\mathrm{V}(\uG^\bullet) = \mathrm{V}(\uG)$,
the same embedding $z^\bullet = z$, 
and with edge set $\mathrm{E}(\uG^\bullet) = \mathrm{E}(\uG) - \mathrm{chords}(\uG)$, 
where $\mathrm{chords}(\uG)$ is the set of all chords in $\uG$.
We call $\uG^\bullet$ the \textbf{redaction}, or \textbf{redacted graph}, of $\uG$.

\begin{Def}\label{DelaunayGr}\end{Def}
A \textbf{weak Delaunay graph} is a polyhedral graph $\uG$ such that for any face $\uf$, the \emph{interior} of the circumdisk  $D_\uf$ (the closed disk whose boundary is the circumcircle $C_\uf$) contains no vertex of $\uG$. The circumcircle itself contains the vertices of $\uf$, and possibly other vertices. A \textbf{Delaunay graph} is a chordless weak Delaunay graph 
\footnote{Note that in the literature the term \emph{Delaunay graph} often denotes what we call here a \emph{weak Delaunay graph}}.

\begin{Def}\label{TriangulGr}\end{Def}A \textbf{triangulation} is an embedded planar graph $\uT$ such that each face is a triangle. 
Obviously, a triangulation is a polyhedral graph.
A \textbf{Delaunay triangulation} is a triangulation which is a Delaunay graph.
A \textbf{weak Delaunay triangulation} is a triangulation which is a weak Delaunay graphs. 

\begin{Def}\label{CompletionTofG}\end{Def}
A triangulation $\uT$ is called a \textbf{completion}  of a weak Delaunay graph $\uG$ if $\mathrm{E}(\uG)\subset\mathrm{E}(\uT)$. 
Such a triangulation is necessarily weakly Delaunay, and is obtained by saturating $\uG$ with a maximal collection of non-crossing chords.
Clearly the redactions $\uT^\bullet$ and $\uG^\bullet$ coincide.
Throughout the paper $\widehat\uG$  will denote a choice of completion of a weakly Delaunay graph $\uG$.

\begin{Rem}
\label{RemFinInfTriang}
The concepts of a polygonal and Delaunay graph can be 
extended to finite graphs embedded in the Riemann sphere.
This is done in \cite{DavidEynard2014} for instance.
Such a graph can be visualized either on the sphere or as an embedded planar Delaunay
graph together with edges (represented as infinite rays) joining 
vertices on the boundary of the convex hull 
of the graph to a vertex situated at $\infty$ (if present). 
Likewise, the Voronoi construction, as well as the Lawson flip algorithm \cite{Lawson77},
can be adapted to construct a unique embedded Delaunay graph
from any finite configuration of points in the Riemann sphere.
A similar approach can be undertaken for
(finite) graphs embedded in a compact Riemann surface;
for example, this is done implicitly in \cite{Kenyon2002}
for the torus. 

M\"obius transformations preserve the Delaunay property
for finite graphs embedded in the Riemann sphere,
and so one can incorporate the $\mathrm{PSL}_2(\Bbb{C})$ symmetry
into the model, as done in \cite{DavidEynard2014}.
Our situation is different: We are chiefly interested in infinite Delaunay graphs in the plane
which are locally finite, i.e. having are only finitely many vertices in any open ball. 
Although the application of a $\mathrm{PSL}_2(\Bbb{C})$ transformation preserves the Delaunay property,
the resulting graph may cease to be locally finite, since a neighborhood of $\infty$ can be mapped to a finite 
radius ball containing an infinite number of vertices.  
This is not a problem for our study, since we shall consider graph deformations which are implemented by 
bounded functions with compact support.
\end{Rem}

\subsubsection{Isoradial graphs}\label{sssIsoRad}
\begin{Def}\label{IsoradGr}\end{Def}An \textbf{isoradial graph} is a polyhedral graph $\uG$ such that the \emph{circumradii} $R(\uf)$ (the radius of the circumcircle $C_\uf$ of  $\uf$) of all the faces $\uf$ of $\uG$ are equal.  

\begin{Def}\label{RegularGr}\end{Def}Following \cite{Kenyon2002}, a face $\uf$ whose circumcenter is inside or on the boundary of $\uf$ 
(considered as a cyclic polyhedron) is called a \textbf{regular face}. 
A polyhedral graph such that all its faces are regular is called a \textbf{regular graph}.

\begin{Rem} 
\label{rem7}
Given an oriented edge $\vec{\ue}$ of a polyhedral graph we 
define the corresponding {\bf north} and {\bf south angles} $\theta_\mathrm{n}(\vec {\ue} \, )$ and
$\theta_\mathrm{s}(\vec{\ue} \,)$ through figure \ref{triangles}
in Sect.~\ref{ssDelGraph}. 
By the inscribed angle theorem
$\theta_\mathrm{n}(\vec {\ue} \, )$ does not depend upon
the choice of vertex $\un \in \uf_\mathrm{n}$ in the north face.
Likewise $\theta_\mathrm{s}(\vec{\ue} \,)$ is independent
of the vertex $\us \in \uf_\mathrm{s}$ in the south face.
Note that reversing the orientation of $\vec{\ue}$ exchanges the roles of north and south
and so the {\bf conformal angle} $\theta(\ue) :=(\theta_\mathrm{n}(\vec{\ue} \, ) +  \theta_\mathrm{s}(\vec{\ue} \, ))/2$
is independent of the choice of edge orientation, hence the notation $\theta(\ue)$. 
\end{Rem}

\begin{Rem}
Given an edge $\ue = \overline{\mathtt{uv}}$ with
vertices $\mathtt{u, v} \in \mathrm{V}(\uG)$ the 
value of the conformal angle $\theta(\ue)$ 
equals the argument of the following cross-ratio involving the (coordinates of the) vertices
$\uu$, $\uv$, $\un$, $\us$:

\begin{equation}
\label{ConfAngleDef}
\begin{array}{c}
\theta(\ue) 
\ = \ {1\over 2} \arg\Big(-\big[ z(\uu),  z(\uv) \, ;  z(\un),  z(\us)  \big]\Big)
\\ \ \\ \text{with the anharmonic cross-ratio}\\ \ \\
 \big[ z_1, z_2 \, ; z_3, z_4 \big] 
\ = \ \displaystyle  {(  z_1 - z_3)  
(  z_2 - z_4 )\over 
{(   z_1- z_4 )  ( z_2 - z_3 )} } 
\end{array}
\end{equation} 
Consequently the conformal angle is $\mathrm{SL}_2(\Bbb{C})$-invariant owing to the fact that the cross-ratio is invariant.
\end{Rem}

\begin{Rem}
We want to reiterate the comments in Sect.~\ref{ssDelGraph}, and stress that the Delaunay condition as stated in Def.~\ref{DelaunayGr} is equivalent to the condition that for any edge $\ue$ of a polyhedral graph 
its conformal angle $\theta(\ue)$ is postive and bounded between
$$0 < \theta(\ue)<\pi/2\ .$$ The weak Delaunay condition holds for a polyhedral graph if and only if 
for any edge $\ue$ of the graph the conformal angle $\theta(\ue)$ is non-negative
and bounded between $0 \leq \theta(\ue)<\pi/2$.
\end{Rem}

\begin{Rem}
\label{Rem9}
In a weak Delaunay graph, an edge $\ue$ is a chord iff
$$\theta_\mathrm{n}(\vec{\ue}) +\theta_\mathrm{s}(\vec{\ue})= 2\, \theta(\ue)=0\ ,$$
i.e. iff the conformal angle $\theta(\ue)$ vanishes.
However in this case, the north and south angles $\theta_\mathrm{n}(\vec{\ue})$
and $\theta_\mathrm{s}(\vec{\ue})$ need not to be both zero.
The special case
$$\theta_\mathrm{n}(\vec{\ue}) =\theta_\mathrm{s}(\vec{\ue})= \theta(\ue)=0$$
occurs only if the edge $\ue$ is a diameter of the circumcircle of the 
cyclic quadrilateral $(\uu, \us, \uv, \un)$
\end{Rem}

\begin{Rem}
Note that $\pi-2\theta(\ue)$ is the intersection angle between the c.w. oriented north and south circumcircles 
$C_\mathrm{n}$ and $C_\mathrm{s}$. In a polyhedral graph we have $\theta(\ue)>0$ iff $z(\un)$ lies outside 
the circumdisk of $C_\mathrm{s}$, 
or equivalently iff $z(\us)$ lies outside the circumdisk of $C_\mathrm{n}$.
\end{Rem}

\subsubsection{Some properties}
Regular graphs will be useful when discussing rhombic graphs (following Kenyon's treatment, see \cite{Kenyon2002}) 
in section \ref{s3RhombicG}, thanks to the following simple result.

\begin{lemma}
\label{chordless-lemma}
Let $\uG_\mathrm{cr}$ be a planar, isoradial Delaunay graph with common circumradius $R_\mathrm{cr}$.
Then $\uG_\mathrm{cr}$ is regular.
\end{lemma}

\begin{proof}
Suppose by contradiction there exists an irregular face $\uf \in \mathrm{F}(\uG_\mathrm{cr})$.
There exists an edge $\ue \in \partial \uf$ with an orientation $\vec{\ue}$ 
such that $\uf = \uf_\mathrm{s}$ and such that {the} 
face $\uf_\mathrm{s}$ is contained in the intersection of the disks of circles $C_\mathrm{s}$ and 
$C$,
where $C$ is the circle of radius $R_\mathrm{cr}$ obtained by  
reflecting $C_\mathrm{s}$ about the line determined by the edge $\ue$.
In virtue of
isoradiality, the vertices $\uv \in \partial \uf_\mathrm{n}$ with $\uv \notin \partial \ue$
must all lie either (1) on the portion of the circle $C$ 
residing in the interior of the disk of circumcircle $C_\uf$
or else (2) on
the circumcircle $C_\mathrm{s}$.
Case (1) is impossible 
because then any vertex $\uv$ of this kind
would violate the Delaunay property with 
respect to the face $\uf_\mathrm{s}$
because edge $\ue$ would form a chord between faces $\mathtt{\uf_\mathrm{n}}$ and $\uf_\mathrm{s}$.
Likewise case (2) is impossible
because edge $\ue$ would form a chord between faces $\uf_\mathrm{n}$ and $\uf_\mathrm{s}$.
So $\uG_\mathrm{cr}$ must be regular.
\end{proof}

\begin{figure}[h]
\begin{center}
\raisebox{-.05in}{\includegraphics[width=4in]{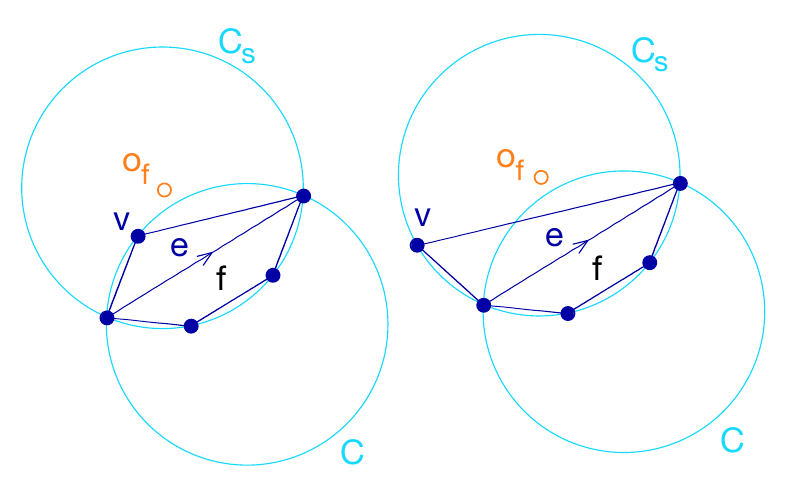}}
\caption{Cases (1) and (2) in the proof of Lemma \ref{chordless-lemma}}
\label{chordless}
\end{center}
\end{figure}

\begin{Cor}
\label{CorChordless}
By Lemma \ref{chordless-lemma} the redacted graph $\uG^\bullet$ is an isoradial, regular Delaunay graph whenever $\uG$ is isoradial and weakly Delaunay.
\end{Cor}

\subsubsection{{Rhombic graphs and abstract rhombic surfaces}}
\label{s3RhombicG} 
We now consider the bipartite kite graph built from the vertices and the face centers of a Delaunay graph, as well as the associated concept of rhombic surface. 
\begin{Def}[\textbf{Kite graphs $\uG^\lozenge$}]
For a Delaunay graph $\uG$ let 
$\uG^\lozenge$ denote the bipartite graph whose vertex set consists of all vertices $\uv$ of $\uG$ (the black vertices $\bullet$)
together with all circumcenters $\mathtt{o}_\uf$ of faces $\uf$ of $\uG$ (the white vertices $\circ$), and whose edges correspond precisely to those
pairs $\{\uv, \mathtt{o}_\uf \}$ for which $\uv \in \partial \uf$. We extend the embedding $z$ of $\uG$ to
$\uG^\lozenge$ by setting $z(\uo_\uf) := z(\uf)$ for each face $\uf \in \mathrm{F}(\uG)$ where 
\begin{equation}
\label{ }
z(\uf) \, := \ {1 \over {4\mathrm{i}}} \,
{|z(\uu)|^2(z(\uv) - z(\uw)) + |z(\uv)|^2(z(\uw) - z(\uu)) + |z(\uw)|^2(z(\uu) - z(\uv)) 
\over {z(\uv)\overline{z}(\uu) - z(\uu)\overline{z}(\uv) + z(\uw)\overline{z}(\uv) - z(\uv)\overline{z}(\uw) +
z(\uu)\overline{z}(\uw) - z(\uw)\overline{z}(\uu)   }}
\end{equation}
is the complex coordinate of the circumcenter of the face $\uf \in \mathrm{F}(\uG)$
with $\uu, \uv, \uw \in \partial \uf$ any choice of three vertices appearing
in counter-clockwise order.
As constructed, each face of the graph $\uG^\lozenge$ 
is quadrilateral (in fact a \emph{kite})
$\lozenge(\overline{\uu\uv}) = (\uu,\mathtt{o}_\us, \uv, \mathtt{o}_\un )$
corresponding to a unique unoriented edge 
$\overline{\mathtt{uv}}$ of the graph. 
\end{Def}

\begin{Rem}
\label{rIsoChless=Rhombic}
For any weak Delaunay graph $\uG$ we define $\uG^\lozenge := (\uG^\bullet)^\lozenge$. 
Clearly $\uG_1^\lozenge = \uG_2^\lozenge$ if and only if
$\uG_1^\bullet= \uG_2^\bullet$ for any two weak Delaunay graphs $\uG_1$ and $\uG_2$.
\end{Rem}

\begin{Def}[\textbf{Rhombic surface $\uS_\uG^\lozenge$}]
\label{dRhomSurf}
Following \cite{DavidEynard2014}, a rhombic surface $\uS_\uG^\lozenge$ can be constructed from a Delaunay {graph} $\uG$ in the following way: 
assign to each unoriented edge $\ue= \overline{\uu \uv}$ a rhombus 
$\lozenge(\ue) = \tilde{\uu} \tilde{\uo}_\mathrm{s} \tilde{\uv}\tilde{\uo}_\mathrm{n}$ 
with unit edge lengths $\ell =1$ and rhombus angle $\angle \tilde{\uo}_\mathrm{s} \tilde{\uu} \tilde{\uo}_\mathrm{n}= 2 \theta(\ue)$
as depicted in fig.~\ref{unit-rhombus}. 

If two edges $\ue_1$ and $\ue_2$ of the graph {share a common vertex and simultaneously} belong to a common face, then
rhombi $\lozenge (\ue_1)$ and $\lozenge(\ue_2)$ are glued together along their common edge.
In this way, we obtain an abstract rhombic surface $\uS_\uG^\lozenge$. 
\begin{figure}[h]
\begin{center}
\raisebox{0.cm}{\includegraphics[width=2.5in]{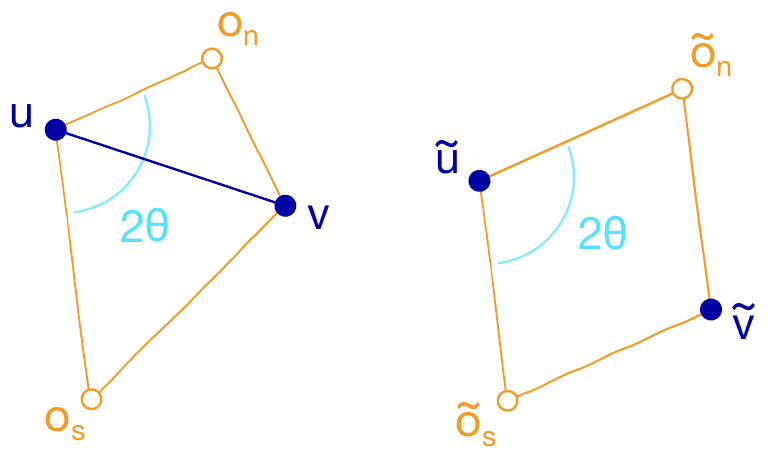}}
\caption{An edge $\ue=\uu\uv$ of $\uG$ and the associated kite in the plane (left), and the associated rhombus $\lozenge(\ue)$ of $\uS_\uG^\lozenge$ (right)}
\label{unit-rhombus}
\end{center}
\end{figure}
\end{Def}

A simple example is depicted in the figure~\ref{fGraphKiteSurf} below. 
In this example an explicit isometric embedding {in $\mathbb{R}^3$} as a tesselated rhombic surface is possible.
(a) is a piece of an Delaunay graph $\uG$, in blue, with the kites associated to each edge (in orange); (b) is the associated kite graph $\uG^\lozenge$ (in orange). (c) is an isometric embedding in $\mathbb{R}^3$ of the associated rhombic surface $\uS_\uG^\lozenge$. In this particular example, the conformal angles $\theta(\ue)$ for each edge of 
$\uG$ equal $\pi/2$, and so the faces of $\uS_\uG^\lozenge$ are in fact squares, and the embedding (c) is a surface in $\mathbb{Z}^3$. In general the rhombic surface 
$\uS_\uG^\lozenge$ cannot be embedded isometrically and rigidly into $\mathbb{R}^3$.

\begin{figure}[h!]
\begin{center}
\includegraphics[width=2.in]{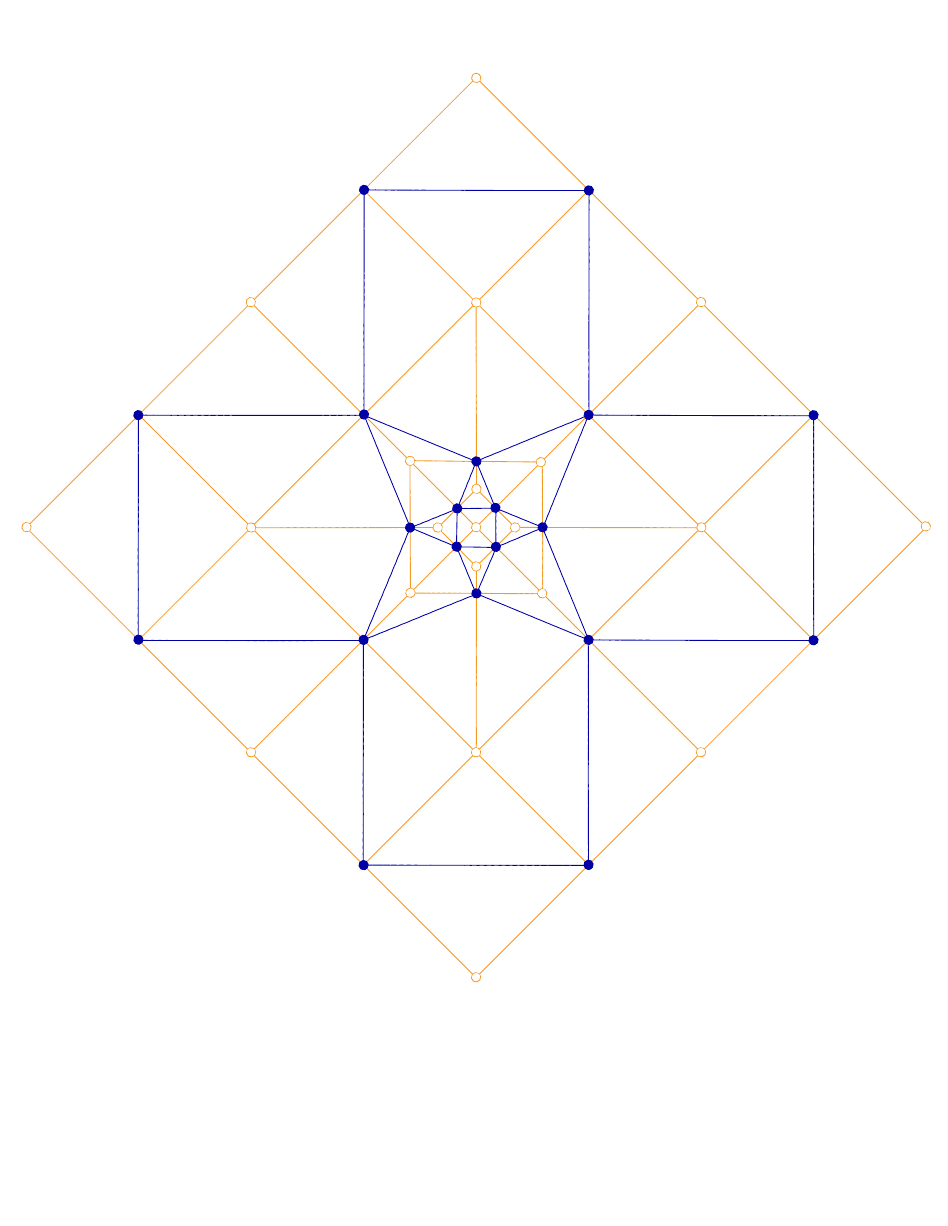}\qquad
\includegraphics[width=2.in]{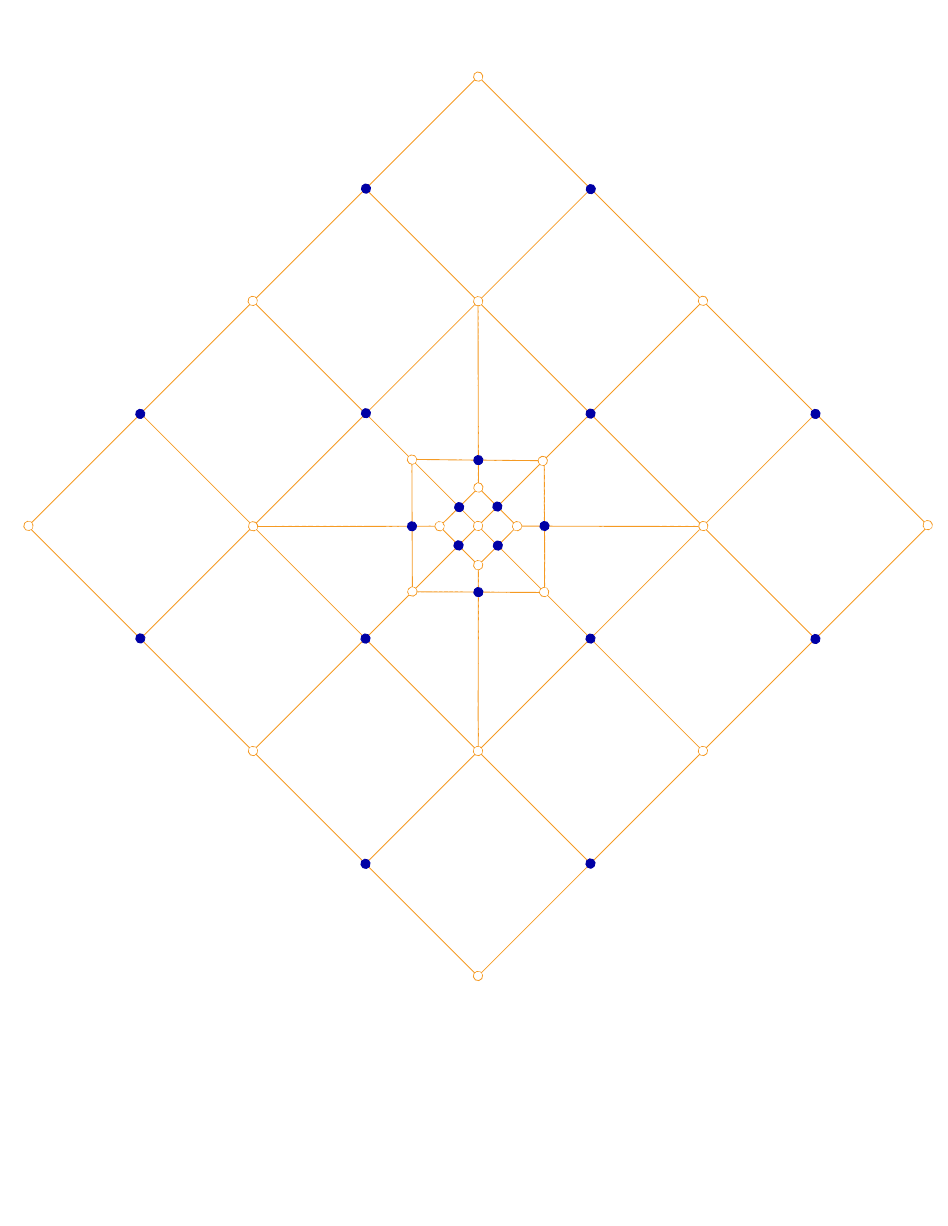}\\
\vskip -.5in (a)\hskip 2.in (b)\\
\vskip -.2 in 
\includegraphics[width=2.5in]{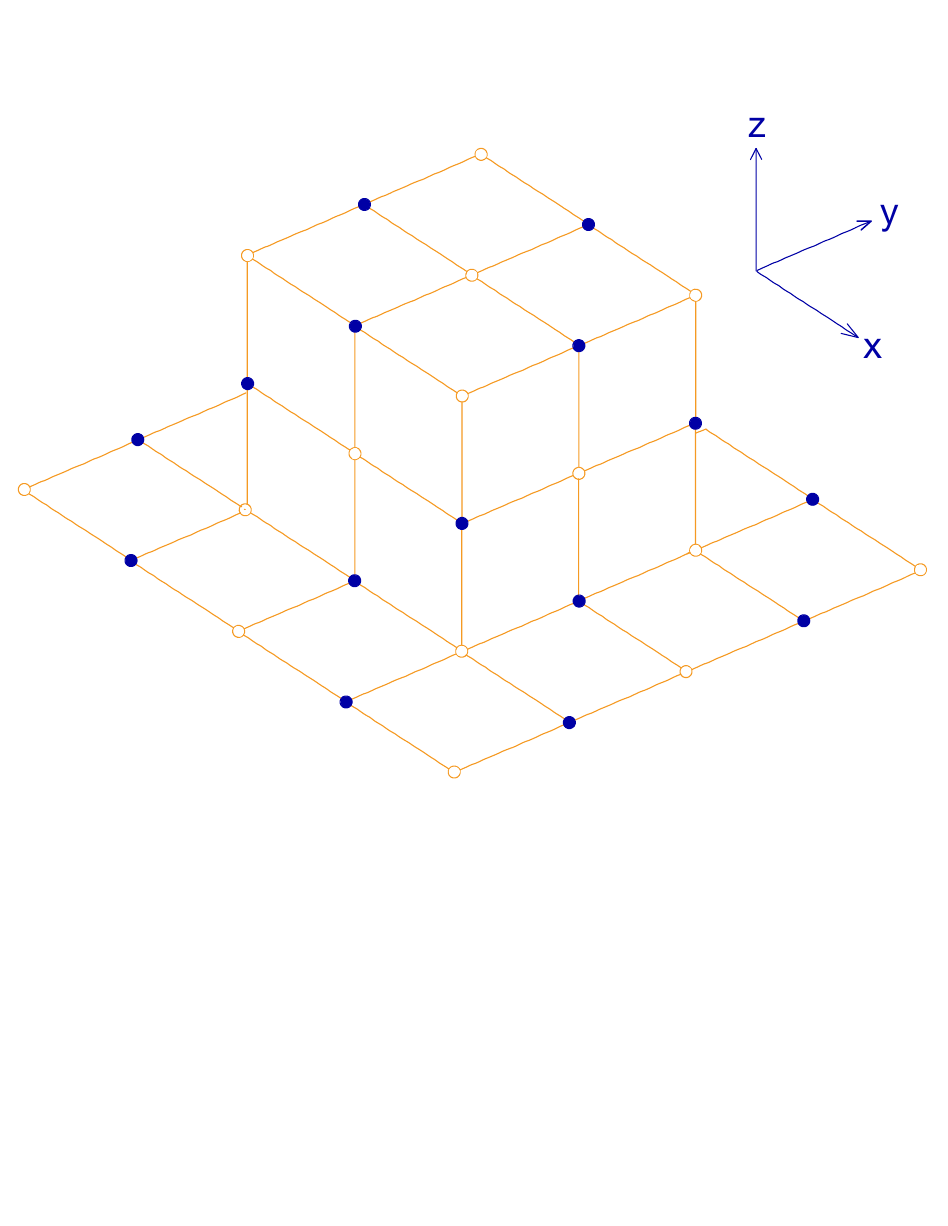}\\
\vskip -.80in(c)
\caption{An example (in blue) of a Delaunay graph $\uG$ (a), of the associated kite graph $\uG^\lozenge$ (b), and of the rhombic surface $\uS_\uG^\lozenge$, here represented as consisting of squares embedding in $\mathbb{R}^3$ (c). The curvature $K$ is localized on the $\circ$ vertices with 3 neighbors $\bullet$ vertices (positive $K$) or 5 neighbors $\bullet$ vertices (negative $K$) .}
\label{fGraphKiteSurf}
\end{center}
\end{figure}

A rhombic surface is flat at each vertex $\tilde{\uu}$ associated to a vertex $\uu$ of $\uG$ but has 
a potential curvature defect at each vertex $\tilde{\uo}_{\uf}$ 
corresponding to a circumcenters $\uo_{\uf}$ of a face $\uf$ of $\uG$, 
with scalar (Ricci) curvature $R_{\scriptscriptstyle{\mathrm{scal}}}$ defined by
\begin{equation}
\label{scalar-curvature}
R_{\scriptscriptstyle{\mathrm{scal}}}(\tilde\uo_\uf) := \ 4\pi - 2 \sum\limits_{ \ue \in\partial \uf} \, \big(\pi - 2 \theta(\ue) \big)
\end{equation}
If $R_{\scriptscriptstyle{\mathrm{scal}}}(\tilde\uo_\uf)=0$ for every {face $\uf$ of the graph}, $\uG$ is said to be \emph{flat}. 
It is easy to see that this is equivalent to saying that the {Delaunay graph} is \emph{isoradial}, namely that all circumradii are equal to some $R$. 
Note that for every oriented edge $\vec{\ue}$ of an {isoradial, polyhedral graph} either $\theta_\mathrm{n} (\vec{\ue} \,)
=\theta_\mathrm{s}( \vec{\ue} \,)= \theta( \ue)>0$ or $\theta_\mathrm{n}(\vec{\ue} \,)= -\theta_\mathrm{s}(\vec{\ue} \,)$ 
in which case $\theta ( \ue)=0$.

When $\uG$ is isoradial (with common circumradius $R$) 
each kite $\lozenge(\overline{\mathtt{uv}})$ will be a rhombus with side length $R$; in this case
we shall refer to $\uG^\lozenge$ as a \textbf{rhombic graph}. Up to a global rescaling $R \to 1$
we have $\uG^\lozenge=\uS_\uG^\lozenge$.
This corresponds to the rhombic graphs discussed in \cite{Kenyon2002}.
\begin{Rem}
\label{DelaunayVsRhombic}
Isoradial Delaunay graphs are in bijection with the rhombic graphs of \cite{Kenyon2002}.
\end{Rem}

\begin{figure}[h]
\begin{center}
\raisebox{-.8in}{\includegraphics[width=2in]{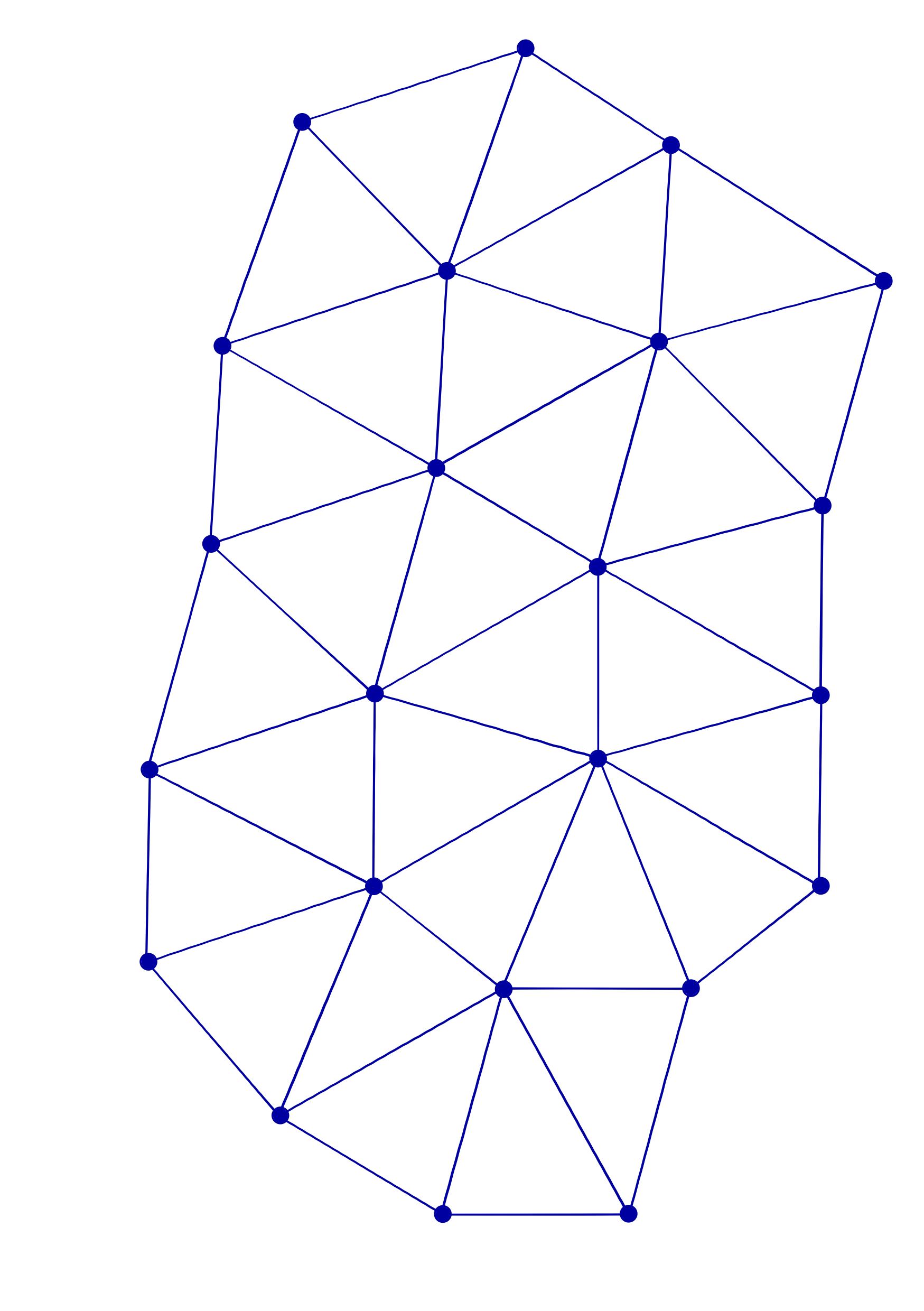}}
\raisebox{-.8in}{\includegraphics[width=2in]{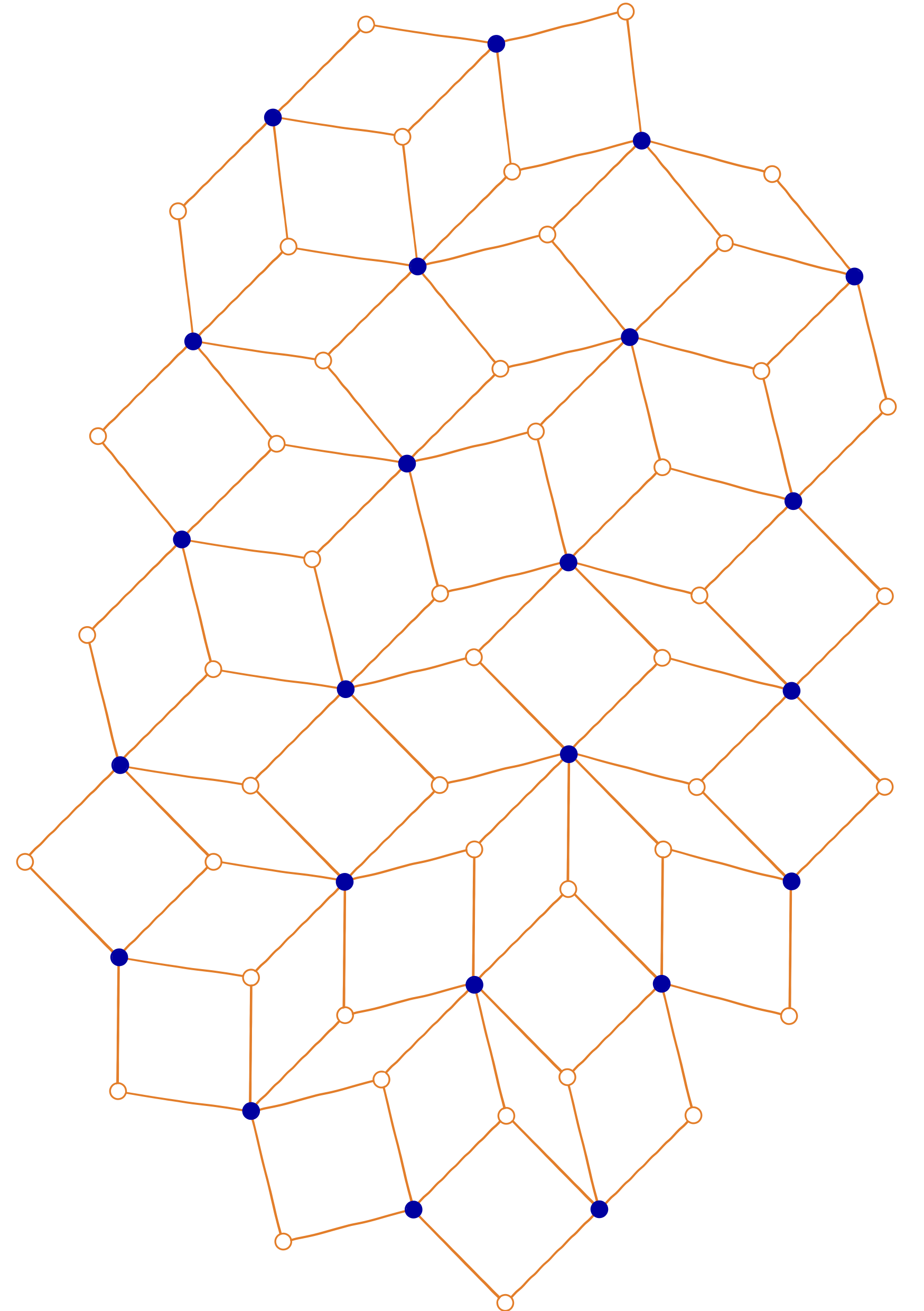}}
\caption{Fragments of an isoradial Delaunay graph
$\uG_\mathrm{cr}$ (on the left) and its rhombic graph $\uG_\mathrm{cr}^\lozenge$ (on the right). }
\label{rhombic-graph}
\end{center}
\end{figure}

\subsection{Geometry on rhombic graphs}
\label{ssRhomGraph}
In the following discussion $\uG_\mathrm{cr}$ will be an isoradial Delaunay graph
with embedding $z_\mathrm{cr}:\mathrm{V}(\uG_\mathrm{cr}) \longrightarrow \Bbb{C}$
and, if not specified otherwise, we shall assume for simplicity that 
the value of the common circumradius is $R_\mathrm{cr} = 1$.
Let us recall some geometrical concepts of \cite{Kenyon2002} and \cite{Kenyon14rhombicembeddings}, with some more material needed in this paper.

\subsubsection{Paths on rhombic graphs}
\label{sssPath}
A path in $\uG_\mathrm{cr}^\lozenge$ is a finite sequence of vertices $\mathbbm{v} =
\big(\uv_0, \dots, \uv_k\big)$ such that for each $1 \leq j \leq k$
 the vertices $\uv_{j-1}$  and $\uv_j$ are joined by an edge $\ue_j$ of $\uG_\mathrm{cr}^\lozenge$; 
in this case we say $\mathbbm{v}$ is a path of length $k$ from $\uv_0$ to $\uv_k$.
Let $\vec{\ue}_j=(\uv_{j-1},\uv_j)$ be the oriented edge corresponding to $\ue_j$, let $\vec{\mathbb{E}}(\mathbbm{v})=(\vec{\ue}_1,\cdots,\vec{\ue}_k)$ be the sequence of oriented edges of $\mathbbm{v}$, and $E \big( \mathbbm{v} \big)=\bigcup_{j}\{\ue_j\}$ the set of edges of $\mathbbm{v}$.
To each edge $\vec{\ue}_j$ of $\mathbbm{v}$ is associated a phase
$e^{\mathrm{i}\theta_j} := z_\mathrm{cr}(\uv_j)  - z_\mathrm{cr}(\uv_{j-1})$. We denote by
$\underline{\theta}(\mathbbm{v})=\big( \theta_1, \dots , \theta_k \big)$ the sequence of angles of these phases.

\begin{figure}[h]
\begin{center}
\raisebox{-.8in}{\includegraphics[width=2.5in]{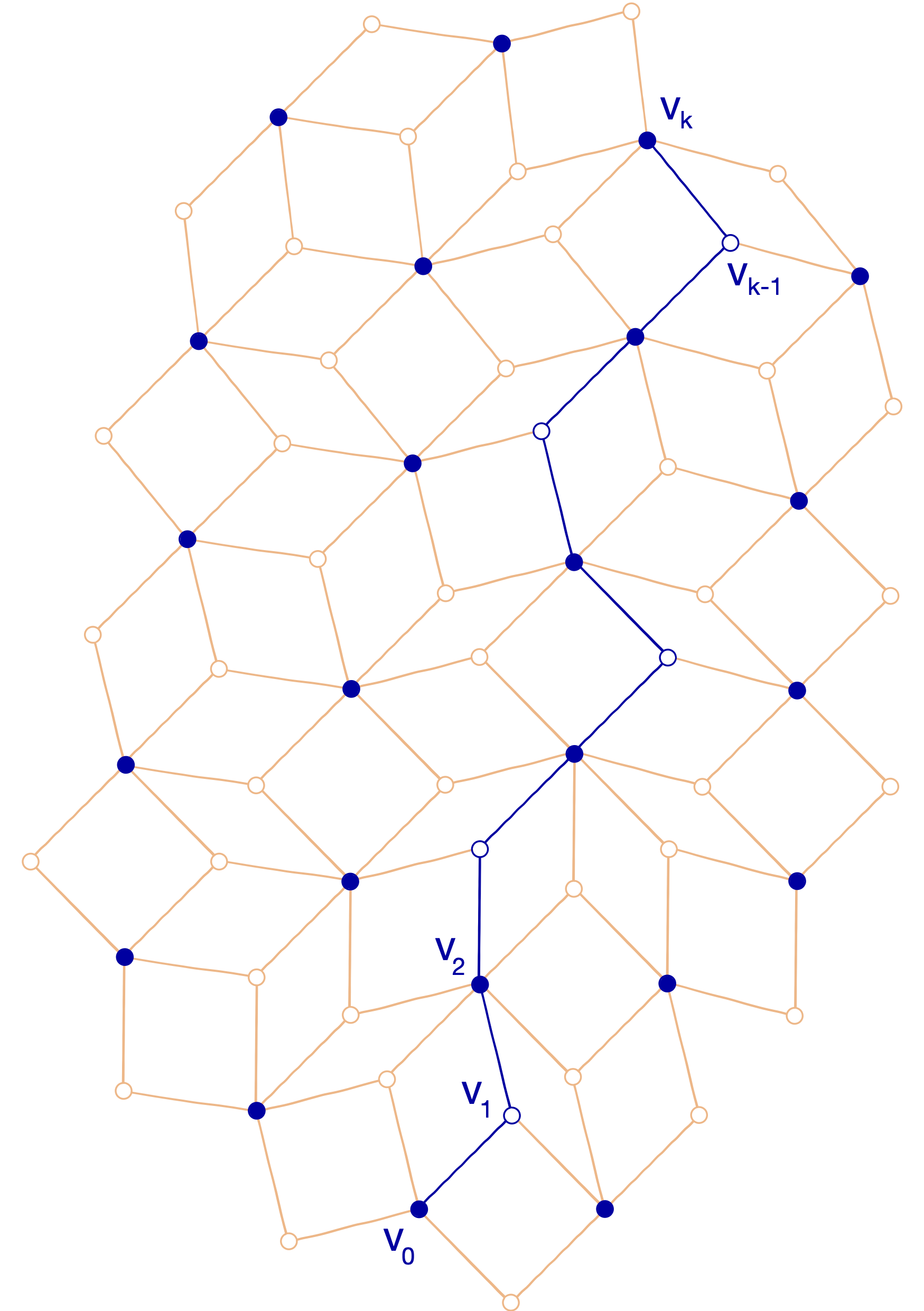}}
\caption{Path $\mathbbm{v}= \big(\uv_0, \dots, \uv_k \big)$ in the rhombic graph $\uG_\mathrm{cr}^{\lozenge}$}
\label{rhombic-path}
\end{center}
\end{figure}

We can regard the rhombic graph $\uG_\mathrm{cr}^\lozenge$ as a cellular decomposition of the plane;
accordingly vertices, oriented edges, and oriented faces of $\uG_\mathrm{cr}^\lozenge$ can be viewed respectively 
as $0$, $1$, and $2$-chains of a cellular complex $\mathcal{X}$ with $\Bbb{Z}$-coefficients. 
For a path $\mathbbm{v}$, let $\vec{\mathbbm{v}}$ denote the $1$-chain 
$\vec{\ue}_1 + \cdots + \vec{\ue}_k$ 
in $\mathrm{C}_1\big(\mathcal{X}; \Bbb{Z} \big)$.
Two paths  $\mathbbm{v}_1$ and $\mathbbm{v}_2$ are said to 
differ by an oriented rhomb $\lozenge^\circlearrowleft$ 
if $\vec{\mathbbm{v}}_2= \vec{\mathbbm{v}}_1 + \partial \lozenge^\circlearrowleft$;
see figure \ref{rhombic-complement} for an example. The vanishing of 
$\mathrm{H}_1\big(\mathcal{X}; \Bbb{Z} \big)$
is equivalent to the fact that any two paths $\vec{\mathbbm{v}}_1$ and $\vec{\mathbbm{v}}_2$
both from a vertex $\uu$ 
to a vertex $\uv$ must differ by a sum of oriented rhombs.

\begin{figure}[h]
\begin{center}
\raisebox{-.1in}{\includegraphics[width=2in]{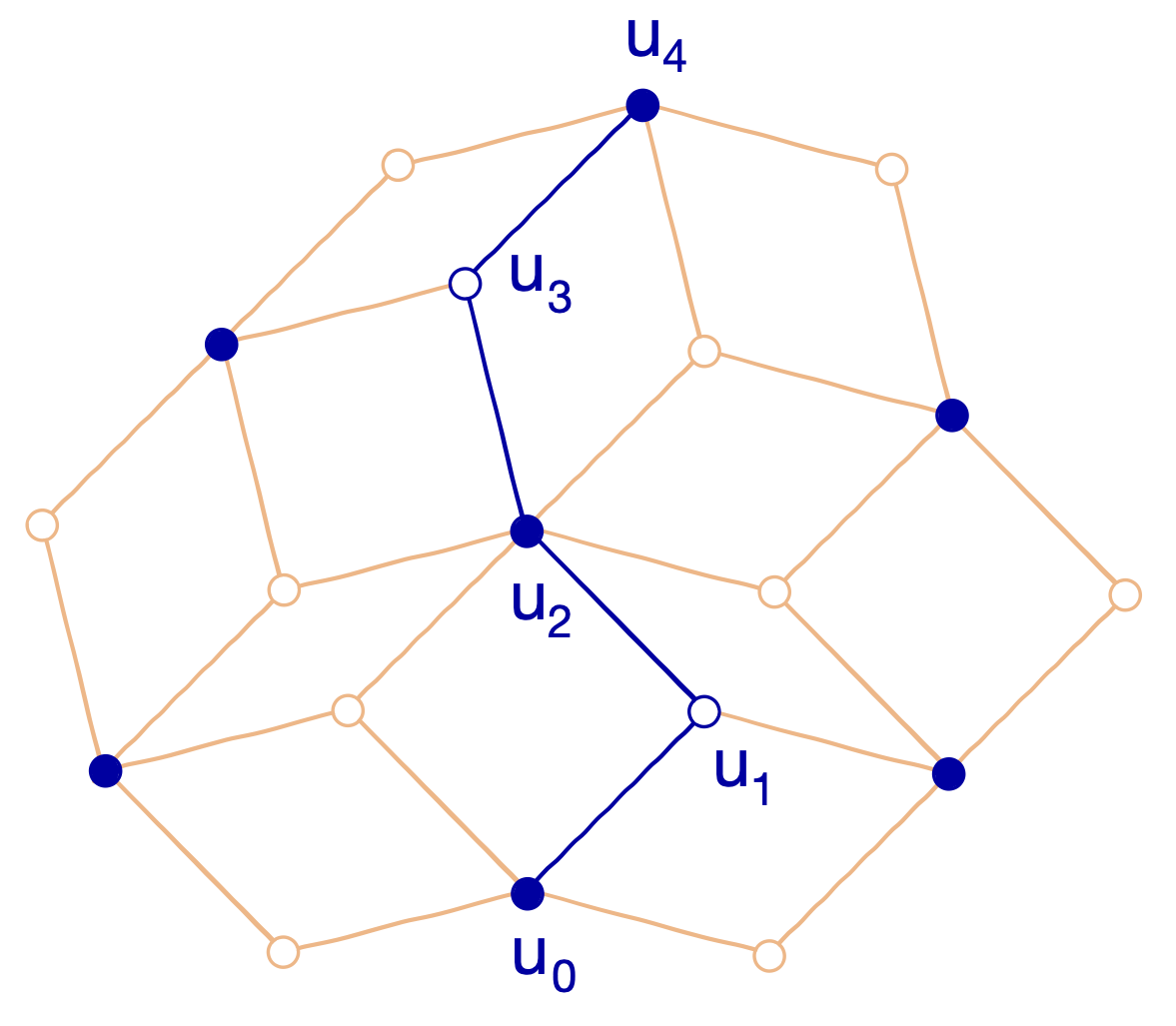}}
\raisebox{-.1in}{\includegraphics[width=2in]{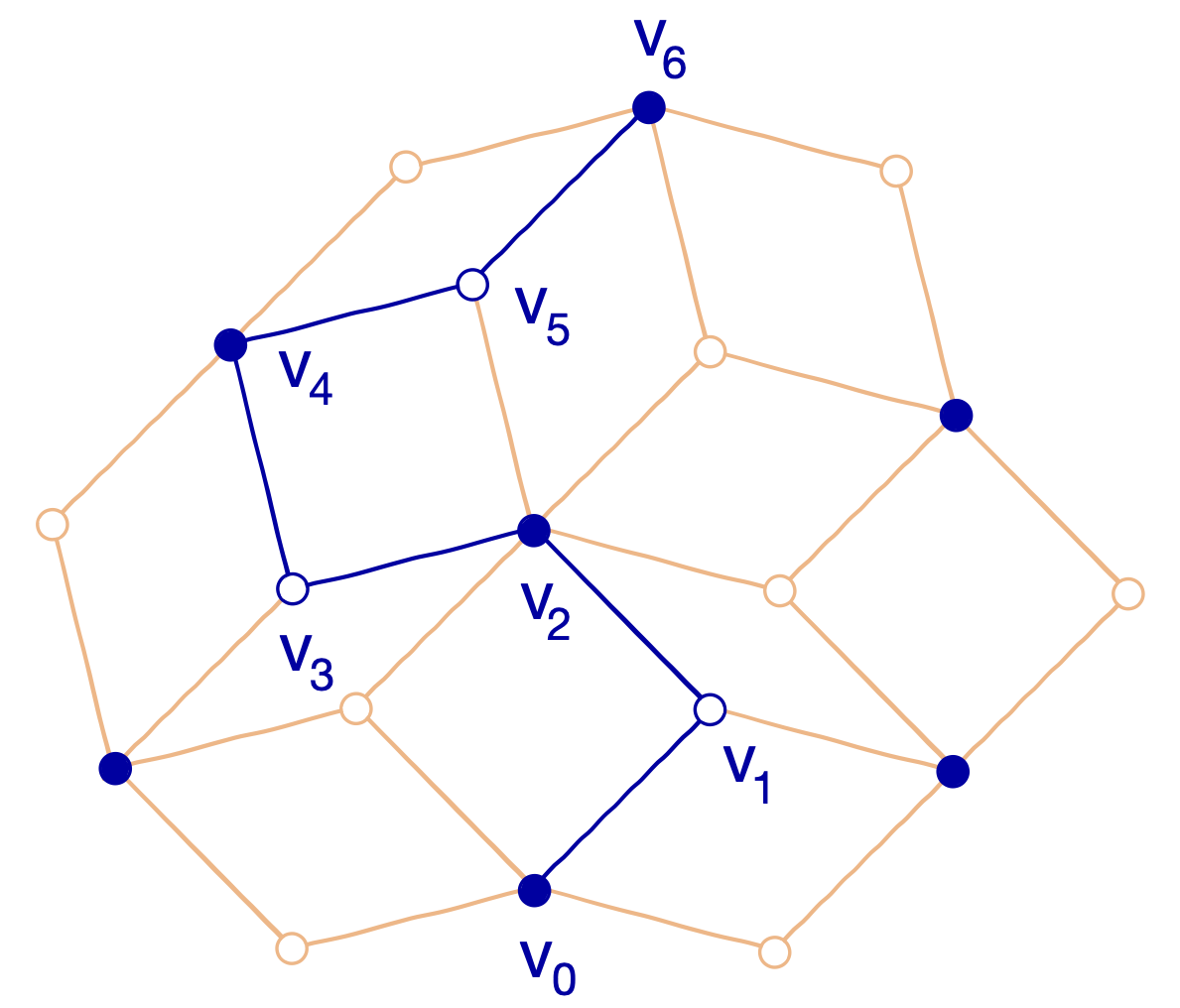}}
\caption{Paths $\mathbbm{u}= \big(\uu_0, \dots, \uu_4 \big)$ 
and $\mathbbm{v}= \big(\uv_0, \dots , \uv_6 \big)$ 
differ by a rhomb.}
\label{rhombic-complement}
\end{center}
\end{figure}

\bigskip
\noindent
For an integer $n$ together with an oriented edge $\vec{\ue}$ joining a vertex $\uu$ to a vertex
$\uv$ set $\lfloor \vec{\ue}  \, \rfloor_n := e^{\mathrm{i} n \theta}$ where $e^{\mathrm{i}\theta} = z_\mathrm{cr}(\uv) - z_\mathrm{cr}(\uu)$
is the phase of the difference of the coordinates of the vertices; extend this by linearity to 1-chains in $\mathrm{C}_1\big(\mathcal{X}; \Bbb{Z} \big)$, and thus define
$\left\lfloor \sum_j  a_j \vec{\ue}_j \right\rfloor_n := \sum_j a_j \, \lfloor \vec{\ue}_j  \rfloor_n$. 
Notice that
$\left\lfloor \lozenge^\circlearrowleft \right\rfloor_n = 0$ for any oriented rhomb $\lozenge^\circlearrowleft$ whenever $n$ is
an odd integer. 
It follows that for any path $\mathbbm{v}$, and for any odd integer $n=2d+1$, $\left\lfloor \vec{\mathbbm{v}} \, \right\rfloor_n$ depends only on the two end-points $( \uv_0,\uv_k )$ of $\mathbbm{v}$.

\begin{Def}
\label{def-p_n} For any pair of vertices $\uu$ and $\uv$ of $\uG_\mathrm{cr}^\lozenge$ and for any odd
integer $n=2d+1$, we define $p_n(\uu,\uv ) := \left\lfloor \vec{\mathbbm{v}} \right\rfloor_n$ where $\mathbbm{v}$
is any path from $\uu$ to $\uv$.
\end{Def}
\bigskip
\noindent
Note that $p_1(\uu,\uv)= z_\mathrm{cr}(\uv) - z_\mathrm{cr}(\uu)$. In addition $p_n(\uu,\uv)= - p_n(\uv,\uu)$.

\subsubsection{Train-tracks}
\label{sssTracks}
\begin{Def} [train-track]
\label{train-track}
A \textbf{train-track} in the rhombic graph $\uG_\mathrm{cr}^\lozenge$
is an infinite sequence of rhombs $\mathbbm{t} = \big( \lozenge_n \, \big| \, n \in \Bbb{Z} \big) $ 
whose consecutive rhombs $\lozenge_n$ and $\lozenge_{n+1}$ are incident along
a common edge $\ue_n$ for each $n \in \Bbb{Z}$
and for which the edges $\ue_n$ and $\ue_{n+1}$ are parallel
for each $n \in \Bbb{Z}$. 
We shall denote these parallel edges ``train-track tie'', or in short ``\textbf{tie}''. 
We consider train-tracks up to shift and inversion, i.e. 
$\mathbbm{t}^{\scriptscriptstyle (1)} = \big\{ \lozenge^{\scriptscriptstyle (1)}_n\, \big| \, n \in \Bbb{Z} \big\}$
is equivalent to 
$\mathbbm{t}^{\scriptscriptstyle (2)} = \big\{ \lozenge^{\scriptscriptstyle (2)}_n\, \big| \, n \in \Bbb{Z} \big\}$ 
if $\lozenge^{\scriptscriptstyle (2)}_n = \lozenge^{\scriptscriptstyle (1)}_{\pm n+d}$ for some $d \in \Bbb{Z}$.
Let $\mathtt{Ties}(\mathbbm{t}) = \big\{ \ue_n \, \big| \, n \in \Bbb{Z} \big\}$ 
denote this set of edges. A train-track $\mathbbm{t}$ has \textbf{inclination} $\theta_\mathbbm{t}
\in [0, \pi)$ if the ties $\ue_n$ are parallel to the phase $\exp \big( \mathrm{i}\theta_\mathbbm{t} \big)$.
\end{Def}

\begin{figure}[h]
\raisebox{-.8in}{\includegraphics[width=2in]{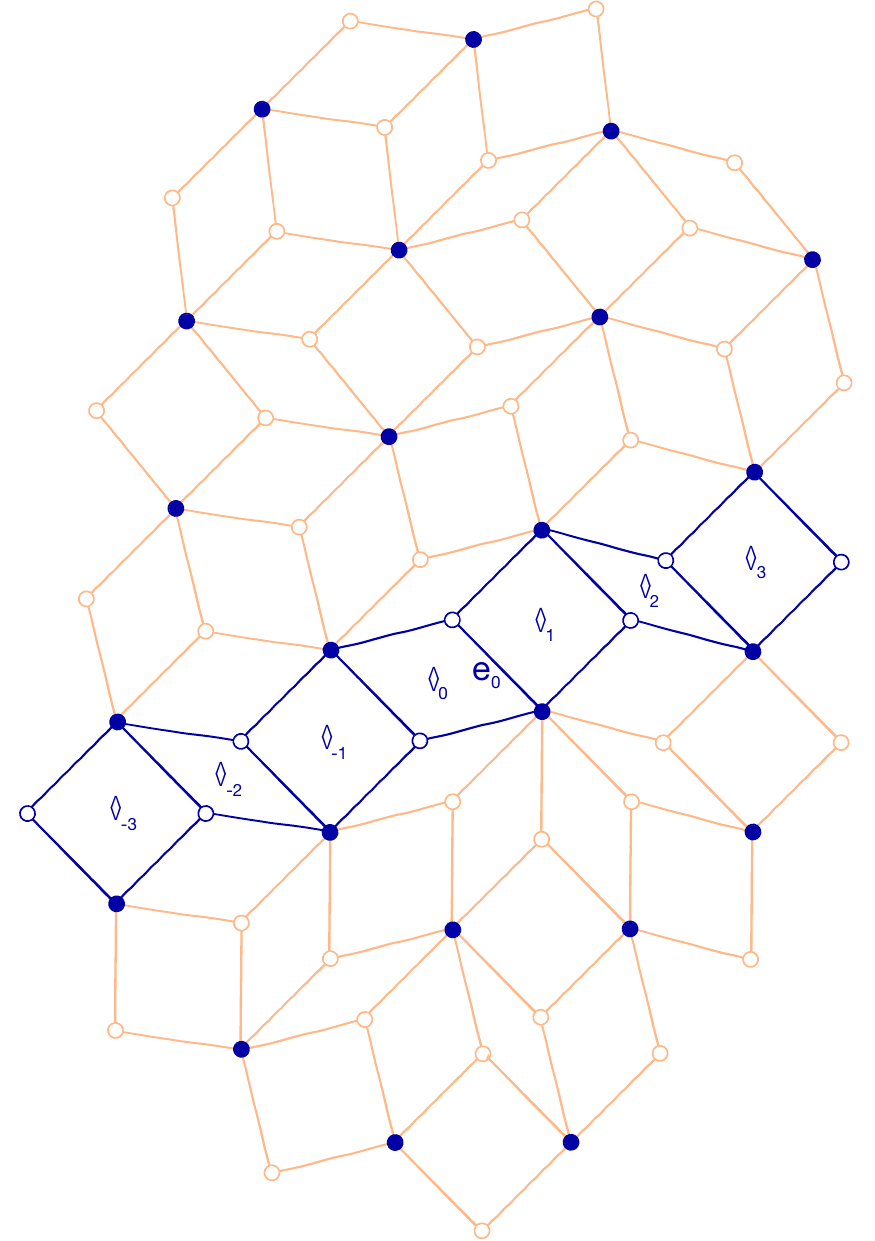}}
\caption{\label{train-track}Train-track $\mathbbm{t}$.}
\end{figure}

Clearly any train-track is determined (but not uniquely) by an initial rhomb $\lozenge_0$ 
together with a choice of one of its edges $\ue_0$.
For any choice of initial edge
$\ue_0$ in $\mathbbm{t}$ the distance of each edge $\ue_n$ from 
the axis determined by $\ue_0$ is monotonically increasing with $n$; i.e. the 
train-track must move forward in the axis perpendicular to $\ue_0$.

We say two train-tracks
$\mathbbm{t}^{\scriptscriptstyle (1)} = \big\{ \lozenge^{\scriptscriptstyle (1)}_n \, \big| \, n \in \Bbb{Z} \big\}$ and 
$\mathbbm{t}^{\scriptscriptstyle (2)} = \big\{ \lozenge^{\scriptscriptstyle (2)}_n\, \big| \, n \in \Bbb{Z} \big\}$ \textbf{intersect} if 
$\lozenge^{\scriptscriptstyle (1)}_m = \lozenge^{\scriptscriptstyle (2)}_n$
for some $m, n \in \Bbb{Z}$.
Two important features of any rhombic-graph
$\uG_\mathrm{cr}^\lozenge$ are

\begin{Fact} 
No train-track can intersect itself, i.e. if 
$\mathbbm{t} =  \big\{ \lozenge_n \, \big| \, n \in \Bbb{Z} \big\}$ then
$\lozenge_m \ne \lozenge_n$  for all integers $m \ne n$. 
\end{Fact}
\begin{Fact}
Any two distinct train-tracks are either disjoint or else intersect once.
\end{Fact} 

The notion of train-track is amenable to any quad-graph (a planar graph
consisting entirely of quadrilateral faces) 
and these two properties characterize rhombic graphs within the broader class of
quad-graphs; specifically any quad-graph satisfying these two properties is a deformation of a rhombic graph
(see \cite{Kenyon14rhombicembeddings}).

\subsubsection{Intersections of train-tracks with paths}
\label{sssIntersections}
A train-track $\mathbbm{t}$ partitions the vertex set $\mathrm{V}\big(\uG_\mathrm{cr}^\lozenge\big)$ into two disjoint subsets $\mathrm{V}'$ 
and $\mathrm{V}''$.
Specifically, the edge set $E\big(\uG_\mathrm{cr}^\lozenge \big) - \mathtt{Ties}(\mathbbm{t})$ 
defines a disconnected subgraph of $\uG_\mathrm{cr}^\lozenge$ with two disjoint components;
$\mathrm{V}'$ and $\mathrm{V}''$ are the respective vertex sets of these components.
Accordingly, we say that two vertices $\uu$ and $\uv$ are \textbf{separated by} $\mathbbm{t}$ if 
they lie in different components; furthermore we say $\mathbbm{t}$ \textbf{separates} the path
$\mathbbm{v}$ if the end-points of the path $\uv_0$ and $\uv_k$ are separated by $\mathbbm{t}$.

Given a path $\mathbbm{v} =\big(\uv_0, \dots \uv_k \big)$ and a train-track $\mathbbm{t}$
let
$I(\mathbbm{v}; \mathbbm{t}) := \big\{ 1 \le j \le k \, \big| \, \ue_j  \in \mathrm{Ties}(\mathbbm{t}) \big\}$ 
be the set of indices of edges common to both $\mathbbm{v}$ and $\mathbbm{t}$. 
If $\mathbbm{t}$ separates $\mathbbm{v}$ then its cardinality $\big| I(\mathbbm{v}; \mathbbm{t}) \big|$ must be odd
due to the fact the path must begin on one side of $\mathbbm{t}$
and end on the other. 
If, on the other hand, $\mathbbm{t}$ does not separate $\mathbbm{v}$ then $\big| I(\mathbbm{v}; \mathbbm{t}) \big|$ is even (and may in fact
be zero if there is no intersection at all).

The edges $\ue_j$ for $j \in I(\mathbbm{v}; \mathbbm{t})$ are clearly parallel (since they all inhabit the train-track $\mathbbm{t}$)
but the oriented edges $\vec{\ue}_j$ for $j \in I(\mathbbm{v}; \mathbbm{t})$ must alternate in direction 
and so their phases $e^{\mathrm{i}\theta_j}$ for $j \in I(\mathbbm{v}; \mathbbm{t})$  must alternate in sign.
Consequently, if $I(\mathbbm{v}; \mathbbm{t}) = \big\{ j_1 < \cdots < j_d \big\} $ and $n$ is odd then
\begin{equation}
\label{stuff1}
\sum_{s=1}^d  e^{\mathrm{i} n \theta_{j_s} } 
= 
\left\{ 
\begin{array}{ll}
e^{\mathrm{i} n \theta_{j_1} } &\text{whenever $\mathbbm{t}$ separates $\mathbbm{v}$ } \\ 
0 &\text{otherwise}
\end{array}
\right.
\end{equation}
If $\mathbbm{t}$ separates $\mathbbm{v}$ their {\it intersection angle} is defined
as $\vartheta(\mathbbm{v}, \mathbbm{t}) := \theta_{j_1}$ and $\Theta(\mathbbm{v}) =
\big\{ \vartheta(\mathbbm{v},\mathbbm{t}) \, \big| \, \mathbbm{t} \ \text{intersects}\ \mathbbm{v} \big\}$ 
is the set of intersections angles of all train-tracks that separate the path $\mathbbm{v}$. 
Define the multiplicity $m_\vartheta := \big| \big\{ \text{$\mathbbm{t}$ separates
$\mathbbm{v}$} \, \big| \, \vartheta = \vartheta(\mathbbm{v},\mathbbm{t}) \big\}  \big|$
for $\vartheta \in \Theta(  \mathbbm{v} )$. It follows from equation \ref{stuff1} that for odd $n$
\begin{equation}
\label{stuff2}
\D p_n(\uu, \uv)
\, = \ \sum_{j=1}^k e^{\mathrm{i} n \theta_j} 
\, = \sum_{ \stackrel{\scriptstyle \text{train-tracks $\mathbbm{t}$} }{\text{separating $\mathbbm{v}$} } } 
e^{\mathrm{i} n \vartheta (\mathbbm{v},\mathbbm{t}) } 
\, =  \sum_{ \vartheta \in \Theta( \mathbbm{v} ) } \, 
m_\vartheta  e^{\mathrm{i} n \vartheta}
\end{equation}
where $\uv_0 = \uu$ and $\uv_k = \uv$ are the beginning and end points of the path $\mathbbm{v}$.
Given a train-track $\mathbbm{t}$ separating $\mathbbm{v}$ with angle
of intersection $\theta=\theta(\mathbbm{v},\mathbbm{t})$ define 
$R^\uu_\theta = z_\mathrm{cr}(\uu) + \Bbb{R}_{>0} \, e^{\mathrm{i} \theta}$ to be the ray (half-line) starting from $z_\mathrm{cr}(\uu)$ in the direction 
$\theta$, and 
$R^\uv_{\theta + \pi} = z_\mathrm{cr}(\uv) + \Bbb{R}_{>0}\, e^{\mathrm{i}(\theta + \pi)}$ be the ray starting from $z_\mathrm{cr}(\uv)$ in the direction $\theta+\pi$.
It is geometrically clear that $\mathbbm{t}$ must intersect the right hand sides of rays $R^\uu_\theta$ and $R^\uv_{\theta + \pi}$, without back tracking in the direction orthogonal to $R^\uu_\theta $ (and without intersecting the opposite rays $R^\uu_{\theta+\pi}$ and $R^\uv_{\theta}$).
See fig.~\ref{Fseparation}.

\begin{figure}[h!]
\raisebox{-.8in}{\includegraphics[width=2in]{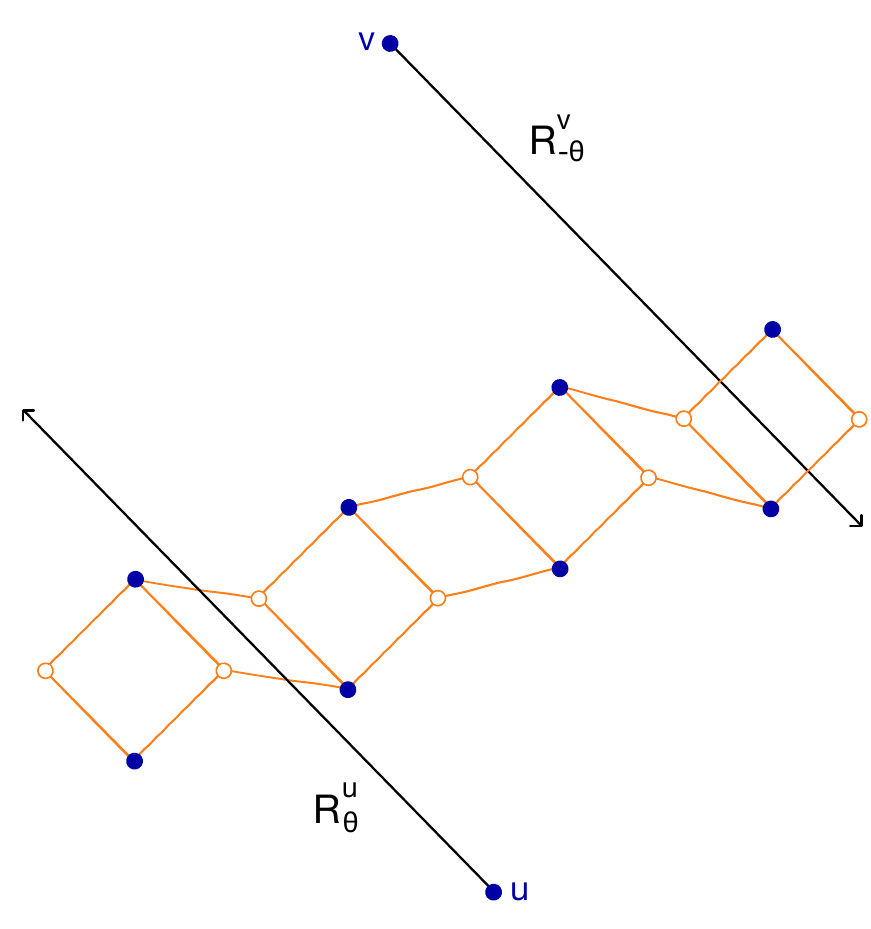}}
\caption{\label{Fseparation}Vertices $\uu$ and $\uv$ separated by a train-track $\mathbbm{t}$.}
\end{figure}

For completeness, one should consider the case where lozenges in $\mathbbm{t}$ become infinitely flat, so that $\mathbbm{t}$ goes to infinity in the $\theta$ direction before intersecting $R^\uu_\theta$ (see Fig.~\ref{asymptotic-train-track-1}).
Then one can consider that $\mathbbm{t}$ crosses $R^\uu_\theta$ at infinity.

\begin{figure}[h]
\raisebox{-.8in}{\includegraphics[width=2in]{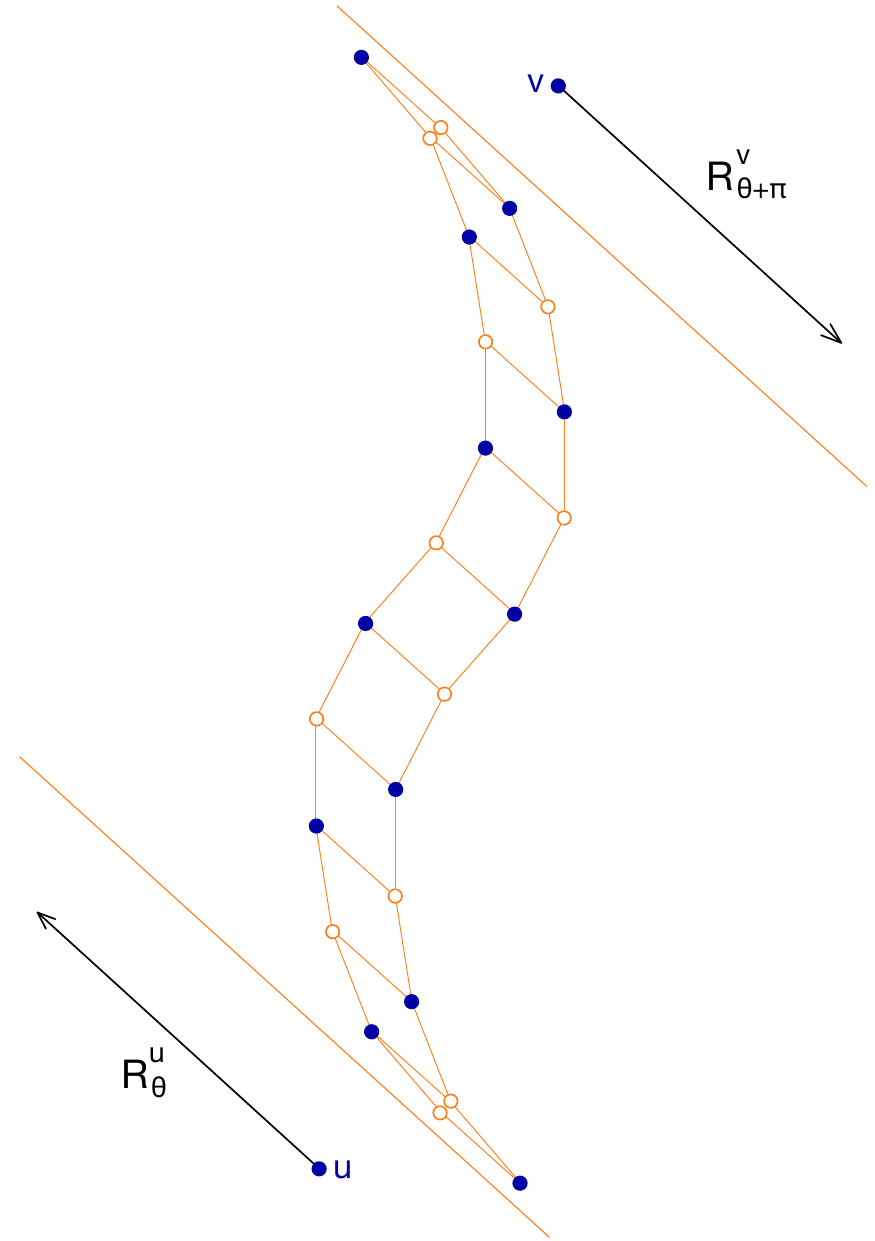}}
\caption{A situation where the vertices $\uu$ and $\uv$ are asymptotically separated by a train-track.}
\label{asymptotic-train-track-1}
\end{figure}

\begin{Prop}
\label{semi-circle-property}
Let $\mathbbm{v}=\big(\uv_0, \dots, \uv_k \big)$ be a path in $\uG_\mathrm{cr}^\lozenge$, let the direction of the path be
$\theta_0 = \arg \big( z_\mathrm{cr}(\uv_k) - z_\mathrm{cr}(\uv_0) \big)$.
Let us fix the determinations of the angles $\vartheta \in\Theta(\mathbbm{v})$ as real numbers in
\begin{equation}
\label{ }
\vartheta \in (\theta_0-\pi,\theta_0+\pi]
\end{equation} and let 
\begin{equation}
\label{ }
\alpha=\max\{\vartheta \in\Theta(\mathbbm{v})\}\ ,\quad \beta=\min\{\vartheta \in\Theta(\mathbbm{v})\}
\end{equation}
Then 
\begin{equation}
\label{ }
\alpha-\beta <\pi\quad\text{and}\quad \beta\le\theta_0\le\alpha
\end{equation}
In other words, the set $\Theta(\mathbbm{v})$ and the angle $\theta_0$ are contained in the open subinterval 
$\big( \theta_\mathbbm{v} - {\pi \over 2} , \theta_\mathbbm{v} + {\pi \over 2} \big) $
where $\theta_\mathbbm{v} = {1 \over 2} (\alpha - \beta )$.
\end{Prop}

\begin{proof}
Set $\theta_0 = \arg \big( z_\mathrm{cr}(\uv_k) - z_\mathrm{cr}(\uv_0) \big) \in \big[0, \pi \big)$.
Each $\vartheta \in \Theta(\mathbbm{v})$ is the intersection angle of at least one 
train-track $\mathbbm{t}$ whose inclination equals
$\vartheta$ (modulo $\pi$) and which separates the vertices $\uv_0$ and $\uv_k$. 

First let us note that the angle $\theta_0 + \pi$ cannot be an element of $\Theta(\mathbbm{v})$. 
Were this the case, there would be train-track joining the righthand sides of the rays $R^\uu_{\theta_0+\pi}$ and $R^\uv_{\theta_0}$ without backtracking. This is impossible, as depicted in figure~\ref{backtracking}.

\begin{figure}[h!]
\includegraphics[width=2in]{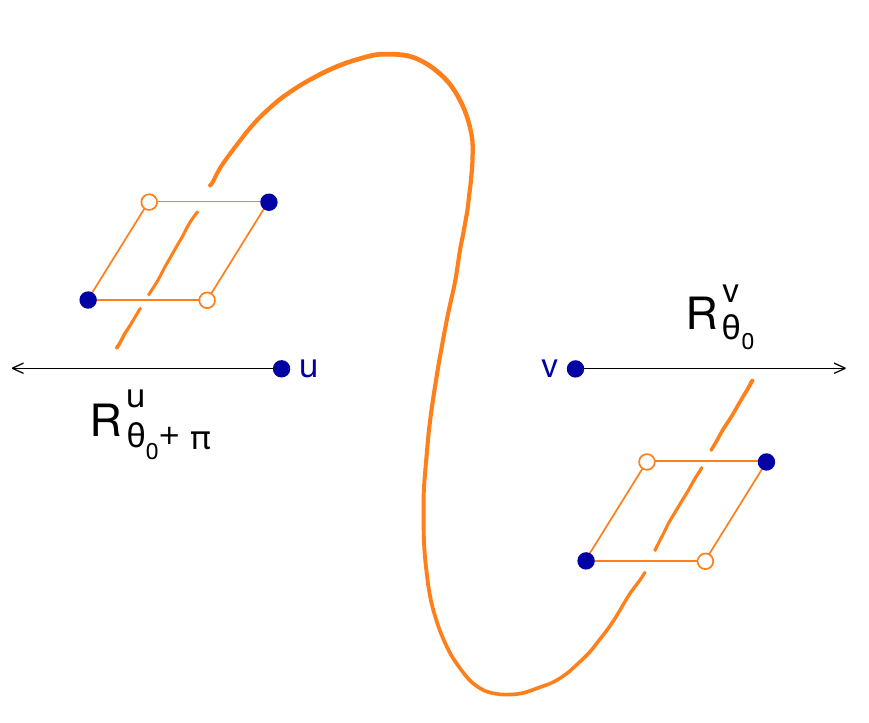}
\caption{\label{backtracking} A track separating $\uu$ and $\uv$ with orientation $\theta_0+\pi$ must backtrack}
\label{ }
\end{figure}

Consequently the angles in $\Theta(\mathbbm{v})$ are in the interval $(\theta_0-\pi, \theta_0+\pi)$. Consider 
$\alpha=\max \Theta(\mathbbm{v})$ and  $\beta=\min \Theta(\mathbbm{v})$. It is enough to prove that $\alpha - \beta \leq \pi$.
Indeed, suppose instead that $\alpha - \beta > \pi$. Both $\alpha$ and $\beta$ are intersection angles for two respective 
train-tracks $\mathbbm{t}_1$ and $\mathbbm{t}_2$ which separate $\uu := \uv_0$ and 
$\uv:= \uv_k$. If we attempt to draw $\mathbbm{t}_1$ and $\mathbbm{t}_2$
bearing in mind monotonicity and their {requisite} intersections with the rays $R_\alpha^\uu$, $R_\beta^\uu$,
$R_{\alpha + \pi}^\uv$, and $R_{\beta + \pi}^\uv$ we will observe that 
the two train-tracks will be forced to intersect at least three times (as depicted on figure~\ref{inconsistent}). 
\begin{figure}[h!]
\raisebox{-.8in}{\includegraphics[width=2in]{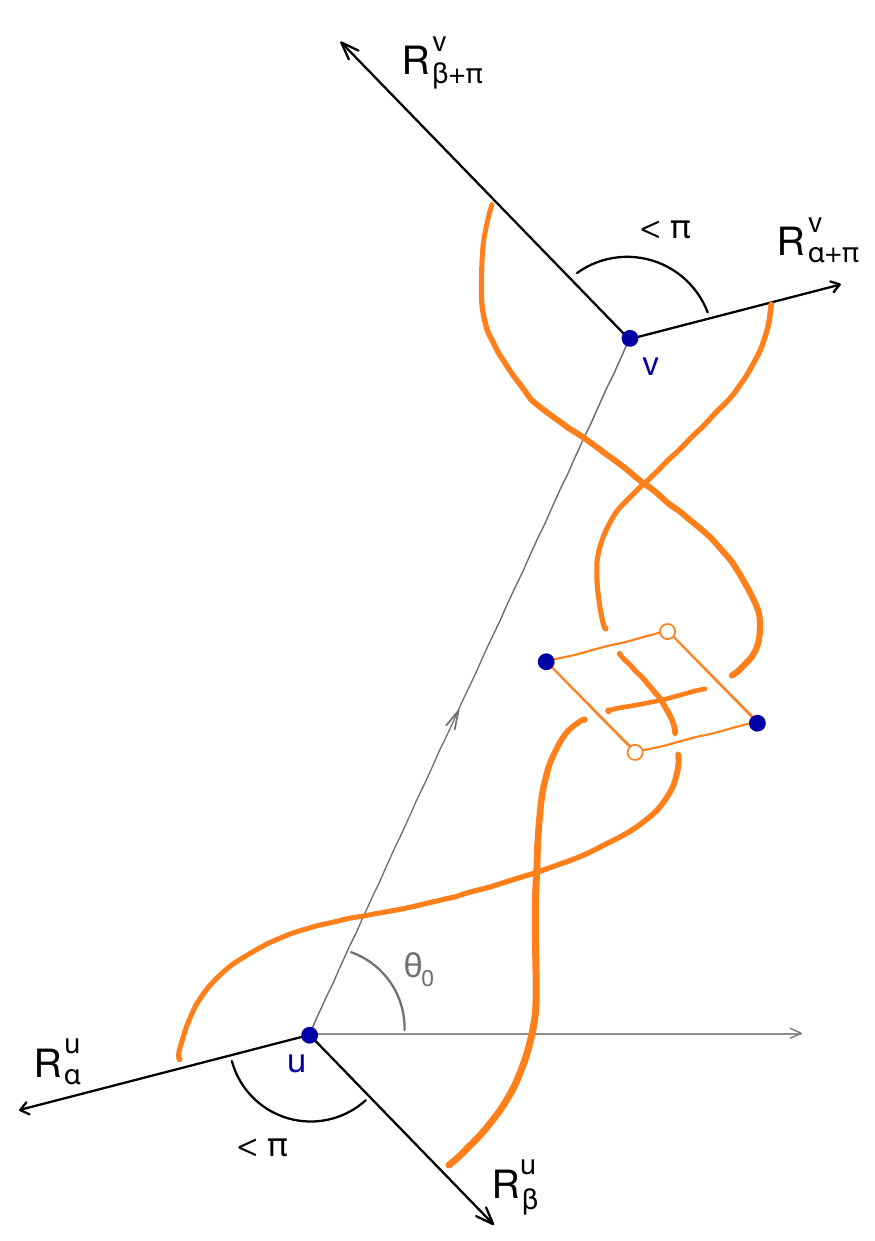}}
\caption{\label{inconsistent}Two train-tracks separating $u$ and $v$ cannot have separating angles differing by more that $\pi$}
\end{figure}

Since two distinct train-tracks may intersect at most once we are forced to conclude that $\alpha - \beta \leq \pi$. 
Finally, by equation \ref{stuff2}, the difference $z_\mathrm{cr}(\uv) - z_\mathrm{cr}(\uu)$ can be written as
$$z_\mathrm{cr}(\uv) - z_\mathrm{cr}(\uu) = \sum_{\vartheta\in\Theta(\mathbbm{v})} m_\vartheta\,e^{\mathrm{i} \vartheta}
\quad\text{with}\qquad m_\vartheta\in\mathbb{Z}_{>0}$$
Any positive combination of phases $e^{\mathrm{i}\vartheta}$ for $\vartheta \in \Theta(\mathbbm{v})$ must lie
in the positive cone $\big\{ ae^{\mathrm{i} \alpha} + be^{ \mathrm{i} \beta} \, \big| \,
a, b \in \Bbb{R}_{>0}  \big\}$ because $\alpha - \beta < \pi$. It follows that $\beta < \theta_0 < \alpha$. 
\end{proof}

For obvious topological reasons the set $\big\{ \text{$\mathbbm{t}$ separates
$\mathbbm{v}$} \, \big| \, \vartheta = \vartheta(\mathbbm{v},\mathbbm{t}) \big\}$
only depends on the end-points $\uv_0$ and $\uv_k$ of the 
path $\mathbbm{v}$. By Proposition \ref{semi-circle-property}
if $\vartheta \in \Theta(\mathbbm{v})$ then $\vartheta + \pi \notin  \Theta(\mathbbm{v})$,
which means that if two distinct train-tracks $\mathbbm{t}_1$ and $\mathbbm{t}_2$ share
the same inclination and both separate $\mathbbm{v}$ then 
$\vartheta(\mathbbm{v}, \mathbbm{t}_1) = \vartheta(\mathbbm{v}, \mathbbm{t}_2)$.
Consequently the set $\Theta(\mathbbm{v})$ together with the multiplicities
$m_\vartheta$ for $\vartheta \in \Theta(\mathbbm{v})$ 
must only depend on the end-points $\uv_0$ and $\uv_k$ of the 
path $\mathbbm{v}$ as well. This observation is consistent with the 
fact that the value of $\left\lfloor \vec{\mathbbm{v}} \, \right\rfloor_n$ 
depends only on the end-points of $\mathbbm{v}$.

\renewcommand{\imath}{\mathrm{i}}

\section{Laplacians and their determinants}
\label{sLaplacians}
\subsection{Laplacians and the critical Laplacian}\ \\
\label{ssLaplacians}
\subsubsection{Laplacians associated to polyhedral graphs and triangulations}\ \\
Given a polyhedral graph $\uG$, we denote by $\Bbb{C}^{\mathrm{V}(\uG)}$, $\Bbb{C}^{\mathrm{E}(\uG)}$, and
$\Bbb{C}^{\mathrm{F}(\uG)}$ the vector spaces of complex-valued functions supported respectively on the vertices, edges, and faces of $\uG$.
The operators $\Delta$, $\Deltaconf$ and $\mathcal{D}$ are associated to a general polyhedral graph {$\uG$}
and were introduced in section~\ref{sss3operators}.
Each operator is a linear map $\Bbb{C}^{\mathrm{V}(\uG)}\to\Bbb{C}^{\mathrm{V}(\uG)}$.
The Laplace-Beltrami operator $\Delta$ is defined as
\begin{equation}
\label{BLDelta1}
\Delta \phi(\uu)=\sum_{\mathrm{edge}\ \vec{\ue}=(\uu, \uv)} c(\vec{\ue}\,) \big(\phi(\uu)-\phi(\uv) \big)
\quad,\quad c(\vec{\ue}\,)= {1\over 2} \big(\tan \theta_\mathrm{n}(\vec{\ue}\,)+\tan\theta_\mathrm{s}(\vec{\ue}\,) \big)
\end{equation}
The conformal Laplacian $\Deltaconf$ is
\begin{equation}
\label{ConfDelta1}
\Deltaconf \phi(\uu)=\sum_{\mathrm{edge}\ \vec{\ue}=(\uu, \uv)} \tan\theta(\ue) \big( \phi(\uu)-\phi(\uv) \big)
\end{equation}
The K\"ahler operator $\mathcal{D}$ is
\begin{equation}
\label{KahlerDelta1}
\mathcal{D} \phi(\uu)=\sum_{\mathrm{edge}\ \vec{\ue}=(\uu, \uv)} 
{1\over 2}\left({\tan \theta_\mathrm{n}(\vec{\ue}\,)+
\mathrm{i} \over R^2_n(\vec{\ue} \,)}+{\tan \theta_\mathrm{n}(\vec{\ue}\,)
-\mathrm{i} \over R^2_\mathrm{s}(\vec{\ue} \,)}\right) \big(\phi(\uu)-\phi(\uv) \big)
\end{equation}
$\theta_\mathrm{n}(\vec{\ue})$, $\theta_\mathrm{s}(\vec{\ue})$ and $\theta(\ue)$ are respectively the north, south and conformal angles associated to the oriented edge $\vec{\ue}=(\uu, \uv)$, while $R_\mathrm{n}(\vec{\ue} \, )$ and $R_\mathrm{s}(\vec{\ue} \, )$ are the circumradii of the respective north $\uf_\mathrm{n}$ and south $\uf_\mathrm{s}$ faces associated to $\vec{\ue}$ (see figure~\ref{triangles}).

\begin{Rem}
The Beltrami-Laplace operator $\Delta$, the conformal Laplacian $\Deltaconf$, 
and the David-Eynard K\"ahler operator $\mathcal{D}$ 
on a polygonal graph $\uG$ agree respectively
with $\Delta$, $\Deltaconf$, and $\mathcal{D}$ when defined on
the associated {\bf redacted} graph $\uG^\bullet$ (see Def.~\ref{defRegGraph}).
By definition the vertex sets of $\uG$ and $\uG^\bullet$ coincide. So
for any pair of vertices $\uu, \uv$ the corresponding matrix entries 
$\Delta_{\uu,\uv}$ and $\Deltaconf_{\uu, \uv}$ and $\mathcal{D}_{\uu,\uv}$ 
are independent of whether we calculate their values with respect
to the graph embedding (and incidence relations) of either $\uG$ or $\uG^\bullet$.
Likewise $\Delta$, $\Deltaconf$, and $\mathcal{D}$
operating on $\uG$ agree with their respective counterparts 
when defined on any {\bf completion} $\widehat{\uG}$ of $\uG$ (see Def.~\ref{CompletionTofG}).
\end{Rem}

\subsubsection{Areas, angles and circumradii formulas}\ \\
We recall some basic geometrical formula for these quantities.
Let $\uf =(\uv_1, \uv_2, \uv_3)$ be a c.c.w. oriented triangle with vertices labelled $\uv_1$, $\uv_2$, $\uv_3$ and respective
coordinates $z_1$, $z_2$, $z_3$ then the area $A(\uf)$ of the triangle is

\begin{equation}
\label{Aform}
A(\uf)={1\over 4 \mathrm{i}}(z_2\bar z_1-z_1\bar z_2+z_3\bar z_2-z_2\bar z_3 + z_1 \bar z_3 - z_3\bar z_1)
\end{equation}
The circumcenter $z(\uf)$ of the triangle is given by
\begin{equation}
\label{Zcform}
z(\uf) ={z_1\bar z_1(z_2-z_3)+ z_2\bar z_2(z_3-z_1)+z_3\bar z_3(z_1-z_2) \over 4 \mathrm{i} A(\uf)}
\end{equation}
and the circumradius $R(\uf)$ of the triangle is given by the trigonometric relation
\begin{equation}
\label{Rcform}
R(\uf) ={|z_1-z_2| |z_2-z_3| |z_3-z_1| \over 4\, A(\uf)}
\end{equation}
while the north angle associated to the oriented edge $\vec{\ue} = (\uv_1, \uv_2)$ is
\begin{equation}
\label{thetaform}
\begin{split}
\theta_\mathrm{n}(\vec{\ue} \, )
=& \, {1\over 2\mathrm{i}}\log\left(-{(\bar z_2-\bar z_3)(z_1-z_3)\over (z_2-z_3)(\bar z_1-\bar z_3)}\right)
\end{split}
\end{equation}
Furthermore $\tan^2\theta_\mathrm{n} (\vec{\ue} \,)$ can be written explicitly in coordinates as
\begin{equation}
\label{TanThetaN}
\begin{split}
\tan^2 \theta_\mathrm{n} (\vec{\ue} \, )&={2+{z_2-z_3\over \bar z_2-\bar z_3}{\bar z_1-\bar z_3\over z_1-z_3}+{z_1-z_3\over \bar z_1-\bar z_3}{\bar z_2-\bar z_3\over z_2-z_3}\over 2-{z_2-z_3\over \bar z_2-\bar z_3}{\bar z_1-\bar z_3\over z_1-z_3}-{z_1-z_3\over \bar z_1-\bar z_3}{\bar z_2-\bar z_3\over z_2-z_3}}\\
&=4\,{|z(\uf) - z_{\scriptscriptstyle{\overline{12}}}|^2\over |z_2-z_1|^2}\qquad\text{with}\quad  z_{\scriptscriptstyle{\overline{12}}}={z_2+z_1\over 2}
\end{split}
\end{equation}
The derivatives of $A(\uf)$, $R(\uf)$ and of the angles $\theta_\mathrm{n}(\vec{\ue} \,)$ under 
a variation of a vertex coordinate are easy to calculate, using for instance 
\begin{equation}
\label{ }
\partial_{z_1} A(\uf)={1\over 4 \mathrm{i}} (\bar z_3-\bar z_2)
\ ,\quad \partial_{z_1} |z_1-z_2|={1\over 2}{\bar z_1-\bar z_2 \over |z_1-z_2|}
\quad\text{with}\quad  \partial_{z_1}={\partial\over\partial z_1}
\end{equation}
and will be discussed later.

\subsubsection{Laplacians on critical (isoradial) graphs}
The critical Laplacian studied by Kenyon in \cite{Kenyon2002} corresponds to the special case of $\Delta$, with edge weight $c(\ue)$ given by \ref{BLDelta1}, defined on a \emph{critical graph} (according to the terminology of \cite{Kenyon2002}), i.e. an  isoradial, Delaunay graph $\uG_\mathrm{cr}$. 
Accordingly we shall use the following terminology for critical Laplacians, given here. 
\begin{Def}
\label{critlaplacian}
 Let $\uG_\mathrm{cr}$ be an isoradial Delaunay graph. The Beltrami-Laplace operator $\Delta$, 
 the conformal Laplacian $\Deltaconf$, and the David-Eynard K\"ahler operator (normalized by the squared isoradius $R_\mathrm{cr}$ of the faces of 
 $\uG_\mathrm{cr}$)
coincide for $\uG_\mathrm{cr}$.
This common operator is called the {\bf critical Laplacian}
associated to $\uG_\mathrm{cr}$ and is denoted $\Delta_\mathrm{cr}$. 
\begin{equation}\label{}
\Delta_\mathrm{cr}= \Delta=\Deltaconf=R^2_\mathrm{cr} \mathcal{D}\qquad\text{on}\quad \uG_\mathrm{cr}
\end{equation}
An explicit formula for $\Delta_\mathrm{cr}$ is also given in equation \ref{crDeltaDef} of Sect.~\ref{sss3operators}.
\end{Def}

\subsection{Factorization of Laplacian using $\nabla$ and $\overline\nabla$ operators}\ \\
\label{FactoLaplNabl}
In the case of a planar triangulation $\uT$, we present an alternative representation of the operators 
$\Delta$ and $\mathcal{D}$ which will be convenient for our calculations. 
We follow the definition and the notations of \cite{DavidEynard2014}.
\begin{Def}
\label{DefNablaBarNabla}
The operators $\nabla$ and $\overline\nabla$ are linear operators from the space of complex-valued functions over the set of vertices 
$\mathrm{V}(\uT)$ of $\uT$, 
onto the space of complex-valued functions over the set of triangles (faces)  $\mathrm{F}(\uT)$ of $\uT$. 
\begin{equation*}
\label{ }
\mathbb{C}^{\mathrm{V}(\uT)} \ \mathop{\longrightarrow}^\nabla\ \mathbb{C}^{\mathrm{F}(\uT)}
\ ,\quad
\mathbb{C}^{\mathrm{V}(\uT)} \ \mathop{\longrightarrow}^{\overline\nabla}\ \mathbb{C}^{\mathrm{F}(\uT)}
\end{equation*}
$\nabla$ is defined as follows.
Given a triangle $\uf$ (a face of the triangulation $\uT$) with vertices $\uv_1$, $\uv_2$, $\uv_3$ (listed in ccw order)
and complex coordinates $z_j := z(\uv_j)$ for $1 \leq j \leq 3$ together with a function $\phi \in \Bbb{C}^{\mathrm{V}(\uT)}$ define

\begin{equation}
\label{nablaDef}
\nabla\phi(\uf)= {\phi(\uv_1)(\bar z_2-\bar z_3)+\phi(\uv_2)(\bar z_3-\bar z_1)+\phi(\uv_3)(\bar z_1-\bar z_2)\over -4 \mathrm{i} \,A(\uf)}
\end{equation}
$\nabla$ corresponds to a discrete linear derivative w.r.t. the embedding $\uv \mapsto z(\uv)$ because 
\begin{equation}
\label{ }
\nabla z =1\ ,\quad \nabla \bar z =0
\end{equation}
Similarly, its conjugate $\overline\nabla$ is defined as
\begin{equation}
\label{barnablaDef}
\overline\nabla\phi(\uf)= {\phi(\uv_1)( z_2- z_3)+\phi(\uv_2)( z_3- z_1)+\phi(\uv_3)( z_1- z_2)\over 4 \mathrm{i} \,A(\uf)}
\end{equation}
and satisfies
\begin{equation}
\label{ }
\overline\nabla z =0\ ,\quad \overline\nabla \bar z =1
\end{equation}
The transpose of these operators are defined accordingly:
\begin{equation*}
\label{ }
\mathbb{C}^{\mathrm{F}(\uT)} \ \mathop{\longrightarrow}^{\nabla^{\!\top}}\ \mathbb{C}^{\mathrm{V}(\uT)}
\ ,\quad
\mathbb{C}^{\mathrm{F}(\uT)} \ \mathop{\longrightarrow}^{\overline\nabla^{\!\top}}\ \mathbb{C}^{\mathrm{V}(\uT)}
\end{equation*}
\end{Def}

\begin{Rem}
\label{remDiffNabla}
It follows from definitions \ref{nablaDef} and \ref{barnablaDef}
and the area formula \ref{Aform} that for any function $\phi \in \Bbb{C}^{\mathrm{V}(\uT)}$
\begin{equation}
\label{diff-formula-nablas}
\begin{array}{ll}
\D \phi(\uv_1) - \phi(\uv_2) 
&\D = \  (z_1 - z_2) \, \nabla \phi(\uf) \ 
+ \ (\overline{z}_1 - \overline{z}_2) \overline{\nabla} \phi(\uf) \\
\end{array}
\end{equation}
\end{Rem}

Note that the discrete derivatives $\nabla$ and $\overline\nabla$ are defined for general triangulations. Even when the triangulation is isoradial,  $\nabla$ and $\overline\nabla$ \emph{do not coincide} with the discrete holomorphic and discrete antiholomorphic derivatives $\partial$ and $\bar\partial$ considered in \cite{Kenyon2002} for isoradial bipartite graphs. Indeed $\nabla$ and $\overline\nabla$ do not even act on the same space of functions than $\partial$ and $\bar\partial$.

Nevertheless, we shall need to bound the difference between the $\nabla \phi$ and the ordinary continuous derivative $\partial \phi$ 
in the case of a smooth complex-valued function $\phi: \Bbb{C} \longrightarrow \Bbb{C}$ with compact support and its restriction
to $\mathrm{V}(\uT)$ given by $\phi(\uv) := \phi(z(\uv))$ where $z: \mathrm{V}(\uT) \longrightarrow \Bbb{C}$ is the embedding of $\uT$. 
This estimate is explained in Lemma \ref{lemmabound} of the introduction and proven in Appendix \ref{prooflemmabound}.

\bigskip
\noindent
In addition, the $\nabla$-operator satisfies a discrete analogue of Green's Theorem 

\[ \iint_\Omega \partial\phi (z, \bar{z}) \, dz \, d\bar{z}  \ = \ \oint_{\partial \Omega} \phi(z, \bar{z}) \, d \bar{z} \] 

\noindent
in complex coordinates, namely:

\begin{lemma}
\label{greens-theorem}
Let $\uT$ be a polyhedral triangulation with embedding $z: \mathrm{V}(\uT) \longrightarrow \Bbb{C}$,
let $\Omega \subset \mathrm{F}(\uT)$ be a finite collection of triangular faces (each taken with
a counter-clockwise orientation), 
let $\partial \Omega \subset \mathrm{E}(\uT)$ be the finite subset of (oriented) edges corresponding to the boundary of 
$\Omega$, and let $\phi \in \Bbb{C}^{\mathrm{V}(\uT)}$ be a complex-valued function, then
\begin{equation}
\label{eDiscrStokes}
\sum_{\mathrm{x} \, \in \, \Omega} A(\mathrm{x}) \nabla \phi (\mathrm{x}) \ = \
\sum_{(\mathrm{u,v})  \, \in \, \partial \Omega} 
\big( \overline{z}(\uv) - \overline{z}(\uu) \big) \, {\phi (\uv) + \phi (\uu)  \over {4\mathrm{i}}}
\end{equation}
\end{lemma}
\begin{proof}
Use definition \ref{nablaDef} for $\nabla$ and observe that the area $A(x)$ defined by \ref{Aform} cancels with the area factor in the denominator of 
$\nabla\phi(x)$ and that for each oriented triangle $x= (\uu,\uv,\uw)$ the term $A(x)\nabla\phi(x)$ can be reorganized as 
$((\bar z(\uv)-\bar z(\uu)(\phi(\uv)+\phi(\uu)) +
((\bar z(\uw)-\bar z(\uv)(\phi(\uw)+\phi(\uv)) +
((\bar z(\uu)-\bar z(\uw)(\phi(\uu)+\phi(\uw)) /(4\mathrm{i})$.
Now sum over the faces of $\Omega$. Note that all internal edges count twice with opposite orientations and cancel, 
and so only the oriented edges on the boundary $\partial\Omega$ contributes and give 
the r.h.s. of \ref{eDiscrStokes}.
\end{proof}
\color{black}
\noindent
The polyhedral condition can in fact be dropped but we assume it to keep the exposition simple. 
Lemma \ref{greens-theorem} implies the following corollary which is relevant to our results.

\begin{Cor}
\label{common-refinement}
Let $\uT_1$ and $\uT_2$ be two polyhedral triangulation which share a common redacted graph $\uG := \uT_1^\bullet = \uT_2^\bullet$.
Given a face $\uf \in \mathrm{F}(\uG)$ with vertex set $\mathrm{V}(\uf)$ 
let $\Omega_i(\uf)$ be the set of triangular faces of $\uT_i$ each of whose vertices 
are in $\mathrm{V}(\uf)$. Then
\begin{equation}
\sum_{x_1  \, \in \, \Omega_1(\uf)} A(\mathrm{x}_1) \nabla \phi(\mathrm{x}_1) \ = \
\sum_{x_2  \, \in \, \Omega_2(\uf)} A(\mathrm{x}_2) \nabla \phi(\mathrm{x}_2)
\end{equation}
\noindent 
for any complex-valued function $\phi \in \Bbb{C}^{\mathrm{V}(\uG)}$.
\end{Cor}

\begin{Def}
The diagonal operators $A=\mathrm{diag}(\{A(\uf);\, \uf \in \mathrm{F}(\uG) \})$ 
(with $A(\uf)$ the area of the face $\uf$ defined by \ref{Aform}), %
and $R=\mathrm{diag}(\{R(\uf);\, \uf \in \mathrm{F}(\uG) \})$
(with $R(\uf)$ the circumradius of the face $\uf$ defined by \ref{Rcform})
map $\mathbb{C}^{\mathrm{F}(\uG)}  \to \mathbb{C}^{\mathrm{F}(\uG)}$ and are defined by their action on all $\psi\in\mathbb{C}^{\mathrm{F}(\uG)}$ as 
\begin{equation}
\label{ARdiag}
{A}\psi(\uf)=A(\uf)\psi(\uf)\ ,\quad {R}\psi(\uf) = R(\uf)\psi(\uf)
\end{equation}
\end{Def}

Then we shall heavily use the following local decompositions for the $\mathcal{D}$ and $\Delta$ operators.
\begin{Rem}
\label{RDnabla}
The K\"ahler operator $\mathcal{D}$ can be factored as 
\begin{equation}
\label{DNablaForm}
\mathcal{D}=4\,\overline\nabla^\top\! {{A}\over{R}^2}\nabla
\end{equation}
\end{Rem}
\noindent This decomposition is shown in the proposition 2.2. of section 2.6 of \cite{DavidEynard2014}. Note that $A$ and $R$ commutes.

\begin{Rem} 
\label{RLBnabla}
The Laplace-Beltrami operator $\Delta$ can be factored as
\begin{equation}
\label{DeltaNablaForm}
\Delta=2\left(\overline\nabla^{\!\top}\! {{A}}\,\nabla+ \nabla^{\!\top}\! {{A}}\,\overline\nabla\right)
\end{equation}
\end{Rem}

\noindent This decomposition can be derived easily using the method of \cite{DavidEynard2014}.
Alternatively one can use the formula \ref{Aform} for $A(\uf)$ and  \ref{diff-formula-nablas} to reorganize terms and show that for any $\phi\in\mathbb{C}^{F(G)}$ one has
\begin{equation}
\sum_{\uu,\uv} \overline\phi(\uu)\Delta_{\uu\uv}\phi(\uv)=2 \sum_\uf A(\uf)  \left(\overline\nabla\bar\phi(\uf)  \nabla\phi(\uf)+\nabla\bar\phi(\uf)  \overline\nabla\phi(\uf)\right)
\end{equation}
which amounts to \ref{DeltaNablaForm}.
\color{black}

\begin{Rem}
No similar decomposition holds for the conformal Laplacian $\Deltaconf$, since the weight $\tan\theta(\ue)$ associated to an 
oriented edge 
$\vec{\ue}$ depends non-additively on the north and south angles $\theta_\mathrm{n}(\vec{\ue} \,)$ and $\theta_\mathrm{s}(\vec{\ue} \,)$.
\end{Rem}

\subsection{Making sense of the log-determinant for infinite lattices}
\label{sLogDetInfLat}
\subsubsection{The problems}
\label{ssProbInfLat}
As explained in the introduction, we are interested in studying the variation of the $\log\det\mathcal{O}$ 
under a variation of the coordinates of the triangulation $\uT$, where $\mathcal{O}$ is any of the Laplace-like 
operators $\Delta$, $\Deltaconf$ and $\mathcal{D}$.
Two potential dangers arise:
(1) These operators have zero modes and some care is needed in imposing boundary conditions in order to exclude them. 
(2) We consider infinite polygonal graphs --- and so by any naive account, the 
log-determinant will be infinite. There is a host of standard methods used to handle these issues; below we discuss two 
situations where problem (1) and (2) can be side stepped.

\subsubsection{Using periodic triangulations:}
\label{ssPerTriang}
Consider a polyhedral graph $\uG$ which is periodic with respect to a lattice
$\Bbb{Z} + \tau \Bbb{Z}$ with $\frak{Im} \, \tau > 0$.
This means there is an action of the additive group $ \Lambda = \Bbb{Z}^2$ 
on $\mathrm{V}(\uG)$ denoted $\uv \mapsto \uv + (a,b)$ such that

\bigskip
\indent \indent
(1) $z\big(\uv + (a,b) \big) = z(\uv) + a + \tau b$ 

\smallskip
\indent \indent
(2) $\uu + (a,b)$ and $\uv + (a,b)$ are joined by an edge whenever $\uu$ and $\uv$ 
\newline \indent \indent \indent \   are joined by an edge (moreover the weights of these edges agree)

\bigskip
\noindent
for all $\uu, \uv \in \mathrm{V}(\uG)$ and $(a,b) \in \Lambda$.
Given a choice of an additive subgroup $\Lambda_{mn} := m\Bbb{Z} \times n\Bbb{Z}$ of $\Lambda$ with $m,n \in \Bbb{Z}_{>0}$,
form the quotient graph $\uG/\Lambda_{mn}$, which we can view as a finite graph embedded in
the torus $\Bbb{T}_{mn} := \Bbb{C}/(m\Bbb{Z} + \tau n \Bbb{Z})$.
Since the edge weights are periodic, the operator $\mathcal{O}$ 
descends to an operator $\mathcal{O}_{mn}$ on the quotient graph $\uG/\Lambda_{mn}$; moreover if we identify
the vertices of $\uG/\Lambda_{mn}$ with the subset $\mathrm{V}_{mn}$ consisting of
vertices $\uv \in \mathrm{V}(\uG)$ 
for which $z(\uv) \in \big\{ s + t \tau \, \big| \, (s,t) \in [0,m) \times [0,n) \big\} $ 
then $\mathcal{O}_{mn}$ is a finite dimensional operator acting on vector space of dimension
$\big|\mathrm{V}_{mn} \big|$.

We define the \emph{reduced log-determinant} $\log\det'\mathcal{O}_{mn}$ as the sum of the logarithms of the non-zero eigenvalues of $\mathcal{O}$ 
(the non-zero part of the spectrum is real and positive since $\mathcal{O}$ will be a positive operator in the cases we consider).
Then the \emph{normalized reduced log-determinant} $\log\det'\mathcal{O}$ is defined as
\begin{equation}
\label{logdetprimeO}
{\log\det}'_*\mathcal{O}_{mn}= {1\over |\mathrm{V}_{mn}|}\log{\det}'\mathcal{O}_{mn}
\end{equation}
The \emph{normalized log-determinant} of $\mathcal{O}$, defined for the entire infinite bi-periodic graph $\uG$ 
is defined simply as the limit
\begin{equation}
\label{limit-normalized-reduced}
\log{\det}_*\mathcal{O}=\lim_{m,n\to\infty}{\log\det}'_*\mathcal{O}_{mn}
\end{equation}
So $\log{\det}_*\mathcal{O}$ corresponds to an ``effective action'' density (free energy density) per vertex on the infinite lattice.

Definition (\ref{limit-normalized-reduced}) agrees with 
the log-determinant
considered  by Kenyon in \cite{Kenyon2002}
when $\mathcal{O}$ is the critical Laplacian 
on a bi-periodic infinite isoradial (critical) graph.

In fact, the limit in formula \ref{limit-normalized-reduced} exists and coincides with the following description in terms of matrix-valued symbols: 
Choose complex parameters $z$ and $w$ and for each pair $(m,n)\in \mathbb{Z}_{>0}^2$ 
define the space of quasi-periodic functions
\begin{equation} 
\label{FmnDef}
\mathcal{F}_{mn}(z,w) = \Bigg\{ \phi: \mathrm{V}(\uG) \longrightarrow \Bbb{C}  \, \Bigg| \,
\begin{array}{l}
\phi \big(\uv + (am, bn) \big) = z^aw^b\phi(\uv) \\ 
\text{for all $\uv \in \mathrm{V}(\uG)$ and all $a,b \in \Bbb{Z}$} 
\end{array}
\Bigg\}  
\end{equation}
This is a finite dimensional vector
space of dimension $\dim \mathcal{F}_{mn}(z,w) = \big|  \mathrm{V}_{mn} \big|$. Clearly $\mathcal{O}\phi \in \mathcal{F}_{mn}$
whenever $\phi \in \mathcal{F}_{mn}$, and consequently the  
operator $\mathcal{O}$ restricts
to a finite dimensional linear operator $\sigma^\mathcal{O}_{mn}$ on $\mathcal{F}_{mn}(z,w)$
which is called the symbol of $\mathcal{O}$. 
As a matrix the entries of $\sigma^\mathcal{O}_{mn}$ are Laurent polynomials in $z$ and $w$,
and for generic values of $z$ and $w$ it will be invertible; indeed the
work of Kassel and Kenyon \cite{KasselKenyon2012} implies that its determinant $\det \sigma^{O}_{mn}$ is non-negative for values of $z$ 
and $w$ each having unit modulus.
One checks that the average value of the log-determinant of this symbol 
agrees with normalized log-determinant of $\mathcal{O}$:

\begin{equation}
\label{norm-log-det}
\log \mathrm{det}_* \mathcal{O} = {1 \over {4 \pi^2 }} \, {1 \over {\big| \mathrm{V}_{mn} \big|}} \, \int_0^{2\pi} 
\int_0^{2\pi} d \zeta \, d \omega \, \log \det \sigma^\mathcal{O}_{mn} \big( e^{\mathrm{i}\zeta}, e^{\mathrm{i} \omega} \big) 
\end{equation}

\begin{Rem}
The value of the right hand side of \ref{norm-log-det} can be evaluated using Jensen's formula (twice) and is 
independent of the choice of $m,n \in \Bbb{Z}_{>0}$.
\end{Rem}

\subsubsection{Using Dirichlet boundary conditions:}
Let us propose the following alternative construction.
For a arbitrary polygonal graph $\uG$ (not necessarily periodic) 
one can consider a sequence of truncated operators $\mathcal{O}_n$ obtained from a nested sequence of domains $\Omega_1\subset\cdots \,\Omega_n\subset\Omega_{n+1}\subset\cdots$ whose union is $\mathbb{C}$: for instance, the sequence of $2n\times 2n$ squares
$\Omega_n=\{z ;\, |\mathrm{Re}(z)|<n, |\mathrm{Im}(z)|<n\}$) where $\mathcal{O}_n$ is the restriction of the operator
$\mathcal{O}$ to the subset of vertices $\mathrm{V}_n= \{ \uv \in \mathrm{V}(\uG) \, | \, z(\uv) \in \Omega_n \}$
with Dirichlet boundary conditions imposed on the complement of $\Omega_n$.
This amounts to setting the $(\uu,\uv)$  matrix entry of $\mathcal{O}_n$ to zero, whenever  $z(\uu),\, z(\uv) \notin\Omega_n$.
Thus the non-zero part of $\mathcal{O}_n$ is a $|\mathrm{V}_n| \times|\mathrm{V}_n|$ submatrix. Since we choose Dirichlet boundary conditions on the boundary of  $\Omega_n$ $\$\mathcal{O}_n$ has no zero modes and $\log\det \mathcal{O}_n$ is well defined. Then we expect that the normalized $\infty$-volume log-determinant, defined in analogy with \ref{logdetprimeO} by
\begin{equation}
\label{dirichlet-log-det}
\lim_{n\to\infty} {1\over |\Omega_n|}\ \log\det \mathcal{O}_n
\end{equation}
exists, at least in the case of a non-periodic graph $\uG$ which is sufficiently ``regular/homogeneous''  (e.g. a quasi-periodic lattice), and agrees with the normalized log-determinant $\log\det_*\! \mathcal{O}$ defined above by \ref{limit-normalized-reduced} when the graph is periodic.
This is to be expected on physical grounds by arguments analogous to those leading to the existence of a unique infinite volume thermodynamical limit for simple classical statistical systems, such as a lattice of classical oscillators, or spin systems, in their high temperature phase, independent of the boundary conditions chosen for the system. We shall not elaborate more, nor attempt to present a complete and fully rigorous proof, since this is not needed for the rest of this work.

\subsubsection{Local variation of $\infty$-volume determinants}
The finite variation of $\infty$-volume determinants (by themselves infinite) under local deformation can be defined properly for the two schemes that we have presented above.
Let us explain the idea in the Dirichlet boundary scheme.
We begin with a polyhedral graph $\uG$ and make perturbation
$\uG \to \uG'$ by moving some of its vertices inside a finite size compact domain $\Omega$. The operator $\mathcal{O}$ changes accordingly
$$\mathcal{O} \to\mathcal{O}'=\mathcal{O}+\delta\mathcal{O}$$
If the incidence relations of $\uG$ do not change, the variation $\delta\mathcal{O}$ will be an operator supported on the finite set $\bar\Omega$ 
consisting of all vertices in $\Omega$ plus their nearest neighbouring vertices 
(any vertex which shares a common face with a vertex in $\Omega$).
Considering a nested sequence of domains $\Omega_1\subset\Omega_2\cdots\subset\Omega_n\subset\cdots \to\mathbb{C}$ such that $ \bar\Omega\subset\Omega_1$, it is clear that one can write the variation series expansion for the restriction of $\mathcal{O}$ in each $\Omega_n$
\begin{equation}
\label{log-expansion2n}
\log \det \mathcal{O}'_n= \  \log \det \mathcal{O}_n \ + \ \tr \Big[  \delta \mathcal{O}_n \cdot   \mathcal{O}_n^{-1} \Big] 
\ - \ {1 \over 2} \, \tr \Big[ \left(\delta \mathcal{O}_n \cdot \mathcal{O}_n^{-1} \right)^2\Big] \ + \ \cdots
\end{equation}
In the $n\to\infty$ limit, since the $\delta \mathcal{O}_n$ extended to $\uG$ are equal to $\delta\mathcal{O}$, every term  in the expansion will converge to its $\infty$-volume limit, so that we have for any positive integer $K$
\begin{equation}
\label{ }
\tr \Big[ \left(\delta \mathcal{O}_n \cdot \mathcal{O}_n^{-1}\right)^K \Big]\ \to\ \tr \Big[ \left(\delta \mathcal{O} \cdot \mathcal{O}^{-1}\right)^K \Big]
\qquad K\in \mathbb{N}_+
\end{equation}
so that, although $\log \det \mathcal{O}'$ and $\log \det \mathcal{O}$ are formally infinite, the difference is finite and one can write
\begin{equation}
\label{log-expansion2}
\log \det \mathcal{O'}= \  \log \det \mathcal{O} \ + \ \tr \Big[  \delta \mathcal{O} \cdot   \mathcal{O}^{-1} \Big] 
\ - \ {1 \over 2} \, \tr \Big[ \left(\delta \mathcal{O} \cdot \mathcal{O}^{-1} \right)^2\Big] \ + \ \cdots
\end{equation}
We shall study the perturbation around an isoradial, Delaunay graph $\uG_\mathrm{cr}$, where we have seen that $\mathcal{O}^{-1}_\mathrm{cr}$ (the Green’s function) can be expressed in a simple contour integral form. Moreover we shall consider infinitesimal transformations \ref{thezdeform}, namely
\begin{equation*}
\label{thezdeform*}
z(\uv)\ \to\  z_\epsilon(\uv) \ = \  z(\uv)  \, + \, \epsilon \, F(\uv)
\end{equation*}
and study the general form of the first order term in \ref{log-expansion2}, and some especially interesting terms in the second order term.

\subsection{Kenyon's local formula for $\log\det\Delta_{\mathrm{cr}}$}

\subsubsection{\textbf{Kenyon's formula for a periodic infinite lattice}}
In \cite{Kenyon2002} Kenyon derived an explicit formula for the normalized log-determinant of $\Delta_\mathrm{cr}$
for periodic, isoradial, Delaunay triangulations $\uT_\mathrm{cr}$. The proof of this result
relies only on the structure of the corresponding rhombic graph $\uT_\mathrm{cr}^\lozenge$
and indeed works for any rhombic graph. For this reason Kenyon's formula implicitly extends
to all periodic, isoradial, Delaunay graphs $\uG_\mathrm{cr}$.
The formula reads  
\begin{equation}
\label{kenyon-log-det}
\log \text{det}_* \Delta_\mathrm{cr} \, = \,
{2 \over {\pi \,  | \mathrm{V}_{\scriptscriptstyle 11} | } }  \sum_{\stackrel{\scriptstyle \mathrm{edges} \, \ue }{\mathrm{of} \, \uG_{\mathrm{cr}}/\Lambda_{11}}}
\, \cyrLL \Big( \theta(\ue)  \Big) \, + \,  \cyrLL \Big( {\pi \over 2} - \theta(\ue)  \Big)
\, +  \, \theta(\ue) \log \tan \theta(\ue) 
\end{equation}
where $| \mathrm{V}_{\scriptscriptstyle 11} |$ is the volume (number of vertices) of the elementary domain of the infinite periodic graph (see section \ref{ssPerTriang}), the sum runs over all edges $\ue$ in the quotient toric graph, and  $\cyrLL$ is the Lobachevsky function (related to the Clausen function $\mathrm{Cl}_2$) defined by
\begin{equation}
\label{lobachevsky}
\cyrLL(x)=-\int_0^x dy\,|2\,\log(y)|= \mathrm{Cl}_2(2x)/2
\end{equation}

\subsubsection{\textbf{Extension to general isoradial (weak) Delaunay graphs}}
\label{log-determinant-cocyclic}
Kenyon's formula can be formally extended to express the (formally infinite) un-normalized 
log-determinant $\log\det \Delta_\mathrm{cr}$ for a general isoradial Delaunay graph $\uG_\mathrm{cr}$ 
as a sum over all edges $\ue \in \mathrm{E}(\uG_\mathrm{cr})$, namely:
\begin{equation}
\label{KenyonLogDet2}
\log\det \Delta_{\mathrm{cr}}={2\over \pi} \sum_{\ue \, \in \, \mathrm{E}(\uG_\mathrm{cr})} \mathcal{L}(\theta(\ue))
\end{equation}
with for compactness the function $\mathcal{L}$ of the conformal angles $\theta(\ue)$ given by
\begin{equation}
\label{LfunDef}
\mathcal{L}(\theta(\ue)) = \cyrLL \Big( \theta(\ue)\Big) \, + \,  \cyrLL \Big( {\pi \over 2} - \theta(\ue)  \Big)\, +  \, \theta(\ue) \, \log \tan \theta(\ue) 
\end{equation}
We may further generalize this formula to any isoradial \emph{weak} Delaunay graph $\uG_\cyrdd$ 
obtained from $\uG_\mathrm{cr}$ by adding chords inside the faces of $\uG_\mathrm{cr}$, i.e. any graph such that 
$\uG_\cyrdd^\bullet =\uG_\mathrm{cr}$.
Indeed, if $\ue$ is a chord in $\uG_\cyrdd$ then $\theta_\mathrm{n}(\vec{\ue} \,)=-\theta_\mathrm{s}(\vec{\ue} \,)$ and 
$\mathcal{L}(\theta_\mathrm{n} (\vec{\ue} \, ))= -\mathcal{L}(\theta_\mathrm{s}( \vec{\ue} \, ))$ where the function $\mathcal{L}(\theta)$  
is analytically extended to an \emph{odd function} of $\theta$ over $(-\pi,\pi)$.
For any isoradial weak Delaunay graph $\uG_\cyrdd$ of this kind, formula \ref{KenyonLogDet2} becomes
\begin{equation}
\label{KenyonLogDet3}
\log\det \Delta_{\mathrm{cr}}={1\over \pi} \sum_{\ue \, \in \,  \mathrm{E}(\uG_\cyrdd)} 
\mathcal{L}(\theta_\mathrm{n}( \vec{\ue} \, ) )+\mathcal{L}(\theta_\mathrm{s}(\vec{\ue} \,))
\end{equation}
since the contribution of any chord is zero. This is true in particular for the isoradial, weak Delaunay
graphs $\uG_{0^+}$ and $\widehat{\uG}_{0^+}$ mentioned in definition $\ref{defIsoRef}$
of the introduction. Note that the derivative of $\mathcal{L}$ is 
\begin{equation}
\label{Lprime}
\mathcal{L}'(\theta)={d \over d\theta} \mathcal{L}(\theta)= {\theta\over\sin\theta\cos\theta}
\end{equation}


\section{The critical Green's function and its asymptotics}
\label{sAsymptotics}

\subsection{Kenyon's formula for the critical Green's function}\ \\
\label{green-subsection}

The Green's function $ \Delta_\mathrm{cr}^{-1}$ studied by Kenyon in \cite{Kenyon2002} is a right-inverse of the critical
laplacian $\Delta_\mathrm{cr}$ characterized uniquely by the following three conditions

\smallskip
\indent \indent 1) $\Delta_\mathrm{cr} \, \Delta_\mathrm{cr}^{-1} = \mathbbm{1}$

\smallskip
\indent \indent 2) $\big[\Delta_\mathrm{cr}^{-1}\big]_{\uu,\uv} = 
\mathrm{O} \big( \log |z_\mathrm{cr}(\uu) - z_\mathrm{cr}(\uv) | \, \big)$ for $|z_\mathrm{cr}(\uu) - z_\mathrm{cr}(\uv) | \gg 0$

\smallskip
\indent \indent 3) $\big[\Delta_\mathrm{cr}^{-1}\big]_{\uu,\uu} = 0$

\bigskip
\noindent
Here $\uG_\mathrm{cr}$ is an isoradial Delaunay graph with embedding $z_\mathrm{cr}$
and $\uG_\mathrm{cr}^\lozenge$ its associated rhombic graph (its embedding is also denoted
$z_\mathrm{cr}$). Kenyon showed that this critical Green function $\Delta_\mathrm{cr}^{-1}$ 
on $\uG_\mathrm{cr}$ is expressed by the explicit integral

\begin{equation}
\label{greens-function}
\big[ \Delta_\mathrm{cr}^{-1} \big]_{\uu,\uv} 
\ = \ - {1 \over {8\pi^2 \mathrm{i}}} \, \oint_\mathcal{C} {dw \over w} \, \log(w) \, \mathrm{E}_{\underline{\theta}(\mathbbm{v}) } (w)
\end{equation}

\noindent
where $\mathbbm{v} = \big(\uv_0, \dots, \uv_k \big)$ is any choice of path 
from $\uv_0 = \uu$ to $\uv_{k} = \uv$ on $\uG_\mathrm{cr}^\lozenge$ and where
$\underline{\theta}(\mathbbm{v}) = \big(\theta_1, \dots, \theta_k  \big)$ is the associated sequence of angles.
$\mathrm{E}_{\underline{\theta}}(w)$ is the meromorphic function in $w$ 
\begin{equation}
\label{discrete-exponential}
\mathrm{E}_{\underline{\theta}}(w) \ := \
\prod_{j=1}^k \, 
{w + e^{\mathrm{i}\theta_j} \over {w - e^{\mathrm{i}\theta_j}} }
\end{equation}
The value of $\mathrm{E}_{\underline{\theta}}(w)$ depends only on the 
end points $\uv_0$ and $\uv_k$ of the path; this follows from
an argument similar to the proof in demonstrating 
that the value of $\lfloor \mathbbm{v} \rfloor_n$ for odd positive integers $n$
also depends only on the end points $\uv_0$ and $\uv_k$ of the path.
If we fix $\uv_0$ and allow the end point $\uv = \uv_k$ of the path to vary
then the mapping $\uv \mapsto \mathrm{E}_{\underline{\theta}}(w)$
is an example of a discrete analytic function on $\uG_\mathrm{cr}^\lozenge$ 
as discussed in \cite{Kenyon2002}. 
By Lemma \ref{infinite-product} the restriction of this
mapping to vertices $\uv \in \mathrm{V}(\uG_\mathrm{cr})$
may be viewed as a lattice approximation of the continuous exponential function
\[ z \,  \mapsto \, \exp \big\{ 2w \,  \big[ \overline{z} -  \overline{z}_\mathrm{cr}(\uv_0) \big] \big\} \]
provided $|w|<1$. For this reason $\mathrm{E}_{\underline{\theta}}(w)$ is referred to
as a {\it discrete exponential function}.
Finally $\mathcal{C}$ is any closed, counter-clockwise oriented contour enclosing the finite set
of phases $ \Phi(\mathbbm{v}) := \big\{ e^{\mathrm{i}\vartheta} \, \big| \, \vartheta \in \Theta(\mathbbm{v}) \big\} $.
As explained in Proposition \ref{semi-circle-property}
the set of angles $\Theta(\mathbbm{v})$, and thus $\Phi(\mathbbm{v})$, are finite and 
depend only on the end-points $\uu$ and 
$\uv$ of the path $\mathbbm{v}$. The set of poles of the integrand in 
formula \ref{greens-function} is precisely $\Phi(\mathbbm{v})$ and $e^{-\mathrm{i}\theta_0} \notin \Phi(\mathbbm{v})$,
so a contour $C$ can be chosen which avoids the 
branch cut $-\theta_0= \arg \big( z_\mathrm{cr}(\uu) - \mathrm{cr}(\uv) \big)$ of the logarithm;
see subsection 4.3 below for details.

\begin{Rem}
\label{rotation-invariance}
Formula (\ref{discrete-exponential}) is invariant under both global translation and rotation of the graph $\uG_\mathrm{cr}$. 
\end{Rem}

\begin{Rem}
\label{edge-value-2}
Let us consider an oriented edge $\vec\ue=(\uu\uv)$ of an {\bf{isoradial, weak Delaunay graph}} $\uG_\mathrm{cr}$.
There are two possible situations.
\begin{description}
  \item[1] Either $\vec\ue=(\uu\uv)$ is not a chord (see Fig.~\ref{Cocyclic-1-TX} left) in which case the north and south angles of $\vec\ue$ are equal (and generically non-zero) and both coincide with the conformal angle of the edge $\ue$.
\begin{equation*}
\label{ }
\thetan(\vec\ue)=\thetas(\vec\ue)=\theta(\ue)
\end{equation*}
\item[2] Or $\vec\ue=(\uu\uv)$ is a chord (see Fig.~\ref{Cocyclic-1-TX} right) in which case the north and south angles of $\vec\ue$ are opposite, 
while the conformal angle of $\ue$ is zero.
\begin{equation*}
\label{ }
\thetan(\vec\ue)=-\thetas(\vec\ue)\neq 0\quad,\qquad\theta(\ue)=0
\end{equation*}
\end{description}

\noindent
In {both cases}, Kenyon's formula for the Green's function for this pair of vertices $\uu, \uv$ reads
\begin{equation}
\label{GreenEdgeChord1}
\big[\Delta_\mathrm{cr}^{-1}\big]_{\uu,\uv}= -{1\over\pi} \, \thetan(\vec\ue) \cot \thetan(\vec\ue)= -{1\over\pi} \, \thetas(\vec\ue) \cot \thetas(\vec\ue)
\end{equation}
\end{Rem}
\begin{figure}[h!]
\begin{center}
\includegraphics[width=1.6in]{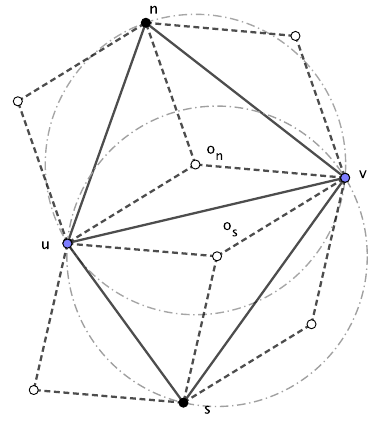}
\hskip 1.cm
\includegraphics[width=1.6in]{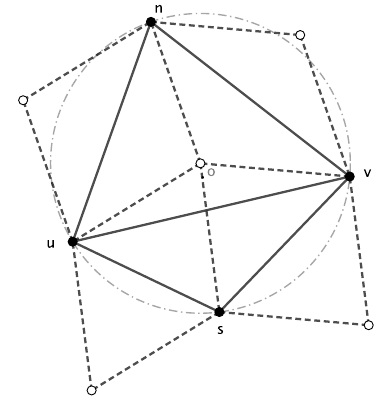}
\caption{ Either the edge $\vec\ue=(\uu\uv)$ is not
a chord, in which case the respective north and south centers $\uo_\mathrm{n}$ and $\uo_\mathrm{s}$ are different and $\theta_\mathrm{n}(\vec\ue)
=\theta_\mathrm{s}(\vec\ue)$ (and generically non-zero). Or else the edge $\vec\ue=(\uu\uv)$ is a chord, in which case the centers 
coincide $\uo_\mathrm{n}=\uo_\mathrm{s}=\uo$ and $\theta_\mathrm{n}( \vec\ue) =-\theta_\mathrm{s}(\vec\ue)$.}
\label{Cocyclic-1-TX}
\end{center}
\end{figure}
\begin{proof}
\newcommand{\uon}{\uo_\mathrm{n}}
\newcommand{\uos}{\uo_\mathrm{s}}
\newcommand{\ufn}{\uf_\mathrm{n}}
Select, for instance, the north face $\ufn$ and let $\uon$ be its center. 
Assume that the isoradius of the critical triangulation is $R_\mathrm{cr}=1$ for simplicity.
Consider the path $\mathbbm{v}=(\uu,\uon,\uv)$.
Set 
$e^{\mathrm{i} \theta_1} = z_\mathrm{cr}(\uon) -z_\mathrm{cr}(\uu) $
and 
$e^{\mathrm{i} \theta_2} = z_\mathrm{cr}(\uv) - z_\mathrm{cr}(\uon)$
and note that 
$\thetan({\overline{\uu \uv}}) = (\theta_1  -\theta_2)/2$. Then
\begin{equation}
\label{green-edge-2} 
\begin{split}
\big[ \Delta_\mathrm{cr}^{-1} \big]_{\uu,\uv}
&=  -{1 \over {8\pi^2\mathrm{i}}} \, \oint_{\mathcal{C}} \, {dw \over w}  \, \log(w) \, {w + e^{\mathrm{i}\theta_1} \over {w - e^{\mathrm{i} \theta_1}} } \, {w + e^{\mathrm{i} \theta_2} \over {w - e^{\mathrm{i} \theta_2}} }  \\ 
&=    -{1 \over {4\pi}} \, \Bigg( 2e^{\mathrm{i} \theta_1} \, {e^{\mathrm{i}\theta_1} + e^{\mathrm{i} \theta_2} \over {e^{\mathrm{i}\theta_1} - e^{\mathrm{i} \theta_2}} } \, {\log\big( e^{\mathrm{i} \theta_1} \big)  \over {e ^{i \theta_1} }} \, + \, 2e^{\mathrm{i} \theta_2} \, {e^{\mathrm{i} \theta_2} + e^{\mathrm{i} \theta_1} \over {e^{\mathrm{i}\theta_2} - e^{\mathrm{i} \theta_1}} } \, {\log\big( e^{\mathrm{i}\theta_2} \big)  \over {e ^{i\theta_2} }} \Bigg)   \\ 
&=     -{1 \over {4\pi}} \, \Bigg(   2 \mathrm{i} (\theta_1-\theta_2) \, {e^{\mathrm{i}\theta_1} + e^{\mathrm{i}\theta_2} \over {e^{\mathrm{i}\theta_1} - e^{\mathrm{i}\theta_2}} } \,   \Bigg) =     -{1 \over {\pi}} \,     {\theta_1-\theta_2\over 2} \cot\left( {\theta_1-\theta_2\over 2}\right)  
\end{split}
\end{equation}  
The calculation with the south face $\uf_\mathrm{s}$ gives the same result, regardless of whether $\uos\neq\uon$ or $\uos=\uon$.
\end{proof}

\subsection{Expansion and bounds for the discrete exponential}
\label{ssExpDiscExp}
\begin{lemma}
\label{lemma-p-estimate}
Consider a finite sequence of angles $\big(\theta_1, \dots,  \theta_k \big)$ contained in the closed
interval of the form $\big[ \vartheta - {\pi \over 2}, \vartheta + {\pi \over 2} \big]$ centered about some
fixed angle $\vartheta$.
{Using Def.~\ref{def-p_n}), consider}
\begin{equation}
\label{p-estimate-def}
p_{2n+1} := \sum_{j=1}^k \, e^{\mathrm{i}(2n+1)\theta_j}
\end{equation}
Then we have the uniform bound
\begin{equation}
\label{p-estimate}
\big| \, p_{2n+1}\big| \, \leq \,
(2n+1) \, \big| \, p_1\big|.
\end{equation}
\end{lemma}

\begin{proof}
Clearly, it is enough to verify the lemma in the case of $\vartheta = 0$, otherwise 
we have $p_{2n+1}  = e^{-\mathrm{i}\vartheta} \tilde{p}_{2n+1}$ where 
$\tilde{p}_{2n+1} = \sum_{j=1}^k \, e^{\mathrm{i}(2n+1) \tilde{\theta}_j }$
and where $\tilde{\theta}_j = \theta_j - \vartheta \in \big[ -{\pi \over 2}, {\pi \over 2} \big]$.

\smallskip
\noindent
Begin with the following polynomial $\D q_{2n+1}(w) := \, 2w^2 \Big( w^{2n} - (-1)^n \Big) \big(w^2 + 1 \big)^{-1}$
and notice that 
\[ \begin{array}{rl} \displaystyle
q_{2n+1}(iw) :=& \displaystyle \, 2\big(iw\big)^2 \Bigg( {{\big(iw\big)^{2n} - (-1)^n} \over {\big(iw\big)^2 + 1}} \Bigg) 
=
\D \big(-1 \big)^n \, 2w^2 \, {{w^{2n} -1} \over {w^2-1}}  \\
=&\D \big(-1 \big)^n \, 2 \, \Big( w^{2n} + w^{2n-2} + \cdots + w^2 + 1 \Big) \\
\end{array}
\]
therefore $\D q_{2n+1}(w) = \big(-1 \big)^{n} 2 \Big( 1 - w^2 + w^4 - w^6 + \cdots + \big( -1\big)^n w^{2n} \Big) $.

\bigskip
\noindent
For $w = e^{\mathrm{i}\theta}$ with $\theta \in 
\big[ -{\pi \over 2}, {\pi \over 2} \big]$
the function $\theta \mapsto q_{2n+1}\big( e^{\mathrm{i}\theta}\big)$ is clearly continuous and its modulus takes maximal value 
$\Big| q_{2n+1}\big( \pm i \big) \Big| = 2n$ and so $\big| q_{2n+1}  \big|_\infty = 2n$. By construction
$e^{\mathrm{i} (2n+1) \theta } = \cos (\theta) \, q_{2n+1}\big( e^{\mathrm{i} \theta} \big) \, + \, (-1)^n e^{\mathrm{i} \theta}$ and so
$p_{2n+1} = \sum_{j=1}^k \cos( \theta_j) \, q_{2n+1}\big( e^{\mathrm{i} \theta_j}\big) \, + \, (-1)^n p_1$. We now proceed 
with a yoga of inequalities:
\[
\begin{array}{rl}
\displaystyle \big| p_{2n+1} \big| &\displaystyle 
\leq \ \Big| \sum_{1 \leq j \leq k} \, \cos( \theta_j) \, q_{2n+1} \big( e^{\mathrm{i} \theta_j}\big) \, \Big| \ + \ \big| \, p_1 \big| \\ 
&\displaystyle \leq \
\sum_{1 \leq j \leq k} \, \Big| \cos( \theta_j) \, q_{2n+1} \big( e^{\mathrm{i} \theta_j}\big) \, \Big| \ + \ \big| \, p_1 \big| \\ 
&\displaystyle \leq \
\sum_{1 \leq j \leq k} \, \cos( \theta_j)  \, \Big| \, q_{2n+1} \big( e^{\mathrm{i} \theta_j}\big) \, \Big| \ + \ \big| \, p_1 \big| 
\quad \left( \begin{array}{ll} \text{Note that $\cos(\theta_j) \geq 0$} \\ 
\text{because $-{\pi \over 2}   \leq \theta_j \leq {\pi \over 2}$} \end{array} \right) \\ 
&\displaystyle \leq \
\sum_{1 \leq j \leq k} \, \cos( \theta_j)  \, \big| \, q_{2n+1}  \big|_\infty \ + \ \big| \, p_1 \big| \\
&\displaystyle \leq \
2n \, \frak{Re} \big[ p_1 \big] \ + \  \big| \, p_1 \big| \\
&\displaystyle \leq \ \big( 2n + 1 \big) \, \big| \, p_1 \big| 
\quad \Big( \text{since $0 \leq \frak{Re} \big[ p_1 \big] \leq \big| \, p_1  \big|$ } \Big)
\end{array}
\]
\end{proof}

\begin{lemma}
\label{infinite-product}
Given a finite sequence of angles $\underline{\theta} = \big(\theta_1, \dots, \theta_k \big)$ and $|w|<1$
the following infinite product expansion of $\mathrm{E}_{\underline{\theta}}(w)$ is valid:

\begin{equation}
\begin{array}{ll}
\D \mathrm{E}_{\underline{\theta}}(w) 
&\D  = \   \big( -1\big)^k \, \prod_{n \, \mathrm{odd} } \exp \Bigg(  {2 \over n} w^n \,  \overline{p}_n  \Bigg)  
\quad \text{where $\D p_n = \sum_{1 \leq j \leq k} e^{\mathrm{i} n \theta_j}$. } 
\end{array}
\end{equation}
\end{lemma}

\begin{proof}
\[ \begin{array}{cl}  \D \prod_{j=1}^k \, {w + e^{\mathrm{i}\theta_j} \over {w - e^{\mathrm{i} \theta_j}} } &\D = \ \big( -1\big)^k \, \prod_{j=1}^k \,
{1 + we^{-\mathrm{i}\theta_j} \over {1 - we^{-\mathrm{i} \theta_j}}  } \\
&\D = \ \big( -1\big)^k \, \exp \, \sum_{j=1}^k \, \log \, \Bigg( {1 + we^{-\mathrm{i}\theta_j} \over {1 - we^{-\mathrm{i} \theta_j}}  } \Bigg) \\
&\D = \   \big( -1\big)^k \, \exp \, \sum_{j=1}^k \, 2\, \Bigg( we^{-\mathrm{i}\theta_j} \, + \,  {1 \over 3} \, w^3e^{-3i \theta_j}  \, + \,
{1 \over 5} \, w^5 e^{-5i \theta_j} \, \cdots  \Bigg) \\ 
&\D = \  \big( -1\big)^k \, \exp \, \Bigg( 2w \, \sum_{j=1}^k \, e^{-\mathrm{i} \theta_j} \ + \  {2 \over 3} \, w^3 \, \sum_{j=1}^k 
\, e^{-3i \theta_j} \ + \
{2 \over 5} \, w^5 \, \sum_{j=1}^k \, e^{-5i \theta_j} \ \cdots  \Bigg) \\
&\D = \   \big( -1\big)^k \, \prod_{n \, \text{odd} } \exp \Bigg(  {2 \over n} w^n \,  \overline{p}_n  \Bigg)  
\end{array}  \]
Note that this can be rewritten as 
$$ \big( -1\big)^k \, \exp \big(  2 w  \, \overline{p}_1  \big) 
\cdot \Bigg( 1 + \sum_{N \geq 3} \, w^N \, \overline{\mathrm{c}}_N \Bigg)$$
with the $\overline{\mathrm{c}}_N $ the coefficients of a series.
\end{proof}

\begin{Rem}
\label{using-p-estimate}
Let $\underline{\theta} = \big(\theta_1, \dots, \theta_n \big)$ be a finite sequence of
angles contained in an interval of the form $\big[ \vartheta - {\pi \over 2}, \vartheta + {\pi \over 2} \big]$
where $n$ is a positive odd integer. Define 
\begin{equation}
\label{def-u-n}
u_n \ = \ {1 \over n} \, {p_n \over {p_1}}
\quad \text{and} \quad 
\D \frak{u}(w) = 
\sum_{\stackrel{\scriptstyle \mathrm{odd}}{n \geq 3} }
\, u_n w^n
\end{equation}
By Lemma \ref{lemma-p-estimate} each $| u_n | \leq 1$ and $\frak{u}(w)$ is analytic in the
unit disk and 
$\mathrm{E}_{\underline{\theta}}(w) =  \big( -1 \big)^k \cdot \exp\big( 2\overline{p}_1w \big) \cdot \exp \big( 2\overline{p}_1 \frak{u}(w) \big)$.
Furthermore we have, through the standard combinatorial vinyasas,  
\begin{equation}
\label{combinatorial-yoga}
\mathrm{E}_{\underline{\theta}}(w) 
\ = \ \big( -1 \big)^k \cdot \exp\big( 2\overline{p}_1w \big) \cdot 
\Bigg( 1 + \sum_{m =1 }^\infty \, \sum_{d = 1}^m \, w^{2m+d} \, \big(  2\overline{p}_1 \big)^d \, \overline{\mathrm{c}}_{m,d} \Bigg)
\end{equation}
with the coefficients $\mathrm{c}_{m,d}$ given by 
\begin{equation}
\label{def-c-md}
\mathrm{c}_{m,d}
= \sum_{\stackrel{\scriptstyle \mathbbm{r} \, \vdash m}{\# (\mathbbm{r})= d} } 
\ \prod_{s \geq 1 } \,
 {1 \over {(r_s)! } } \, (u_{1+2s})^{r_s} 
\end{equation}
and where the sum is taken over infinite tuples $\mathbbm{r} = \big(r_1, r_2, r_3, \dots \big) \in \Bbb{Z}_{\geq 0}^\Bbb{N}$
with $\sum_{s \geq 1} r_s = d$ and such that $\sum_{s \geq 1} \,s\, r_s = m$.
\end{Rem}

Let $\uu$ and $\uv$ be distinct vertices of $\uG_\mathrm{cr}$ and let $\mathbbm{v} =
\big(\uv_0, \dots, \uv_k \big)$ be 
a path from $\uu$ to $\uv$. 
Translation and rotation invariance of the Green's function allows us  
to assume without loss of generality
that $\uu$ is situated at the origin and that
the phases $e^{\mathrm{i} \theta_j} := z_\mathrm{cr}(\uv_j)  - z_\mathrm{cr}(\uv_{j-1})$ of the path
lie in the open interval $\big( {\scriptstyle -}{\pi \over 2}, {\pi \over 2} \big)$; 
if not the embedding of $\uG_\mathrm{cr}$ may be 
shifted $z \mapsto z- z_\mathrm{cr}(\uu)$ and rotated $z \mapsto  z \exp \big(-i\theta_\mathbbm{v} \big)$ 
to achieve these features; see Proposition \ref{semi-circle-property} for 
a definition of $\theta_\mathbbm{v}$.

\subsection{Contour integral for the expansion}
\label{ssContIntExp}

In \cite{Kenyon2002} Kenyon handles the asymptotic behaviour of the Green's function with
respect to the distance $|\uu - \uv|$
using a {\it keyhole} contour $C$ with a corridor of width $\epsilon >0$
avoiding the cut of the logarithm $\arg(w) = -\pi$.
Paraphrasing Kenyon, this contour $C_\epsilon$ runs counter-clockwise along the circle of radius $R$ 
about the origin (connecting $-R \pm i\epsilon$), 
then travels horizontally above the $x$-axis from $-R + i \epsilon$ to $-r+ i \epsilon$, runs clockwise
along the circle of radius $r$ about the origin (connecting $-r \pm i \epsilon$), 
and finally returns horizontally from $-r - i \epsilon$ to $-R - i \epsilon$ below the $x$-axis. 
Here $R \gg |u-v|$ and $r  \ll  |\uu - \uv|^{-1}$ (see figure \ref{keyhole-contour}).
The following lemma allows us to compute the Green's function 
by integrating along the cut of the logarithm provided we subtract off
the logarithmic divergences. 

\begin{figure}[h]
\begin{center}
\raisebox{-.8in}{\includegraphics[width=3in]{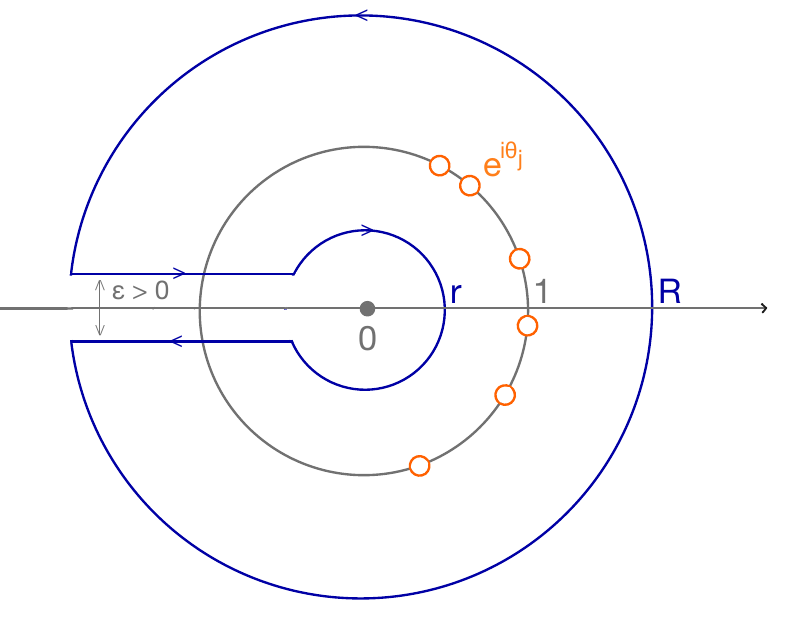}}
\caption{keyhole contour $C$.}
\label{keyhole-contour}
\end{center}
\end{figure}

\begin{lemma}
\label{integrating-on-the-cut}
Let $\mathrm{F}(w)$ be a function which is holomorphic on the extended complex plane $\Bbb{C} \cup \{ \infty\}$
outside a subset $S$ contained in the 
interior of the keyhole contour $C$ for some values of $R$, $r$, and $\epsilon$, and such that
$F(0)=F(\infty)=1$ then

\begin{equation}
\oint_C  {dw \over w} \, \log(w) \, \mathrm{F}(w) \ = \
- 2 \pi i  \int_0^\infty \, \Big( \mathrm{F}({\scriptstyle -}t) - 1\Big) \, {dt \over t}
\end{equation}
\end{lemma}
\begin{proof}

\[
\begin{array}{ll}
\D \oint_C  {dw \over w} \, \log(w) \, \mathrm{F}(w) 
&\D = \ \lim_{  \stackrel{\scriptstyle r \rightarrow 0}{R \rightarrow \infty} } \ \lim_{\epsilon \rightarrow 0} \, 
\oint_C  {dw \over w} \, \log(w) \, \mathrm{F}(w) \\ \\
&\D = 
\lim_{  \stackrel{\scriptstyle \, r \rightarrow 0}{R \rightarrow \infty}    } 
\left\{
\begin{array}{cl}
\D i \int_{\pi}^{-\pi} \log \big( re^{\mathrm{i} \phi} \big)
\, \mathrm{F} \big(re^{\mathrm{i}\phi} \big) \, d \phi &\begin{array}{l} \text{integral (1):} \\
\text{contribution of} \\ \text{circle radius $r$} \end{array}  \\ 
+ & \\ 
\D 2 \pi i  \int_{-R}^{-r}  \, \mathrm{F}(t) \, {dt \over t} & \begin{array}{l} \text{integral (2):} \\
 \text{contribution} \\ \text{along the cut} \end{array} \\ 
+ & \\ 
\D i \int_{-\pi}^\pi \log \big( Re^{\mathrm{i} \phi} \big)  
\, \mathrm{F} \big(Re^{\mathrm{i}\phi} \big) \, d \phi &\begin{array}{l} \text{integral (3):} \\
\text{contribution of} \\ \text{circle radius $R$} \end{array} \\ 
\end{array} 
\right. \\ \\
&\D =
\lim_{  \stackrel{\scriptstyle \, r \rightarrow 0}{R \rightarrow \infty}    } 
\left\{
\begin{array}{c}
\D -2 \pi i  \log(r) - 2 \pi i  \sum_{N \geq 1} \,  {1 \over N} \big( {\scriptstyle-}r\big)^N  a_N \\ 
 +  \\ 
\D -2 \pi i  \int_{r}^{R}  \, \mathrm{F}(-t)  \, {dt \over t} \\ 
+  \\ 
\D  2 \pi i \log(R) + 2 \pi i  \sum_{N \geq 1} \,  {1 \over N} \big( {\scriptstyle-}R\big)^{-N}  b_N
\end{array} 
\right. \\ \\
&\D = - 2 \pi i  \int_0^{\infty} \, {dt \over t} \, \Big( \mathrm{F}(-t) - 1 \Big)
\end{array}
\]
where $1 + \sum_{N \geq 1} a_N w^N$ and $1 + \sum_{N \geq 1} b_N w^N  $ are the power series
expansions of $\mathrm{F}(w)$ at $0$ and $\infty$ respectively.
\end{proof}

\begin{Cor}
\label{integral-over-(0,1)}
For vertices $\uu$ and $\uv$ in $\uG_\mathrm{cr}$ the value of the Green's function is
\begin{equation}
\label{integral-over-(0,1)}
\D \big[\Delta_\mathrm{cr}^{-1}\big]_{\uu,\uv} \ = \
{1 \over {2\pi}} \, \frak{Re} \int_0^1 \, \Big( \mathrm{E}_{\underline{\theta}(\mathbbm{v})}(-t) \, - \,  1\Big) \, {{dt} \over t}
\end{equation}
\end{Cor}

\begin{proof}
We begin with the observation that $\mathrm{E}_{\underline{\theta}}\big(w^{-1} \big) 
= \big( -1 \big)^k \, \overline{\mathrm{E}}_{\underline{\theta}}(w)$ for any finite sequence
of angles $\underline{\theta} = \big(\theta_1, \dots, \theta_k \big)$. 
Since $\uu$
and $\uv$ are vertices in $\uG_\mathrm{cr}$, the length
$k$ of any path $\mathbbm{v}= \big(\uv_0, \dots, \uv_k \big)$ 
from $\uv_0= \uu$ to $\uv_k= \uv$ in $\uG_\mathrm{cr}^\lozenge$
must be even. Thus $\mathrm{E}_{\underline{\theta}}\big(w^{-1} \big) 
= \overline{\mathrm{E}}_{\underline{\theta}(\mathbbm{v}) }(w)$.  Then
\[
\begin{array}{ll} 
\D \big[ \Delta^{-1}_\mathrm{cr} \big]_{\uu,\uv}  &\D = 
- {1 \over {8 \pi^2 i}} \, \oint_C  {dw \over w} \, \log(w) \, \mathrm{E}_{\underline{\theta}(\mathbbm{v})}(w) \\
&\D = {1 \over {4 \pi}} \int_0^\infty \Big( \mathrm{E}_{\underline{\theta}(\mathbbm{v})}(-t) -1 \Big) \, {dt \over t} \\
&\D =  {1 \over {4 \pi}}  \int_0^1 \, \Big( \mathrm{E}_{\underline{\theta}(\mathbbm{v})}(-t) -1 \Big) \, {dt \over t} \ 
+ \  {1 \over {4 \pi}} \int_1^\infty \, \Big( \mathrm{E}_{\underline{\theta}(\mathbbm{v})}(-t)  -1 \Big)
\, {dt \over t} \\
&\D =  {1 \over {4 \pi}}  \int_0^1 \, \Big( \mathrm{E}_{\underline{\theta}(\mathbbm{v})}(-t) - 1 \Big) \, {dt \over t} \ + 
\  {1 \over {4 \pi}} \int_0^1 \, \Big( \overline{\mathrm{E}}_{\underline{\theta}(\mathbbm{v})}(-t) 
-1 \Big) \, {dt \over t} \\
&\D =  {1 \over {2 \pi}} \,  \frak{Re} \Bigg[ \int_0^1 \, \Big( \mathrm{E}_{\underline{\theta}(\mathbbm{v})}(-t)  - 1 \Big)  
\,  {{dt} \over t} \, \Bigg] \\
\end{array}
\]
\end{proof} 

\begin{Rem}
Since $|t| < 1$ in formula \ref{integral-over-(0,1)} we may use the presentation of $\mathrm{E}_{\underline{\theta}(\mathbbm{v})}(t)$
given in Remark \ref{using-p-estimate} and write
\[
\D \big[ \Delta^{-1}_\mathrm{cr} \big]_{\uu,\uv}  \ = \
{1 \over {2\pi}} \, \frak{Re} \int_0^1 \, \Big(    
\exp\big( -2p_1t \big) \cdot \exp \big( 2p_1 \frak{u}(-t) \big) 
 \, - \,  1\Big) \, {{dt} \over t} 
\]
where $\frak{u}(t)=\sum\limits_{n>0} u_{2n+1} (t)^{2n+1}$ is the function defined from the momenta $p_{2n+1}$ in eq. \ref{def-u-n}.
 
 \noindent
We {may} adopt the view that $p_1$ and $\overline{p}_1$ are independent variables on the plane
and that $\big[ \Delta^{-1}_\mathrm{cr} \big]_{\uu,\uv} $ is a smooth function of
 $p_1$ and $\overline{p}_1$.
\end{Rem}

\subsection{The general asymptotics}

\begin{Prop}
\label{pGasymptotics}
The Green's function $\big[ \Delta^{-1}_\mathrm{cr} \big]_{\uu,\uv}$ has a series expansion
at $\infty$ given by:
\begin{equation}
-{1 \over {2\pi}} \Bigg( \log \big(2 | p_1|\big) \, + \,  \gamma_\mathrm{euler} 
\, - \sum_{m \geq d \geq 1}  \, \big( {\scriptstyle -} 1 \big)^d  \big(2m+d-1 \big)! \, \frak{Re} \Big[ 
\mathrm{c}_{m,d} \, \big(2p_1\big)^{-2m}  \Big] \Bigg)
\end{equation}
\noindent
where the coefficients $\mathrm{c}_{m,d}$ are defined in equation \ref{def-c-md} in terms of the $u_{1+2s}$ defined by \ref{def-u-n}, which are themselves bounded in terms of $p_1$ by Lemma \ref{lemma-p-estimate}.
\end{Prop}

\begin{proof} 
\begin{align} 
\D \big[ \Delta^{-1}_\mathrm{cr} \big]_{\uu,\uv}  
&\D = 
{1 \over {2 \pi}} \,  \frak{Re} \Bigg[  \int_0^1 \Big( \mathrm{E}_{\underline{\theta}(\mathbbm{v})}(-t) - 1 \Big) \, {dt \over t} \, \Bigg] \nonumber\\
&\D = 
\left\{
\begin{array}{c}
\D {1 \over {2 \pi}} \,  \frak{Re} \Bigg[  \int_0^1 \, \Big( \exp \big( {\scriptstyle-}2 p_1t \big) -1 \Big) \, {dt \over t} \, \Bigg] \\
\ \\
+ \\
\ \\
\D {1 \over {2 \pi}}  \sum_{m \geq d \geq 1}      
\, \frak{Re} \Bigg[  
\, \mathrm{c}_{m,d} \big(2p_1\big)^d 
\int_0^1
\, {\scriptstyle -} \big({\scriptstyle -} t\big)^{2m+d-1} \exp \big( {\scriptstyle -}2 p_1t \big)  \, dt  \, \Bigg]
\end{array} 
\right. \nonumber\\
&\D =
\left\{
\begin{array}{c}
\D -{1 \over {2\pi}} \, \frak{Re} \Bigg[ 
\log\big( 2p_1 \big) + \gamma_\mathrm{euler}  +
\underbrace{\int_{2p_1}^\infty \, \exp(-t) \,  {dt \over t}}_{ \stackrel{\scriptstyle \text{null power series}}{\scriptstyle \text{development at $\infty$}}} 
\, \Bigg]
\ \\
+ \\
\ \\
\D -{1 \over {2 \pi}}  \sum_{m \geq d \geq 1} \, \frak{Re} \Bigg[ \,
\mathrm{c}_{m,d}  \big({\scriptstyle -} 1 \big)^d \, \big(2 p_1\big)^{-2m}
\sum_{i=0}^{2m+d-1} 
{{\big( 2m + d {\scriptstyle -} 1 \big)!} \over {i!}} \, 
\underbrace{\big( 2p_1 t \big)^i \, \exp\big( {\scriptstyle -} 2p_1t \big)}_{ \stackrel{\scriptstyle \text{null power series}}{\scriptstyle \text{development at $\infty$}}} 
\, \Bigg] \, \Bigg|_0^1
\end{array} 
\right. \nonumber\\
&\D = -{1 \over {2\pi}} \Bigg( \log \big(2 | p_1|\big) + \gamma_\mathrm{euler} 
- \sum_{m \geq d \geq 1}  \, \big( {\scriptstyle -} 1 \big)^d  \big(2m+d-1 \big)! \, \frak{Re} \Big[ 
\mathrm{c}_{m,d} \, \big(2p_1\big)^{-2m}  \Big] \Bigg)
\nonumber\\
\end{align}
\end{proof}


\section{Deforming Delaunay lattices and operators}
\label{sVarOp}

\subsection{Setup and problems for deformations of isoradial Delaunay graphs}
\label{sswwococycl}
\renewcommand{\imath}{\mathrm{i}}
We start to address the main problem of this work, which is to study geometric deformations of isoradial Delaunay graphs and their associated operators defined in Sect.~\ref{sss3operators}. 

Begin with an initial (not necessarily isoradial) \textbf{Delaunay graph} $\uG_0$ as defined in Def.~\ref{DelaunayGr} with vertex set $\mathrm{V}(\uG_0)$, edge set $\mathrm{E}(\uG_0)$, and face set $\mathrm{F}(\uG_0)$.
We deform the initial vertex embedding $\uv \mapsto z_0(\uv)$ for $\uv \in \mathrm{V}(\uG_0)$ by
\begin{equation}
\label{zDeform}
 z_\epsilon(\uv) \, := \ z_0(\uv) \ + \ \epsilon \, F(\uv)
\end{equation}
where $\epsilon$ a positive real parameter and the displacements 
$F(\uv)$ are implemented by a
complex-valued function
$$F:\ \mathrm{V}(\uG_0)\to\mathbb{C}$$
with {\bf finite support}, i.e. a 
finite subset $\boldsymbol{\Omega}_F\subset \mathrm{V}(\uG_0)$ such that  $\uv\in\boldsymbol{\Omega}_F\iff F(\uv)\neq 0$.

If the deformation parameter $\epsilon$ is unconstrained displaced vertices may potentially collide, i.e. the mapping $\uv \mapsto z_\epsilon(\uv)$ may fail to be one-to-one. The following simple lemma allows us 
to avoid this situation.
\begin{lemma}
\label{MFbound}
For any pair of distinct vertices
$\uu,\uv \in \uG_0$ the corresponding perturbed coordinates
$z_\epsilon(\uu)$ and $z_\epsilon(\uv)$ will always remain
distinct provided 
\begin{equation}
\label{ }
0 \leq \epsilon < \epsilon'_F=M_F^{-1} 
\end{equation}
where
\begin{equation}
\label{MFDef}
M_F = \ \max_{\uu \ne \uv } \, 
\Big| dF\big( \overline{\uu \uv} \big) \Big| 
\quad \text{with} \quad
\D d F \big( \uu, \uv \big) \, = \ 
 { {F(\uu) - F(\uv) }  \over {z_0(\uu) -z_0(\uv)}} 
\end{equation}
\end{lemma}
\begin{proof}
The mapping $\uv \mapsto F(\uv)$ has finite support, so the set of pairs $\uu, \uv \in \uG_0$ 
such that $dF(\overline{\uu\uv})\neq 0$ is finite, and $M_F$ is well defined and finite.
The coordinates 
$z_\epsilon(\uu)$ and $z_\epsilon(\uv)$ are distinct so
$$|z_\epsilon(\uu)-z_\epsilon(\uv)| > 0$$
provided
$1 + \epsilon \, dF(\uu,\uv)$ is non-zero, which is clearly the case 
whenever $\epsilon \leq M_F^{-1}$. 
\end{proof}

\begin{Def}
\label{DefDelDeform}
Let $\uG_0$ be a Delaunay graph with embedding $\uv \mapsto z_0(\uv)$ and let $F:\mathrm{V}(\uG_0) \rightarrow \Bbb{C}$
be a displacement function as above. Let $\epsilon \geq 0$ be a value for which the 
mapping $\uv \mapsto z_\epsilon(\uv)$ given by \ref{zDeform} is one-to-one.
The corresponding \textbf{Delaunay deformation} $\uG_\epsilon$ of $\uG_0$ is the unique Delaunay graph with vertex set 
$\mathrm{V}(\uG_\epsilon) =\mathrm{V}(\uG_0)$ for which the map  $\uv \to z_\epsilon(\uv)$ is a
planar graph embedding.
\end{Def}

Generically, the edge set of $\uG_\epsilon$ differs from 
the edge set of the initial graph $\uG_0$. This is caused by the Delaunay constraints, and this difference 
can occur spontaneously for $\epsilon>0$.
We offer two (not unrelated) examples.
Consider first a cyclic face $\uf$ of $\uG_0$ with $n>3$ vertices. As soon as $\epsilon>0$ these vertices may cease to be concyclic. In this case the Delaunay condition imposed on $\uG_\epsilon$ will force the appearance of new edges which will subdivide 
the initial face $\uf$ into new cyclic sub-faces. An example is depicted on Fig.~\ref{fmultiplesplit}.
In the limit $\epsilon\to 0_+$ these new edges would become \emph{chords} of the original face $\uf$ 
if they were adjoined to edge set of $\uG_0$ (see  Def.~\ref{ChordEdge} for the concept of chords and edges of a weak Delaunay graph).
\begin{figure}[h!]
\begin{center}
\includegraphics[height=1.6in]{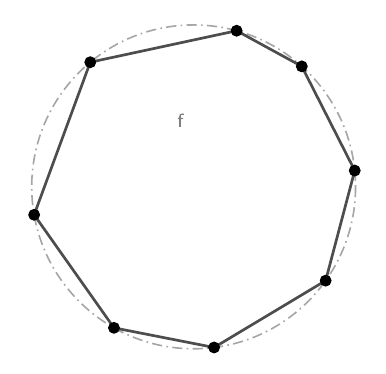}
\qquad \raisebox{.75 in}{$\longrightarrow$}
\qquad
\includegraphics[height=1.6in]{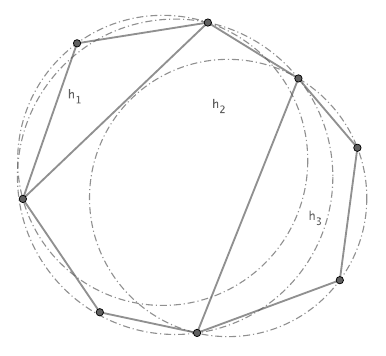}
\caption{Example of deformation of a general cyclic face of the Delaunay graph $\uG_0$ into several cyclic faces; here a cyclic octogon $f$ ($n=8$) splits into 3 cyclic polygons $h_1$, $h_2$ an $h_3$, a triangle ($n_1=3$), a pentagon ($n_2=5$) and a quadrilateral ($n_3=4$).}
\label{fmultiplesplit}
\end{center}
\end{figure}

This phenomenon can also occur around intermediate thresholds $\epsilon_0 > 0$ of the deformation parameter. 
Two (or more) faces of $\uG_\epsilon$ which are distinct
for $\epsilon < \epsilon_0$ may become concyclic and merge into a single face (the
boundary edges having vanished) when $\epsilon = \epsilon_0$. 
For $\epsilon>\epsilon_0$ this larger face may cease to be cyclic and instantaneously split into sub-faces caused by the appearance of new edges, 
possibly different from those which existed for $\epsilon < \epsilon_0$.
The prototypical example is depicted on Fig.~\ref{fLawsonFlip}. Two triangular faces for $\epsilon<\epsilon_0$ merge into a cyclic quadrilateral at 
$\epsilon=\epsilon_0$ and split again along the opposite diagonal of the quadrilateral for $\epsilon>\epsilon_0$. This is, of course, an example of a \textbf{Lawson flip} (well known from flip algorithm used to construct Delaunay triangulations) or, more generally, of a \textbf{Pachner move} on a two-dimensional simplicial complex.
\begin{figure}[h]
\begin{center}
\includegraphics[height=1.5in]{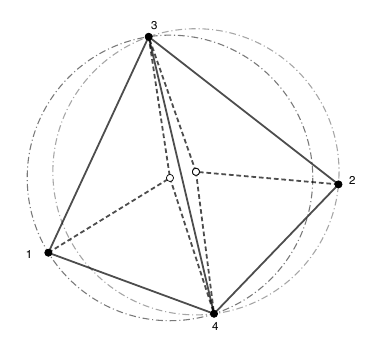}
\includegraphics[height=1.5in]{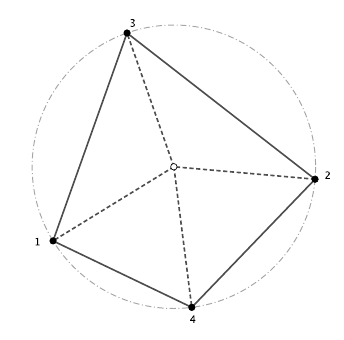}
\includegraphics[height=1.5in]{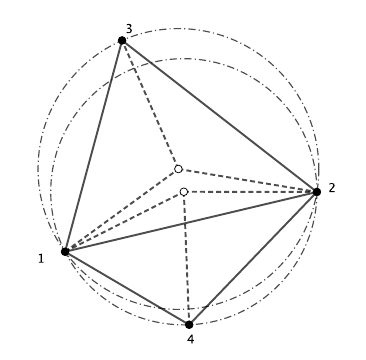}
\caption{Example of flip when two triangular faces become concyclic}
\label{fLawsonFlip}
\end{center}
\end{figure}

In the next section, we shall discuss how to control this phenomenon of face splittings and edge flips.
Let us first introduce two other concepts of graph deformations, which shall be used later.
\begin{Def}
\label{DefRigidDef}
Let $\uT_0$ be an initial, planar triangulation with embedding $\uv \mapsto z_0(\uv)$ 
and let $F\!:\mathrm{V}(\uT_0) \to\mathbb{C}$ be a displacement function as above.
For $\epsilon \geq 0$ the triangulation $\uT_{:\epsilon}$ with vertex set $\mathrm{V}(\uT_{:\epsilon})=\mathrm{V}(\uT_0)$
and edge set $\mathrm{E}(\uT_{:\epsilon})=\mathrm{E}(\uT_0)$ is called a \textbf{rigid deformation}
of $\uT_0$ if the mapping $\uv \mapsto z_\epsilon(\uv)$ given by \ref{zDeform} defines a planar embedding of $\uT_{:\epsilon}$. 
In particular this implies that the face set $\mathrm{F}(\uT_{:\epsilon})$ induced by the 
embedding coincides with the initial face set $\mathrm{F}(\uT_0)$.
Stated simply, no flips are allowed during a rigid deformation.
\end{Def}

For a Delaunay deformation
the mapping $\uv \mapsto z_\epsilon(\uv)$ will be an embedding
provided it is one-to-one since, by construction,
the edges determined by the Delaunay constraints
will never cross in the plane.
Injectivity can be achieved, for example, by
bounding the deformation parameter $0 \leq \epsilon < \epsilon_F' $ 
as prescribed in Lemma  \ref{MFbound}. For rigid deformations
the mapping $\uv \mapsto z_\epsilon(\uv)$ must be a planar embedding
with respect to a predetermined edge set $\mathrm{E}(\uT_0)$. 
One way to ensure this is to
regulate the deformation parameter $\epsilon \geq 0$
so that the area of a triangle $(z_\epsilon(\uu),z_\epsilon(\uv),z_\epsilon(\uw))$
given by formula (\ref{Aform}) remains positive
whenever $\uf=(\uu,\uv,\uw)$ is a triangle of $\uT_0$. This is addressed
in the following lemma:

\begin{lemma}
\label{LemRigidDef}
For an embedded planar triangulation $\uT_0$ and a displacement function $F$ as above, let
\begin{equation}
\label{M'FDef}
M'_F= \max_{\uf\in \mathrm{F}(\uT_0)}\ \max \big\{ |\nabla F(\uf)|,|\overline\nabla F(\uf)| \big\}
\end{equation}
where $\nabla$ and $\overline\nabla$ are the discrete derivative operators which map $\mathbb{C}^{\mathrm{V}(\uT_0)}\to\mathbb{C}^{\mathrm{F}(\uT_0)}$ as defined by \ref{nablaDef} and \ref{barnablaDef} in Sect.~\ref{FactoLaplNabl}.
Then the rigid deformation $\uT_{:\epsilon}$ is an embedded planar triangulation if
\begin{equation}
\label{RigidBound}
0\le\epsilon < \epsilon''_F={1\over 2\,M'_F}
\end{equation}
\end{lemma}
\begin{proof}
Let us consider a face $\uf_0$ of $\uT_0$  with embedding $(z_0(\uu),z_0(\uv),z_0(\uw))$ and the deformed face $\uf_\epsilon$ with embedding $(z_\epsilon(\uu),z_\epsilon(\uv),z_\epsilon(\uw))$.
Using formula \ref{Aform} for the area $A$ of a triangle, and the result \ref{diff-formula-nablas} of Rem.~\ref{remDiffNabla}, it is easy to show that the area of $\uf_\epsilon$ is related to the area of $\uf_0$ by
\begin{equation}
\label{ }
\begin{split}
A(\uf_\epsilon)&=A(\uf_0)\,( 1 +\epsilon (\nabla\! F(\uf_0) + \overline\nabla\! \bar F(\uf_0)) + \epsilon^2 (\nabla\! F(\uf_0) \overline\nabla\! \bar F(\uf_0) -\overline\nabla\! F(\uf_0) \nabla\! \bar F(\uf_0)))
\\
&=A(\uf_0)\,\left( ( 1 +\epsilon \nabla\! F(\uf_0) ) (1 +\epsilon\overline\nabla\! \bar F(\uf_0)) - \epsilon^2 \overline\nabla\! F(\uf_0) \nabla\! \bar F(\uf_0) \right)
\end{split}
\end{equation}
$A(\uf_0)$ is positive and for $0\le\epsilon<1/M'_F$ we have the inequality
\begin{equation}
\label{ }
A(\uf_\epsilon)\ge A(\uf_0)\left( ( 1- \epsilon\, M'_F)^2 - \epsilon^2 \, {M'}_F^2  \right)= A(\uf_0) (1-2\, \epsilon\,M'_F)
\end{equation}
Therefore $2\,\epsilon\, M'_F<1$ implies that $A(\uf_\epsilon) > 0$, so that the face $\uf_\epsilon$ is c.w. oriented as $\uf_0$. This bound is valid for all the faces of $\uT_{:\epsilon}$. This ends the proof.
\end{proof}
Note that the concept of rigid deformation can be extended from triangulations to more general embedded planar graphs, but we shall not need it. We shall, however, use the following concept.

\begin{Def}
\label{StableDef}
Let $\uG_0$ be a Delaunay graph and let $F$ a displacement function as above. 
A Delaunay deformation $\uG_\epsilon$ of $\uG_0$ as defined in Def.~\ref{DefDelDeform}
is said to be \textbf{stable} 
iff $\uG_\varepsilon$ is a Delaunay deformation of $\uG_0$ for all $0 \leq \varepsilon \leq \epsilon$ and
\begin{equation}
\label{StableDefCond}
\mathrm{E}(\uG_0) \subset\ \mathrm{E}(\uG_{\varepsilon}) \quad \text{and} \quad \mathrm{E}(\uG_{\varepsilon}) = \mathrm{E}(\uG_\epsilon)\quad 
\text{for all}\ 0<\varepsilon\le\epsilon
\end{equation}
This means that edges of the initial graph $\uG_0$ remain edges in $\uG_\varepsilon$ and
that the deformation creates a common set of new edges 
in $\uG_\varepsilon$ which persist (i.e. do not flip) within the range $0 < \varepsilon \leq \epsilon$. \end{Def}

\subsection{Keeping control of stable deformations}

\label{ssEpsilonBounds}
We now give results which allows controlling stable deformations of Delaunay graphs.

\begin{lemma}
\label{epsilonFboundT}
Let $\uT_0$ be an initial planar triangulation with embedding $\uv \mapsto z_0(\uv)$ 
and let $F\!:\mathrm{V}(\uT_0) \to\mathbb{C}$ be a displacement function
with finite support $\boldsymbol{\Omega}_F\subset \mathrm{V}(\uG_0)$, as in the previous section.
Let $\vec\ue=(\uu,\uv)$ be a given oriented edge of $\uT_0$, and let $\theta_0(\ue)$ be its conformal angle defined by  \ref{ConfAngleDef}. We do
not assume $\uT_0$ to be Delaunay, so the conformal angle $\theta_0(\ue)$ can be positive, zero, or negative.
Let  $\uf_\mathrm{n}=(\uu,\uv,\un)$ and  $\uf_\mathrm{s}=(\uv,\un,\us)$ be the adjacent north and south triangular faces of $\vec\ue$, and let
\begin{equation}
\label{MFEDef}
M_F(\ue)=\max\left\{  |dF(\uv,\un)|,|dF(\uv,\us)|,|dF(\uu,\un)|,|dF(\uu,\us)|\right\}
\end{equation}
with $dF(\uu_1,\uu_2)=(F(\uu_1)-F(\uu_2))/(z_0(\uu_1)-z_0(\uu_2))$ as in Lemma~\ref{MFbound}.
Take $\epsilon$ such that  $$0<\epsilon<\epsilon''_F$$ with $\epsilon''_F$ defined as in Lemma~\ref{LemRigidDef} and let us consider the rigid deformation $\uT_{:\epsilon}$ of $\uT_0$ (which is an embedded planar triangulation by Lemma~\ref{LemRigidDef}).
Let
$\theta_{\epsilon}(\ue)$ be the deformed conformal angle of the edge $\ue$ in $\uT_{:\epsilon}$.
Then we have the bound
\begin{equation}
\label{theLemma10}
0<\epsilon\,M_F(\ue) <\mathtt{b}\ \implies\ 
\left| \theta_{\epsilon}(\ue) - \theta_{0}(\ue) \right| \ <\ {1\over 2}\,\arcsin\left(\epsilon\,M_F(\ue)/\mathtt{b}\right)
\end{equation}
with the constant
\begin{equation}
\label{ }
\mathtt{b}= \sqrt{10}-3 =0.162278\cdots
\end{equation}
\end{lemma}

\begin{proof}
Consider the triangulation $\uT_0$ and an oriented edge $\vec{\ue}_0=(\uu,\uv)$ with adjacent north and south faces $\uf_\mathrm{n}=(\uu,\uv,\un)$ 
and $\uf_\mathrm{s}=(\uv,\un,\us)$.
Let $\uT_{:\epsilon}$ be the rigid deformation and let
$$z_\epsilon(\uu)=z_0(\uu) + \epsilon F(\uu)\ $$
be the corresponding embedding. By \ref{ConfAngleDef} the conformal angle of the edge $\ue$ in the initial triangulation 
can be expressed as 
\begin{equation}
\label{ }
\theta_0(\ue)={1\over 2}\,\arg\left(-{(z_0(\uu)-z_0(\un) (z_0(\uv)-z_0(\us))\over (z_0(\uu)-z_0(\us) (z_0(\uv)-z_0(\un))}\right)
\end{equation}
and the deformed conformal angle in $\uT_{:\epsilon}$ is given by
\begin{equation}
\label{thetaepsfrom0}
\begin{split}
\theta_\epsilon(\ue)&={1\over 2}\,\arg\left(-{(z_\epsilon(\uu)-z_\epsilon(\un) (z_\epsilon(\uv)-z_\epsilon(\us))\over (z_\epsilon(\uu)-z_\epsilon(\us) (z_\epsilon(\uv)-z_\epsilon(\un))}\right)
\\
&= \theta_0(\ue)+{1\over 2}\,\arg\left[ 
{
 \left( 1+\epsilon{F(\uu)-F(\un)\over z_0(\uu)-z_0(\un)}\right) 
 {
 \left( 1+\epsilon{F(\uv)-F(\us)\over z_0(\uv)-z_0(\us)}\right)
}
\over
{
 \left( 1+\epsilon{F(\uu)-F(\us)\over z_0(\uu)-z_0(\us)}\right)
}
{
 \left( 1+\epsilon{F(\uv)-F(\un)\over z_0(\uv)-z_0(\un)}\right)
}
}
\right]
\\
&= \theta_0(\ue)+{1\over 2}\,\arg\left[1+X(\ue)\right]
\end{split}
\end{equation}
where
\begin{equation}
\label{ }
X(\ue) = {
\epsilon\,
X_1(\ue)
+
\epsilon^2\,
X_2(\ue)
\over
{
 \left( 1+\epsilon{F(\uu)-F(\us)\over z_0(\uu)-z_0(\us)}\right)
}
{
 \left( 1+\epsilon{F(\uv)-F(\un)\over z_0(\uv)-z_0(\un)}\right)
}
}
\end{equation}
with
\begin{equation}
\label{ }
X_1(\ue) =  {F(\uu)-F(\un)\over z_0(\uu)-z_0(\un)} + {F(\uv)-F(\us)\over z_0(\uv)-z_0(\us)} - {F(\uu)-F(\us)\over z_0(\uu)-z_0(\us)} -  {F(\uv)-F(\un)\over z_0(\uv)-z_0(\un)}
\end{equation}
and
\begin{equation}
\label{ }
X_2(\ue) = {F(\uu)-F(\un)\over z_0(\uu)-z_0(\un)} \,{F(\uv)-F(\us)\over z_0(\uv)-z_0(\us)} - {F(\uu)-F(\us)\over z_0(\uu)-z_0(\us)} \, {F(\uv)-F(\un)\over z_0(\uv)-z_0(\un)}
\end{equation}
Consider $M_F(\ue)$ defined by \ref{MFEDef}.
Provided that $\epsilon\,M_F(\ue)<1$ we have
\begin{equation}
\label{boundX}
|X_1(\ue)|\le 4 M_F(\ue)\ ,\ \ |X_2(\ue)| \le 2 M_F(\ue)^2\ \implies\ \ |X(\ue)|\le 
{ 4\,\epsilon\,M_F(\ue) +2\,\epsilon^2\,M_F(\ue)^2 \over (1-\epsilon M_F(\ue))^2}
\end{equation}Define the function $Y(x)$ by
\begin{equation}
\label{YofXdef}
Y(x)= {4x + 2 x^2\over (1-x)^2}
\end{equation}
It is a monotone, convex function on the interval $x\in [0,1)$ with $Y(0)=0$ 
satisfying
\begin{equation}
\label{bttdef}
0\le x \le \mathtt{b}=\sqrt{10}-3 = 0.162278\cdots\quad\implies\quad 0\le Y(x)\le  x/\mathtt{b}
\end{equation}
Now we use the fact that for any complex number $x \in \Bbb{C}$ 
\begin{equation}
\label{arg1+Xbound}
|x|\le 1\ \implies\ | \!\!\arg(1+x)|\ \le\ \arcsin(|x|)
\end{equation}
Combining these inequalities, we deduce that
\begin{equation}
\label{CrucialBoundLemma}
\epsilon M_F(\ue) \le  \mathtt{b}\quad\implies\quad \left|\arg\left(1+X(\ue))\right] \right|\le \arcsin\left(Y(\epsilon M_F(\ue))\right)
\ \le\ \arcsin\left(\epsilon M_F(\ue)/\mathtt{b}\right)
\end{equation}
Combining this with \ref{thetaepsfrom0} we get \ref{theLemma10}.
\end{proof}
We now use this {lemma} to get our first result for a Delaunay deformation of a Delaunay graph.
\begin{lemma}
\label{Lemma11}
Let $\uG_0$ be a Delaunay graph and let $F$ be a displacement function 
$F:\,\mathrm{V}(\uG_0)\to\mathbb{C}$ with finite support $\boldsymbol{\Omega}_F\subset \mathrm{V}(\uG_0)$, as above.
To each edge $\ue\in \mathrm{E}(\uG_0)$ of $\uG_0$ we associate its conformal angle $\theta(\ue)$ defined by \ref{ConfAngleDef}.
Define $\vartheta_F $ as
\begin{equation}
\label{varthetaDef}
\vartheta_F  = \ \min \Big\{ \theta(\ue)  \, \Big| \,\ue= \overline{\uu\uv} \in \mathrm{E}(\uG_0) \ \text{such that}\ \uu\ \text{or}\ \uv\in\boldsymbol{\Omega}_F\Big\}
\end{equation}
and $M_F$ as defined by \ref{MFDef} in Lemma~\ref{MFbound}. 
Let $\uG_\epsilon$ be the Delaunay deformation of $\uG_0$ as defined in Def.~\ref{DefDelDeform}.
Then the following bound ensures that the edges of $\uG_0$ remain edges of $\uG_\epsilon$, namely:
\begin{equation}
\label{equLemma11}
\epsilon < \bar\epsilon_F=\sin(2\,\vartheta_F)\,{\mathtt{b}\over M_F}\ \implies\ \mathrm{E}(\uG_0)\subset\mathrm{E}(\uG_\epsilon)
\end{equation}
\end{lemma}

\begin{proof}
The proof uses Lemma~\ref{epsilonFboundT} and the Lawson flip algorithm.

Given our initial Delaunay graph $\uG_0$, let us consider a \textbf{triangular completion} $\uT_0$ of $\uG_0$, as introduced in Def.~\ref{CompletionTofG}.
In otherwords 
\begin{equation}
\label{ }
\uT_0\ \text{is a triangulation and} \ \, \mathrm{E}(\uG_0)\subset \mathrm{E}(\uT_0)
\end{equation}
Any completion $\uT_0$ is a weak Delaunay graph (see Def.~\ref{DelaunayGr})
and the edges of $\uT_0$ which are not edges of $\uG_0$ are chords; 
consequently 
\begin{equation}
\label{ }
\ue\notin\mathrm{E}(\uG_0)
\ \iff\  \theta_0(\ue)=0\ ,\quad \ue\in\mathrm{E}(\uG_0)
\ \iff\  \theta_0(\ue)>0
\end{equation}
for any edge  $\ue \in \mathrm{E}(\uT_0)$

In general $\uG_0$ may have multiple (possibly infinitely many) triangular completions. Let {\Fontauri{T}}$\,(\uG_0)$ denote the set of triangular completions of $\uG_0$, and let us extend the bounds $M'_F$ and $\epsilon''_F$ of Lemma~\ref{LemRigidDef} (valid for triangulations) to Delaunay graphs:
\begin{equation}
\label{MFDef2}
M'_F = \max_{\uT_0\in\,\text{\Fontauri{T}}\,(\uG_0)}\,\max_{\uf\in F(\uT_0)}\ \max \big\{ |\nabla F(\uf)|,|\overline\nabla F(\uf)| \big\} 
\end{equation}
and then as in Lemma~\ref{LemRigidDef}
\begin{equation}
\label{RigidBound2}
\epsilon''_F={1\over 2\,M'_F}
\end{equation}

Now we start the proof.
We choose an arbitrary triangular completion $\uT_0$ of $\uG_0$, and consider the rigid deformation $\uT_{:\epsilon}$ of $\uT_0$ for 
$\epsilon>0$ bounded by
\begin{equation}
\label{epsboundsLemma11}
\epsilon < \sin(2\,\vartheta_F)\,{\mathtt{b}\over M_F}\quad\text{and}\quad \epsilon < \epsilon''_F
\end{equation}
For any edge $\ue$ of $\uT_0$ notice that
\begin{equation}
\label{ }
\epsilon < \sin(2\,\vartheta_F)\,{\mathtt{b}\over M_F}\ \implies\ \epsilon< {\mathtt{b}\over M_F(\ue)}
\end{equation} 
and by Lemma~\ref{epsilonFboundT} we have
\begin{equation}
\label{ }
\theta_\epsilon(\ue) > \theta_0(\ue)-{1\over 2}\ \arcsin(\epsilon\,M_F(\ue)/\mathtt{b})
\end{equation}
If the edge $\ue$ of $\uT_{:\epsilon}$ is also an edge of $\uG_0$ then $\theta_0(\ue)\ge \vartheta_F$.
Cleary $M_F(\ue)\le M_F$ and therefore
\begin{equation}
\label{theteps>0}
\theta_\epsilon(\ue)>\vartheta_F-{1\over 2}\ \arcsin(\epsilon\,M_F/\mathtt{b})>0
\end{equation}
So the initial edges of $\uG_0$ still satisfy the Delaunay condition in $\uT_{:\epsilon}$.

Now we consider whether or not the deformed triangulation $\uT_{:\epsilon}$ is weakly Delaunay, i.e. if
\begin{equation*}
\label{ }
\theta_\epsilon(\ue)\ge 0\quad\text{for all}\quad \ue\in \mathrm{E}(\uT_{:\epsilon})
\end{equation*}
If $\uT_{:\epsilon}$ is weakly Delaunay, it is sufficient to remove all its chords, namely all edges such that
$\theta_\epsilon(\ue)= 0$. We obtain the redacted graph $\uT_{:\epsilon}^\bullet$ (see Def.~\ref{defRegGraph}) which is a Delaunay graph 
with the same vertex set as $\uG_0$ and with embedding $\uv\to z_\epsilon(\uv)= z_0(\uv)+\epsilon\,F(\uv)$. Hence it is the Delaunay deformation 
$\uG_\epsilon$ of $\uG_0$, and it contains the original edges of $\uG_0$ in light of \ref{theteps>0}. In short
\begin{equation}
\label{ }
\uG_0\to\uT_0\to\uT_{:\epsilon}\to \uT_{:\epsilon}^\bullet = \uG_\epsilon
\end{equation}
If $\uT_{:\epsilon}$ is not weakly Delaunay, there must exist edges $\ue$ of $\uT_{:\epsilon}$ such that
$$\theta_\epsilon(\ue)<0\ .$$
In this case we recursively apply the Lawson flip algorithm to construct from $\uT_{:\epsilon}$ 
a weak Delaunay triangulation $\uT^{\scriptscriptstyle{\mathtt{Del}}}_{:\epsilon}$ which still completes $\uG_0$ (see \cite{Lawson77} and standard textbooks such as  \cite{deBergetalBook2008} or \cite{Gartner-Hoffmann}). Let us describe the first iterative step.
\begin{enumerate}
  \item Choose an edge $\ue$ of $\uT_{:\epsilon}$ such that $\theta_\epsilon(\ue)<0$, and consider the quadrilateral $(\uu,\us,\uv,\un)$ made of its north and south faces.
  \item  {Flip} the edge $\ue$, i.e. perform the replacement
  $$\ue= \overline{\uu\uv} \ \to\ \ue'=  \overline{\un\us}$$ so that one obtains a new triangulation $\uT'_{:\epsilon}$.
\end{enumerate}
Some conformal angles in  $\uT'_{:\epsilon}$ have changed. Specifically
$$
\theta'_\epsilon(\ue')=-\theta_\epsilon(\ue)>0
$$
and the conformal angles $\theta'_\epsilon$ of the edges $\overline{\uu\us}$,  $\overline{\uv\un}$, $\overline{\uv\us}$, and $\overline{\uv\un}$ 
as measured in $\uT'_{:\epsilon}$ may differ from their corresponding measures $\theta_\epsilon$ in $\uT_{:\epsilon}$.
The new triangulation $\uT'_{:\epsilon}$ is the rigid deformation of \emph{another triangular completion} $\uT'_0$ of $\uG_0$, namely the 
triangulation with edge set
$$\mathrm{E}(\uT'_{0}) = \mathrm{E}(\uT_{0})\backslash\{\ue\}\cup\{\ue'\}\ .$$
Therefore $\mathrm{E}(\uG_0)\subset\mathrm{E}(\uT'_{:\epsilon})$ and inequality \ref{theteps>0} is still valid for the edges of $\uG_0$, i.e.
$$\ue\in\mathrm{E}(\uG_0)\ \implies\ \theta'_\epsilon(\ue)>0\ .$$

The Lawson flip algorithm entails iterating this process: Choose an edge $\ue$ in $\uT'_{:\epsilon}$ such that $\theta'_\epsilon(\ue) < 0$ and 
perform the edge flip $\ue \to \ue''$ to
obtain a new triangulation $\uT''_{:\epsilon}$ with $\theta''_\epsilon(\ue'') > 0$. Repeat. This process is known to stop after a finite number of iterations, 
and the final triangulation $\uT^{\scriptscriptstyle{\mathtt{Del}}}_{:\epsilon}$ 
will have no edge $\ue$ with ${\theta_{:\epsilon}}^{\!\!\!\!\scriptscriptstyle{\mathtt{Del}}}(\ue)<0$, and so it will be weakly Delaunay.
Clearly $\uT^{\scriptscriptstyle{\mathtt{Del}}}_{:\epsilon}$ is the 
rigid deformation of a triangular completion $\uT^{\scriptscriptstyle{\mathtt{Del}}}_{0}$ of $\uG_0$. Now
take the redaction ${\uT^{{\scriptscriptstyle{\mathtt{Del}}}\,\bullet}_{:\epsilon}}$ by removing any chords. Schematically 
$$\uG_0\to\uT_{:\epsilon}\to \uT'_{:\epsilon}\to\cdots\to \uT^{\scriptscriptstyle{\mathtt{Del}}}_{:\epsilon}\to {\uT^{{\scriptscriptstyle{\mathtt{Del}}}\,\bullet}_{:\epsilon}}=\uG_\epsilon$$
The redacted graph ${\uT^{{\scriptscriptstyle{\mathtt{Del}}}\,\bullet}_{:\epsilon}}$ is a 
Delaunay deformation of $\uG_0$ which coincides with $\uG_\epsilon$ and, to be sure,
$$\mathrm{E}(\uG_0)\subset \mathrm{E}(\uG_\epsilon)\ $$
as long as the initial bounds \ref{epsboundsLemma11} on $\epsilon$ are satisfied.
\end{proof}

Lemma \ref{Lemma11} says nothing about the additional edges which can appear and flip within the initial faces of $\uG_0$ during the deformation. 
The following proposition establishes that these additional edges are themselves stable, i.e. undergo no flips, for 
values of the deformation parameter $\epsilon > 0$ which are sufficiently small.

\begin{Prop}
\label{epsilontildeF}
Let $\uG_0$ be a Delaunay graph and $F$ a displacement function as above. There exists a deformation threshold $\tilde{\epsilon}_F > 0$ such that
for any $0<\epsilon<\tilde{\epsilon}_F$ the deformation $\uG_{\epsilon}$ is stable (see Def.~\ref{StableDef}). 
As a consequence, the limit of the Delaunay graph $\uG_\epsilon$ when $\epsilon\to 0^+$ is unambiguously defined, 
and is denoted $\uG_{0^+}$
\begin{equation}
\label{G0+def}
\uG_{0^+}=\lim_{\epsilon\to 0^+} \uG_\epsilon
\end{equation} 
$\uG_{0^+}$ is a weak-Delaunay graph sharing the same vertex set and embedding as $\uG_0$. Its redacted 
graph (see Def.~\ref{defRegGraph}) is the initial Delaunay graph, i.e.  $\uG_{0^+}^\bullet=\uG_0$.
\end{Prop}

\begin{proof}
Let us consider $\uT_0$ be a triangular completion of $\uG_0$ (an element of {\Fontauri{T}}$\,(\uG_0)$) and an edge $\ue= \overline{\uu\uv}$ of 
$\uT_0$ which is not an edge of $\uG_0$ (i.e. a chord such that its conformal angle is $\theta_0(\ue)=0$.
Now as in the proof of Lemma~\ref{Lemma11} consider the rigid deformation $\uT_{:\epsilon}$ of $\uT_0$. The deformed conformal angle of $\ue$ is given by
\begin{align}
\label{thetaepsilon}
 \theta_{:\epsilon} (\ue)
&\D= \ {1\over 2}\mathrm{Arg}\left(-[z_\epsilon(\uu),z_\epsilon(\uv);z_\epsilon(\un),z_\epsilon(\us)]\right)\nonumber\\
&\D= \ 
{1\over 2}
\mathrm{Arg}
\left[{ {\Big( 1 \, + \, \epsilon \, dF \big( \uu,\un \big)  \Big) \cdot 
\Big( 1 \, + \, \epsilon \, dF \big( \uv,\us \big)  \Big)} \over 
{\Big( 1 \, + \, \epsilon \, dF \big( \uu, \us \big)  \Big) \cdot 
\Big( 1 \, + \, \epsilon \, dF \big( \uv, \un \big)  \Big)} } \right]
\end{align}
with $\un$ and $\us$ the north and south vertices for the north and south faces $\uf_\mathrm{n}$ and $\uf_\mathrm{s}$ of the edge $\ue$ in $\uT_{:\epsilon}$ (remember that for $\epsilon=0$ this is zero).

$\theta_{:\epsilon}$ is a regular function of $\epsilon$ (for $\epsilon$ small enough). We are interested in the values of $\epsilon$ 
for which $\theta_{:\epsilon} (\ue)$ vanishes. Clearly this occurs if ($\epsilon$ is taken real)
\begin{equation}
\label{ }
{ {\Big( 1  +  \epsilon \, dF \big( \uu,\un \big)  \Big) 
\Big( 1  +  \epsilon \, dF \big( \uv,\us \big)  \Big)} \over 
{\Big( 1  +  \epsilon \, dF \big( \uu, \us \big)  \Big) 
\Big( 1  +  \epsilon \, dF \big( \uv, \un \big)  \Big)} } 
{ {\Big( 1  +  \epsilon \, \overline{dF} \big( \uu,\us \big)  \Big) 
\Big( 1  +  \epsilon \, \overline{dF} \big( \uv,\un \big)  \Big)} \over 
{\Big( 1  +  \epsilon \, \overline{dF} \big( \uu, \un \big)  \Big) 
\Big( 1  +  \epsilon \, \overline{dF} \big( \uv, \us \big)  \Big)} } 
=1
\end{equation}
This amounts to solving a quartic real polynomial equation in $\epsilon$ of the form
\begin{equation}
\label{ }
\text{\Fontauri{P}}_4(\epsilon)=0\ ,\quad \text{\Fontauri{P}}_4\ \text{a degree 4 real polynomial}
\end{equation}
with
\begin{equation}
\label{ }
\begin{split}
\text{\Fontauri{P}}_4(\epsilon)=&{\Big( 1  +  \epsilon \, dF \big( \uu,\un \big)  \Big) 
\Big( 1  +  \epsilon \, dF \big( \uv,\us \big)  \Big)}
{\Big( 1  +  \epsilon \, \overline{dF} \big( \uu,\us \big)  \Big) 
\Big( 1  +  \epsilon \, \overline{dF} \big( \uv,\un \big)  \Big)}\\
-&{\Big( 1  +  \epsilon \, dF \big( \uu, \us \big)  \Big) 
\Big( 1  +  \epsilon \, dF \big( \uv, \un \big)  \Big)}
{\Big( 1  +  \epsilon \, \overline{dF} \big( \uu, \un \big)  \Big) 
\Big( 1  +  \epsilon \, \overline{dF} \big( \uv, \us \big)  \Big)}
\end{split}
\end{equation}
{\Fontauri{P}}$_4$ has at least one zero (with multiplicity) at $\epsilon=0$, and at most three other zeros, unless it is identically zero.
Let us define $\epsilon_c(\ue)$ by
\begin{equation}
\label{ }
\epsilon_c(\ue)=
\begin{cases}
      &+\infty\  \text{if {\Fontauri{P}}$_4$ is identically zero}, \\
      & \text{the smallest strictly positive root of {\Fontauri{P}}$_4$, if it exists}\\
      &+\infty\ \text{if it does not exist}.
\end{cases}
\end{equation}
For the completion $\uT_0$ define
\begin{equation}
\label{ }
\epsilon_c(\uT_0)=\min_{\ue\in\mathrm{E}(\uT_0)} \epsilon_c(\ue)
\end{equation}

$\epsilon_c(\uT_0)$ is strictly positive (possibly infinite) since 
$F$ has finite support on $\mathrm{V}(\uT_0)=\mathrm{V}(\uG_0)$ and there are only finitely many 
chords $\ue$ affected by the deformation.
For the initial Delaunay graph $\uG_0$ define
\begin{equation}
\label{ }
\epsilon_c(\uG_0)=\min_{\uT_0\in\text{\Fontauri{T}}\,(\uG_0)} \epsilon_c(\uT_0)
\end{equation}
Again $F$ has finite support on $\mathrm{V}(\uG_0)$ and since the $\epsilon_c(\uT_0)$'s can only take a 
finite number of distinct (strictly positive) values it must be the case that $\epsilon_c(\uG_0)$ is strictly positive.
\begin{equation}
\label{ }
\epsilon_c(\uG_0)>0
\end{equation}
Now define 
\begin{equation}
\label{ }
\tilde\epsilon_F=\min\{\epsilon_c(\uG_0), \epsilon''_F,\bar\epsilon_F \}
\end{equation}
Take $\epsilon$ within the range
\begin{equation}
\label{ }
0<\epsilon<\tilde\epsilon_F
\end{equation}
and construct $\uT^{\mathtt{Del}}_{:\epsilon}$ according to the proof of Lemma~\ref{Lemma11}. 
It is weakly Delaunay, hence $\theta_\epsilon(\ue)\ge 0$ for each edge. By the argument above, the conformal angles cannot change sign 
in the interval $0 \leq \epsilon \leq \epsilon_c(\uG_0)$. 
Therefore, for any $\epsilon'\le \epsilon$,  each conformal angle must stay nonnegative, 
and so  $\uT^{\mathtt{Del}}_{:\epsilon'}$ remains weakly Delaunay.
This implies that its redacted graph is the Delaunay deformation $\uG_{\epsilon'}$ of $\uG_0$ and that 
$\mathrm{E}(\uG_{\epsilon'}) = \mathrm{E}(\uG_\epsilon)$.
In other words, $\uG_{\epsilon}$ is a stable deformation of $\uG_0$.

The graph $\uG_{0^+}$ stipulated in \ref{G0+def}, exists: 
{it shares the same vertex set and embedding as $\uG_0$, while its edge set
coincides with $\mathrm{E}(\uG_\epsilon)$ for any $0 < \epsilon < \tilde{\epsilon}_F$ by 
stability. In particular each edge $\ue$ of $\uG_0$ is an edge of $\uG_{0^+}$ while
the remaining edges of $\uG_{0^+}$ are all chords.}
\end{proof}

\begin{Rem}
\label{tildEpsvsEps}
The bound $\tilde\epsilon_F$ (which defines an interval $0<\epsilon<\tilde\epsilon_F$ where no flips occur) may be much smaller than $\epsilon_F$. In fact, even for a fixed initial Delaunay graph $\uG_0$ and a generic displacement function $F$, the threshold $\tilde\epsilon_F$ may be arbitrarily small w.r.t. $\epsilon_F$. This point will become relevant when discussing the scaling limit and the problem of obtaining uniform bounds with respect to the choice of 
$\uG_0$. We return to this issue in Sect.~\ref{aVarOpFlips}.
\end{Rem}

\begin{Rem}
As discussed in the proof of Prop.~\ref{epsilontildeF}, the edge sets
$\mathrm{E}(\uG_\epsilon)$ and $\mathrm{E}(\uG_{0^+})$ coincide for $0\le\epsilon<\tilde\epsilon_F$.
Consequently any stable Delaunay deformation $\uG_\epsilon$ of $\uG_0$
is also a rigid deformation of the corresponding limit graph $\uG_{0^+}$
within the range $0\le\epsilon<\tilde\epsilon_F$. Accordingly
the notions of stable and rigid deformation agree for the limit graph 
$\uG_{0^+}$ provided we work with sufficiently small values of the 
deformation parameter.
\end{Rem}

\begin{Rem}
Since $\mathrm{E}(\uG_\epsilon)=\mathrm{E}(\uG_{0^+})$ for $0\le\epsilon<\tilde\epsilon_F$ and since the faces of $\uG_{0^+}$ are cyclic polygons, the conformal angle $\theta_\epsilon(\ue)$ of any edge $\ue \in \mathrm{E}(\uG_\epsilon)$ is unambiguously defined and strictly positive
as $\epsilon$ varies in the interval $0\le\epsilon<\tilde\epsilon_F$.
\end{Rem}

\subsection{Variation of operators under rigid deformations}
\label{ssVarOp}
We now study the variations of the operators $\Delta(\epsilon)$, $\mathcal{D}(\epsilon)$ and $\Deltaconf(\epsilon)$
and of the associated local geometrical quantities arising from {\bf rigid} deformations (see Def.~\ref{DefRigidDef}) of triangulations.
{It will not be necessary to assume that the triangulations are Delaunay at \underline{this} stage.}

Let $\uT$ be an initial triangulation (possibly Delaunay) with vertex set $\mathrm{V}(\uT)$, edge set $\mathrm{E}(\uT)$ and  face set (triangles) $\mathrm{F}(\uT)$. Let $\uT_{:\epsilon}$ be the rigid deformation of $\uT$
induced by the deformed embedding
\begin{equation}
\label{zPertF}
z_\epsilon(\uv) \, = \  z(\uv) + \epsilon \, F(\uv)
\end{equation}
where $F \in \mathbb{C}^{\mathrm{V}(\uT)}$ is a displacement function
and where $\epsilon \geq 0$ is bounded by the threshold $\epsilon_F''$ 
defined in Lemma~\ref{LemRigidDef}.
Recall that $\uT_{:\epsilon}$ and $\uT$ share the same set of vertices, edges and faces.
The deformed discrete differential operators are denoted $\nabla_\epsilon, \overline{\nabla}_\epsilon: 
\Bbb{C}^{\mathrm{V}(\uT)} \longrightarrow \Bbb{C}^{\mathrm{F}(\uT)}$ while
the deformed area and radius operators are denoted
$A_\epsilon, R_\epsilon:  \Bbb{C}^{\mathrm{F}(\uT)}
\longrightarrow  \Bbb{C}^{\mathrm{F}(\uT)}$. They are obtained by
making the substitution $z \mapsto z_\epsilon$
in formulae \ref{nablaDef}, \ref{barnablaDef}, \ref{Aform},
and \ref{Rcform} respectively. This allows us to
unambiguously define deformed versions
$\Delta(\epsilon)$ and $\mathcal{D}(\epsilon)$
of the Beltrami-Laplace and discrete K\"ahler 
operators using the factorizations
\ref{DeltaNablaForm} and \ref{DNablaForm}, namely:
\begin{equation}
\label{deformed-factorizations}
\Delta(\epsilon) = 2 \left(\overline\nabla_\epsilon^{\!\top}\! {{A_\epsilon}}\,\nabla_\epsilon
+ \nabla_\epsilon^{\!\top}\! {{A_\epsilon}}\,\overline{\nabla}_\epsilon\right)
\quad \text{and} \quad \mathcal{D}(\epsilon) = 4\,\overline{\nabla}_\epsilon^\top\! {{A_\epsilon}\over{R_\epsilon^2}} \, \nabla_\epsilon
\end{equation}
We may expand all the relevant operators 
as (formal) series in $\epsilon$ (they are in fact meromorphic in $\epsilon$). Up to first order in $\epsilon$, the terms in these
developments can be compactly expressed using the discrete derivatives $\nabla F$
and $\overline{\nabla} F$ with respect to the triangulation $\uT$.

\begin{Prop}
\label{PVarDelta}
The variation of the Laplace-Beltrami operator is
\begin{equation}
\label{VarDelta}
\Delta(\epsilon) \ = \ \Delta - 4 \epsilon \left( \nabla^{\!\top} (A\,\overline\nabla F)\,\nabla + \overline\nabla^{\!\top}(A\,\nabla \bar F)\,\overline\nabla\, \right)+\mathrm{O}(\epsilon^2)
\end{equation}
\end{Prop}

\begin{Prop}
\label{PVarKahler}
The variation of the K\"ahler operator is
\begin{equation}
\label{VarD}
\begin{split}
\mathcal{D}(\epsilon) \ = \ \mathcal{D}- 4 \epsilon \left[\overline\nabla^{\!\top} {A\over R^2} \left(\nabla F +\overline\nabla \bar F + C\,\overline\nabla F + \bar C\,\nabla\bar F\right) \nabla \right.\ &
\\
+\left. \nabla^{\!\top} {A\over R^2} (\overline\nabla F)\nabla+\overline\nabla^{\!\top} {A\over R^2} (\nabla \bar F)\overline\nabla\right]
&+\mathrm{O}(\epsilon^2)
\end{split}
\end{equation}
with the diagonal function $C \in \mathbb{C}^{\mathrm{F}(\uT)}$ and its conjugate $\bar C$ which are given for a triangle 
$\uf =(\uu, \uv, \uw)$ by
\begin{equation}
\label{CDef}
C(\uf)=\left( {\bar{z}(\uu) -\bar{z}(\uv) \over {z(\uu) -z(\uv)}} +  {\bar{z}(\uv) -\bar{z}(\uw) \over {z(\uv) -z(\uw)}}  
+ {\bar{z}(\uw) -\bar{z}(\uu) \over {z(\uw) -z(\uu)}} \right)
\ , \quad
\bar C(\uf)= \overline{C(\uf)}
\end{equation}
\end{Prop}

Before deriving these two equations, let us note that the variation for $\Delta$ is rather simple, while the variation for $\mathcal{D}$ is more complicated, since we have not found a simple interpretation for the quantities $C$ and $\overline C$ in terms of the geometry of the triangle $\uf$.

\begin{proof}
From \ref{diff-formula-nablas}, 
for a pair of vertices $\uu$ and $\uv$ of a triangle $\uf = (\uu, \uv, \uw)$ in $\mathrm{F}(\uT)$
\begin{equation}
\label{varZ}
\begin{split}
z_\epsilon(\uu) - z_\epsilon(\uv)
&= z(\uu) -z(\uv) +\epsilon \, \Big( (z(\uu) - z(\uv))\nabla F(\uf) + (\bar{z}(\uu) -\bar{z}(\uv)) \overline\nabla F(\uf) \Big) \\
\bar{z}_\epsilon(\uu) - \bar{z}_\epsilon(\uv)
&= \bar{z}(\uu) - \bar{z}(\uv) + \epsilon \, \Big( (z(\uu) -z(\uv)) \nabla \bar F(\uf) + (\bar{z}(\uu) -\bar{z}(\uv))\overline\nabla \bar F(\uf) \Big)
\end{split}
\end{equation}
Inserting this in \ref{Aform} gives the variation of the area of the triangle $\uf$
\textcolor{black}{
\begin{equation}
\label{varAfirst}
A_\epsilon(\uf) = A(\uf)+\epsilon\, A(\uf) \big(\nabla F(\uf)+{\overline\nabla}\bar F(\uf) \big) 
+ \mathrm{O}(\epsilon^2)
\end{equation}}
which we can succinctly express as
\textcolor{black}{
\begin{equation}
\label{var1A}
A_\epsilon = A+\epsilon\, A (\nabla\! F+{\overline\nabla}\!\bar F) 
+ \mathrm{O}(\epsilon^2)
\end{equation}}
where we view $A$, $\nabla\!F$, and  $\overline\nabla\!\bar F$
as functions in $\Bbb{C}^{\mathrm{F}(\uT)}$ or, alternatively,
as diagonal operators mapping $\Bbb{C}^{\mathrm{F}(\uT)} \rightarrow \Bbb{C}^{\mathrm{F}(\uT)}$.

Using \ref{Rcform} we can write the variation of the circumradius $R(\uf)$ of the face $\uf$. 
We write only the leading term of order $\mathrm{O}(\epsilon)$ with the same compact notation and with $C$, and $\bar C$ defined by \ref{CDef}
\begin{equation}
\label{var1AR2}
{A_\epsilon \over R_\epsilon^2} = {A\over R^2}-\epsilon\, {A\over R^2} \left(C\,\overline\nabla\! F +\bar C\, \nabla\! \bar F \right)+\mathrm{O}(\epsilon^2)
\end{equation}
Similarly, we get the variation of the matrix elements of the $\nabla$ operator.  At first order
\begin{equation}
\label{VarNablaExpl}
\big[\nabla_\epsilon\big]_{\uf,\uv} =  \nabla_{\uf, \uv}-\epsilon\, \left(\nabla\! F(\uf)\, \nabla_{\uf, \uv} + \nabla\!\bar F(\uf)\, \overline\nabla_{\uf,\uv} \right)+\mathrm{O}(\epsilon^2)
\end{equation}
When read as operators, formula~\ref{VarNablaExpl} for $\nabla_\epsilon$
and its complex conjugate become
\begin{equation}
\label{VarNabla}
\begin{split}
\nabla_\epsilon = \nabla-\epsilon\, \left(\nabla\! F \,\nabla + \nabla\!\bar F \,\overline\nabla \right)
+\mathrm{O}(\epsilon^2)\\
\overline{\nabla}_\epsilon =  \overline\nabla- \epsilon\, \left(\overline\nabla\! \bar F \,\overline\nabla + \overline\nabla\! F \,\nabla \right)
+\mathrm{O}(\epsilon^2)
\end{split}
\end{equation}
Combining this with \ref{deformed-factorizations} and the Leibnitz product rule
we get \ref{VarDelta} and \ref{VarD}. 
\end{proof}

\begin{Rem}
Note that the exact formulae (to all orders in $\epsilon$) for these variations \ref{var1A}-\ref{VarNabla} are derived in Section~\ref{aVarOpFlips} 
(see in particular Eqn. \ref{full-variation-nabla-epsilon} in section \ref{ssFullVar}).
\end{Rem}

\begin{Rem}
\label{rVArConf}
There is no such a compact expression for the variation of the conformal Laplacian $\Deltaconf$ in the general case. In particular, the variation of the weight associated to an edge $\ue= \overline{\uu\uv}$ will depend on the discrete derivatives of $F$ both at the north triangle $\uf_\mathrm{n}$ and the south triangle 
$\uf_\mathrm{s}$ of the oriented edge $\vec{\ue} = (\uu, \uv)$, which are a priori independent (see figure~\ref{triangles}).
\end{Rem}
We can, of course, make the substitution $z \mapsto z_\epsilon$ in formula \ref{thetaform} for the north and south 
angles (which express the angles as a difference of arguments of edge vectors) 
\begin{equation}
\begin{array}{ll}
\D \theta_\mathrm{n}(\vec{\ue}, \epsilon ) 
&\D := \ 
{1\over 2\imath} \, \log \left(-{(\bar{z}_\epsilon(\uv) - \bar{z}_\epsilon(\un))( z_\epsilon(\uu)  
- z_\epsilon(\un))\over (z_\epsilon(\uv) -z_\epsilon(\un)) (\bar{z}_\epsilon(\uu) - \bar{z}_\epsilon(\un))}\right) \\ \\
\D \theta_\mathrm{s}(\vec{\ue}, \epsilon ) 
&\D := \ 
{1\over 2\imath} \, \log \left(-{(\bar{z}_\epsilon(\uu) - \bar{z}_\epsilon(\us))( z_\epsilon(\uv)  
- z_\epsilon(\us))\over (z_\epsilon(\uu) -z_\epsilon(\us)) (\bar{z}_\epsilon(\uv) - \bar{z}_\epsilon(\us))}\right)
\end{array}
\end{equation}
where $\un \in \uf_\mathrm{n}$ and $\us \in \uf_\mathrm{s}$ are the
respective north and south vertices of the adjacent triangles $\uf_\mathrm{n}$ and $\uf_\mathrm{s}$ to the edge $\vec\ue=(\uu,\uv)$ as depicted on Fig.~\ref{triangles}.
The order zero and order one terms in the formal $\epsilon$ series expansion read from \ref{varZ}
\begin{equation}
\label{VarThetaNS}
\begin{array}{ll}
\D \theta_\mathrm{n}(\vec{\ue}, \epsilon) \ = \
&\D \theta_\mathrm{n}(\vec{\ue} \, ) \, + \, \epsilon \, {\imath\over 2}
\Big( \overline\nabla F(\uf_\mathrm{n} ) \, \mathcal{E}_\mathrm{n}(\vec{\ue} \,)
-\nabla\bar F(\uf_\mathrm{n}) \,  \overline{\mathcal{E}}_\mathrm{n}(\vec{\ue} \,) \Big)
\, + \, \mathrm{O}(\epsilon^2)  \\ \\
\D \theta_\mathrm{s}(\vec{\ue}, \epsilon) \ = \
&\D \theta_\mathrm{s}(\vec{\ue} \, ) \, + \, \epsilon \, {\imath\over 2}
\Big( \overline\nabla F(\uf_\mathrm{s} ) \, \mathcal{E}_\mathrm{s}(\vec{\ue} \,)
-\nabla\bar F(\uf_\mathrm{s}) \,  \overline{\mathcal{E}}_\mathrm{s}(\vec{\ue} \,) \Big)
\, + \, \mathrm{O}(\epsilon^2) 
\end{array}
\end{equation}
with complex coefficients $\mathcal{E}_\mathrm{n}(\vec{\ue})$ and $\mathcal{E}_\mathrm{s}(\vec{\ue})$ given by
\begin{equation}
\label{mathcalE-terms}
\begin{array}{lll}
\D \mathcal{E}_\mathrm{n}(\vec{\ue} \, ) 
&\D := \ { \overline{z}(\uv) - \overline{z}(\un) \over {z(\uv) - z(\un) } } - 
{ \overline{z}(\uu) - \overline{z}(\un) \over {z(\uu) - z(\un) } } 
&\D = \ {-4 A(\uf_\mathrm{n}) \over {\big(z(\uv) - z(\un) \big) \big(z(\uu) - z(\un) \big)}}
\\ \\
\D \mathcal{E}_\mathrm{s}(\vec{\ue} \, ) 
&\D := \ { \overline{z}(\uu) - \overline{z}(\us) \over {z(\uu) - z(\us) } } - 
{ \overline{z}(\uv) - \overline{z}(\us) \over {z(\uv) - z(\us) } } 
&\D = \ {-4 A(\uf_\mathrm{s}) \over {\big(z(\uv) - z(\us) \big) \big(z(\uu) - z(\us) \big)}}
\\ \\
\end{array}
\end{equation}
The corresponding first order variation of the edge weight $\tan\theta(\ue)$ with $\theta(\ue)=(\theta_\mathrm{n}(\vec{\ue})+\theta_\mathrm{s}(\vec{\ue}))/2$  can be written explicitly in term of
the discrete derivatives $\overline\nabla F(\uf_\mathrm{n})$ and $\overline\nabla F(\uf_\mathrm{s})$, the coefficients $\mathcal{E}_\mathrm{n}(\vec{\ue})$ and
$\mathcal{E}_\mathrm{s}(\vec{\ue})$, and their complex conjugates.
{We shall not write the formula here, since it lacks the simplicity and geometrical interpretation 
of our results for $\Delta$ and $\mathcal{D}$.}

\subsection{Generic notation for derivatives under graph deformations} 
\label{ssGenVarNot}
We shall use the following compact notation for derivatives and variations of general objects $\Obj$ associated to a rigid deformation 
$\uG\to\uG_\epsilon$ of a polygonal (Delaunay) graph induced by a deformed coordinate embedding $z\to z_\epsilon=z+\epsilon\, F$ 
as defined in Def.~\ref{DefRigidDef}.
The object $\Obj$ can be a local quantity such as the angle $\theta$, $\theta_{\mathrm{n}}$, $\theta_{\mathrm{s}}$ associated to oriented edge $\vec\ue$ 
of $\uG$ or the area $A$ and circumradius $R$ of a face $\uf$. Other objects include the operators $\Delta$, $\Deltaconf$ and 
$\mathcal{D}$.

If the object $\Obj$ is defined on the unperturbed graph $\uG$, the corresponding object on the deformed graph $\uG_\epsilon$ is denoted
\begin{equation}
\label{ }
\Obj(\epsilon)\quad\text{or sometimes}\quad\Obj_\epsilon\ \text{ (for clarity or brevity)}
\end{equation}
This is consistent with the notations of Sect.~\ref{ssVarOp}.
The variation of $\Obj$ for finite $\epsilon$ is denoted
\begin{equation}
\label{ }
\delta\Obj(\epsilon)=\Obj(\epsilon)-\Obj
\end{equation}
The initial derivatives w.r.t. $\epsilon$ are denoted
\begin{equation}
\label{ }
{\partial\over\partial\epsilon}\Obj(\epsilon) = \deltae\Obj(\epsilon)
\ ,\quad
{\partial^2\over\partial\epsilon^2}\Obj(\epsilon) = \deltaee\Obj(\epsilon)
\ ,\quad
\text{etc.}
\end{equation}
while their evaluations at zero are denoted
\begin{equation}
\label{ }
\left.{\partial\over\partial\epsilon}\Obj(\epsilon)\right|_{\epsilon=0} = \deltae\Obj
\ ,\quad
\left.{\partial^2\over\partial\epsilon^2}\Obj(\epsilon)\right|_{\epsilon=0} = \deltaee\Obj
\ ,\quad
\text{etc.}
\end{equation}
Accordingly, the Taylor expansion of $\Obj$ reads
\begin{equation}
\label{ }
\Obj(\epsilon)=\Obj+\epsilon \deltae\Obj+{1\over 2}\epsilon^2 \deltaee\Obj+\mathrm{O}(\epsilon^3)
\end{equation}
The terms of order $\epsilon$ obtained in the previous section \ref{ssVarOp} give the explicit formula of the first derivatives $\deltae$ for the objects considered there. We do not rewrite them explicitly.

\newpage 

\section{Variations of log-determinants}
\label{sCalculations}

\subsection{First order variations of determinants}
\label{subsection-first-order}

\subsubsection{The setup}
Here we compute the first order term in the $\epsilon$-expansion of the (formally infinite) 
logarithm of the determinant of $\mathcal{O}(\epsilon)$. This first order term is on general grounds

\begin{equation}
\label{ }
\delta \log \det \mathcal{O}(\epsilon) =\tr\left[\delta\mathcal{O}(\epsilon) \cdot\mathcal{O}_{\mathrm{cr}}^{-1}\right]
\end{equation}

\noindent
The results are each expressed as a sum of 
local terms over the weak Delaunay graph $\uG_{0^+}$
arising from the critical graph $\uG_\mathrm{cr}$ and the displacement function $F$. 
For both the Laplace-Beltrami and K\"ahler operators,
there is a local term associated to each edge of $\uG_{0^+}$; there
is an additional local term attached to each face
of $\uG_{0^+}$ for the K\"ahler operator.
In the case of the conformal Laplacian the local terms associated to chords 
of $\uG_{0^+}$ differ from local terms of the regular edges of $\uG_{0^+}$. For this reason 
formula \ref{varTrLlC1} 
is expressed as two sums: one over the regular edges $\ue\in\mathrm{E}(\uG_{0^+}^\bullet)= \mathrm{E}(\uG_\mathrm{cr})$
and another over the set of chords $\ue\in\mathrm{C}(\uG_{0^+})=\mathrm{E}(\uG_{0^+}) \backslash \mathrm{E}(\uG_\mathrm{cr})$.

\subsubsection{The results for first order variations}
We first give the results; their derivations are given in the following sections.

\begin{Prop}
\label{T1stVar}
\textbf{Laplace-Beltrami.}
For the Laplace-Beltrami operator $\Delta(\epsilon)$, the first order variation of $\log \det \Delta(\epsilon)$ with respect to the deformation \ref{zDeform} can be 
expressed simply in terms of the variations of the  north and south angles $\theta_\mathrm{n}(\vec{\ue}, \epsilon )$ and
$\theta_\mathrm{s}(\vec{\ue}, \epsilon)$
of edges $\ue \in \mathrm{E}(\uG_{0^+})$.
\begin{equation}
\label{varTrLlLB1}
\tr\left[\delta\Delta(\epsilon)\cdot\Delta_{\mathrm{cr}}^{-1}\right]=
\ {\epsilon\over\pi}
\sum_{\stackrel{\scriptstyle \mathrm{edges}}{\ \, \ue \in \uG_{0^+}}}
\deltae \theta_\mathrm{n}(\vec{\ue} \,)\, {\mathcal{L}'(\theta_\mathrm{n}(\vec{\ue} \, )}) + \deltae\theta_\mathrm{s}(\vec{\ue} \,) 
\,{\mathcal{L}'(\theta_\mathrm{s}(\vec{\ue} \,) )}\quad +\quad \mathrm{O}(\epsilon^2)
\end{equation}
The function $\mathcal{L}'$, given by \ref{Lprime}, is the derivative of the function $\mathcal{L}$ given by \ref{LfunDef}.
$\theta_\mathrm{n}(\vec{\ue}, \epsilon)$  and $\theta_\mathrm{s}(\vec{\ue}, \epsilon)$ are given by \ref{VarThetaNS}.
\end{Prop}

\begin{Rem}
\label{RTrLoDelta}
Owing to the extended form \ref{KenyonLogDet3} of Kenyon's result for $\log\det \Delta_{\mathrm{cr}}$, it is interesting to note that, up to terms of order 
$\epsilon^2$, $\log\det \Delta(\epsilon)$ can still be written as a sum of local terms involving the local geometry of the deformed Delaunay graph $\uG_{\epsilon}$, similar to Kenyon's result although the graph is not isoradial
\begin{equation}
\label{KenyonLogDetLB}
\log\det \Delta(\epsilon)={1\over \pi}	
\sum_{\stackrel{\scriptstyle \mathrm{edges}}{\ \, \ue \in \uG_{0^+}}}
\!\! 
\mathcal{L}\big(\theta_\mathrm{n}(\ue,\epsilon)\big)
+
\mathcal{L}\big(\theta_\mathrm{s}(\vec{\ue},\epsilon)\big)
\quad+\quad\mathrm{O}(\epsilon^2)
\end{equation}
\end{Rem}

\begin{Rem}
Equivalently, formula \ref{varTrLlLB1} can be written as a sum over triangles $\uf$ of any \emph{triangular completion} 
$\widehat{\uG}_{0^+}$ of $\uG_{0^+}$, namely
\begin{equation}
\label{varTrLlLB2}
\tr\left[\delta\Delta(\epsilon) \cdot\Delta_{\mathrm{cr}}^{-1}\right]=
- 4\epsilon\sum_{\stackrel{\scriptstyle \mathrm{faces}}{\ \ \, \uf \in \widehat{\uG}_{0^+}}} A(\uf)\Big( \overline\nabla\! F(\uf) Q(\uf) + c.c.\Big) 
\quad +\quad \mathrm{O}(\epsilon^2)
\end{equation}
where $Q(\uf)=[\nabla\Delta_{\mathrm{cr}}^{-1}\nabla^\top]_{\uf \uf}$ is a diagonal matrix entry.
{This is a direct consequence of the variational formula in Proposition \ref{PVarDelta}. 
Note that value of \ref{varTrLlLB2} is independent of the choice of triangular completion.} 
\end{Rem}

\begin{Prop}
\label{T1stVarC}
\textbf{Conformal Laplacian.} For the conformal Laplacian $\Deltaconf(\epsilon)$, the first order variation of $\log \det \Deltaconf(\epsilon)$ with
respect to the deformation \ref{zDeform} can also be expressed simply in terms of the variations of the north and south angles 
$\theta_\mathrm{n}(\vec{\ue}, \epsilon)$ and 
$\theta_\mathrm{s}(\vec{\ue}, \epsilon)$
of edges $\ue \in \mathrm{E}(\uG_{0^+})$. 
However, we must distinguish between the contributions made by regular edges versus chords in $\uG_{0^+}$. Keep in mind that
the set of regular edges $\mathrm{E}(\uG_{0^+}^\bullet)$ coincides with the edge set $\mathrm{E}(\uG_\mathrm{cr})$ of 
the critical graph.
\begin{equation}
\label{varTrLlC1}
\begin{split}
\tr\left[\delta \Deltaconf(\epsilon) \cdot \Delta_{\mathrm{cr}}^{-1}\right] = 
\  {2\,\epsilon\over\pi}\!\!\!
\, &\sum_{\stackrel{\scriptstyle \mathrm{edges}}{\  \ue \in \uG_\mathrm{cr}}}
\deltae\theta(\ue) {\mathcal{L}' \big(\theta(\ue) \big)}
\\
+{\epsilon\over \pi}
&\sum_{\stackrel{\scriptstyle \mathrm{chords} }{\, \, \ue \in \uG_{0^+} }}
\deltae\theta_\mathrm{n}(\vec{\ue} \,)\mathcal{H}'\big(\theta_\mathrm{n}(\vec{\ue} \, )\big)
+\deltae\theta_\mathrm{s}(\vec{\ue} \,)\mathcal{H}'\big(\theta_\mathrm{s}(\vec{\ue} \,)\big)
 \ +\ \mathrm{O}(\epsilon^2)
\end{split}   
\end{equation}
where $\mathcal{H}'(\theta)=\theta \cot\theta$ is the derivative of the function
\begin{equation}
\label{ }
\mathcal{H}(\theta)
= 2\,\theta\log(2\sin\theta) + \cyrLL(\theta)
\end{equation}
and where $\cyrLL(\theta)$ is the Lobachevsky function defined in \ref{lobachevsky}.
Remember that $\theta(\ue)=(\theta_\mathrm{n}(\vec{\ue} \, )+\theta_\mathrm{s}(\vec{\ue} \, ))/2$ is the conformal edge angle for general triangulations.
\end{Prop}

\begin{Rem}
\label{R1stVarC}
Up to order $\epsilon^2$, $\log\det \Deltaconf(\epsilon)$ can still be written as a sum of terms reflecting
the local geometry of the weak Delaunay graph $\uG_{0^+}$.
(see \ref{sswwococycl}).
\begin{equation}
\label{KenyonLogDetC}
\begin{split}
\log\det\,\Deltaconf(\epsilon) \ =\ &
{2\over \pi}
\sum_{\stackrel{ \scriptstyle \mathrm{edges}}{\ \ue \in \uG_\mathrm{cr}}} 
\mathcal{L}\big(\theta(\ue,\epsilon)\big)
\\
+ &{1\over \pi}
\sum_{\stackrel{\scriptstyle \mathrm{chords}}{\, \, \ue \in \uG_{0^+} }} \!\!
\mathcal{H}\big(\theta_\mathrm{n}(\ue,\epsilon)\big)
+
\mathcal{H}\big(\theta_\mathrm{s}(\vec{\ue},\epsilon)\big) \ +\ \mathrm{O}(\epsilon^2)
\end{split}
\end{equation}
\end{Rem}

\begin{Prop}
\label{T1stVarD}
\textbf{K\"ahler operator.}
For the K\"ahler operator $\mathcal{D}(\epsilon)$, a local formula also holds at order $\epsilon$. It involves the variations of the angles 
$\theta_\mathrm{n}(\vec{\ue}, \epsilon )$ 
and $\theta_\mathrm{s}(\vec{\ue}, \epsilon)$ for edges $\ue \in \mathrm{E}(\uG_{0^+})$, but also the variations of the circumradii 
$R(\uf,\epsilon)$ for faces $\uf \in \mathrm{F}(\uG_{0^+})$. 
We note that $R(\uf,\epsilon) = R_\mathrm{cr} + \delta R(\uf, \epsilon)= R_\mathrm{cr}+ \epsilon\,\delta_\epsilon R(\uf)+\mathrm{O}(\epsilon^2)$.
\begin{equation}
\label{varTrLlLK1}
\begin{split}
\tr\left[\delta\mathcal{D}(\epsilon) \cdot\mathcal{D}_{\mathrm{cr}}^{-1}\right]\ =&\ 
{\epsilon\over\pi}\sum_{\stackrel{\scriptstyle \mathrm{edges}}{\ \, \ue \in \uG_{0^+}}} 
\deltae\theta_\mathrm{n}(\vec{\ue} \, )\, {\mathcal{L}'(\theta_\mathrm{n}(\vec{\ue} \, ))}
+ \deltae\theta_\mathrm{s}(\ue) \,{\mathcal{L}'(\theta_\mathrm{s}(\vec{\ue} \, ))}\\
\ -&\ {\epsilon} \ \sum_{\stackrel{\scriptstyle \mathrm{faces}}{\ \ \, \uf \in \uG_{0^+}}}{\deltae R(\uf)\over R_\mathrm{cr}}\quad+\quad\mathrm{O}(\epsilon^2)
\end{split}
\end{equation}
\end{Prop}
\begin{Rem}
\label{rLDKvar}
Up to order $\epsilon^2$, $\log\det \mathcal{D}(\epsilon)$ can still be written as a sum of terms reflecting the local geometry of 
$\uG_{0^+}$
\begin{equation}
\label{KenyonLogDetD}
\begin{split}
\log\det \mathcal{D}(\epsilon)=&\ {1\over \pi} \sum_{\stackrel{\scriptstyle \mathrm{edges}}{\ \, \ue \in \uG_{0^+}}}
\mathcal{L}\big(\theta_\mathrm{n}(\ue,\epsilon)\big)
+\mathcal{L}\big(\theta_\mathrm{s}(\ue,\epsilon) \big) \\
&-  \, \sum_{\stackrel{\scriptstyle \mathrm{faces}}{\ \, \, \uf \in \uG_{0^+}}}
\log R(\uf,\epsilon) \quad +\quad \mathrm{O}(\epsilon^2)
\end{split}
\end{equation}
{Again, we obtain a nice local expression involving the angles $\theta_\mathrm{n}(\vec{\ue} \,)$ and $\theta_\mathrm{s}(\vec{\ue} \,)$ and the circumradii 
$R(\uf)$. Like the conformal Laplacian, the global conformal invariance properties of the K\"ahler operator are not evident in the result. However, concyclic configurations and chords do not play any special role.}
\end{Rem}

\subsubsection{Proof of Proposition~\ref{T1stVar}}
We first consider the variation of the Laplace-Beltrami operator $\Delta$ under a deformation of the form
\ref{zDeform}.
One can use \ref{VarDelta} to compute explicitly the first order variation of $\log\det\Delta$, but it is simpler to start from its definition in terms of angles \ref{BLDelta}.
For an edge $\ue= \overline{\uu \uv}$ of $\uG_\epsilon$
\begin{equation}
\label{ }
\big[\Delta(\epsilon) \big]_{\uu \uv}=-c(\ue,\epsilon) = \, - \, {\tan \theta_\mathrm{n}(\vec{\ue}, \epsilon ) + \tan \theta_\mathrm{s}(\vec{\ue}, \epsilon ) \over 2 }
\end{equation}
This implies that the variation is
\begin{equation}
\label{ }
\big[\delta\Delta(\epsilon)\big]_{\uu \uv} =-{\epsilon\over 2}
\Big(\deltae\theta_\mathrm{n}(\vec{\ue} \, )\,\sec^2\theta_\mathrm{n}(\vec{\ue} \, )
+ \deltae\theta_\mathrm{s}(\vec{\ue} \, )\,\sec^2\theta_\mathrm{s}(\vec{\ue} \, )\Big)\  +\ \mathrm{O}(\epsilon^2)
\end{equation}
where $\deltae\theta_\mathrm{n}(\vec{\ue} \, )$ and $\deltae\theta_\mathrm{n}(\vec{\ue} \, )$ are of order $\mathrm{O}(1)$.
The limit graph $\uG_{0^+}$ is weakly Delaunay and isoradial so 
either $\theta_\mathrm{n}(\vec{\ue} \, )=\theta_\mathrm{s}(\vec{\ue} \, )$ or $\theta_\mathrm{n}(\vec{\ue} \, )=-\theta_\mathrm{s}(\vec{\ue} \, )$. In both case $\sec^2\theta_\mathrm{n}(\vec{\ue} \, ) =\sec^2 \theta_\mathrm{s}(\vec{\ue} \, )$ so that at first order
\begin{equation}
\label{ }
\begin{array}{ll}
\D \big[\delta\Delta(\epsilon)\big]_{\uu \uv} 
&\D = -{\epsilon} \, {\deltae\theta_\mathrm{n}(\vec{\ue} \, )+\deltae\theta_\mathrm{s}(\vec{\ue} \, )\over 2}\sec^2\theta_\mathrm{n}(\vec{\ue} \, )
+\mathrm{O}(\epsilon^2) \\ \\
&\D =-{\epsilon} \, {\deltae\theta_\mathrm{n}(\vec{\ue} \, )+\deltae\theta_\mathrm{s}(\vec{\ue} \, )\over 2}\sec^2\theta_\mathrm{s}(\vec{\ue} \, )
+\mathrm{O}(\epsilon^2)
\end{array}
\end{equation}
It remains to combine this with the propagator 
$\left[\Delta^{-1}_\mathrm{cr} \right]_{\uv  \uu}$
which for regular edges $\ue = \overline{\uu \uv}$ of $\uG_{0^+}^\bullet = \uG_\mathrm{cr}$ is 
\begin{equation}
\label{prop1edge}
\left[\Delta^{-1}_\mathrm{cr} \right]_{\uv \uu}= -{1\over\pi} \, \theta(\ue) \cot\theta(\ue)
\end{equation}
A similar relation is in fact valid for chords of $\uG_{0^+}$
\begin{equation}
\label{ }
\left[\Delta^{-1}_\mathrm{cr} \right]_{\uv \uu}=-{1\over\pi} \,
\theta_\mathrm{n}(\vec{\ue} \,) \cot\theta_\mathrm{n}(\vec{\ue} \,)= -{1\over\pi} \, \theta_\mathrm{s}(\vec{\ue} \,) \cot\theta_\mathrm{s}(\vec{\ue} \,)
\end{equation}
Thus the first order variation is
\begin{equation}
\label{bigVarDelta}
\begin{split}
  \tr\left[\deltae \Delta\cdot \Delta^{-1}_\mathrm{cr} \right]  =
&\sum_{\stackrel{\scriptstyle \text{vertices}}{\ \  \, \uu,\uv \, \in \, \uG_{0^+}}} \deltae \Delta_{\uu \uv}\left[\Delta^{-1}_\mathrm{cr} \right]_{\uv \uu}\,\\
=& \  {1\over\pi} \sum_{\stackrel{\scriptstyle \text{edges}}{\ \ \, \ue \, \in \, \uG_{0^+}}} 
\deltae\theta_\mathrm{n}(\vec{\ue} \, )  \theta_\mathrm{n}(\vec{\ue} \,) \cot \theta_\mathrm{n}(\vec{\ue} \,)
\sec^2 \theta_\mathrm{n}( \vec{\ue} \, ) + \delta\theta_\mathrm{s}(\vec{\ue} \,) \theta_\mathrm{s}(\vec{\ue} \,) \cot\theta_\mathrm{s}(\vec{\ue} \,)
\sec^2 \theta_\mathrm{s}(\vec{\ue} \,) \\
=&\ {1\over\pi} \sum_{\stackrel{\scriptstyle \text{edges}}{\ \ \, \ue \, \in \, \uG_{0^+}}} 
\deltae\theta_\mathrm{n}(\vec{\ue} \,) {\theta_\mathrm{n}(\vec{\ue} \,)\over \sin \theta_\mathrm{n}(\vec{\ue} \,)\cos \theta_\mathrm{n}(\vec{\ue} \,) }+ \deltae\theta_\mathrm{s}(\vec{\ue} \,) {\theta_\mathrm{n}(\vec{\ue} \,)\over \sin \theta_\mathrm{s}(\vec{\ue} \,) \cos \theta_\mathrm{s}(\vec{\ue} \,) }\\
=&\ {1\over\pi} \sum_{\stackrel{\scriptstyle \text{edges}}{\ \ \, \ue \, \in \, \uG_{0^+}}} 
\deltae\theta_\mathrm{n}(\vec{\ue} \,) {\mathcal{L}' \big(\theta_\mathrm{n}(\vec{\ue} \,) \big)}+ \deltae\theta_\mathrm{s}(\vec{\ue} \,) {\mathcal{L}' \big( \theta_\mathrm{n}(\vec{\ue} \,) \big) }\\
=&\ \deltae\left[{1\over\pi} \sum_{\stackrel{\scriptstyle \text{edges}}{\ \ \, \ue \, \in \, \uG_{0^+}}} 
\mathcal{L}\big(\theta_\mathrm{n}(\vec{\ue}, \epsilon) \big)+\mathcal{L}\big( \theta_\mathrm{s}(\vec{\ue}, \epsilon) \big) \right]
\end{split}   
\end{equation}
This, together with \ref{KenyonLogDet3}, leads to \ref{KenyonLogDetLB}, and this ends the proof of Proposition~\ref{T1stVar}.


\subsubsection{Proof of Proposition \ref{T1stVarC}}
\label{ssProofConfLap1}
For an edge $\ue= \overline{\uu \uv}$ of $\uG_\epsilon$ the matrix element the conformal Laplacian $\Deltaconf$ is 
\begin{equation}
\label{ }  
\big[\Deltaconf(\epsilon) \big]_{\uu \uv}=-\tan \theta(\ue, \epsilon) \ ,\ \ \theta(\ue, \epsilon)={ \theta_\mathrm{n}(\vec{\ue}, \epsilon)+\theta_\mathrm{s}(\vec{\ue}, \epsilon)  \over 2}
\end{equation}
This implies that the variation is
\begin{equation}
\label{ }
\big[\delta \Deltaconf(\epsilon) \big]_{\uu \uv} =
- \epsilon \, {\deltae\theta_\mathrm{n}(\vec{\ue} \, )+ \deltae\theta_\mathrm{s}(\vec{\ue} \, ) \over 2} 
\, \sec^2\left({\theta_\mathrm{n}(\vec{\ue} \, )+\theta_\mathrm{s}(\vec{\ue} \, )\over 2}\right)\  +\ \mathrm{O}(\epsilon^2)
\end{equation}
Keep in mind that the limit graph $\uG_{0^+}$ is weakly Delaunay and isoradial and so either $\theta_\mathrm{n}(\vec{\ue} \, )=\theta_\mathrm{s}(\vec{\ue} \, )$ 
or $\theta_\mathrm{n}(\vec{\ue} \, )=-\theta_\mathrm{s}(\vec{\ue} \, )$.
The first case corresponds to a regular edge, while
the second case corresponds to a chord.
Thus to first order in $\epsilon$ the matrix entry is
\begin{equation}
\label{varDelConf}
\begin{array}{ll}
\D \big[ \deltae \Deltaconf \big]_{\uu \uv} 
&\D =
\begin{cases}
      \D \ - \, { \deltae\theta_\mathrm{n}(\vec{\ue} \, )\sec^2\theta_\mathrm{n}(\vec{\ue} \, )
      + \deltae\theta_\mathrm{s}(\vec{\ue} \, ) \sec^2\theta_\mathrm{s}(\vec{\ue} \, ) \over 2}
      & \text{if}\ \theta_\mathrm{n}(\vec{\ue} \, )=\theta_\mathrm{s}(\vec{\ue} \, ) \\
      \ \\
      \D \  - \, { \deltae\theta_\mathrm{n}(\vec{\ue} \, )+ \deltae\theta_\mathrm{s}(\vec{\ue} \, ) \over 2}
      & \text{if}\ \theta_\mathrm{n}(\vec{\ue} \, )=-\theta_\mathrm{s}(\vec{\ue} \, )
\end{cases} \\ \\
&\D =
\begin{cases}      
      \D \ \big[\deltae \Delta \big]_{\uu \uv}
      & \D \text{if}\ \theta_\mathrm{n}(\vec{\ue} \, )=\theta_\mathrm{s}(\vec{\ue} \, ) \\
      \ \\
      \D \ \big[\deltae \Delta \big]_{\uu \uv}\!\!\! \ + \ {\deltae\theta_\mathrm{n}(\vec{\ue} \, ) \tan^2\theta_\mathrm{n}(\vec{\ue} \, )
      + \deltae\theta_\mathrm{s}(\vec{\ue} \, ) \tan^2\theta_\mathrm{s}(\vec{\ue} \, ) \over 2}
      & \D \text{if}\ \theta_\mathrm{n}(\vec{\ue} \, )=-\theta_\mathrm{s}(\vec{\ue} \, )
\end{cases} \\ \\
\end{array}
\end{equation}
The first order variation of the log-determinant 
reads as a sum over the edges of $\uG_{0^+}$, but it is different for the edges in $\uG_{0^+}^\bullet=\uG_\mathrm{cr}$ and the chords of  $\uG_{0^+}$.
Combining with \ref{prop1edge} we get at first order
\begin{equation}
\label{}
\begin{split}
  \tr\left[\deltae \Deltaconf\cdot \Delta_{\mathrm{cr}}^{-1}\right]  
  &=  \sum_{\stackrel{\scriptstyle \text{vertices}}{\ \ \, \uu,\uv \, \in \, \uG_{0^+}}} 
\big[\deltae \Deltaconf \big]_{\uu \uv}
\left[\Delta_\mathrm{cr}^{-1}\right]_{\uv \uu}\,\\ 
  =\ {2\over\pi}&
  \sum_{\substack{\mathrm{edges}\\ \ \, \ue \, \in \, \uG_\mathrm{cr}}}\deltae\theta(\ue) {\mathcal{L}'(\theta(\ue))}
  +{1\over \pi}\sum_{\substack{\mathrm{chords}\\ \ \, \ue \, \in \, \uG_{0^+} }}\deltae\theta_\mathrm{n}(\vec{\ue} \,)
  \mathcal{H}'(\theta_\mathrm{n}(\vec{\ue} \,))+\deltae\theta_\mathrm{s}(\vec{\ue} \,)\mathcal{H}'(\theta_\mathrm{s}(\vec{\ue} \,))
\end{split}   
\end{equation}
with the function $\mathcal{H}(\theta)$ given by
\begin{equation}
\label{ }
\mathcal{H}(\theta)=\int_0^\theta dt\, t \, \cot (t) = 2\theta\log(2\sin\theta) + \cyrLL(\theta)
\end{equation}
This leads to \ref{KenyonLogDetC} and the proof of Proposition~\ref{T1stVarC}.

\subsubsection{Proof of Proposition~\ref{T1stVarD}}
The variation of the K\"ahler operator $\mathcal{D}$
starts from the expression of the matrix elements $\mathcal{D}_{\uu,\uv}$ of an edge $\vec{\ue}= (\uu,\uv)$ in terms of the angles 
$\theta_\mathrm{n}(\vec{\ue}, \epsilon)$ and $\theta_\mathrm{s}(\vec{\ue}, \epsilon)$ and of the circumradii $R_\mathrm{n}(\vec{\ue}, \epsilon)$ and 
$R_\mathrm{s}( \vec{\ue}, \epsilon)$ given by 
\ref{ConfDelta}, namely
\begin{equation}
\label{ }
\big[\mathcal{D}(\epsilon) \big]_{\uu \uv}=-{1\over 2}\left({\tan\theta_\mathrm{n}(\vec{\ue}, \epsilon )+i\over R^2_\mathrm{n}(\vec{\ue} , \epsilon) }
+{\tan\theta_\mathrm{s}(\vec{\ue}, \epsilon )-i\over R^2_\mathrm{s}(\vec{\ue}, \epsilon) } \right)
\end{equation}
Its variation is therefore 
\begin{equation}
\label{deltaDuv}
\begin{split}
\big[ \delta\mathcal{D}(\epsilon)\big]_{\uu \uv}=&\ \deltae\mathcal{D}_{\uu \uv}^{(1)}+\deltae\mathcal{D}_{\uu \uv}^{(2)}\ +\ \mathrm{O}(\epsilon^2)\\
\deltae\mathcal{D}_{\uu \uv}^{(1)}=&\  -\left({1\over 2\,R_\mathrm{n}^2(\vec{\ue} \,)}\,\deltae \tan\theta_\mathrm{n}(\vec{\ue} \, )
+{1\over 2\,R_\mathrm{s}^2(\vec{\ue} \,)}\,\deltae \tan\theta_\mathrm{s}(\vec{\ue} \, ) \right)\\
\deltae\mathcal{D}_{\uu \uv}^{(2)}=&\ \left({\tan\theta_\mathrm{n}(\vec{\ue} \, )+ i\over R^3_\mathrm{n}(\vec{\ue} \,) }\,\deltae R_\mathrm{n}(\vec{\ue} \,)+{\tan\theta_\mathrm{s}(\vec{\ue} \, )- i\over R^3_\mathrm{s}(\vec{\ue} \, )}\,\deltae R_\mathrm{s}(\vec{\ue} \,)\right)
\end{split}
\end{equation}
For an isoradial triangulation (critical case), $R_\mathrm{n}(\vec{\ue} \,)=R_\mathrm{s}(\vec{\ue} \, )=R_\mathrm{cr}$, therefore one has $\mathcal{D}_{\mathrm{cr}}=R^{-2}_\mathrm{cr} \Delta_\mathrm{cr}$.
Thus in the critical case, the contribution made by the first term in \ref{deltaDuv} to the variation is
\begin{equation}
\label{ }
\deltae\mathcal{D}_{\uu \uv}^{(1)}=R_\mathrm{cr}^{-2}\, \deltae\Delta_{\uu \uv}\ \implies\quad \tr\left[\deltae\mathcal{D}^{(1)}\cdot \mathcal{D}^{-1}_\mathrm{cr}  \right]=\tr\left[\deltae\Delta\cdot \Delta^{-1}_\mathrm{cr} \right]
\end{equation}
The second term's contribution can be reorganized as a sum over faces of $\widehat{\uG}_{0^+}$, i.e. counter-clockwise oriented 
triangles $\uf=(\uu,\uv,\uw)$
\begin{equation}
\label{ }
\begin{split}
\tr&\left[\deltae\mathcal{D}^{(2)}\cdot \mathcal{D}^{-1}_\mathrm{cr} \right]=
\ \sum_{\stackrel{\scriptstyle \text{vertices}}{\ \, \, \,  \uu, \uv \, \in \, \uG_{0^+}}}
\  \deltae\mathcal{D}^{(2)}_{\uu  \uv}\ \big[\mathcal{D}^{-1}_\mathrm{cr}\big]_{\uv \uu}
\\
&=\sum_{\substack{\scriptstyle \text{triangles}  \\ \uf \, = \, (\uu, \uv, \uw) \\ \text{in $\widehat{\uG}_{0^+}$} }} 
{\deltae R(\uf)\over R_\mathrm{cr}^3}\left( 
\begin{split}&
(\tan\theta_\mathrm{n}( \vv{\uu \uv})+\mathrm{i}) \big[\mathcal{D}^{-1}_\mathrm{cr}\big]_{\uv\uu} \ 
+ (\tan\theta_\mathrm{s}(\vv{\uv \uu})-\mathrm{i}) 
\big[\mathcal{D}^{-1}_\mathrm{cr} \big]_{\uu \uv} \\
+& (\tan\theta_\mathrm{n}(\vv{\uv \uw})+\mathrm{i}) 
 \big[\mathcal{D}^{-1}_\mathrm{cr}\big]_{\uw \uv} \ 
+(\tan\theta_\mathrm{s}(\vv{\uw \uv})-\mathrm{i}) 
\big[\mathcal{D}^{-1}_\mathrm{cr}\big]_{\uv \uw} \\
+& (\tan\theta_\mathrm{n}(\vv{\uw \uu})+\mathrm{i})  
\big[\mathcal{D}^{-1}_\mathrm{cr}\big]_{\uw \uu} \ 
+(\tan\theta_\mathrm{s}(\vv{\uu \uw})-\mathrm{i})  
\big[\mathcal{D}^{-1}_\mathrm{cr}\big]_{\uw \uu}
\end{split}\right)
\end{split}
\end{equation}
Using the fact that $\theta_\mathrm{n}(\vv{\uu \uv})=\theta_\mathrm{s}(\vv{\uv \uu})$ and that for the critical case $$ \big[\mathcal{D}^{-1}_\mathrm{cr}\big]_{\uv \uu}= \big[\mathcal{D}^{-1}_\mathrm{cr}\big]_{\uu \uv}= -{1\over\pi} \, R_\mathrm{cr}^2 \, \theta_\mathrm{n}(\vv{\uu  \uv})\cot\theta_\mathrm{n}(\vv{\uu \uv})$$
and the fact that for a triangle $\uf=(\uu,\uv,\uw)$, one has
$$\theta_\mathrm{n} ( \vv{\uu \uv} ) \ + \ \theta_\mathrm{n} ( \vv{\uv \uw} ) \ +\ \theta_\mathrm{n} ( \vv{\uw \uu})  \ = \ \pi/2$$
we obtain
\begin{equation}
\label{ }
\tr\left[\deltae\mathcal{D}^{(2)}\cdot \mathcal{D}^{-1}_\mathrm{cr} \right]=
-\sum_{\substack{\text{faces} \\ \ \ \ \uf \, \in \, \widehat{\uG}_{0^+} }} 
{\deltae R(\uf)\over R_\mathrm{cr}} = 
-\sum_{\substack{\text{faces} \\ \ \ \ \uf \, \in \, \widehat{\uG}_{0^+} }} 
\deltae \log R(\uf)
\end{equation}
This leads to \ref{KenyonLogDetD} and to Proposition~\ref{T1stVarD}.

\subsection{Second order variations}
\label{subsection-second-order}
\subsubsection{Principle of the calculation}

In order to probe the second trace term in the perturbative expansion 
\ref{log-expansion}, we consider a {\it bi-local} deformation $\uG_{\underline{\epsilon}}$
of the underlying critical graph $\uG_\mathrm{cr}$
with an embedding of the form
\begin{equation}
z_{\underline{\epsilon}}(\uv) \, := \ z_\mathrm{cr}(\uv) \, + \, \epsilon_1 F_1(\uv) \, + \, \epsilon_2 F_2(\uv)
\end{equation}
where $\underline{\epsilon} = (\epsilon_1, \epsilon_2)$ is a pair of independent deformation parameters
and where $F_1, F_2 \in \Bbb{C}^{\mathrm{V}(\uG_\mathrm{cr})}$ 
are two functions, with finite supports $\Omega_1:= \Omega_{F_1}$ and $\Omega_2 := \Omega_{F_2}$
in $\mathrm{V}(\uG_\mathrm{cr})$ and
whose respective {\it lattice closures} $\overline{\Omega}_1$ and $\overline{\Omega}_2$ are 
disjoint; see \ref{omegaclosure}.
To insure that $\uG_{\underline{\epsilon}}$ is a stable Delaunay deformation, we restrict
the parameters $\epsilon_1, \epsilon_2$ within the range $[0, \tilde{\epsilon}_F )$ where
$\tilde{\epsilon}_F := \mathrm{min} (\tilde{\epsilon}_{F_1},  \tilde{\epsilon}_{F_2})$
and $\tilde{\epsilon}_{F_1}, \tilde{\epsilon}_{F_1}$ are the 
thresholds dictated by Proposition \ref{epsilontildeF}.
Furthermore we shall assume that 
the {\it distance} $d$ between the two supports is large, i.e. $d \gg R_\mathrm{cr}$ where
\begin{equation}
d \, =  \ \mathtt{dist}\big( \overline{\Omega}_1, \overline{\Omega}_2 \big) \, :=
\ \inf \Big\{  \big| z_\mathrm{cr} (\uw_1) - z_\mathrm{cr} (\uw_2) \big|  \, \Big| \, \uw_1 \in \overline{\Omega}_1 \, , \, \uw_2 \in \overline{\Omega}_2  \Big\} 
\end{equation}
We want to isolate and then examine the long-range behavior of the $\epsilon_1 \epsilon_2$ cross-term occurring within
the perturbative expansion of the log-determinant, namely in
\[ \log \det \mathcal{O}(\underline{\epsilon}) \, = \,
\log \det \mathcal{O}_\mathrm{cr} \, + \, \mathrm{tr} \Big[ \delta \mathcal{O}(\underline{\epsilon}) \cdot \mathcal{O}_\mathrm{cr}^{-1} \Big] 
\, - \, 
{1 \over 2} \tr \Big[ \big( \delta \mathcal{O}(\underline{\epsilon}) \cdot \mathcal{O}_\mathrm{cr}^{-1} \big)^2 \Big] 
\, + \, \cdots
\]
\noindent
The first trace term $\mathrm{tr} [\delta \mathcal{O}(\underline{\epsilon}) \cdot \mathcal{O}_\mathrm{cr}^{-1} ] $
contributes nothing of order $\epsilon_1 \epsilon_2$ since the lattice closures of the supports
$\overline{\Omega}_1$ and $\overline{\Omega}_2$ are disjoint. The only non-vanishing contribution
comes from the second trace term 
\begin{equation}
\label{bilinear-bilocal}
-{1 \over 2} \tr \Big[ \big( \delta \mathcal{O}(\underline{\epsilon}) \cdot \mathcal{O}_\mathrm{cr}^{-1} \big)^2 \Big] 
\end{equation}
which is bilinear in the total variation $\delta \mathcal{O}(\underline{\epsilon})$. Accordingly,
the coefficient of $\epsilon_1 \epsilon_2$ can be expressed as 
\begin{equation}
\label{ }
\begin{split}
&-\tr \Big[ \deltaeo \mathcal{O}\cdot \mathcal{O}_\mathrm{cr}^{-1} \cdot \deltaet \mathcal{O} \cdot \mathcal{O}_\mathrm{cr}^{-1} \Big]
= - \sum_{\stackrel{\scriptstyle \uu, \uv \, \in \, \overline{\Omega}_1}{\scriptstyle \up, \uq  \, \in \, \overline{\Omega}_2}}
\big[\deltaeo \mathcal{O}\big]_{\uu \uv}\, 
\big[ \mathcal{O}_\mathrm{cr}^{-1} \big]_{\uv \up} \, 
\big[ \deltaet \mathcal{O} \big]_{\up \uq}\, 
\big[\mathcal{O}_\mathrm{cr}^{-1} \big] _{\uq \uu}
\end{split}
\end{equation}
where $\deltaeo \mathcal{O}$ and $\deltaet \mathcal{O}$ are first order variations of the Laplace-like operator
$\mathcal{O}(\underline{\epsilon})$
following the notations set in Sect.~\ref{ssGenVarNot}.
The sum on the left hand side is taken over vertices $\uu$, $\uv$, $\up$, $\uq$ such that both
matrix entries $\big[\deltaeo \mathcal{O}\big]_{\uu \uv}$ and $\big[\deltaet \mathcal{O}\big]_{\up \uq}$
are non-zero. In particular this implies that
$\overline{\uu \uv}$ is an edge in  $\uG_{0^+}$
with vertices $\uu, \uv \in \overline{\Omega}_1$. Likewise
$\overline{\up \uq}$ must be an edge in $\uG_{0^+}$
with vertices $\up, \uq \in \overline{\Omega}_2$.

Provided the two zones of support $\Omega_1$ and $\Omega_2$ are far enough apart,
the matrix entries $\big[\mathcal{O}_\mathrm{cr}^{-1}\big]_{\uv \up} $ 
and $\big[\mathcal{O}_\mathrm{cr}^{-1}\big]_{\uq \uu} $
of the critical Green's function will only involve pairs of vertices
with
$|z_\mathrm{cr}(\uv) -z_\mathrm{cr}(\up) | \simeq d $
and $|z_\mathrm{cr}(\uq) -z_\mathrm{cr}(\uu) | \simeq d$.
Under these circumstances we may estimate the contributions made by these matrix entries
using the asymptotic expansion \ref{SharpAsymptotics1} for the Green's function.

It will be convenient to take a triangular completion $\widehat{\uG}_\epsilon$ of the deformed Delaunay graph 
$\uG_\epsilon$ as defined in Def.~\ref{CompletionTofG}.
Likewise $\widehat{\uG}_{0^+}$ will be the completion of the limit graph $\uG_{0^+}$
induced from $\widehat{\uG}_\epsilon$.
This will allow us to use the
variational formulae 
\ref{PVarDelta}
and \ref{PVarKahler}
for the Laplace-Beltrami and K\"ahler operators.
In general such a completion $\widehat{\uG}_\epsilon$ will not be unique.
Nevertheless the redactions satisfy
$\widehat{\uG}_\epsilon^\bullet= \uG_\epsilon$
and $\widehat{\uG}_{0^+}^\bullet = \uG_\mathrm{cr}$
regardless of the choice of completion.
The Laplace-Beltrami operator, K\"ahler operator, and the conformal Laplacian
will not be affected by this choice, since the weights assigned by the operators to any chords, 
introduced by the completion, must vanish.

\subsubsection{The Laplace-Beltrami operator}
\label{ssLaplSecOrder} 
 \ \\
 The simplest case is the Laplace-Beltrami $\Delta$ operator. We shall need two intermediate  results.
 \begin{lemma}
 \label{lnablabound}
 Let $\uf=({\uv}_1,{\uv}_2,{\uv}_3)$ be a c.c.w. oriented triangle with circumcenter $z_\mathrm{cr}(\uf)$ and circumradius $R=1$. Define
$e^{\mathrm{i}\theta_j} := z_\mathrm{cr}(\uv_j) - z_\mathrm{cr}(\uf)$ for $j=1,2,3$. Let $\nabla_{\uf \uv_j}$ be the matrix elements of the discrete derivative operator $\nabla$ restricted to the triangle $\uf$. For any integer $m\in\mathbb{Z}$ one has the uniform bounds
\begin{equation}
\label{boundnablam}
\left|\sum_{j=1}^3  \,  \nabla_{\! \uf \uv_j} \, e^{\mathrm{i} m\theta_j}\right| \le {m(m+1)\over 2}\quad, \qquad m\in\mathbb{Z}
\end{equation} 
and \begin{equation}
\label{boundnablabarm}
\left|\sum_{j=1}^3  \,  \overline\nabla_{\! \uf \uv_j} \, e^{\mathrm{i} m\theta_j}\right| \le {m(m-1)\over 2}\quad, \qquad m\in\mathbb{Z}
\end{equation} 

\end{lemma}
\begin{proof}
Using the definition \ref{nablaDef} of $\nabla$ one can rewrite
\begin{equation}
\label{nablaexpl}
\begin{split}
\sum_{j=1}^3 \, \nabla_{\! \uf \uv_j}\,e^{\mathrm{i} m\theta_j}&=\det\begin{pmatrix}
1&e^{-\mathrm{i}\theta_1}& e^{\mathrm{i} m\theta_1}   \\
1&e^{-\mathrm{i}\theta_2}& e^{\mathrm{i} m\theta_2} \\
1&e^{-\mathrm{i}\theta_3}&e^{\mathrm{i} m\theta_3}
\end{pmatrix}{\Bigg /}\det\begin{pmatrix}
1&e^{-\mathrm{i}\theta_1}& e^{\mathrm{i} \theta_1}   \\
1&e^{-\mathrm{i}\theta_2}& e^{\mathrm{i} \theta_2} \\
1&e^{-\mathrm{i}\theta_3}&e^{\mathrm{i} \theta_3}
\end{pmatrix}
\end{split}
\end{equation}
For $m>0$ we can rewrite the numerator as
\begin{equation}
\label{ }
e^{-\mathrm{i}(\theta_1+\theta_2+\theta_3)}
\det\begin{pmatrix}
1&e^{\mathrm{i}\theta_1}& e^{\mathrm{i} (m+1)\theta_1}   \\
1&e^{\mathrm{i}\theta_2}& e^{\mathrm{i} (m+1)\theta_2} \\
1&e^{\mathrm{i}\theta_3}&e^{\mathrm{i} (m+1)\theta_3}
\end{pmatrix}
\end{equation}
which involves a special case of the following Vandermonde-like determinant:
\begin{equation}
\label{ }
\det\begin{pmatrix}
1&z_1& z_1^{m+1}   \\
1&z_2& z_2^{m+1}  \\
1&z_3& z_3^{m+1} 
\end{pmatrix}
=
(z_1-z_2)(z_2-z_3)(z_3-z_1)S_{m-1}(z_1,z_2,z_3)
\end{equation}
where $S_n$ is the complete homogeneous symmetric polynomial of degree $n$ (a Schur polynomial), 
\begin{equation}
\label{ }
S_n(z_1,z_2,z_3)=\sum_{\substack{p_1,p_2,p_3\in\mathbb{N}\\p_1+p_2+p_3=n}}z_1^{p_1} z_2^{p_2} z_3^{p_3}
\end{equation}
which consists of $(n+1)(n+2)/2$ monomials. 
The numerator equals the denominator in the r.h.s. of \ref{nablaexpl} when $m=1$
and since $S_0(z_1,z_2,z_3)=1$ we get 
\begin{equation}
\label{ }
\sum_{j=1}^3 \, \nabla_{\! \uf \uv_j}\,e^{\mathrm{i} m\theta_j}= S_{m-1}(e^{\mathrm{i}\theta_1},e^{\mathrm{i}\theta_2},e^{\mathrm{i}\theta_3})
\end{equation}
when $m>0$. It is clear that for $m>0$ we have the bound
\begin{equation}
\label{ }
|S_{m-1}(e^{\mathrm{i}\theta_1},e^{\mathrm{i}\theta_2},e^{\mathrm{i}\theta_2})|\le (m-1+1)(m-1+2)/2=m(m+1)/2
\end{equation}
which is saturated when $\theta_1=\theta_2=\theta_3$.
For $m=0$ or $m=-1$ it is clear that $\sum_{j=1}^3 \, \nabla_{\! \uf \uv_j}\,e^{\mathrm{i} m\theta_j}=0$.
When $m\le -2$, we can rewrite 
\begin{equation}
\label{ }
\sum_{j=1}^3 \, \nabla_{\! \uf \uv_j}\,e^{\mathrm{i} m\theta_j}= 
e^{-\mathrm{i}\theta_1} e^{-\mathrm{i}\theta_2} e^{-\mathrm{i}\theta_3}
\ S_{-m-2}(e^{-\mathrm{i}\theta_1},e^{-\mathrm{i}\theta_2},e^{-\mathrm{i}\theta_3}) 
\end{equation}
by a similar trick. Since $-m-2\ge 0$ we get the bound
\begin{equation}
\label{ }
\left |S_{-m-2}(e^{-\mathrm{i}\theta_1},e^{-\mathrm{i}\theta_2},e^{-\mathrm{i}\theta_3}) \right| \le {(-m-2+1)(-m-2+2)/2}=m (m+1)/2
\end{equation}
Thus we get \ref{boundnablam}.
To obtain \ref{boundnablabarm} one uses simply
\begin{equation}
\label{ }
\sum_{j=1}^3  \,  \overline\nabla_{\! \uf \uv_j} \, e^{\mathrm{i} m\theta_j}=\overline{\sum_{j=1}^3  \,  \nabla_{\! \uf \uv_j} \, e^{-\mathrm{i} m\theta_j}}
\end{equation}
and \ref{boundnablam}.
\end{proof}

Now we can get uniform asymptotic estimates for the discrete derivatives of the Green function.
\begin{lemma}
\label{lNGN}
Let $\Delta_\mathrm{cr}^{-1}$ be the critical Green's function on an isoradial,  weak Delaunay 
triangulation $\uT_\mathrm{cr}$, let $\uf$ and $\ug$ be two faces (triangles),
and let $z_\mathrm{cr}(\uf)$ and $z_\mathrm{cr}(\ug)$ be the complex coordinates of their respective circumcenters $\uo_\uf$ and $\uo_\ug$.
Let $d=|z_\mathrm{cr}(\uf)-z_\mathrm{cr}(\ug)|$ be the distance between the centers.
Then the discrete double derivatives of the Green function have the following large distance asymptotics
\begin{equation}
\label{nnadjbnnt}
\begin{array}{l}
   \D \Big[ \, \nabla\   \Delta^{-1}_\mathrm{cr} \  {\overline\nabla}^\top \Big]_{\uf \ug} \    = 
   \ {1 \over {4 \pi }}  \,  \Bigg( {\prod_{\uv \in \ug} \, e^{\mathrm{i}\theta_\uv} \over { \big(z_\mathrm{cr}(\uf) -z_\mathrm{cr}(\ug) \big)^3}}
   \ - \  {\prod_{\uu \in \uf} \, e^{-\mathrm{i} \theta_\uu} \over { \big( \bar z_\mathrm{cr}(\uf) - \bar z_\mathrm{cr}(\ug) \big)^3}} \Bigg)
   \ + \  \mathrm{O}(d^{-4}) \\ \\
   \D \Big[ \, \overline{\nabla}\  \Delta^{-1}_\mathrm{cr}     \  \nabla^\top \Big]_{\uf \ug} \  = 
   \ {1 \over {4 \pi }}  \,  \Bigg( {\prod_{\uv \in \ug} \, e^{-\mathrm{i}\theta_\uv} \over { \big(\bar z_\mathrm{cr}(\uf) - \bar z_\mathrm{cr}(\ug) \big)^3}}
   \ - \  {\prod_{\uu \in \uf} \, e^{\mathrm{i} \theta_\uu} \over { \big( z_\mathrm{cr}(\uf) - z_\mathrm{cr}(\ug) \big)^3}} \Bigg)
   \ + \  \mathrm{O}(d^{-4}) 
   \end{array}
\end{equation}
and 
\begin{equation}
\label{nntbnnadj}
\begin{split}
\Big[ \, \nabla\   \Delta^{-1}_\mathrm{cr}    \  \nabla^\top \Big]_{\uf \ug} \  &   =\  - {1\over 4 \pi} {1\over (z_\mathrm{cr}(\uf) - z_\mathrm{cr}(\ug))^2} \ + \ \mathrm{O}(d^{-3})\\
\Big[ \, \overline{\nabla}\   \Delta^{-1}_\mathrm{cr}   \  {\overline\nabla}^\top\Big]_{\uf \ug} \  &   = 
\  - {1\over 4 \pi} {1\over (\bar z_\mathrm{cr}(\uf) - \bar z_\mathrm{cr}(\ug) )^2}\ +\ \mathrm{O}(d^{-3})
\end{split}
\end{equation}
 \end{lemma}
 \begin{proof}
Let $\uf=(123)$ and $\ug=(456)$ be the vertices of $\uf$ and $\ug$ respectively. The triangulation $\uT_\mathrm{cr}$ is isoradial, so denote 
$z_\mathrm{cr}(\uu) - z_\mathrm{cr}(\uf) = e^{\mathrm{i}\theta_\uu}$ ($\uu=1,2,3$) and $z_\mathrm{cr}(\uv) - z_\mathrm{cr}(\ug) =e^{\mathrm{i}\theta_\uv}$ 
($\uv=4,5,6$). Use \ref{SharpAsymptotics1} to separate the Green function $ \big[\Delta^{-1}_\mathrm{cr} \big]_{\uu \uv}$ into its leading large distance term (continuous limit term) of order $\log d$, its subleading large distance correction of order $d^{-2}$, and the rest of its large distance expansion of order $d^{-4}$.
\begin{equation}
\label{propdecomp}
\big[ \Delta^{-1}_\mathrm{cr}\big]_{\uu \uv} =G_{\uu \uv}^{(0)}+G_{\uu \uv}^{(2)}+G_{\uu \uv}^{(4)}
\end{equation}
with
\begin{equation}
\label{Gparts}
\begin{split}
G_{\uu \uv}^{(0)}=& -{1\over 2\pi}\,\left( \log( 2 |z_\mathrm{cr}(\uu)-z_\mathrm{cr}(\uv)|) +\gamma_{\mathrm{euler}} \right)\\
G_{\uu \uv}^{(2)}=& -{1\over 24 \pi} \left({p_3(\uu, \uv)\over (z_\mathrm{cr}(\uu)-z_\mathrm{cr}(\uv))^3}+ {\bar p_3(\uu, \uv)\over 
(\bar z_\mathrm{cr}(\uu)-\bar z_\mathrm{cr}(\uv))^3}\right)\\
G_{\uu \uv}^{(4)}=& \ \mathrm{O}\left(|z_\mathrm{cr}(\uu)-z_\mathrm{cr}(\uv)|^{-4}\right)
\end{split}
\end{equation}
Begin by writing
\begin{equation}
\label{ }
z_\mathrm{cr}(\uu)-z_\mathrm{cr}(\uv) \ = \ z_\mathrm{cr}(\uf)-z_\mathrm{cr}(\ug) \, + \, e^{\mathrm{i}\theta_\uu}-e^{\mathrm{i}\theta_\uv}
\end{equation}
and expand the logs and powers of $(z_\mathrm{cr}(\uu)-z_\mathrm{cr}(\uv))$ and $(\bar z_\mathrm{cr}(\uu)-\bar z_\mathrm{cr}(\uv))$ 
in formulae \ref{Gparts} as {power} series 
in $(z_\mathrm{cr}(\uf) - z_\mathrm{cr}(\ug))$ and $(\bar z_\mathrm{cr}(\uf) - \bar z_\mathrm{cr}(\ug))$
where $d=|z_\mathrm{cr}(\uf)-z_\mathrm{cr}(\ug)|\gg 1$ is large. For example:
\[ G^{(0)}_{\uu \uv} \ = \
-{1 \over {2 \pi}} \, \Big(
\log \big( 2 | z_\mathrm{cr}(\uf) - z_\mathrm{cr}(\ug) | \big)  \, + \, \gamma_\mathrm{euler} \Big)    
\ + \ {1 \over {2 \pi}} \, \frak{Re} \, \sum_{r \geq 1} \, {1 \over r} \, 
\left( { e^{\mathrm{i}\theta_\uv} - e^{\mathrm{i}\theta_\uu} \over {z_\mathrm{cr}(\uf) - z_\mathrm{cr}(\ug)} } \right)^r
\]
The coefficients in these expansions involve the phases $e^{\mathrm{i}\theta_\uu}$ and $e^{\mathrm{i}\theta_\uv}$
and so the matrix entries in formulae \ref{nnadjbnnt} and \ref{nntbnnadj} can be computed
using the basic identities
\begin{equation}
\label{ }
\sum_{\uu \in \uf}\nabla_{\! \uf \uu}\,e^{\mathrm{i} \theta_\uu}=1 \quad \text{and} \quad \sum_{\uu \in \uf}\nabla_{\! {\uf \uu}} \,e^{-\mathrm{i} \theta_\uu}= \, \sum_{\uu \in \uf}\nabla_{\! \uf \uu} = 0
\end{equation} 
along with values of
$\nabla_{\uf \uu}$, $\overline\nabla_{\uf \uu}$ and $\nabla^{\top}_{\uv \ug}=\nabla_{\ug \uv}$, 
$\nabla^{\dagger}_{\uv \ug}=\overline\nabla_{\ug \uv}$ 
explicitly given in \ref{nablaDef} and \ref{barnablaDef}).
As an illustration:
\[
\begin{array}{ll}
\D \Big[ \, \nabla G^{(0)}  \overline{\nabla}^\top \Big]_{\uf \ug} 
&\D =  \sum_{\uu \in \uf} \sum_{\uv \in \ug} 
\nabla_{\uf \uu} \overline{\nabla}_{\ug \uv} \, G^{(0)}_{\uu \uv} \\ \\
&\D =  
\left\{
\begin{array}{l}
\D  {1 \over {4 \pi }}  \,  \Bigg( {\prod_{\uu \in \uf} \, e^{\mathrm{i}\theta_\uu} \over { \big(z_\mathrm{cr}(\uf) -z_\mathrm{cr}(\ug) \big)^3}}
\ - \  {\prod_{\uv \in \ug} \, e^{-\mathrm{i} \theta_\uv} \over { \big( \bar z_\mathrm{cr}(\uf) - \bar z_\mathrm{cr}(\ug) \big)^3}} \Bigg) \ +  \\ \\
\D {1 \over {2\pi}}  \, \sum_{r \geq 4} 
\sum_{\uu \in \uf} \sum_{\uv \in \ug} 
\nabla_{\uf \uu} \overline{\nabla}_{\ug \uv} \, {1 \over r} \, \frak{Re}
\Big( { e^{\mathrm{i}\theta_\uv} - e^{\mathrm{i}\theta_\uu} \over {z_\mathrm{cr}(\uf) - z_\mathrm{cr}(\ug)} } \Big)^r
\end{array}
\right.
\end{array} 
\]
The vanishing of the coefficients of order $r \leq 2$ is straightforward. We present
the calculation of the coefficient of $(z_\mathrm{cr}(\uf) - z_\mathrm{cr}(\ug))^{-3}$ occurring in
 $\big[ \, \nabla G^{(0)}  \overline{\nabla}^\top \big]_{\uf \ug} $ here:
\[
\begin{array}{ll}
\D  
{1 \over 3} \, \sum_{\uu \in \uf} \sum_{\uv \in \ug} \,
\nabla_{\uf \uu} \overline{\nabla}_{\ug \uv} \, 
\Big( e^{\mathrm{i}\theta_\uv} - e^{\mathrm{i}\theta_\uu} \Big)^3
&\D = 
\left\{
\begin{array}{l}
\D {1 \over 3} \ 
\overbrace{\Big( \sum_{\uu \in \uf}  \, \nabla_{\uf \uu}  \Big)}^{\text{vanishes}} 
\cdot 
\Big(  \sum_{\uv \in \ug} \, \overline{\nabla}_{\ug \uv} \, e^{3i\theta_\uv} \Big) \\ \\
\D - \  
\overbrace{\Big( \sum_{\uu \in \uf}  \, \nabla_{\uf \uu} \, e^{\mathrm{i}\theta_\uu} \Big)}^{\text{equals 1}} 
\cdot 
\overbrace{\Big(  \sum_{\uv \in \ug} \, \overline{\nabla}_{\ug \uv} \, e^{2i\theta_\uv} \Big)}^{
-\prod_{\uv \in \ug}  e^{\mathrm{i} \theta_\uv}  }
\\ \\
\D + \  \Big( \sum_{\uu \in \uf}  \, \nabla_{\uf \uu} \, e^{2i\theta_\uu} \Big) \cdot 
\overbrace{\Big(  \sum_{\uv \in \ug} \, \overline{\nabla}_{\ug \uv} \, e^{\mathrm{i}\theta_\uv} \Big)}^{\text{vanishes}} 
\\ \\
\D - \ {1 \over 3} \, \Big( \sum_{\uu \in \uf}  \, \nabla_{\uf \uu} \, e^{3i\theta_\uu} \Big) \cdot 
\overbrace{\Big(  \sum_{\uv \in \ug} \, \overline{\nabla}_{\ug \uv}  \Big)}^{\text{vanishes}} 
\end{array}
\right. 
\end{array}
\]
Thanks to Lemma~\ref{lnablabound} (or in this case through a direct estimate) its norm is uniformly bounded by  a constant independent 
of the shape of the faces.
For $G^{(0)}$, which is a smooth function of the vertex coordinates, these calculations
amount to replacing $\nabla$ and $\overline\nabla$ by their corresponding continuous derivatives 
$\partial$ and $\bar\partial$, up to subdominant terms of order $\mathrm{0}(d^{-3})$.
This is in agreement with Lemma~\ref{lemmabound}.
The result is that the asymptotics \ref{nnadjbnnt} and \ref{nntbnnadj} are valid for $G^{(0)}$ alone.

To end the proof of the lemma, one must show that the corresponding derivative terms for $G^{(2)}$ and $G^{(4)}$ are $\mathrm{O}(d^{-3})$.
This is clear for $G^{(4)}$, which is itself  $\mathrm{O}(d^{-4})$, hence its discrete derivatives are also $\mathrm{O}(d^{-4})$.
But this is not obvious for $G^{(2)}$ which is only $\mathrm{O}(d^{-2})$. We must use the explicit form of $G^{(2)}$.
Let us consider the term 
\begin{equation*}
\label{ }
\sum_{\uu \in \uf} \sum_{\uv \in \ug}\nabla_{\uf \uu} \left({p_3(\uu, \uv)
\over (z_\mathrm{cr}(\uu)-z_\mathrm{cr}(\uv))^3} \right)\nabla^\top_{\uv \ug}
\end{equation*}
which appears in $\nabla G^{(2)}\nabla^\top$.
One has
$$ p_3(\uu,\uv)=p_3( \uo_\uf, \uo_\ug)+ e^{-3 i \theta_\uu} - e^{- 3 i \theta_\uv}$$
So we have to consider three terms. 
The first term is
\begin{equation*}
\label{ }
\begin{split}
& \sum_{\uu \in \uf} \sum_{\uv \in \ug} \nabla_{\uf \uu} \left({p_3( \uo_\uf, \uo_\ug)
\over (z_\mathrm{cr}(\uu)-z_\mathrm{cr}(\uv))^3}\right)\nabla^\top_{\uv \ug} \\
& =p_3(\uo_\uf, \uo_\ug) \sum_{\uu \in \uf} \sum_{\uv \in \ug} \nabla_{\uf \uu} 
\left({1\over (z_\mathrm{cr}(\uu)-z_\mathrm{cr}(\uv))^3}\right)\nabla^\top_{\uv \ug} \\
& = p_3(\uo_\uf, \uo_\ug) \left( {
-12 
\over (z_\mathrm{cr}(\uf) -z_\mathrm{cr}(\ug))^5} +\mathcal{O}(d^{-6})\right)\\
& = \mathrm{O}(d^{-4})
\end{split}
\end{equation*}
In the last step we used the uniform bound from Lemma~\ref{lemma-p-estimate} $$|p_3(\uo_\uf, \uo_\ug)|\le 3\, 
|z_\mathrm{cr}(\uf)-z_\mathrm{cr}(\ug) | = 3 d$$
The second term is
\begin{equation*}
\label{ }
\begin{split}
& \sum_{\uu \in \uf} \sum_{\uv \in \ug} \nabla_{\uf \uu} 
\left({e^{-3 i \theta_\uu} \over (z_\mathrm{cr}(\uu)-z_\mathrm{cr}(\uv))^3}\right)\nabla^\top_{\uv \ug} 
 =\sum_{\uu \in \uf} \nabla_{\uf \uu} \left({ 3\,e^{-3 i \theta_\uu} 
 \over (z_\mathrm{cr}(\uf) -z_\mathrm{cr}(\ug))^4}+\mathrm{O}(d^{-5})\right) \\
& =  { 3\,\left(\sum\limits_{\uu \in \uf} \nabla_{\uf \uu} \, e^{-3 i \theta_\uu}\right)}
{1\over (z_\mathrm{cr}(\uf)-z_\mathrm{cr}(\ug))^4}+\mathrm{O}(d^{-5}) 
\end{split}
\end{equation*}
From Lemma~\ref{lnablabound}
$$\left|\sum\limits_{\uu \in \uf} \nabla_{\uf \uu} \, e^{-3 i \theta_\uu}\right|\le 6$$ 
hence the second term is of order $\mathrm{O}(d^{-4})$.
By the same argument, the third term is
\begin{equation*}
\label{ }
\begin{split}
& - \sum_{\uu \in \uf} \sum_{\uv \in \ug} \nabla_{\! {\uf \uu}} 
\left( {e^{-3 i \theta_\uv} \over (z_\mathrm{cr}(\uu)-z_\mathrm{cr}(\uv))^3}\right)\nabla^\top_{\uv \ug}  = \mathrm{O}(d^{-4})
\end{split}
\end{equation*}
This ends the derivation of \ref{nntbnnadj} (the second equation is the c.c.).
The derivation of \ref{nnadjbnnt} proceeds in a similar way.
\end{proof}
We are now in a position to state the main result.
\begin{Prop}
\label{ThDelta}
The second order variation for the Laplace-Beltrami operator $\Delta(\underline{\epsilon})$ on an isoradial, Delaunay graph $\uG_\mathrm{cr}$ is
\begin{equation}
\label{OPEDelta}
\begin{split}
   &\tr\left[\deltaeo \Delta\cdot \Delta^{-1}_\mathrm{cr} \cdot \deltaet \Delta\cdot \Delta^{-1}_\mathrm{cr} \right] =  \\
 &\qquad{1\over \pi^2} \sum_{\uf \in \overline{\Omega}_1} \sum_{ \ug \in \overline{\Omega}_2} 
   \,A(\uf)\,A(\ug) \,\left[{\overline\nabla  F_1(\uf)  \overline\nabla  F_2(\ug)\over \big(z_\mathrm{cr}(\uf)-z_\mathrm{cr}(\ug) \big)^4}
   + {\nabla  \bar F_1(\uf)  \nabla  \bar F_2(\ug)\over \big(\bar z_\mathrm{cr}(\uf)-\bar z_\mathrm{cr}(\ug) \big)^4}\right]\ +\ \mathrm{O}(d^{-5})
\end{split}
\end{equation}
where the double sum is taken over pairs of triangles $\uf, \ug \in \mathrm{F}(\widehat{\uG}_{0^+})$ 
such that all vertices of $\uf$ reside in $ \overline{\Omega}_1$ 
and all vertices of $\ug$ reside in $ \overline{\Omega}_2$.
\end{Prop}

\begin{proof}
We start from the local form of the $\Delta(\epsilon)$ operator \ref{DeltaNablaForm}, which implies that the first order variation on 
$\Delta(\epsilon)$ is
$$
\deltae\Delta= 2\left(\deltae\overline\nabla^{\!\top} {{A}}\,\nabla+\overline\nabla^{\!\top} {{\deltae A}}\,\nabla+\overline\nabla^{\!\top} {{A}}\,\deltae\nabla
+ \deltae\nabla^{\!\top} {{A}}\,\overline\nabla+ \nabla^{\!\top} {{\deltae A}}\,\overline\nabla+ \nabla^{\!\top} {{A}}\,\deltae\overline\nabla \right)
$$
We use the formula for the variation of $A$
\begin{equation*}
\label{ }
\deltae A= A (\nabla\! F+{\overline\nabla}\!\bar F)
\end{equation*}
as well as the formulae for the variations of the $\nabla$ and $\overline\nabla$ operators given by \ref{VarNabla}, which read
\begin{equation*}\label{VarNabla2}
\begin{split}
\deltae\nabla = -\, \left(\nabla\! F \,\nabla + \nabla\!\bar F \,\overline\nabla \right)\\
\deltae\overline\nabla=- \, \left(\overline\nabla\! \bar F \,\overline\nabla + \overline\nabla\! F \,\nabla \right)
\end{split}
\end{equation*}
to get
\begin{equation}
\label{ }
\deltae\Delta=-4\left( {\overline\nabla}^\top (\nabla\!\bar F) A\,\overline\nabla +\nabla^{\!\top} (\overline\nabla\!F)A\,\nabla\right)
\end{equation}
One uses this and the cyclicity of the trace to rewrite the second order variation as
\begin{equation*}
\begin{split}
\tr\left[\deltaeo \Delta\cdot \Delta^{-1}_\mathrm{cr} {\cdot} \deltaet \Delta{\cdot} \Delta^{-1}_\mathrm{cr} \right] &=  16\ \Big[
\tr\left(A\,\nabla\!\bar F_1\cdot \overline\nabla \Delta^{-1}_\mathrm{cr} {\overline\nabla}^\top\!{\cdot} A\,\nabla\!\bar F_2\cdot \overline\nabla 
\Delta^{-1}_\mathrm{cr} \, {\overline\nabla}^\top\right)\\
&+
\tr\left(A\,\overline\nabla\! F_1\cdot\nabla \Delta^{-1}_\mathrm{cr} {\overline\nabla}^\top\!\cdot A\,\nabla\!\bar F_2\cdot \overline\nabla \Delta^{-1}_\mathrm{cr} \, \nabla^{\!\top}\right)
\\
&+
\tr\left(A\,\nabla\!\bar F_1\cdot\overline\nabla \Delta^{-1}_\mathrm{cr} \nabla^{\!\top}\!\cdot A\,\overline\nabla\! F_2\cdot \nabla \Delta^{-1}_\mathrm{cr} \, {\overline\nabla}^\top\right)
\\&
+
\tr\left(A\,\overline\nabla\! F_1\cdot\nabla \Delta^{-1}_\mathrm{cr} \nabla^\top\!\cdot A\,\overline\nabla\! F_2\cdot \nabla \Delta^{-1}_\mathrm{cr} \, \nabla^{\!\top}\right)
\Big]
\end{split}
\end{equation*}
Note that the trace on the l.h.s. is a sum over vertices, while the trace on the r.h.s. is a sum over faces (triangles).
Using the large distances asymptotics \ref{nnadjbnnt} and \ref{nntbnnadj}, and writing the trace explicitly as a double sum over faces $\uf$ and $\ug$ gives the theorem.
\end{proof}

We now consider the other operators. The case of the conformal Laplacian is more complicated, so let us first discuss the K\"ahler operator.

\subsubsection{The K\"ahler operator $\mathcal{D}$}
\label{ssKalOpSecOrd}
\begin{Prop}
\label{The3}
The second order variation for the K\"ahler operator $\mathcal{D}(\underline{\epsilon})$ on an isoradial, Delaunay graph 
$\uG_\mathrm{cr}$ has the same form as the second order variation 
for the Laplacian $\Delta(\underline{\epsilon})$
\begin{equation}
\label{OPEKaehler}
\begin{split}
   &\tr\left[\deltaeo \mathcal{D}\cdot \mathcal{D}^{-1}_\mathrm{cr} \cdot \deltaet \mathcal{D}\cdot \mathcal{D}^{-1}_\mathrm{cr} \right] =  \\
 &\quad{1\over \pi^2} \sum_{\uf \in \overline{\Omega}_1} \sum_{\ug \in \overline{\Omega}_2} 
   \,A(\uf)\,A(\ug) \,\left[{\overline\nabla  F_1(\uf)  \overline\nabla  F_2(\ug)\over \big( z_\mathrm{cr}(\uf) - z_\mathrm{cr}(\ug) \big)^4}
   + {\nabla  \bar F_1(\uf)  \nabla  \bar F_2(\ug)\over \big( \bar{z}_\mathrm{cr}(\uf) -
   \bar{z}_\mathrm{cr}(\ug) \big)^4}\right]\ +\ \mathrm{O}(d^{-5})
\end{split}
\end{equation}
where the double sum is taken over pairs of triangles $\uf, \ug \in \mathrm{F}(\widehat{\uG}_{0^+})$ 
such all vertices of $\uf$ reside in $ \overline{\Omega}_1$ 
and all vertices of $\ug$ reside in $ \overline{\Omega}_2$ 
\end{Prop}
\begin{proof}
The derivation goes along the same line. We start from Prop.~\ref{PVarKahler} which gives the explicit form \ref{VarD} of the first order variation of 
$\mathcal{D}(\epsilon)$. The graph $\uG_\mathrm{cr}$ is isoradial, so all circumradii are equal $R(\uf)=R_{\mathrm{cr}}$ and thus
$$\mathcal{D}_\mathrm{cr} = R_\mathrm{cr}^{-2} \Delta_{\mathrm{cr}}$$
This implies that the first order variation of $\mathcal{D}(\epsilon)$ has the special form
\begin{equation}
\label{ }
\deltae\mathcal{D} = R^{-2}_\mathrm{cr} \deltae \Delta
 - 4 R^{-2}_\mathrm{cr} \overline\nabla^\top \Big( 
 A \big( \nabla\! F + \overline\nabla\! \bar F \big) + C\, \overline\nabla\! F + \bar C\, \nabla\!\bar F \Big) \nabla
\end{equation}
with $C$ and $\bar C$ defined by \ref{CDef}.
Formula \ref{OPEKaehler} follows by
repeating the analysis made in the proof of Proposition~\ref{ThDelta}, which relies on the asymptotics of Lemma~\ref{lNGN}.
One can check that the new terms involving $C$ and $\bar C$ do not change the asymptotics  \ref{OPEDelta} obtained for $\Delta$. 
\end{proof}

\subsection{The case of the conformal Laplacian: the anomalous term}
\label{ssConfLapl2nd}
\subsubsection{Second order variation for the conformal Laplacian $\Deltaconf$}
\label{ssConfLApDEl}
\newcommand{\scrz}{\mathscr{\scriptstyle{Z}}}

By formula \ref{varDelConf} in the proof of Proposition \ref{T1stVarC}
the contribution made by regular edges
$\ue \in \mathrm{E}(\uG_{0^+}^\bullet)$ to the first order variation $\deltae\Deltaconf$ of the conformal Laplacian is identical to the variation 
$\deltae\Delta$  of the Laplace-Beltrami Laplacian. There is, however, an additional term in the first order variation $\deltae\Deltaconf$ 
coming from the chords of $\uG_{0^+}$. We call it the ``anomalous term'' and denote  it $\delta \mathbb{A}$:
\begin{equation}
\label{ }
 \deltae \Deltaconf= \deltae \Delta + \deltae \mathbb{A}
\end{equation}
The non-diagonal elements of $\deltae \mathbb{A}$ are non-zero only for chords. 
From \ref{varDelConf}, for vertices $\uu \neq \uv$, they are 
\begin{equation}
\label{dAoff}
\deltae \mathbb{A}( \vec{\ue} \, ) = 
\big[\deltae \mathbb{A}\big]_{\uu \uv}=
\begin{cases}
      {1\over 2}\big(\deltae\theta_\mathrm{n}(\vec{\ue} \, )  \tan^2 \theta_\mathrm{n}(\vec{\ue} \, )
      +\deltae\theta_\mathrm{s}(\vec{\ue} \, ) \tan^2 \theta_\mathrm{s}(\vec{\ue} \, ) \big)
      &\! \!\begin{array}{l} \text{if $\ue = \overline{\uu \uv}$ is a} \\ \text{chord in $\mathrm{E}(\uG_{0^+})$,} \end{array} \\
      \qquad 0& \ \,  \text{otherwise}.
\end{cases}
\end{equation}
Here $\ue = \overline{\uu \uv}$ is an edge of $\uG_{0^+}$ and $\vec{\ue} = (\uu, \uv)$ is an orientation. 
The graph $\uG_{0^+}$ is isoradial and weakly Delaunay and so
$\theta_\mathrm{n}(\vec{\ue} \, ) = \pm \, \theta_\mathrm{s}(\vec{\ue} \, )$ for any edge. 
In particular 
$\tan^2 \theta_\mathrm{n}(\vec{\ue} \,) = \tan^2 \theta_\mathrm{s}(\vec{\ue} \,)$
and so $\deltae \mathbb{A}( \vec{\ue} \, ) = \deltae \mathbb{A}( \vec{\ue}^{\, *}  )$ where
$\vec{\ue}^{\, *} = (\uv, \uu)$ is the opposite orientation. As for the diagonal terms we have
\begin{equation}
\label{dAdiag}
\big[\deltae \mathbb{A}\big]_{\uu \uu}=-\sum_{\uv \neq \uu} \big[\deltae \mathbb{A}\big]_{\uu \uv}
\end{equation}
In the case of a chord  $\vec{\ue}$
we may use \ref{VarThetaNS} for the angle variations
$\deltae\theta_\mathrm{n}(\vec{\ue} \,)$ and $\deltae\theta_\mathrm{s}(\vec{\ue} \, )$
and re-express the anomalous term $\deltae \mathbb{A}(\vec{\ue} \, )$ given in
formula \ref{dAoff} as
\begin{equation}
\label{delAexpl}
\deltae \mathbb{A}(\vec{\ue} \, )={1\over 2} \, \frak{Im} \Big[ \overline\nabla\! F(\uf_\mathrm{n} )\, 
\mathscr{E}_\mathrm{n}(\vec{\ue} \, ) \tan^2 \theta_\mathrm{n}(\vec{\ue} \, ) +
\overline\nabla\! F(\uf_\mathrm{s})\,\mathscr{E}_\mathrm{s}( \vec{\ue} \,) \tan^2 \theta_\mathrm{s}(\vec{\ue} \,) \Big]
\end{equation}
where the functions
$\mathscr{E}_\mathrm{n}(\vec{\ue})$ and $\mathscr{E}_\mathrm{s}(\vec{\ue})$ are defined in
\ref{mathcalE-terms} and where $\uf_\mathrm{n}$ and $\uf_\mathrm{s}$ are the respective 
north and south triangles abutting $\vec{\ue}$ in the triangulation $\widehat{\uG}_{0^+}$ which completes
$\uG_{0^+}$. 

The second order variation 
\begin{equation}
\tr\left[  \deltaeo \, \Deltaconf \cdot \Delta_\mathrm{cr}^{-1}
\cdot  \deltaet \, \Deltaconf  \cdot \Delta_\mathrm{cr}^{-1}\right] 
\end{equation} 
is the sum of the second order variation made by the Laplace-Beltrami Laplacian, namely
\begin{equation}
\label{varDD}
\tr\left[\deltaeo \Delta\cdot \Delta_\mathrm{cr}^{-1}\cdot\deltaet\Delta\cdot \Delta_\mathrm{cr}^{-1}\right]\,
\end{equation}
\noindent
along with three anomalous trace terms which we can express (in light of \ref{dAdiag}) as follows:
\begin{equation}
\label{var-anomalous}
\begin{array}{lr}
\D 
\underbrace{
\tr\left[\deltaeo \mathbb{A}\cdot \Delta_\mathrm{cr}^{-1}
\cdot\deltaet\Delta\cdot \Delta_\mathrm{cr}^{-1}\right]}_{\text{chord-edge term}}
&\D = \ 
\sum_{\substack{\text{chords} \ \vec{\ue}_1 \, \in \, \uG_{0^+} \\ \text{edges} \ \vec{\ue}_2 \, \in \, \widehat{\uG}_{0^+} }}
\deltaeo \mathbb{A}(\vec{\ue}_1) \, K(\vec{\ue}_1, \vec{\ue}_2) \,  \deltaet \Delta (\vec{\ue}_2)  \,  K(\vec{\ue}_2, \vec{\ue}_1) \\ \\
\D \underbrace{
\tr\left[\deltaeo \Delta \cdot \Delta_\mathrm{cr}^{-1}
\cdot\deltaet \mathbb{A} \cdot \Delta_\mathrm{cr}^{-1}\right]}_{\text{edge-chord term}}
&\D = \
\sum_{\substack{\text{edges} \ \vec{\ue}_1 \, \in \, \widehat{\uG}_{0^+} \\ \text{chords} \ \vec{\ue}_2 \, \in \, \uG_{0^+} }}
\deltaeo \Delta (\vec{\ue}_1) \, K(\vec{\ue}_1, \vec{\ue}_2) \,  
\deltaet \mathbb{A}(\vec{\ue}_2) \,  K(\vec{\ue}_2, \vec{\ue}_1 ) \\ \\
\D \underbrace{
\tr\left[\deltaeo \mathbb{A}\cdot \Delta_\mathrm{cr}^{-1}
\cdot\deltaet \mathbb{A}\cdot \Delta_\mathrm{cr}^{-1}\right]}_{\text{chord-chord term}}
&\D = \
\sum_{\substack{\text{chords} \\  \vec{\ue}_1 , \, \vec{\ue}_2 \, \in \,  \uG_{0^+}}}
\ \deltaeo\mathbb{A}(\vec{\ue}_1 )\ K(\vec{\ue}_1, \vec{\ue}_2)
\ \deltaet\,\mathbb{A}(\vec{\ue}_2 )\ K(\vec{\ue}_2, \vec{\ue}_1)
\end{array}
\end{equation}
where $\vec{\ue}_1 = (\uu_1, \uv_1)$ and $\vec{\ue}_2 =(\uu_2, \uv_2)$ are oriented edges
of the triangulation $\widehat{\uG}_{0^+}$ whose vertices $\uu_1, \uv_1$ and $\uu_2, \uv_2$
lie in $\overline{\Omega}_1$ and $\overline{\Omega}_2$ respectively and where
\begin{equation}
\label{ }
K(\vec{\ue}_1, \vec{\ue}_2)
\ := \ \big[\Delta_\mathrm{cr}^{-1}\big]_{\uv_1 \uv_2}-\big[\Delta_\mathrm{cr}^{-1}\big]_{\uu_1 \uv_2}
-\big[\Delta_\mathrm{cr}^{-1}\big]_{\uv_1 \uu_2}+\big[\Delta_\mathrm{cr}^{-1}\big]_{\uu_1 \uu_2} 
\end{equation}
Note that $K(\vec{\ue}_1, \vec{\ue}_2)=K(\vec{\ue}_2, \vec{\ue}_1)=-K(\vec{\ue}^{\, *}_1, \vec{\ue}_2)$ 
where $\vec{\ue}^{\, *}_1=(\uv_1, \uu_1)$ has the reverse orientation. 
Applying two rounds of formula \ref{diff-formula-nablas} we obtain
\begin{equation}
\label{K-formula-nabla}
\begin{array}{ll}
\D K(\vec{\ue}_1, \vec{\ue}_2)
&\D = \ \left\{
\begin{array}{c}
p_1(\uu_2, \uv_2) \Big[ \Delta_\mathrm{cr}^{-1}  \nabla^\top \Big]_{\uu_1 \uf_2}  \  - \
p_1(\uu_2, \uv_2) \Big[ \Delta_\mathrm{cr}^{-1}  \nabla^\top \Big]_{\uv_1 \uf_2}  \\ \\
\D + \\ \\
\D 
\overline{p}_1(\uu_2, \uv_2) \Big[ \Delta_\mathrm{cr}^{-1}  \overline{\nabla}^\top \Big]_{\uu_1 \uf_2}  \  - \
\overline{p}_1(\uu_2, \uv_2) \Big[ \Delta_\mathrm{cr}^{-1}  \overline{\nabla}^\top \Big]_{\uv_1 \uf_2} 
\end{array}
\right. \\ \\
&\D = \ 2 \, \frak{Re}  \left[
\begin{array}{c}
p_1(\uu_1, \uv_1) \, p_1(\uu_2,\uv_2) \Big[ \nabla \Delta_\mathrm{cr}^{-1} \nabla^\top \Big]_{\uf_1 \uf_2} 
\\ \\
\D + \\ \\
p_1(\uu_1, \uv_1) \, \overline{p}_1(\uu_2,\uv_2) \Big[ \nabla  \Delta_\mathrm{cr}^{-1} \overline{\nabla}^\top \Big]_{\uf_1 \uf_2} 
\end{array}
\right]
\end{array}
\end{equation}
where $\uf_i$ is a triangle of $\widehat{\uG}_{0^+}$, north or south, containing the edge $\vec{\ue}_i$
for $i=1,2$.
By assumption 
$\overline{\Omega}_1$ and $\overline{\Omega}_2$
are separated by a large distance $d\gg R_\mathrm{cr}$
and so 
we can estimate
$K(\vec{\ue}_1, \vec{\ue}_2)$ 
as presented in formula \ref{K-formula-nabla}
using 
asymptotic expansions
 \ref{nnadjbnnt} and \ref{nntbnnadj}
of Lemma~\ref{lNGN}. We end up with
\begin{equation}
\label{KAsympt}
K(\vec{\ue}_1, \vec{\ue}_2)\ =\ {1\over 2\pi}\,
\frak{Re}
\left[ { p_1(\uu_1, \uv_1) \, p_1(\uu_2, \uv_2)  \over { \big( 
z_\mathrm{cr}(\uf_1) -  z_\mathrm{cr}(\uf_2) \big)^2 } } \right]  
\ + \ \mathrm{O}\left( {1 \over { \big|  z_\mathrm{cr}(\uf_1) -  z_\mathrm{cr}(\uf_2) \big|^3}} \right)
\end{equation}
where $p_1(\uu, \uv) = z_\mathrm{cr}(\uv) - z_\mathrm{cr}(\uu)$ as introduced in Definition \ref{def-p_n}.

\subsubsection{The chord-chord term}
Let us begin by examining the chord-chord term of \ref{var-anomalous}.
It involves the contribution of two (oriented) chords
$\vec{\ue}_1 = (\uu_1, \uv_1)$ and $\vec{\ue}_2 =( \uu_2, \uv_2)$ 
whose vertices of $\uu_1, \uv_1$ and $\uu_2, \uv_2$
are contained in $\overline{\Omega}_1$ and $\overline{\Omega}_2$ respectively.
Since $\vec{\ue}_i = (\uu_i, \uv_i)$ is a chord for $i=1 ,2$ in $\uG_{0^+}$ 
the corresponding north and south triangles ${\uf_i}_\mathrm{n}$ and ${\uf_i}_\mathrm{s}$ in $\widehat{\uG}_{0^+}$
are concyclic and therefore share a common circumcenter 
whose complex coordinate we denote  $\scrz_\mathrm{cr}(\vec{\ue}_i)=z({\uf_i}_\mathrm{n})=z({\uf_i}_\mathrm{s})$.
This is depicted in Fig~\ref{f2chords}. 

\begin{figure}[h!]
\begin{center}
\includegraphics[width=4.5in,angle=0]{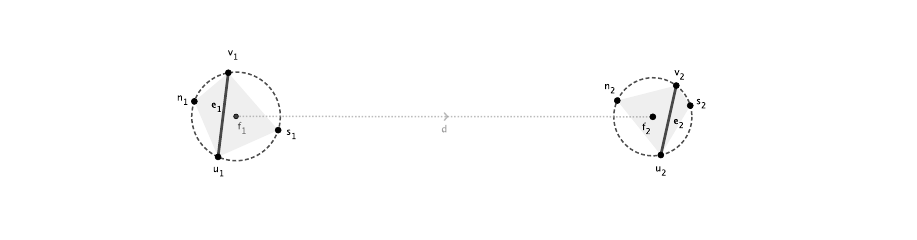}
\caption{Two far apart chords $\vec\ue_1=(\uu_1 \uv_1)$ and $\vec\ue_2=(\uu_2\uv_2)$ at distance $d\gg 1$}
\label{f2chords}
\end{center}
\end{figure}

Putting things together, we see that the contribution made by a pair of (oriented) chords 
$(\vec{\ue}_1, \vec{\ue}_2)$ to the chord-chord anomalous trace term in \ref{var-anomalous} is
\begin{equation}
\label{chord2chord}
{1\over 16\,\pi^2}\,
\deltaeo \mathbb{A}(\vec{\ue}_1) \, 
\deltaet  \mathbb{A}(\vec{\ue}_2) 
\left( \frak{Re}
\left[ { p_1( \uu_1,\uv_1) \, p_1(\uu_2,\uv_2)  \over { \big( 
\scrz_\mathrm{cr}(\vec{\ue}_1) -  \scrz_\mathrm{cr}(\vec{\ue}_2) \big)^2 } } \right]  \right)^2 
+ \, \mathrm{O} \left({1 \over {\big|  \scrz_\mathrm{cr}(\vec{\ue}_1) -  \scrz_\mathrm{cr}(\vec{\ue}_2) \big|^5 }} \right)
\end{equation}
with $\deltaeo \mathbb{A}(\vec{\ue}_1)$ and  $\deltaet  \mathbb{A}(\vec{\ue}_2) $ given by \ref{delAexpl}, that we recall for completeness.
\begin{equation*}
\label{delAexpl*}
\deltae \mathbb{A}(\vec{\ue} \, )={1\over 2} \, \frak{Im} \Big[ \overline\nabla\! F(\uf_\mathrm{n} )\, 
\mathscr{E}_\mathrm{n}(\vec{\ue} \, ) \tan^2 \theta_\mathrm{n}(\vec{\ue} \, ) +
\overline\nabla\! F(\uf_\mathrm{s})\,\mathscr{E}_\mathrm{s}( \vec{\ue} \,) \tan^2 \theta_\mathrm{s}(\vec{\ue} \,) \Big]
\end{equation*}
with
\begin{equation*}
\label{ }
\mathcal{E}_\mathrm{n}(\vec\ue \, ) 
= \ { \overline{z}(\uv) - \overline{z}(\un) \over {z(\uv) - z(\un) } } - 
{ \overline{z}(\uu) - \overline{z}(\un) \over {z(\uu) - z(\un) } } 
= \ {-4 A(\uf_\mathrm{n}) \over {\big(z(\uv) - z(\un) \big) \big(z(\uu) - z(\un) \big)}}
\end{equation*}
and a similar form for $\mathcal{E}_\mathrm{s}(\vec\ue \, )$.
Any triangulation $\widehat{\uG}_{0^+}$ which completes the limit graph $\uG_{0^+}$ is itself isoradial and weakly Delaunay
consequently $\tan^2\theta_\mathrm{n} (\vec{\ue} \,) = \tan^2\theta_\mathrm{s} (\vec{\ue} \,)$ the value of which is given by \ref{TanThetaN}.

The result \ref{chord2chord} for the anomalous chord-chord contribution to the variation of $\log\det\Deltaconf(\underline{\epsilon})$ does not have the same form as the ``regular'' contribution \ref{varDD} which is similar to the variation of the Laplace-Beltrami operator $\Delta$, which is a sum over triangles of terms
\begin{equation*}
\label{tr2trterm*}
 A(\uf_1) A(\uf_2) {\overline\nabla\!F_1(\uf_1)\cdot\overline\nabla\!F_2(\uf_2)\over (z_{\mathrm{cr}}(\uf_1)-z_{\mathrm{cr}}(\uf_2))^4}\ +\  \text{c.c.}
\end{equation*}

\noindent
First, besides harmonic terms in the coordinate of the circumcenters of the form
$$  \big( \scrz_\mathrm{cr}(\vec{\ue}_1) -  \scrz_\mathrm{cr}(\vec{\ue}_2) \big)^{-4} \quad \text{and} \quad 
\big( \overline\scrz_\mathrm{cr}(\vec{\ue}_1) -  \overline\scrz_\mathrm{cr}(\vec{\ue}_2) \big)^{-4}\
$$
it contains non-harmonic terms of the form 
$$
{\big | \scrz_\mathrm{cr}(\vec{\ue}_1) -  \scrz_\mathrm{cr}(\vec{\ue}_2) \big |}^{-4}
$$
which are problematic with conformal invariance and an interpretation in term of CFT, as will be discussed in Sect.~\ref{sDiscussion}.

Second, from the form of $\deltaeo \mathbb{A}(\vec{\ue}_1)$ and  $\deltaet  \mathbb{A}(\vec{\ue}_2) $, it does not contain only terms of the form
$$\overline\nabla F_1(\uf_1) \cdot \overline\nabla F_2(\uf_2)\quad\text{and}\quad \nabla \bar F_1(\uf_2) \cdot \nabla \bar F_2(\uf_2)$$
but also terms of the form
$$\overline\nabla F_1(\uf_1) \cdot \nabla \bar F_2(\uf_2)\quad\text{and}\quad \nabla \bar F_1(\uf_2) \cdot \overline\nabla F_2(\uf_2)$$

Third, the geometric terms associated to the faces (the triangles $\uf_1$ and $\uf_2$) are not simply the area terms $A(\uf_1)$ and $A(\uf_2)$, but they depend of the detailed geometry and orientation of the chords and the triangles through the terms $\mathcal{E}_\mathrm{n/s}(\vec\ue \, )$ and $ \tan^2 \theta_{\mathrm{n/s}}(\vec{\ue} \, )$ .

\subsubsection{The chord-edge term}
We now discuss briefly the chord-edge term 
present in \ref{var-anomalous} which involves the anomalous variation term $[\deltaeo \mathbb{A}]_{\uu_1 \uv_1}$
of a chord $\vec{\ue}_1 = (\uu_1, \uv_1)$ and the ordinary variation term  $[ \deltaet \Delta ]_{\uu_2 \uv_2}$ of  
an edge $\vec{\ue}_2 = (\uu_2 ,\uv_2)$.
It will be simpler to group together the terms
made by a single chord $\vec{\ue}_1 = \vec{\ue} = (\uu, \uv)$ 
and the edges $\vec{\ue}_2$ forming the boundary of a fixed (counter-clockwise oriented) triangle $\uf$ and then sum the contributions
as the chord $\vec{\ue}$ in $\uG_{0^+}$ and triangle $\uf$ in $\widehat{\uG}_{0^+}$ both vary;
see the illustration in Fig.~\ref{fchord-triangle}. 
\begin{figure}[h!]
\begin{center}
\includegraphics[width=5in,angle=2.2]{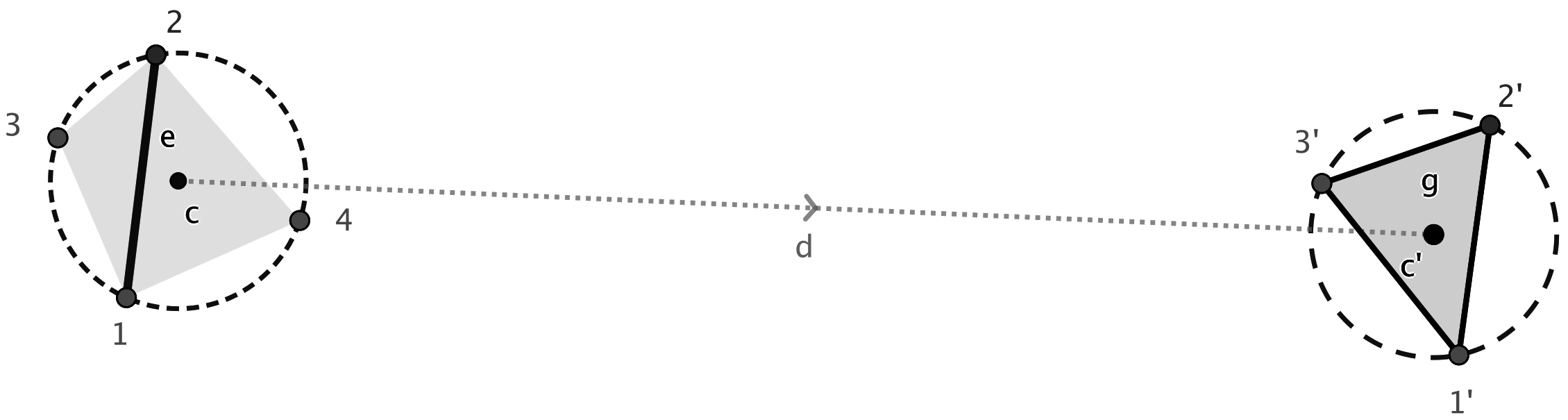}
\caption{A chord $\ue=(12)$ and a triangle $\ug=(1'2'3')$ at distance $d$}
\label{fchord-triangle}
\end{center}
\end{figure}

Accordingly, the contribution
made by a chord-triangle pair $(\vec{\ue}, \uf)$ is found to be
\begin{equation}
\label{chord2edge}
{1\over 4 \pi^2} \, \deltaeo \, \mathbb{A}(\vec{\ue} \, ) \, \frak{Re}
\left[ {p_1^2(\uu, \uv) \, A(\uf)\,\overline\nabla\!F_2(\uf) \over 
{\big(\scrz_\mathrm{cr}(\vec{\ue} \, )- z_\mathrm{cr}(\uf) \big)^4}} \right] \, + 
\,\mathrm{O}
\left({1\over {\big| \scrz_\mathrm{cr}(\vec{\ue})- z_\mathrm{cr}(\uf) \big|^5}}\right)
\end{equation}
\noindent
This term is again different from the regular term. Now it is harmonic in the coordinates of the circumcenters, since it does not contain the non-harmonic term  
$$
{\big | \scrz_\mathrm{cr}(\vec{\ue}_1) -  \scrz_\mathrm{cr}(\vec{\ue}_2) \big |}^{-4}
$$
However, it still contains the terms of the form
$$\overline\nabla F_1(\uf_1) \cdot \nabla \bar F_2(\uf_2)\quad\text{and}\quad \nabla \bar F_1(\uf_2) \cdot \overline\nabla F_2(\uf_2)$$
and it depends on the detailed geometry and orientation of the chord, as for the chord-chord term discussed previously.

\subsubsection{A simplification for specific deformations}\ \\
Finally, let us note that the anomalous term $\deltae \mathbb{A}(\vec{\ue} \,)$ for a chord $\vec{\ue}$  \ref{delAexpl} takes a simpler form in the special case when the discrete derivatives of $F$ coincides on the north and south triangles $\uf_\mathrm{n}(\vec\ue \, )$ and $\uf_\mathrm{s}(\vec\ue \, )$
thanks to the following lemma, 

\begin{lemma}
\label{nabla-and-flip}
Consider two triangles $\mathtt{N} = (\uv_1, \uv_2, \uv_3)$ and $\mathtt{S}= (\uv_2, \uv_1, \uv_4)$ sharing the edge $\overline{\uv_1 \uv_2}$ and the flipped triangles 
$\mathtt{E} =(\uv_3, \uv_4, \uv_2)$ and $\mathtt{W}=(\uv_4, \uv_3, \uv_1)$
sharing the edge $\overline{\uv_3 \uv_4}$
as depicted on Fig.~\ref{2flippedT}. Let $\uv \mapsto F(\uv)$ be a function defined on the vertices. Then the four following equalities are equivalent
\begin{equation}
\label{ }
\nabla\! F(\mathtt{N})=\nabla\! F(\mathtt{S})
\ ,\ \ 
\nabla\! F(\mathtt{E})=\nabla\! F(\mathtt{W})
\ ,\ \ 
\overline\nabla\! F(\mathtt{N})=\overline\nabla\! F(\mathtt{S})
\ ,\ \ 
\overline\nabla\! F(\mathtt{E})=\overline\nabla\! F(\mathtt{W})
\end{equation}
Note that the four points are not necessarily concyclic.
\end{lemma}
\begin{proof}
The proof follows from the definitions \ref{nablaDef} and \ref{barnablaDef}, and it is left to the reader. It has a simple geometric interpretation. Again, note that this is valid for any pair of triangles sharing an edge.
\end{proof}
\begin{figure}[h]
\label{2flippedT}
\begin{center}
\includegraphics[width=1.5in,angle=0]{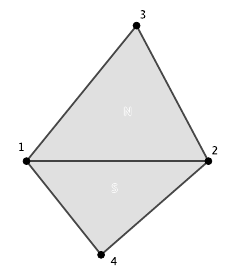}
\quad
\includegraphics[width=1.5in,angle=0]{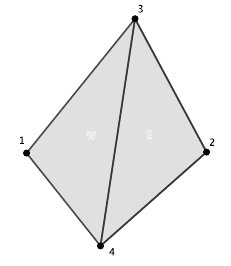}
\caption{\texttt{N}, \texttt{S}, \texttt{E} and \texttt{W} triangles}
\label{ }
\end{center}
\end{figure}
In this case,
a single pair of discrete derivatives $\big(\nabla \!F(\mathtt{c}),\overline\nabla\! F(\mathtt{c}) \big)$ of $F$
is associated to a cocyclic configuration of points, namely a simple cyclic polygon $P=(z_1,z_2,\cdots z_k)$ ($k\ge 4$)
with circumcenter $\mathtt{c}$. The variation $\deltaeo\mathbb{A}(\vec{\ue} \, )$ for a chord $\vec\ue$ is then given by
\begin{equation}
\label{delAexpl2}
\deltaeo\mathbb{A}(\vec{\ue} \, )={1\over 2} \, \frak{Im} \Big[ \overline\nabla\! F (\mathtt{c}) \, 
\big(\mathscr{E}_\mathrm{n}(\vec{\ue} \, )+\mathscr{E}_\mathrm{s}(\vec{\ue} \, ) \big) \Big] \tan^2 \theta_\mathrm{n/s} (\vec{\ue} \, ) 
\end{equation}

\subsection{Curvature dipoles and the anomalous chord term} 
\label{subsubsection-dipole}
Let us discuss a possible explanation of the anomalous terms corresponding to deformations of cocyclic vertex configurations. The adjective "anomalous" indicates that these contributions are not present for either the Laplace-Beltrami operator $\Delta$ or the K\"ahler operator $\mathcal{D}$, both of which admits a smooth continuum limit consistent with the predictions of conformal invariance.

As discussed in the definition \ref{ConfDelta} the conformal Laplacian $\Deltaconf$ for a Delaunay graph $\uG$
can be viewed as the discretized \emph{Laplace-Beltrami operator} on the \emph{rhombic surface} $\uS_\uG^\lozenge$ 
introduced in Def.~\ref{dRhomSurf}. 
The construction of $\uS_\uG^\lozenge$ is illustrated in Fig.~\ref{Kite2LozengeFlat} for an isoradial Delaunay graph $\uG$ 
and in  Fig.~\ref{Kite2Lozenge} for a generic (non-isoradial) Delaunay triangulation $\uG$.

\begin{figure}[h!]
\begin{center}
\includegraphics[scale=.5]{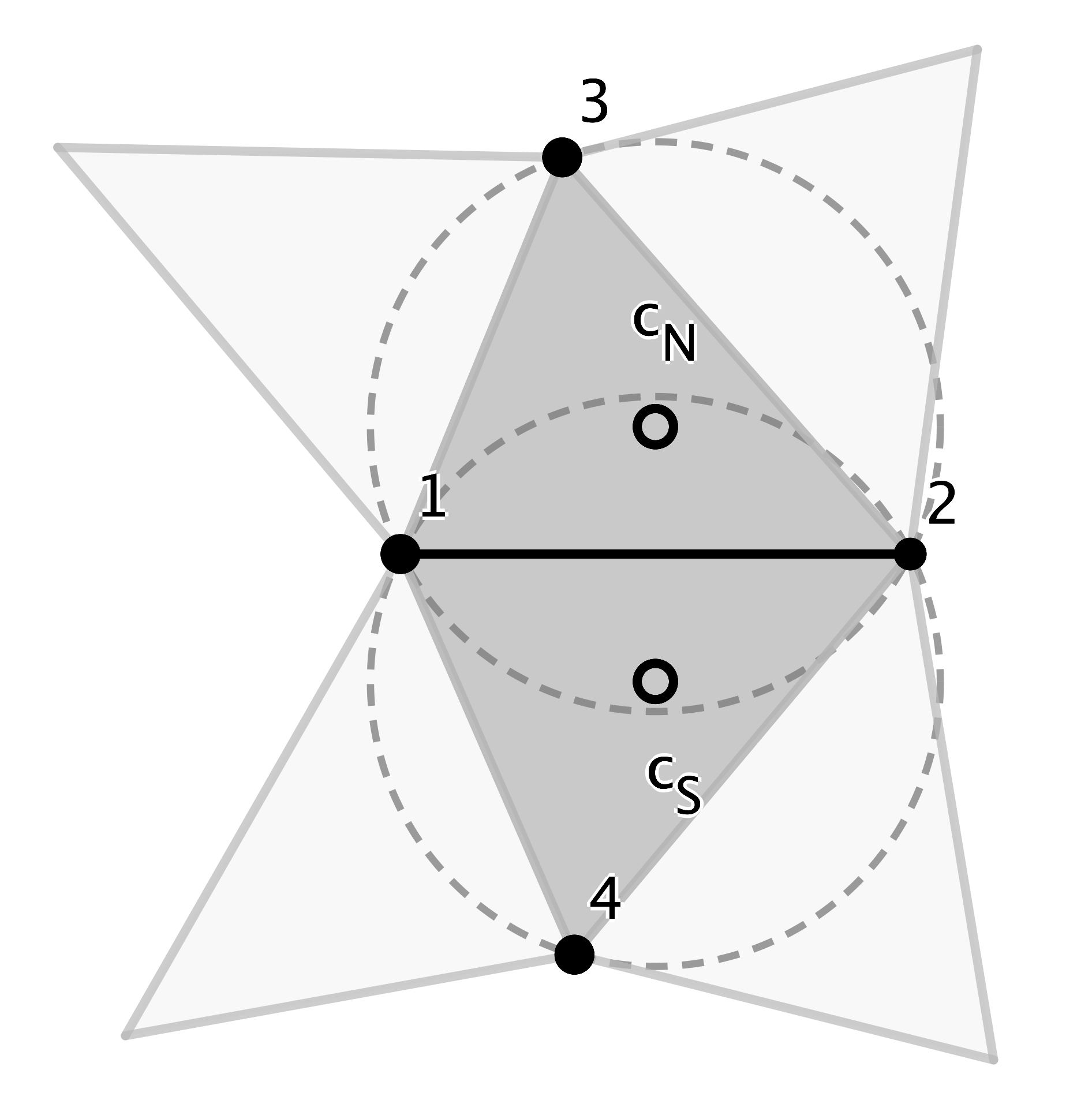}\qquad \includegraphics[scale=.5]{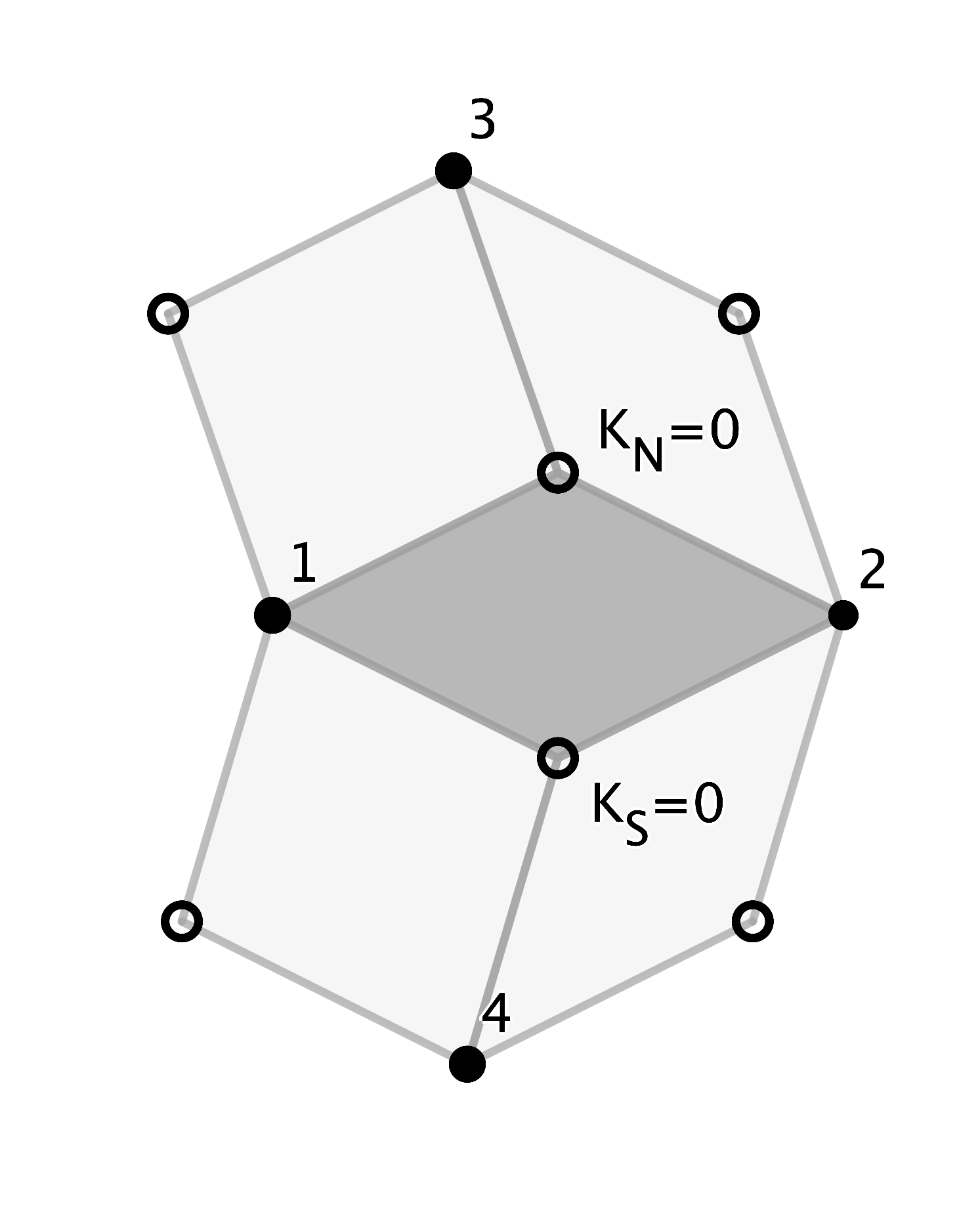}
\caption{A regular edge $\ue=(12)$ of a critical triangulation $\uG$ (left) and its associated rhombic lattice ${\uG}^\lozenge$ (right), the curvature $K$ associated to each face of $\uG$, i.e. its white $\tilde{\uo}$-vertices of ${\uG}^\lozenge$, is zero.}
\label{Kite2LozengeFlat}
\end{center}
\end{figure}
\begin{figure}[h!]
\begin{center}
\includegraphics[scale=.5]{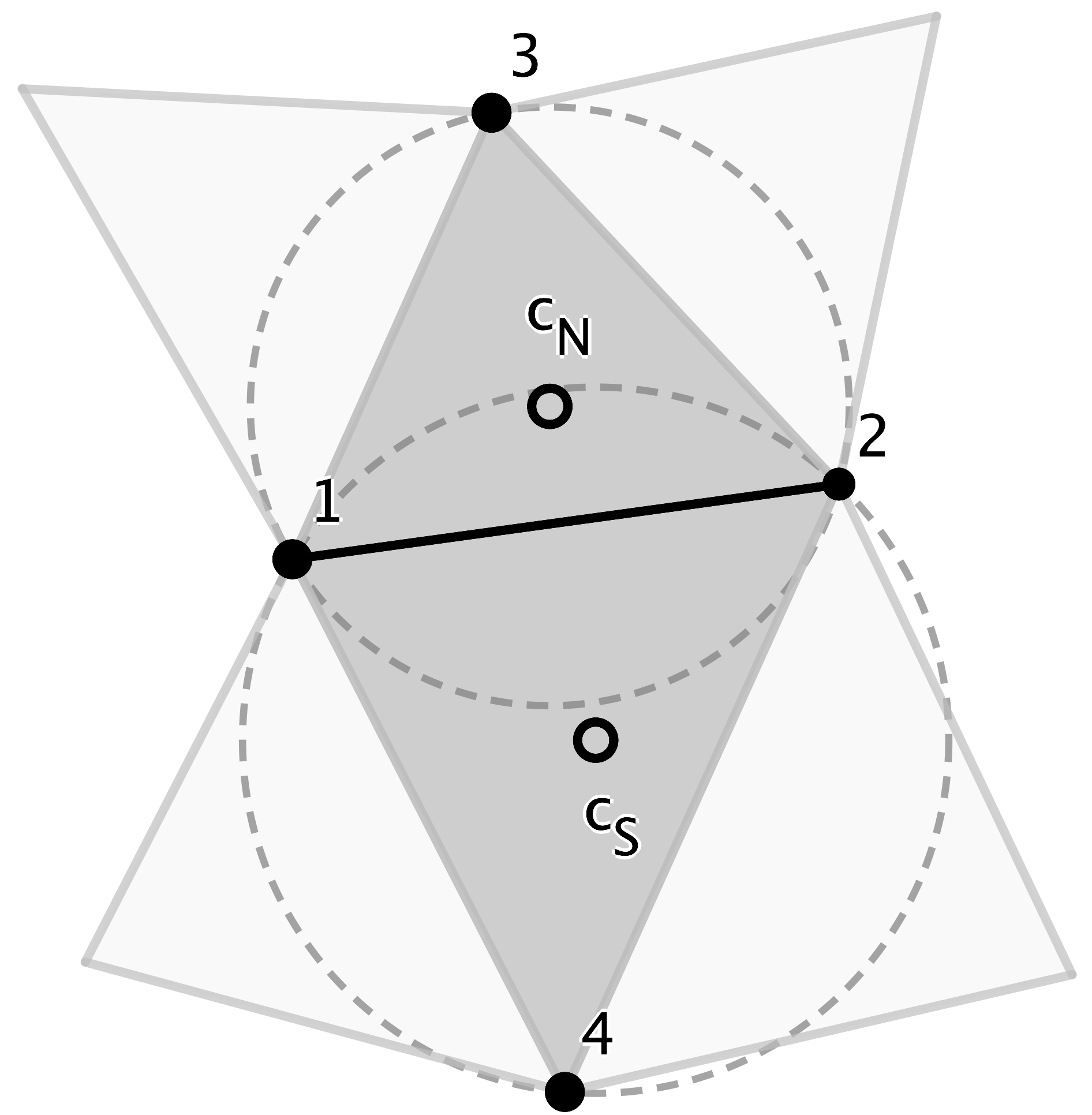}\qquad  \includegraphics[scale=.5]{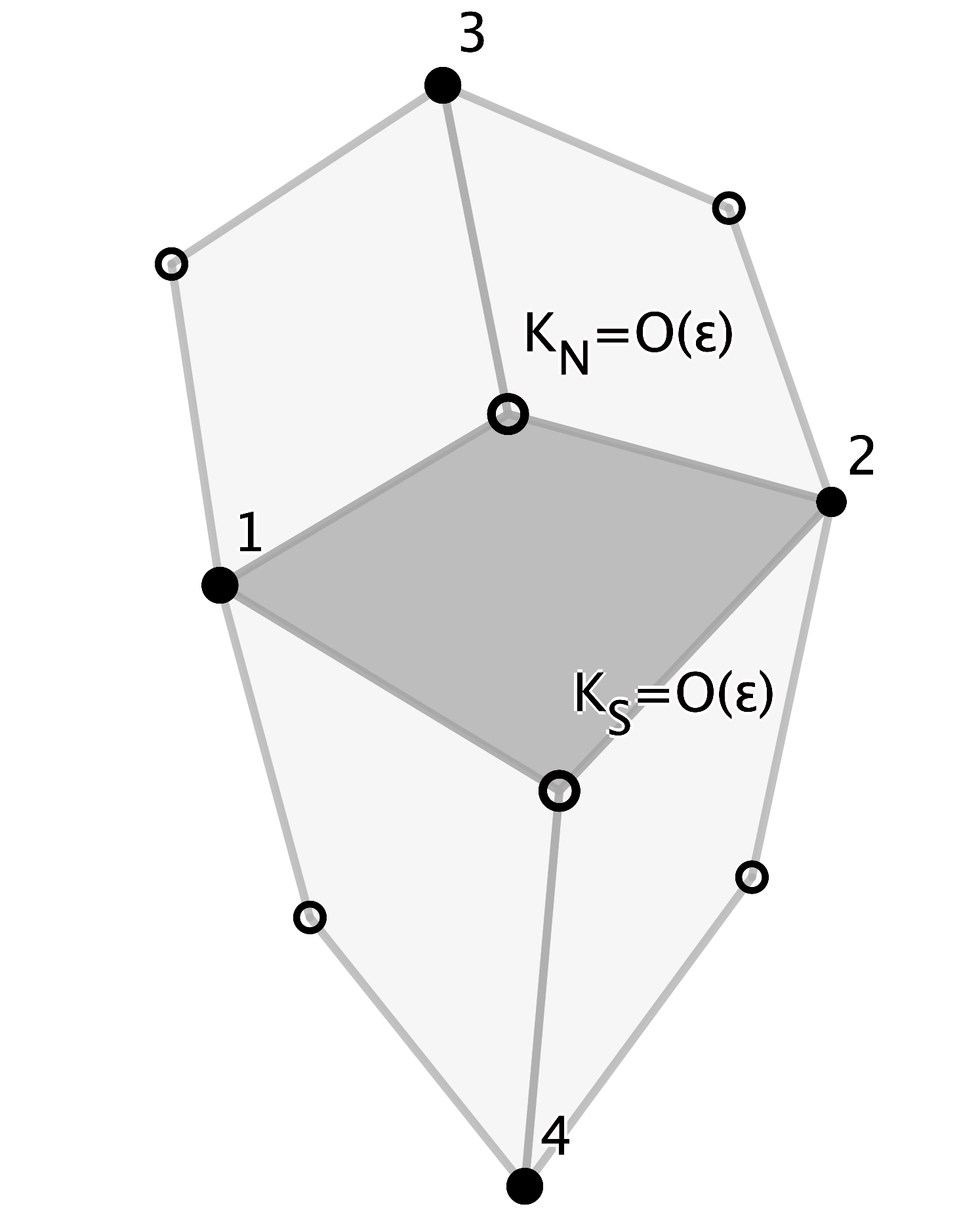}
\caption{ A $\mathrm{O}(\epsilon)$ deformation of $\uG$. The Gauss curvatures $K$ of the $N$ and $S$ faces are non-zero, but of order $\mathrm{O}(\epsilon)$.}
\label{Kite2Lozenge}
\end{center}
\end{figure}
It is easy to see that the surface $\uS_{\uG}^\lozenge$ is piecewise flat, with curvature defects (i.e. conical singularities) localized at the vertices $\tilde{\uo}_\uf$ 
associated to circumcenters of faces $\uf$ in $\uG$. 
The defect angle $K(\uf)$ corresponds to a localized curvature defect at $\tilde{\uo}_\uf$ and its value is given in terms of the 
conformal angles $\theta(\ue)$ of the edges forming the boundary of the face $\uf$.
Recall from \ref{scalar-curvature} that the associated scalar curvature $R_{\scriptscriptstyle{\mathrm{scal}}}(\tilde\uo_\uf)$ at 
a vertex $\tilde{\uo}_\uf$ is 

\begin{equation}
R_{\scriptscriptstyle{\mathrm{scal}}}(\tilde\uo_\uf)=4\pi - 2 \sum\limits_{ \ue \in\partial \uf} \, \big(\pi - 2 \theta(\ue) \big)
\end{equation}
or equivalently 
twice the measure of the defect angle around the circumcenter $\uo_\uf$ of the face $\uf$, i.e.
the Gauss curvature 
\begin{equation}
\label{}
\underbrace{K(\uf) \ :=\ 2\pi - \sum_{\ue \in \uf} \, \big( \pi - 2 \theta(e) \big)}_{\text{discrete Gauss curvature} }
\end{equation}
For an isoradial Delaunay graph $\uG$ the rhombic surface $\uS_{\uG}^\lozenge$ coincides with 
the planar kite graph $\uG^\lozenge$ whose faces, in this case, are all rhombs.
Furthermore, the scalar curvature $R_{\scriptscriptstyle{\mathrm{scal}}}(\tilde\uo_\uf)$ associated to each face 
$\uf$ in $\uG$ is zero.
For a generic Delaunay graph $\uG$ the scalar curvature $R_{\scriptscriptstyle{\mathrm{scal}}}(\tilde\uo_\uf)$ 
will be non-zero (see Fig.~\ref{Kite2Lozenge}).
Indeed, consider a cyclic quadrilateral face $\uf$
in an isoradial triangulation $\uG_\mathrm{cr}$ 
depicted in Fig.~\ref{4PointsLozenge}
and the effects of a generic deformation $\uG_\mathrm{cr} \rightarrow \uG_\epsilon$
depicted in Fig.~\ref{4PointsKite}. In $\uS_{\uG_\mathrm{cr}}^\lozenge$ four lozenges meet at $\tilde{\uo}_\uf$
where the scalar curvature $R_{\scriptscriptstyle{\mathrm{scal}}}(\tilde\uo_\uf)$ vanishes.

\begin{figure}[h!]
\begin{center}
\includegraphics[scale=.3]{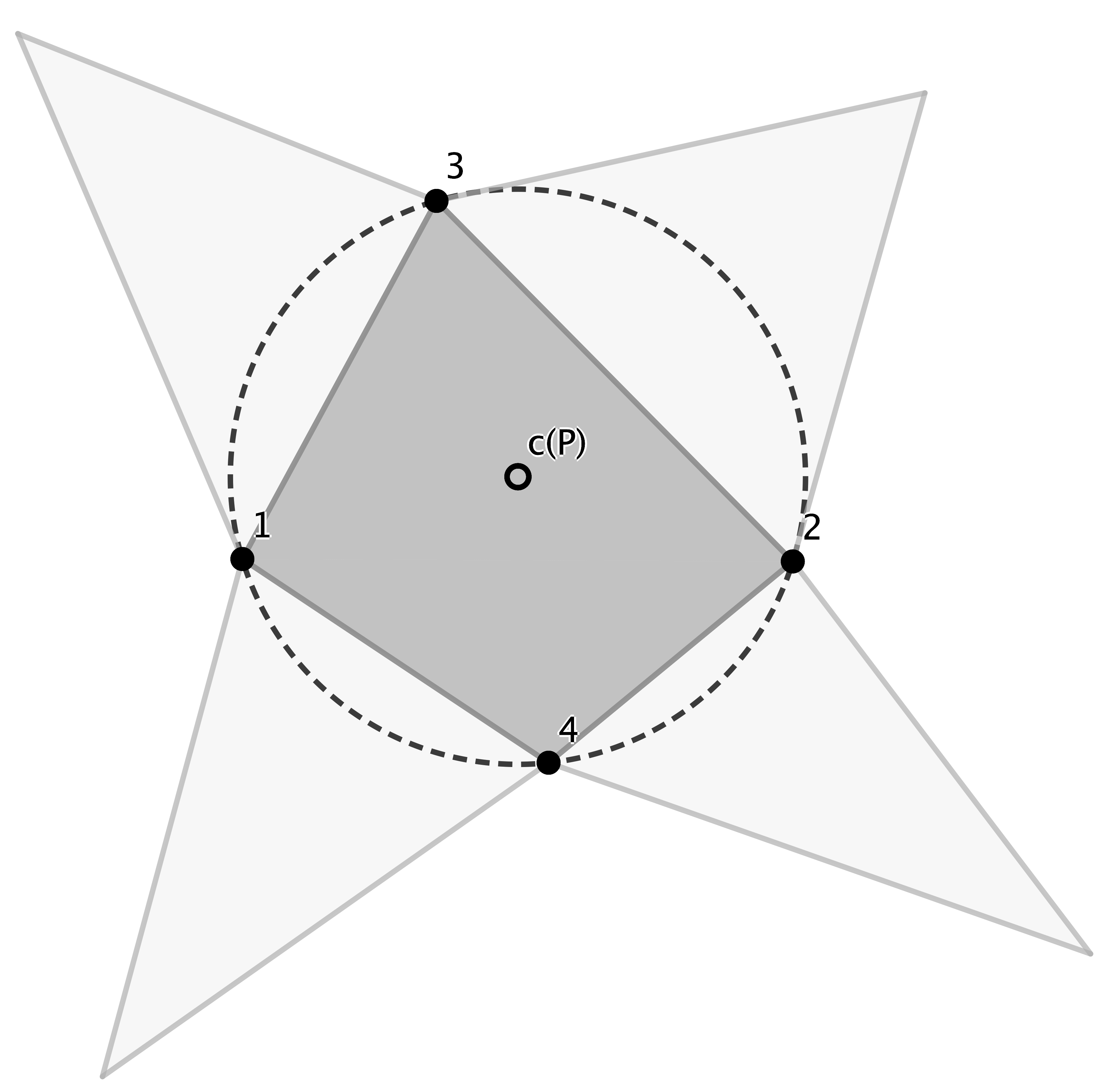}\qquad \raisebox{7.ex}{\includegraphics[scale=.3]{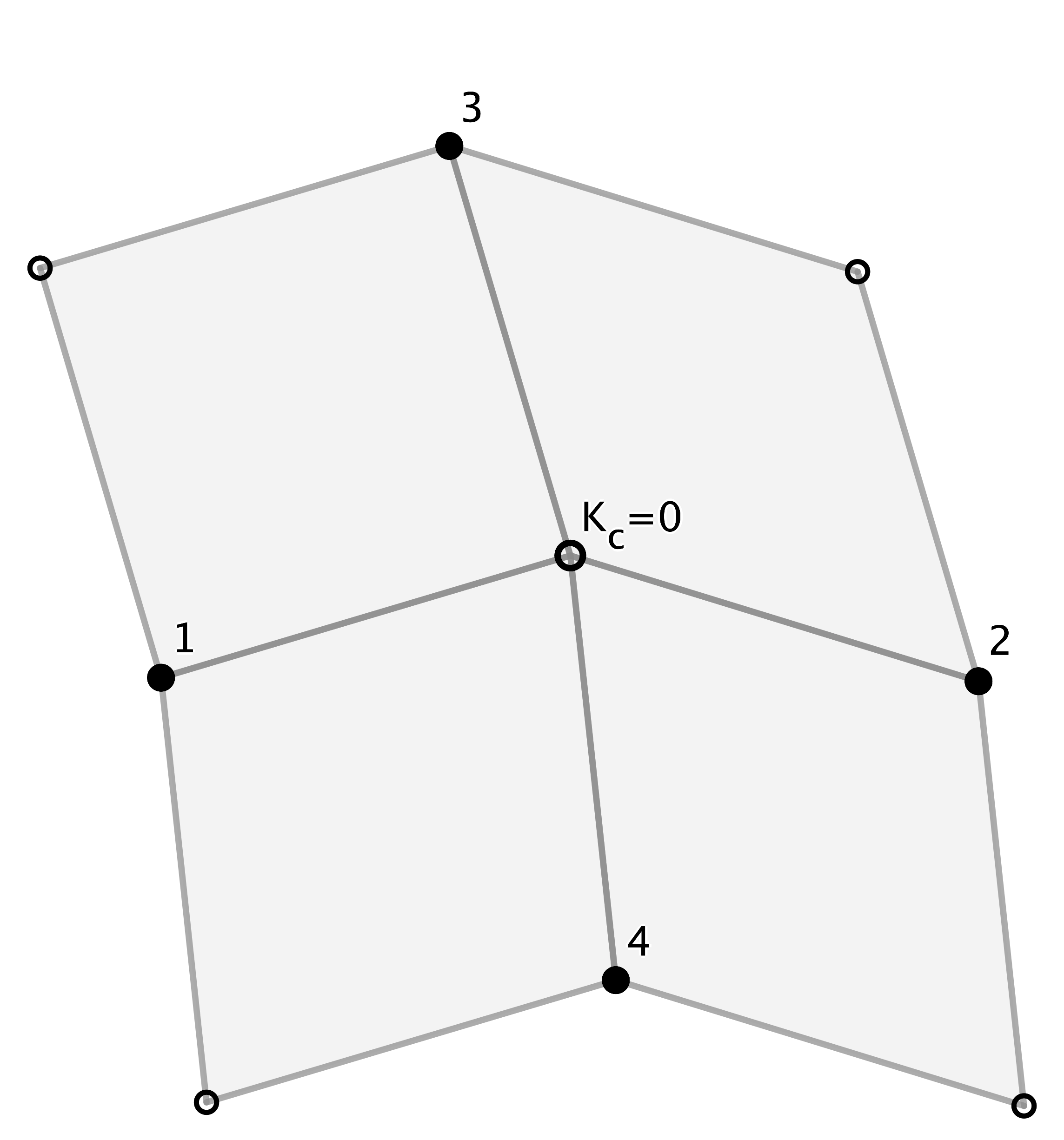}}
\caption{A cocylic face  $P=(1423)$ of a critical triangulation $\uG_\mathrm{cr}$ (left) and its associated rhombic lattice ${\uG}^\lozenge$ (right), the curvature $K$ associated to each face of $\uG$, i.e. its white $\uo$\,-vertices of ${\uG}^\lozenge$, is zero.}
\label{4PointsLozenge}
\end{center}
\end{figure}

As illustrated in Fig.~\ref{4PointsKite}, as soon as we deform this cyclic quadrilateral, a diagonal edge $\ue$ generically emerges in $\uG_\epsilon$ 
(infinitesimally a chord $\ue$ in $\uG_{0^+}$)
which subdivides the quadrilateral $\uf$ into two triangles 
$\uf_\mathrm{n}$ and $\uf_\mathrm{s}$, while the circumcenter $\uo_\uf$ splits into two circumcenters $\uo_\mathrm{n}$ and $\uo_\mathrm{s}$. In the deformed rhombic surface $\uS_{\uG_\epsilon}^\lozenge$ a new lozenge appears between $\tilde{\uo}_\mathrm{n}$ and $\tilde{\uo}_\mathrm{s}$.
However this new lozenge is ``flat'' i.e. to first order in $\epsilon$ its angles are $(0,\pi,0,\pi)$.
Therefore the  Gaussian curvatures $K(\uf_\mathrm{n})$ and 
$K(\uf_\mathrm{s})$ have opposite sign and they are both of order $\mathrm{O}(1)$, not of order $\mathrm{O}(\epsilon)$.
In terms of the north and south angles of the chord $\vec\ue$ they read
\begin{equation}
\label{ }
K(\uf_\mathrm{n})=- 2 \theta_\mathrm{n}(\vec{\ue} \,) \ + \ \mathrm{O}(\epsilon)
\quad,\qquad
K(\uf_\mathrm{s}) 
= - 2 \theta_\mathrm{s}(\vec{\ue} \, ) \ + \ \mathrm{O}(\epsilon) 
\end{equation}

{Thus the deformation produces a \emph{curvature dipole} associated to the chord $\ue$, 
i.e. neighboring curvature defects with non-zero but opposite signs.}
{Said differently, }
the smooth deformation  $\uG_\mathrm{cr} \rightarrow \uG_\epsilon$ manifests a
\emph{discontinuity} in the curvature. 
Generically when one smoothly deforms a cyclic face $\uf$ of $\uG_\mathrm{cr}$ with four or more vertices, a curvature dipole
will emerge for each chord $\ue \in \uG_{0^+}$ which subdivides the face $\uf$.
\begin{figure}[h!]
\begin{center}
\includegraphics[scale=.3]{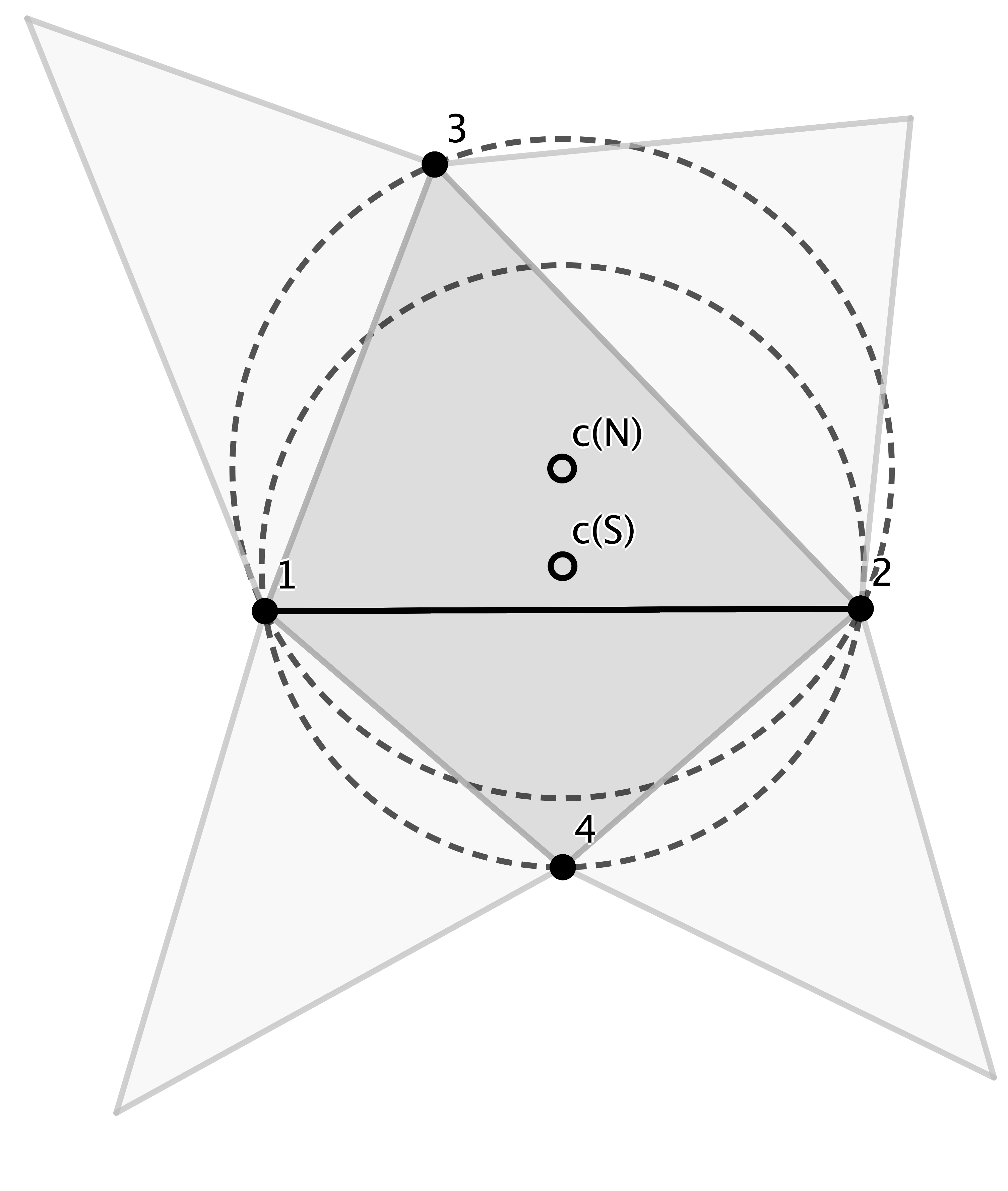}\qquad \raisebox{7.ex}{\includegraphics[scale=.3]{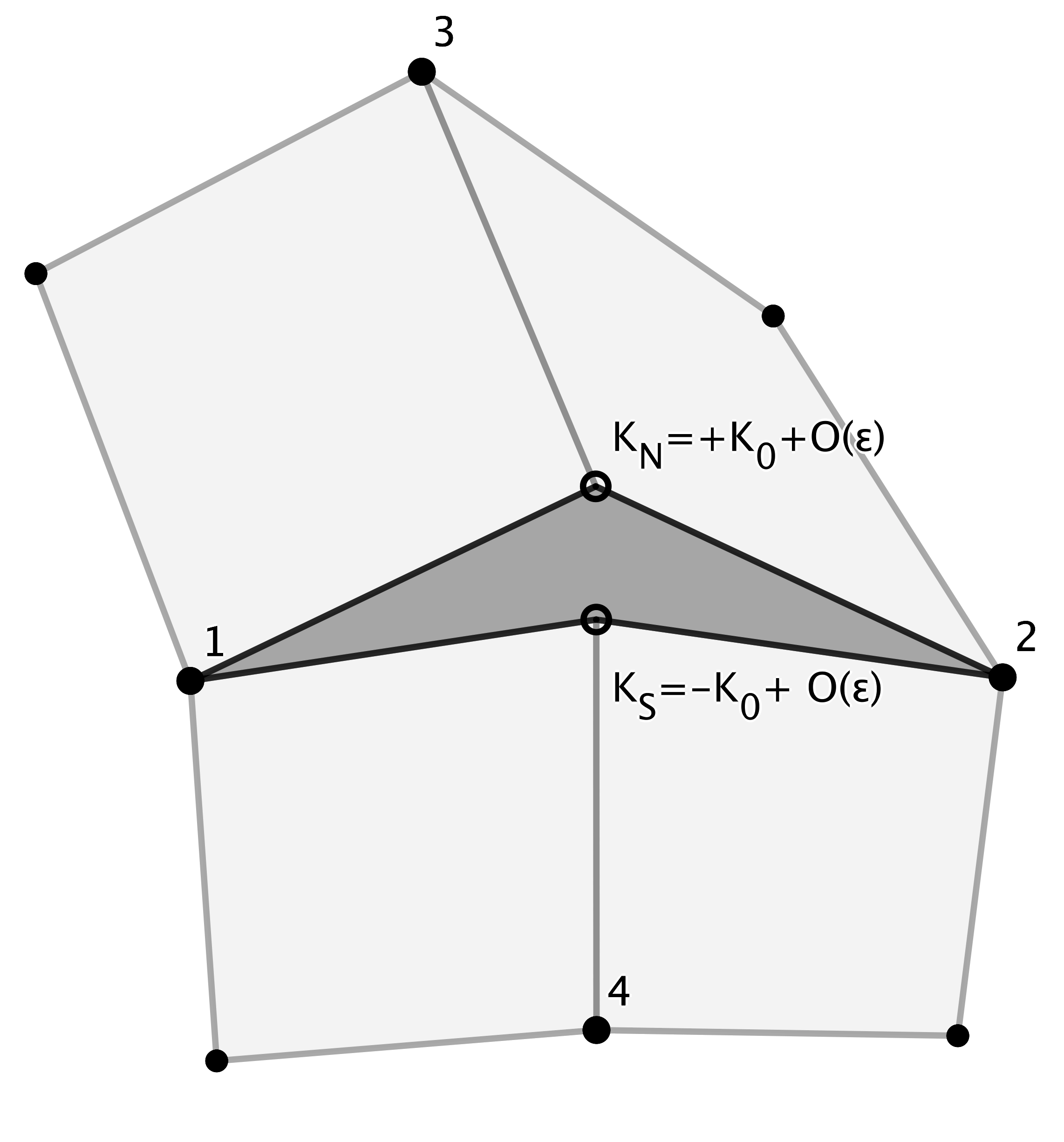}}
\caption{ A $\mathrm{O}(\epsilon)$ deformation of the cocyclic configuration. The Delaunay condition select a chord $\ue=(12)$, which splits the face $P=(1423)$ into two triangles $N=(123)$ and $S=(214)$. A flat lozenge $(1S2N$) appear in the rhombic lattice ${\uG}^\lozenge$. The curvatures $K$ of the $N$ and $S$ faces are non-zero, but of order $\mathrm{O}(1)$ and opposite.
$N$ and $S$ form a ``curvature dipole''.}
\label{4PointsKite}
\end{center}
\end{figure}

Finally, let us stress that a curvature dipole appears if the anomalous term $\deltae \mathbb{A} (\vec{\ue} \, )$ discussed above in \ref{ssConfLApDEl} is non-zero. 
Indeed this anomalous term is proportional to $\tan^2\theta_\mathrm{n}( \vec{\ue} \,)$, while the dipole is proportional to $\theta_\mathrm{n}(\vec{\ue} \, )$. 
Thus for a chord $\ue \in \uG_{0^+}$
with $\theta_\mathrm{n}(\vec{\ue} \, )=\theta_\mathrm{s}(\vec{\ue} \, )=0$, no anomalous term is present and so no curvature dipole appears at first order in the deformation. 
This occurs iff the circumcenter $\uo_\uf$ of the face $\uf$ lies on the edge $\ue$.
Notice that if $\uf$ is a quadrilateral (as in Fig.~\ref{4PointsKite}) where the north and south angles of both $\ue=(12)$ and the flipped edge 
$\ue^*=(34)$ are zero then the face $\uf$ is a rectangle. 
In this case, to first order in $\epsilon$, the deformation is $\mathtt{isoradial}\to\mathtt{isoradial}$, not $\mathtt{isoradial}\to\mathtt{non-isoradial}$.
These $\mathtt{isoradial}\to\mathtt{isoradial}$ deformations are the ones considered by Kenyon in the seminal paper \cite{Kenyon2002}.

\section{The scaling limit of variations}
\label{sScalingLimit}

\subsection{Rescaling smooth deformations}
\label{ssRescaling}

As explained in the introduction, we incorporate a scaling factor
$\ell > 0$ into the deformation in order to define and study a continuum limit.
We may view the scaling parameter $\ell>0$ as imparting a {\it resolution} on the
critical graph, i.e. we get a rescaled 
embedding $z^{\scriptscriptstyle{1\!/\!\ell}}_\mathrm{cr} := z_\mathrm{cr}/\ell$ of $\uG_\mathrm{cr}$, 
under which vertices become closer and denser in any compact region of the plane
as $\ell >0$ increases.
In particular the area $A(\uf)$ of a face $\uf\in \mathrm{F}(\uG_\mathrm{cr})$
shrinks by a factor of $1/\ell^2$ under the rescaled embedding
while its circumcenter coordinate $z_\mathrm{cr}(\uf)$ is rescaled by a factor of $1/\ell$.
In this way, the scaling parameter $\ell >0$ allows us to interpret
the critical graph as a planar partition and can be used to define a Riemann sum.
More specifically, given any continuous
complex-valued function $H: \Bbb{C} \longrightarrow \Bbb{C}$ with compact support 
$\Omega = \supp H$ then

\begin{equation}
\lim_{\ell \rightarrow \infty}
\, \sum_{\ux \in \mathrm{F}(\uG_\mathrm{cr})}
A(\ux)/\ell^2 \cdot H\big(z_\mathrm{cr}(\ux)/\ell \big)
\ = \
\int_{\Omega} \, 
d^2x \, H(x)
\end{equation}
Given a smooth complex-valued function $F: \Bbb{C} \longrightarrow \Bbb{C}$ with compact support, and $\ell >0$ a scaling real parameter, we
set $F_\ell(z) := \ell F(z/\ell)$. When deforming a critical isoradial Delaunay graph $\uG_\mathrm{cr}$ (with unit 
{circumradius}
$R_\mathrm{cr}=1$), we
shall consider the restriction of $F_\ell$ to (the coordinates of) the vertices of the critical graph.
By abuse of notation, we shall write $F_\ell(\mathrm{\uv}) := \, \ell \, F \big(z_\mathrm{cr}(\uv)/\ell \big)$
for each vertex $\uv \in \mathrm{V}(\uG_\mathrm{cr})$. We use $F_\ell$ to displace the
coordinates of the critical graph and define a deformed embedding, namely
$$z_{\epsilon,\ell}(\uv) := \, z_\mathrm{cr}(\uv) \, + \,  \epsilon \, F_\ell \big( \uv \big)  $$

\subsection{Rescaling bi-local deformations}
\label{ssRsBilDef}
Our analysis of second order variations (for the log-determinants which we consider) involve 
a bi-local deformation implemented by two smooth, complex-valued
functions $F_1$ and $F_2: \Bbb{C} \longrightarrow \Bbb{C}$ whose respective supports $\Omega_1$
and $\Omega_2$ 
are {\it compact} and have lattice closures $\overline{\Omega}_1$ and $\overline{\Omega}_2$ which are {\it disjoint}. 
Set 
$$d := \mathtt{dist}(\Omega_1,\Omega_2)=\text{inf} \big\{ |w_1 - w_2 | \, \big| \, w_i \in \Omega_i \big\} $$
to be the distance between the supports $\Omega_1$ and $\Omega_2$.
Obviously $0<d<\infty$.
The corresponding deformed embedding 
$z_{\underline{\epsilon}, \ell}: \mathrm{V}(\uG_\mathrm{cr}) \longrightarrow \Bbb{C}$ of the critical lattice is given by
\[ z_{\underline{\epsilon},\ell} ( \uv) := \, z_\mathrm{cr} (  \uv ) \,+ \, \epsilon_1 F_{1;\ell} (  \uv) 
\, +  \,  \epsilon_2 F_{2;\ell} ( \uv) \]
where $\underline{\epsilon} = (\epsilon_1, \epsilon_2)$ is a pair of deformation parameters $\epsilon_1, \epsilon_2 \geq 0$
and where we use the notation $F_{i ; \ell}(z) := \ell F_i (z/\ell)$ and by abuse of notation
$F_{i;\ell}( \uv) := F_{i;\ell} \big( z_\mathrm{cr}( \uv) \big)$ for a vertex $ \uv \in \mathrm{V}(\uG_\mathrm{cr})$
and $i=1,2$.
The results of Lemma \ref{epsilontildeF}
still hold for the bi-local deformed embedding $z_{\underline{\epsilon}, \ell}$; simply apply the Lemma to $F_1$ and $F_2$ independently
and take $\tilde{\epsilon}_F = \min (\tilde{\epsilon}_{F_1}, \tilde{\epsilon}_{F_2})$. Let us denote by $\uG_{\underline{\epsilon},\ell}$ the Delaunay
graph uniquely determined by the vertex set $\mathrm{V}(\uG_{\underline{\epsilon},\ell}) := \mathrm{V}(\uG_\mathrm{cr})$ 
together with the deformed embedding $z_{\underline{\epsilon},\ell}$. As we have seen, the one-sided limit $\epsilon_i \rightarrow 0^+$
for $i=1,2$ induces the structure of a weak Delaunay graph $\uG_{0^+\!,\ell}$ on the vertex set $\mathrm{V}(\uG_\mathrm{cr})$
with respect to the critical embedding $z_\mathrm{cr}$.
In general, the edge set $\mathrm{E}(\uG_{0^+\!,\ell})$ will vary as the scaling parameter $\ell >0$ evolves; nevertheless
$\mathrm{E}(\uG_\mathrm{cr}) \subseteq \mathrm{E}(\uG_{0^+\!,\ell})$ for all $0 < \ell \leq \infty$. 
For each value of $\ell > 0$ select
a weak Delaunay triangulation $\widehat{\uG}_{0^+\!,\ell}$ which completes $\uG_{0^+\!,\ell}$.
Because $\mathrm{E}(\uG_\mathrm{cr}) \subseteq \mathrm{E}(\uG_{0^+\!,\ell}) \subseteq \mathrm{E}(\widehat{\uG}_{0^+\!,\ell})$
for each $0 < \ell \leq \infty$ we may
always perform the following resummation

\begin{equation}
\sum_{\ux \in \mathrm{F} (\widehat{\uG}_{0^+\! , \ell} )} A(\ux) \, H (z_\mathrm{cr}(\ux))
\ = \
\sum_{\uy \in \mathrm{F} (\uG_\mathrm{cr})} A(\uy) \, H \big(z_\mathrm{cr}(\uy) \big)
\end{equation}
where we combine terms on the left hand side involving
triangles of $\widehat{\uG}_{0^+\!,\ell}$ which share a common circumcenter
and where $H \big(\ux\big)$ is any quantity which depends only upon
the circumcenter $z_\mathrm{cr}(\ux)$ of $\ux \in \mathrm{F} (\widehat{\uG}_{0^+\!,\ell} )$.
Consequently the choice of triangulation $\widehat{\uG}_{0^+\!,\ell}$ completing $\uG_{0^+\!,\ell}$ will not affect our calculations.

\subsection{Scaling limit and derivation of Theorem~\ref{TfinalOPElike1}}
\label{sScalLimit}
We now are in a position to study the scaling limit of the bilocal deformation terms \ref{OPEDelta} (Prop. \ref{ThDelta}) and \ref{OPEKaehler} (Prop.~\ref{The3})) and to derive Theorem~\ref{TfinalOPElike1}.
For $\uG = \uG_\mathrm{cr}$ or $\uG = \widehat{\uG}_{0^+\!,\ell}$ let 
$\mathrm{F}_{\overline{\Omega}_i(\ell)}(\uG)$ denote the subset of faces $\ux$ of $\uG$ whose
vertices belong to the lattice closure $\overline{\Omega}_i(\ell)$ 
of the support $\Omega_i (\ell):= \supp F_{i;\ell}$ for $i=1,2$.

\subsubsection{The initial $\ell$ finite term}
Let $\mathcal{O}(\underline{\epsilon} \, , \ell)$ denote either the Laplace-Beltrami operator $\Delta(\underline{\epsilon} \, , \ell)$ or the K\"ahler operator 
$\mathcal{D}(\underline{\epsilon} \, , \ell)$ 
on the Delaunay graph $\uG_{\underline{\epsilon},\ell}$.
From prop.~\ref{ThDelta} and \ref{The3} the $\epsilon_1 \epsilon_2$ cross-term of $\log \det \mathcal{O}(\underline{\epsilon} \, , \ell)$
is given by the trace term
\begin{equation}
\label{deps1eps2detlog}
 \ =\  -\,\mathrm{tr}\big[ \deltaeo \mathcal{O}(\ell) \cdot \Delta_\mathrm{cr}^{-1} \cdot \deltaet \mathcal{O}(\ell) \cdot \Delta_\mathrm{cr}^{-1}  \big]
\end{equation}

\newcommand{\thebigsumOO}{\sum_{\stackrel{\scriptstyle \ux_1 \in \, \mathrm{F}_{\overline{\Omega}_1(\ell)} \, (\widehat{\uG}_{0^+\!,\ell})}
{\ux_2 \in \, \mathrm{F}_{\overline{\Omega}_2(\ell)} \, (\widehat{\uG}_{0^+\!,\ell})}}} 

\newcommand{\thebigsumOOO}{\sum_{\stackrel{\scriptstyle \ux_1 \in \, \mathrm{F}_{\overline{\Omega}_1(\ell)} \, (\uG_\mathrm{cr})}
{\ux_2 \in \, \mathrm{F}_{\overline{\Omega}_2(\ell)} \, (\uG_\mathrm{cr})}}} 

\noindent
which can be expressed as the following double sum over triangles in $\widehat{\uG}_{0^+\!,\ell}$
\begin{equation}
\label{OPEnabla}
-{2\over {\pi^2}} \!\!\!
\thebigsumOO
\hskip -2.em A(\ux_1)A(\ux_2) \Bigg( 
\frak{Re} \Bigg[ 
{\overline{\nabla} F_{1;\ell} (\ux_1) \, \overline{\nabla} F_{2;\ell} (\ux_2) 
\over { \big( z_\mathrm{cr}(\ux_1) - z_\mathrm{cr}(\ux_2) \big)^4}   }
\Bigg]
+ \mathrm{O} \Big( \, \big| z_\mathrm{cr}(\ux_1) - z_\mathrm{cr}(\ux_2) \big|^{-5} \Big)
\Bigg)
\end{equation}
where $z_\mathrm{cr}(\ux_i)$ is the circumcenter of $\ux_i$ for $i = 1,2$.  
Both $F_1$ and $F_2$ have
compact support so by Lemma~\ref{lemmabound}
we have that $\overline{\nabla}F_{i;\ell} (\ux) = \overline{\partial} F_i \big( z_\mathrm{cr}(\ux)/\ell \big) \, + \, R_\mathrm{cr}/\ell \cdot E_i(\ux )$
where $\big| E_i(\ux) \big|$ is bounded by a constant $B_i>0$ independent of both $\ux$ and $\ell>0$.
We begin by breaking \ref{OPEnabla} into two pieces and evaluate 
their large $\ell$ limits separately.

\subsubsection{The subleading term} The large $\ell$ limit of the second part of \ref{OPEnabla}, vanishes as the following computation shows:
\begin{equation}
\begin{split}
&\D  \ \ \,
\Big|
\thebigsumOO
\, A(\ux_1)  A(\ux_2)  
\cdot \mathrm{O} 
\Big( \,  \big| z_\mathrm{cr}(\ux_1) - z_\mathrm{cr}(\ux_2) \big|^{-5} \Big)
\Big| 
 \\
&\D\leq
\thebigsumOO
\, A(\ux_1)  A(\ux_2) \cdot
\Big| \mathrm{O} 
\Big( \,  \big| z_\mathrm{cr}(\ux_1) - z_\mathrm{cr}(\ux_2) \big|^{-5} \Big)
\Big| 
\\
&\D\leq 
\thebigsumOOO
\, A(\ux_1)  A(\ux_2) \cdot
\Big| \mathrm{O} 
\Big( \,  \big| z_\mathrm{cr}(\ux_1) - z_\mathrm{cr}(\ux_2) \big|^{-5} \Big)
\Big| 
 \\
&\D \leq
{1 \over d} \, {1\over \ell} \  
\thebigsumOOO
\, A( \ux_1)/\ell^2 \,  A( \ux_2)/\ell^2  \cdot 
\Big| \mathrm{O} 
\Big( \,  \big| z_\mathrm{cr}(\ux_1)/\ell - z_\mathrm{cr}(\ux_2)/\ell \big|^{-4} \Big) 
\Big| \\
\end{split}
\end{equation}
\noindent
In the large $\ell$ limit the sum over the triangles becomes a standard Riemann integral
\begin{equation}
\begin{split}
&\D \leq 
\lim_{\ell \rightarrow \infty}
\thebigsumOOO
\, A( \ux_1)/\ell^2 \,  A( \ux_2)/\ell^2 \cdot
\Big| \mathrm{O} 
\Big( \,  \big| z_\mathrm{cr}(\ux_1)/\ell - z_\mathrm{cr}(\ux_2)/\ell \big|^{-4} \Big) \Big|
 \\
&\D =
\ \iint_{\Omega_1  \times \Omega_2}\hskip -1 em d^2x_1 \, d^2x_2 \cdot  \Big| \text{O} \Big( \, \big| x_1 - x_2 \big|^{-4} \Big)  \Big| \  = \ \text{O}(1) 
\end{split}
\end{equation}
Hence
\begin{equation}
\label{ }
\lim_{\ell \rightarrow \infty}
\thebigsumOOO
\, A(\ux_1)  A(\ux_2)  \cdot \mathrm{O} \Big( \,  \big| z_\mathrm{cr}(\ux_1) - z_\mathrm{cr}(\ux_2) \big|^{-5} \Big)\ =\  0
\end{equation}

\subsubsection{The leading term} To evaluate the first part in \ref{OPEnabla}
we consider the norm of the difference between the original term with discrete derivative and the corresponding term with continuous derivatives, and use the previous results to get the bounds

\begin{align}
\label{partI}
&
\left|
\thebigsumOO
\hskip -2.em
A(\ux_1)  A(\ux_2)  
\, \frak{Re} \Bigg[{\overline{\nabla} F_{1;\ell}( \ux_1 ) \, 
\overline{\nabla} F_{2;\ell} (\ux_2 ) - 
\overline{\partial}F_1\big(z_\mathrm{cr}(\ux_1)/\ell \big) \, \overline{\partial}F_2\big(z_\mathrm{cr}(\ux_2)/\ell \big)
\over \big(z_\mathrm{cr}(\ux_1) - z_\mathrm{cr}(\ux_2) \big)^4} \Bigg]  \, 
\right|\nonumber\\
&\ \nonumber\\
&
\leq \left\{
\begin{array}{l}
\D 
\ \ \  {R_\mathrm{cr} \over \ell} \,  
\thebigsumOO
A(\ux_1)/\ell^2 \,  A(\ux_2)/\ell^2  
\,  { \big|E_1( \ux_1) \big| \cdot
\big| \overline{\partial} F_2 \big(z_\mathrm{cr}(\ux_2)/\ell \big) \big| 
\over \big| z_\mathrm{cr}(\ux_1)/\ell - z_\mathrm{cr}(\ux_2)/\ell \big|^4}   \,  \\ \\
\D+ 
\, {R_\mathrm{cr} \over \ell} \, 
\thebigsumOO
A(\ux_1)/\ell^2 \,  A(\ux_2)/\ell^2  
\, { \big| \overline{\partial} F_1\big( z_\mathrm{cr}(\ux_1) /\ell \big) \big| \cdot
\big| E_2 (\ux_2) \big| \over \big| z_\mathrm{cr}(\ux_1)/\ell - z_\mathrm{cr}(\ux_2)/\ell \big|^4} \\ \\
\D+
\, {R^2_\mathrm{cr} \over {\ell^2}} \, 
\thebigsumOO
A(\ux_1)/\ell^2 \,  A(\ux_2)/\ell^2  
\, { \big| E_1( \ux_1) \big| \cdot
\big| E_2 (\ux_2) \big| \over \big| z_\mathrm{cr}(\ux_1)/\ell - z_\mathrm{cr}(\ux_2)/ \ell \big|^4}  
\end{array}
\right. \nonumber\\
& \nonumber\\
&
\leq \left\{
\begin{array}{l}
\ \ \  \D {R_\mathrm{cr} \over \ell} \,  
\thebigsumOO
A(\ux_1)/\ell^2 \,  A(\ux_2)/\ell^2  
\,  {B_1  \cdot
\big| \overline{\partial} F_2 \big(z_\mathrm{cr}(\ux_2)/\ell \big) \big| 
\over \big| z_\mathrm{cr}(\ux_1)/\ell - z_\mathrm{cr}(\ux_2)/\ell \big|^4}   \,  \\ \\
\D+ 
\, {R_\mathrm{cr} \over \ell} \, 
\thebigsumOO
A(\ux_1)/\ell^2 \,  A(\ux_2)/\ell^2  
\, { \big| \overline{\partial} F_1\big( z_\mathrm{cr}(\ux_1) /\ell \big) \big| \cdot
B_2 \over \big| z_\mathrm{cr}(\ux_1)/\ell - z_\mathrm{cr}(\ux_2)/\ell \big|^4} \\ \\
\D +
\, {R^2_\mathrm{cr} \over {\ell^2}} \, 
\thebigsumOO
A(\ux_1)/\ell^2 \,  A(\ux_2)/\ell^2  
\, {B_1 \cdot B_2 \over \big| z_\mathrm{cr}(\ux_1)/\ell - z_\mathrm{cr}(\ux_2)/ \ell \big|^4}  
\end{array}
\right. \nonumber\\
& \nonumber\\
&\leq \left\{
\begin{array}{l}
\D 
\ \ \  {R_\mathrm{cr} \over \ell} \,  
\thebigsumOOO
A(\ux_1)/\ell^2 \,  A(\ux_2)/\ell^2  
\,  {B_1  \cdot
\big| \overline{\partial} F_2 \big(z_\mathrm{cr}(\ux_2)/\ell \big) \big| 
\over \big| z_\mathrm{cr}(\ux_1)/\ell - z_\mathrm{cr}(\ux_2)/\ell \big|^4}   \,  \\ \\
\D+ 
\, {R_\mathrm{cr} \over \ell} \, 
\thebigsumOOO
A(\ux_1)/\ell^2 \,  A(\ux_2)/\ell^2  
\, { \big| \overline{\partial} F_1\big( z_\mathrm{cr}(\ux_1) /\ell \big) \big| \cdot
B_2 \over \big| z_\mathrm{cr}(\ux_1)/\ell - z_\mathrm{cr}(\ux_2)/\ell \big|^4} \\ \\
\D+
\, {R^2_\mathrm{cr} \over {\ell^2}} \, 
\thebigsumOOO
A(\ux_1)/\ell^2 \,  A(\ux_2)/\ell^2  
\, {B_1 \cdot B_2 \over \big| z_\mathrm{cr}(\ux_1)/\ell - z_\mathrm{cr}(\ux_2)/ \ell \big|^4}  
\end{array}
\right. 
\end{align}

In the large $\ell$ limit each sum over triangles becomes a Riemann integral, hence the large $\ell$ limit of the l.h.s. of \ref{partI}  is bounded by
\begin{equation}
\begin{array}{l}
\leq 
\left\{
\begin{array}{l}
\D \ \ \ \, \lim_{\ell \rightarrow \infty}
{c\over {2 \pi^2}}  {B_1R_\mathrm{cr} \over \ell} \cdot \iint_{\Omega_1 \times \Omega_2}
{dx_1 \, dx_2 \over \big| x_1 - x_2 \big|^4} 
\, \big| \overline{\partial} F_2 (x_2) \big| \\ \\
\D + 
\, \lim_{\ell \rightarrow \infty}
{c\over {2 \pi^2}} {B_2R_\mathrm{cr} \over \ell} \cdot \iint_{\Omega_1 \times \Omega_2}  
{dx_1 \, dx_2 \over \big| x_1 - x_2 \big|^4} 
\, \big| \overline{\partial} F_1(x_1) \big| \\ \\
\D + 
\, \lim_{\ell \rightarrow \infty}
{c\over {2 \pi^2}} \, {B_1B_2R^2_\mathrm{cr} \over {\ell^2}}  \cdot \iint_{\Omega_1 \times \Omega_2} 
{dx_1 \, dx_2 \over \big| x_1 - x_2 \big|^4}  
\end{array}
\right. \\ \\
=\  0
\end{array}
\end{equation}

\subsubsection{Summing up}
From this it follows that 
\begin{equation}
\begin{split}
\D 
&\ \lim_{\ell \rightarrow \infty}
\thebigsumOO
A(\ux_1)  A(\ux_2)  
\ \frak{Re} \Bigg[{\overline{\nabla} F_{1;\ell}( \ux_1 ) \, 
\overline{\nabla} F_{2;\ell} (\ux_2 ) 
\over \big( z_\mathrm{cr}(\ux_1)- z_\mathrm{cr}(\ux_2) \big)^4} \Bigg] \\
\D = &\ \lim_{\ell \rightarrow \infty}
\sum_{\stackrel{\scriptstyle \ux_1 \in \, \mathrm{F}(\widehat{\uG}_{0^+\!,\ell} )
}{\ux_2 \in \, \mathrm{F}(\widehat{\uG}_{0^+\!,\ell})}}
A(\ux_1)  A(\ux_2)  
\ \frak{Re} \Bigg[{\overline{\nabla} F_{1;\ell}( \ux_1 ) \, 
\overline{\nabla} F_{2;\ell} (\ux_2 ) 
\over \big( z_\mathrm{cr}(\ux_1)- z_\mathrm{cr}(\ux_2) \big)^4} \Bigg] \\
\D = & \ \lim_{\ell \rightarrow \infty}
 \sum_{\stackrel{\scriptstyle \ux_1 \in \, \mathrm{F}(\widehat{\uG}_{0^+\!,\ell} )
}{\ux_2 \in \, \mathrm{F}(\widehat{\uG}_{0^+\!,\ell})}}
A(\ux_1)/\ell^2 \,  A(\ux_2)/\ell^2  
\ \frak{Re} \Bigg[{\overline{\partial}F_1 \big( z_\mathrm{cr}(\ux_1)/\ell \big) \, 
\overline{\partial}F_2 \big( z_\mathrm{cr}(\ux_2)/\ell \big)
\over \big( z_\mathrm{cr}(\ux_1)/\ell - z_\mathrm{cr}(\ux_2)/\ell \big)^4} \Bigg]  \\
\D = & \ \lim_{\ell \rightarrow \infty}
 \sum_{\stackrel{\scriptstyle \ux_1 \in \, \mathrm{F}(\uG_\mathrm{cr})
}{\ux_2 \in \, \mathrm{F}(\uG_\mathrm{cr})}}
A(\ux_1)/\ell^2 \,  A(\ux_2)/\ell^2   
\ \frak{Re} \Bigg[{\overline{\partial}F_1 \big( z_\mathrm{cr}(\ux_1)/\ell \big) \, 
\overline{\partial}F_2 \big( z_\mathrm{cr}(\ux_2)/\ell \big)
\over \big( z_\mathrm{cr}(\ux_1)/\ell - z_\mathrm{cr}(\ux_2)/\ell \big)^4} \Bigg]  \\
\D
= &\ 
\iint_{\Omega_1 \times \Omega_2}
dx_1 \, dx_2
\  \frak{Re} \Bigg[{
\overline{\partial}F_1(x_1) \, \overline{\partial}F_2(x_2)
\over (x_1-x_2)^4} \Bigg]  \\
\end{split}
\end{equation}

\noindent
Thus we have
\begin{equation}
\label{scalim12LB}
\lim_{\ell\to\infty} \mathrm{tr}\big[ \deltaeo \mathcal{O}(\ell)  \cdot \Delta_\mathrm{cr}^{-1} \cdot \deltaet \mathcal{O}(\ell) \cdot \Delta_\mathrm{cr}^{-1}  \big]=
{2\over \pi^2}\iint_{\Omega_1 \times \Omega_2}
dx_1 \, dx_2
\  \frak{Re} \Bigg[{
\overline{\partial}F_1(x_1) \, \overline{\partial}F_2(x_2)
\over (x_1-x_2)^4} \Bigg] 
\end{equation}
This settles the proof of Theorem~\ref{TfinalOPElike1} by establishing eq.~\ref{finalOPElike1} .

\subsection{Controlling the geometry of the lattice for small deformations}
\label{ssBeyondCheck}
\subsubsection{The limits we considered}
Let us summarize what we did previously, up to sect.~\ref{sScalLimit}. 
We begin with an infinite critical graph $\uG_{\mathrm{cr}}$ and two displacement functions $F_1$ and $F_2$ 
whose respective supports $\Omega_1$ and $\Omega_2$ are compact 
and whose lattice closures $\overline{\Omega}_1$ and $\overline{\Omega}_2$ are disjoint.
We construct the stable Delaunay deformation $\uG_{\underline{\epsilon}}$
with embedding $z_{\underline{\epsilon}} = z_\mathrm{cr} + \epsilon_1 F_1+\epsilon_2 F_2$ 
along with a corresponding deformed operator $\mathcal{O}(\underline{\epsilon})$ where 
$\underline\epsilon=(\epsilon_1,\epsilon_2)$ is a pair of independent parameters. 
We then proceed to isolate the coefficient of $\epsilon_1 \epsilon_2$ in the Taylor series 
of $\log \det \mathcal{O}(\underline{\epsilon})$. Since the lattice closures of the
supports of $F_1$ and $F_2$ are disjoint the first trace term $\mathrm{tr} [\mathcal{O}(\underline{\epsilon}) \cdot 
\mathcal{O}_\mathrm{cr}^{-1} ]$ contributes nothing.
The only non-vanishing contribution to $\epsilon_1 \epsilon_2$ comes from the second trace and can be expressed as
\begin{equation}
\label{limepsto0}
- \tr \Big[ \deltaeo\mathcal{O} \cdot \mathcal{O}_\mathrm{cr}^{-1} 
\cdot \deltaet \mathcal{O}  \cdot \mathcal{O}_\mathrm{cr}^{-1}   \Big]
\end{equation} 
\noindent
defined on the weak Delaunay graph $\widehat{\uG}_{0^+}$ (a completion of the isoradial refinement of the initial graph $\uG_{\mathrm{cr}}$
relative to the deformation).
We then rescale the deformation by $\ell$ and consider the family of deformations 
$z_\mathrm{cr} \to z_\mathrm{cr} + \epsilon_1 F_{1;\ell}+\epsilon_2 F_{2;\ell}$
and show the scaling limit $\ell\to\infty$ of \ref{limepsto0} exists and is independent of the choice of initial critical graph $\uG_{\mathrm{cr}}$. 
Stated simply, we study the nested limit
\begin{equation}
\label{nested-limit}
\lim_{\ell \rightarrow \infty}
\,
\lim_{\stackrel{\scriptstyle \epsilon_1 \rightarrow 0}{\epsilon_2 \rightarrow 0}}
\,  
\left(
\tr \Bigg[ 
\deltaeo\mathcal{O}(\underline{\epsilon} \, , \ell)
\cdot \mathcal{O}_\mathrm{cr}^{-1} 
\cdot 
\deltaet\mathcal{O}(\underline{\epsilon} \, , \ell)
\cdot \mathcal{O}_\mathrm{cr}^{-1}   \Bigg]
\right)
\end{equation}

An interesting question is whether these two limits can be interchanged.
A positive answer would be a first step in understanding if one can define a continuum limit of (the total variation of) 
$\log\det\mathcal{O}(\underline{\epsilon} \, , \ell)$ starting from an infinite Delaunay graph which is not isoradial, but rather obtained by a small, smooth deformation of a Delaunay graph which is isoradial.
A simpler question is the following: We know that for a given critical graph $\uG_{\mathrm{cr}}$, the limit \ref{nested-limit} makes sense when $\epsilon_1,\epsilon_2 \to 0$. Is the convergence 
uniform w.r.t. all critical graphs $\uG_{\mathrm{cr}}$ ?
We return to this issue in Section~\ref{aVarOpFlips}.

\subsubsection{The problem with flips}
The geometrical effects of a finite $\epsilon$-deformation of a Delaunay graph $\uG$ have already been discussed in Sections \ref{sswwococycl} and \ref{ssEpsilonBounds}.
Lemma~\ref{epsilonFboundT} and Prop.~\ref{epsilontildeF} ensure that, for a given initial graph $\uG_\mathrm{cr}$ and a given displacement function $F$ (with compact support), there exist a strictly positive bound $0<\tilde \epsilon_F$ such that no flip occurs in the interval $0<\epsilon <\tilde \epsilon_F$. However 
$\tilde \epsilon_F$ depends non-trivially on $F$ and on the geometry of $\uG_\mathrm{cr}$. Furthermore it is 
clear that such a bound cannot be made uniform w.r.t. all critical graphs $\uG_\mathrm{cr}$.
This means that given any small value $\epsilon > 0$ of the deformation parameter,
flips will occur in $\uG_\epsilon$ for some critical graph $\uG_\mathrm{cr}$
within the class of all critical graphs.
Consequently the (matrix entries of the) operators $\deltae\mathcal{O}(\epsilon)$ are discontinuous functions of $\epsilon$,
and it will be difficult to control them as $\epsilon$ varies.

\subsection{A simple restriction to control small deformations: enforcing a global lower bound on the edge angles}
\label{brutal-approach}
A na\"\i ve but brutal way to manage the ``flip problem'' is to consider only  a subclass of graphs
$\uG_\mathrm{cr}$ such that the bound $\tilde\epsilon_F$ of Prop.~\ref{epsilontildeF}  can be controlled explicitly, so that no flip occurs. 
Similar constraints \ref{varThetGtVarCK} on the  geometry of $\uG_{\mathrm{cr}}$ have been used in the literature for other problems involving isoradial lattices, see e.g. the paper by U. B\"ucking \cite{Bucking}.
Our solution is given by the following Lemma.

\medskip

\begin{lemma}
\label{epsFbound2}
Let $F:\Bbb{C} \longrightarrow \Bbb{C}$ be a non-zero, smooth complex-valued function
with compact support $\Omega_F$. We define
\begin{equation}
\label{MchekF}
\check{M}_F  \,= \, \max_{z\in\mathbb{C}} |\partial F(z)| + \max_{z\in\mathbb{C}} |\overline{\partial}F(z)| 
\end{equation}
This is a simple modification of the bound $M_F$ of Lemma~\ref{MFbound} given by \ref{MFDef}
which is now independent of the triangulation.
For a generic Delaunay \emph{triangulation} $\uT$ 
we define, in analogy with $\vartheta_F$ given by \ref{varthetaDef} in Lemma~\ref{Lemma11},
\begin{equation}
\label{CkVarthetaF}
\check\vartheta(\uT) \,= \, \mathrm{inf} \, \Big\{ \theta( \ue )  \, \Big| \ue \in \mathrm{E}(\uT)  \Big\}
\end{equation}
For a fixed, strictly positive $\check\vartheta >0$, 
define the subset of Delaunay triangulations 
\begin{equation}
\label{varThetGtVarCK}
\text{\calligra{T}}_{\check\vartheta}=\left\{\text{Delaunay triangulation}\ \uT\ :\ \check\vartheta(\uT)\,\ge\, \check\vartheta \right\}
\end{equation}
and the strictly positive bound
\begin{equation}
\label{ }
\check\epsilon_F = \, \mathtt{b}\,{\sin( 2\,\check\vartheta}){{ \check{M}_F}^{-1} }
\end{equation}
with $\mathtt{b}=\sqrt{10}-3$ as in Lemma~\ref{epsilonFboundT}.

For any triangulation $\uT\in \text{\calligra{T}}_{\check\vartheta}$ and any scaling parameter $\ell>0$, the Delaunay deformation 
$z  \to z_{\epsilon;\ell}= z +\epsilon\, F_\ell(z) $ of $\uT$ preserves all the edges of $\uT$ if $0<\epsilon\le\check\epsilon_F$.
\begin{equation}
\label{ }
0<\epsilon\le\check\epsilon_F\ ,\quad\ell>0\quad\text{and}\quad\uT\in \text{\calligra{T}}_{\check\vartheta}\quad\implies\quad
\mathrm{E}(\uT_{\epsilon,\ell})=\mathrm{E}(\uT) 
\end{equation}
In other words, no flip occurs as long as $0<\epsilon\le\check\epsilon_F$.
\end{lemma}

\begin{proof} 
The mapping $z_{\epsilon,\ell}: \mathrm{V}(\uT_{\epsilon,\ell}) \longrightarrow \Bbb{C}$ is an embedding 
provided there are no ``collisions'', that is $z_{\epsilon,\ell}(\uu) \ne z_{\epsilon,\ell}(\uv)$ whenever
$\uu \ne \uv$ are distinct vertices in $\mathrm{V}(\uT_\mathrm{cr})$.
Equivalently $1 + \epsilon \, dF_{(\ell)}(\uu, \uv)$ must not vanish.
Apply the fundamental theorem of calculus using
$\gamma_{\uu\uv}(\tau) := \tau z_\mathrm{cr}(\uu)/\ell + (1-\tau)z_\mathrm{cr}(\uv)/\ell$.
\[
\begin{array}{ll}
\D \big| dF_\ell (\uu, \uv) \big|
&\D = \left| { {F \big(z_\mathrm{cr}(\uu) /\ell \big) - F \big(z_\mathrm{cr}(\uv)/\ell \big)}  \over {z_\mathrm{cr}(\uu)/\ell 
-z_\mathrm{cr}(\uv)/\ell}}  \right| \\ \\
&\D = {1 \over  {| z_\mathrm{cr}(\uu)/\ell -z_\mathrm{cr}(\uv)/\ell  \, | }} \cdot
\left| \,  \int_0^1 \, d\tau \, {d \over {d \tau}} F\big(\gamma_{\uu\uv}(\tau) \big) \,  \right| \\ \\
&\D = \left| \, \int_0^1 \, d\tau \, \partial F(\gamma_{\uu\uv}(\tau)) \ + \ 
{{\overline{z}_\mathrm{cr}(\uu) -\overline{z}_\mathrm{cr}(\uv)} \over {z_\mathrm{cr}(\uu) - z_\mathrm{cr}(\uv) }} 
\, \int_0^1 \, d\tau \, \overline{\partial} 
F\big(\gamma_{\uu\uv}(\tau) \big) \, \right| \\ \\
&\D \leq \left| \, \int_0^1 \, d\tau \, \partial F\big(\gamma_{\uu\uv}(\tau) \big) \, \right|  \ + \ \left| \, 
{{\overline{z}_\mathrm{cr}(\uu) -\overline{z}_\mathrm{cr}(\uv)} \over {z_\mathrm{cr}(\uu) - z_\mathrm{cr}(\uv) }}  \, \right| \cdot
\, \left| \,  \int_0^1 \, d\tau \, \overline{\partial} F\big(\gamma_{\uu\uv}(\tau) \big) \, \right| \\ \\
&\D \leq \, \mathrm{max} \, \big| \partial F \big| \ + \ \mathrm{max} \, \big| \overline{\partial} F \big| \ = \, \check M_F
\end{array}
\]
\noindent
By construction $\check\vartheta\le\vartheta_{F_\ell}$ and taken together with the fact that $M_{F_\ell}\le \check M_{F}$ 
we can conclude that $\check\epsilon_{F}\le \bar\epsilon_{F_\ell}$.
As long as $\epsilon< \check\epsilon_F$ we can apply Lemma~\ref{Lemma11} and conclude that the edge set $\mathrm{E}(\uT)\subset \mathrm{E}(\uT_{\epsilon, \ell})$.
Since $\uT$ is a triangulation no chords appear and hence $\mathrm{E}(\uT)= \mathrm{E}(\uT_{\epsilon, \ell})$.
We stress that this bound on $\epsilon$ is valid and independent on all values of the scaling parameter$\ell >0$, including $\ell = \infty$.

\end{proof}

\section{Finite $\epsilon$ variations, beyond the linear approximation}
\label{aVarOpFlips}

\subsection{Outline of the section}
\label{ssFullVarIntro}

In this section, we now consider deformations of an initial critical lattice $\uG_\mathrm{cr}$ implemented by a local diffeomorphism of the plane
$$ z\ \to z +\epsilon\, F(z,\bar z)
$$ 
for small values of a deformation parameter $\epsilon$, and a fixed smooth (but non-analytic) displacement function $F$  with compact support. 
We shall look for uniform bounds for the variation of the operators $\Delta$ and $\mathcal{D}$ with respect to $\epsilon$, independent of the particular geometry of the initial critical graph  $\uG_\mathrm{cr}$, except for its isoradius $R_\mathrm{cr}$. 

We therefore need to consider generic Delaunay deformations and take into account the occurrence of edge flips in the deformed Delaunay graph $\uG_\mathrm{cr}\to\uG_\epsilon$. 
These flips were avoided in the stable deformation scheme studied in sections \ref{sVarOp}, \ref{sCalculations} and \ref{sScalingLimit}
by imposing tight bounds on the parameter $\epsilon$.

For a fixed smooth displacement function $F$ and a deformation parameter $\epsilon$, it will be necessary to compare the corresponding Delaunay and rigid deformations, as explained in Sect.~\ref{sVarOp}.
We discuss this in sect.~\ref{ssFlipsNoFlips}, as well as the concept of ``backtracking a deformation without flips''.

In Sect.~\ref{ssFullVar} we give explicit variational formulas for the various operators $\nabla$, $\overline\nabla$, $\Delta$ and $\mathcal{D}$ as well as the circumradii $R$ in the case of a rigid deformation of the graph; see Def.~\ref{DefRigidDef}.

In Sect.~\ref{ssFullVarFlip} we derive integral representations of the variations of these objects taking flips into account.

In Sect.~\ref{BoundNablaF} and \ref{ssVArRGen} we give variational formulas for the discrete derivatives $\nabla$ and $\overline\nabla$ as well as the circumradius of a face. The later result,  given in Prop.~\ref{PropReps}, is important and 
leads to uniform bounds on the variations of $\nabla$, $\overline\nabla$, $\Delta$ and $\mathcal{D}$
with respect to $\epsilon$; see Prop.~\ref{propboundswab}.

In Sect.~\ref{ssScalLimD}
we deduce strong results on the uniform convergence of the $\ell\to\infty$ scaling limit for $\Delta$ (prop.~\ref{PropLimDpr} and \ref{propboundDelta})
and of the scaling limit of the corresponding second order bi-local trace term (which leads to the OPE) (prop.~\ref{PropLimTr2}).
 
In sect.~\ref{ssSLbilocDelta} we finally address the problem of interchanging the
$\underline{\epsilon} \to 0$ deformation limit and $\ell \to \infty$ scaling limit
when evaluating the bi-local trace term of $\log\det\Delta(\underline{\epsilon} \, , \ell)$.
Specifically we consider  the scaling limit $\ell\to\infty$ of the
bi-local term in the variation of $\log\det\Delta$ for non-zero deformation parameters. 
The uniformity of this limit depends on a technical bound on the discrete derivatives of the function $p_3(\uu,\uv)$ defined for isoradial graphs by  \ref{stuff2}. We explicate this condition and conjecture that the bound is valid for general isoradial graphs in Conjecture~\ref{Conf3bound}. Provided the bound is satisfied, we prove in Prop.~\ref{PropLimTr2} that the bi-local trace term has a uniform scaling limit, and that the $\ell\to\infty$  scaling limit and the 
 $\underline{\epsilon}\to 0$ deformation parameter limit both exist, are uniform, and commute (see Prop.~\ref{prvsepsComm}).

Finally, in Sect.~\ref{ssScalingKahler} we address the same questions for deformations of the K\"ahler operator $\mathcal{D}$.
Prop.~\ref{BoundOnDprime} gives a uniform bound on the variation of $\mathcal{D}$, but it implies that there is no general scaling limit $\ell\to\infty$ for 
$\mathcal{D}$ for non-zero values of the deformation parameters $\underline{\epsilon}$ (Prop.~\ref{NoLimitDprime}). 
This is different from the situation for $\Delta$. 
We argue that the best uniform convergence result to be expected for the bi-local trace term is a scaling limit where both $\ell\to\infty$ and 
$\underline{\epsilon}\to 0$ simultaneously, keeping $\ell \underline{\epsilon} = \underline{\mathtt{c}}$ constant (Prop.~\ref{PropLimTr2D}).

\subsection{Deforming triangulations with and without flips}
\label{ssFlipsNoFlips}
We now define and compare Delaunay deformations of graphs and connectivity-fixed deformations of the same graphs.

\newcommand{\epszero}{{\epsilon: {\scriptscriptstyle 0}}}
\newcommand{\zedzero}{z_{\scriptscriptstyle{0}}}
\newcommand{\zeroeps}{{{\scriptscriptstyle 0}: \epsilon}}

\subsubsection{Delaunay deformations (with flips)}
\label{aAssflips}

We start from an (isoradial) Delaunay graph $\uG_{0}=\uG_{\mathrm{cr}}$ and then deform its embedding
$\uv \mapsto \zedzero(\uv)$ using a smooth function $F:\Bbb{C} \longrightarrow \Bbb{C}$ with compact support
to obtain a mapping
\begin{equation}
\label{embedding-deformation-again}
\uv \mapsto z_\epsilon(\uv) = \zedzero(\uv) + \epsilon\, F(\zedzero(\uv))
\end{equation}
\noindent
for vertices $\uv$ of $\uG_{0}$. Using the method for proving Lemma~\ref{MFbound}, it is simple to prove that the mapping $\uv \mapsto z_\epsilon(\uv)$ defines an embedding of the vertex set $\mathrm{V}(\uG_0)$
as long as $\epsilon$ is small enough, namely:
\begin{equation}
\label{epsforF}
|\epsilon|< 
\dot{\epsilon}_F= 
(\max(|\partial F |)+\max(|\bar\partial F |))^{-1}
\end{equation}
Indeed we have 
\begin{equation*}
\label{ }
\left|{z_\epsilon(\uu)-z_\epsilon(\uv)\over \zedzero(\uu)-\zedzero(\uv)}\right|= \left|1 -\epsilon\,{F(\zedzero(\uu))-F(\zedzero(\uv))\over \zedzero(\uu)-\zedzero(\uv)}\right| \ge \left|1-\epsilon\left(\max |\partial F|+\max |\bar\partial F|\right)\right|
\end{equation*}
This ensures that if $\uu \neq \uv$, $|z_\epsilon(\uu)- z_\epsilon(\uv)|>0$ at least as long as \ref{epsforF} holds.

As in Def.~\ref{DefDelDeform}, the Delaunay graph $\uG_\epsilon$ is obtained by applying the Delaunay construction to the
set of deformed coordinates $z_\epsilon(\uv)$ for  $\uv \in \uG_{0}$. The vertices of $\uG_\epsilon$
and $\uG_0$ are identical by definition, however 
the edges and the faces of $\uG_\epsilon$ may differ from those of $\uG_{0}$
since the Delaunay constraints may force flips to occur during the deformation. Unlike 
the setup of Lemma~\ref{Lemma11}, 
the inclusion
$\mathrm{E}(\uG_0) \subset \mathrm{E}(\uG_\epsilon)$
may now fail.
Generically $\uG_\epsilon$ will be a triangulation regardless of whether 
the initial graph $\uG_0$ is a triangulation.

$\Delta(\epsilon)$, $\mathcal{D}(\epsilon)$ and ${\Deltaconf}(\epsilon)$ are the Laplacian operators relative to the lattice $\uG_\epsilon$, and  
act on the same space of functions $\Bbb{C}^{\mathrm{V}(\uG_\epsilon)}=\Bbb{C}^{\mathrm{V}(\uG_{0})}$ irrespective of $\epsilon$ since, by construction, the vertex sets $\mathrm{V}(\uG_\epsilon) = \mathrm{V}(\uG_{0})$ agree. 
Similarly, we denote by $\nabla_\epsilon$ and $\overline\nabla_\epsilon$ the discrete derivative operators relative to the faces of $\uG_\epsilon$, both of which are operators $\Bbb{C}^{\mathrm{V}(\uG_\epsilon)} \to \Bbb{C}^{\mathrm{F}(\uG_\epsilon)}$. Note that, in general, the set of deformed and critical faces differ,
i.e. $\mathrm{F}(\uG_\epsilon) \neq \mathrm{F}(\uG_{0})$. 
Similarly we denote by $A_\epsilon$ and $R_\epsilon$ the area and circumradius functions for the faces of $\uG_\epsilon$.

\subsubsection{Geometric Back-Deformation: deforming without flip}

We define the {\it rigid back-deformation} $\uG_\epszero$ of the Delaunay graph $\uG_\epsilon$ to be the 
graph whose vertex set and embedding are identical to
those of our initial (weak) Delaunay graph $\uG_0$, but whose edge and face sets 
coincide with those of the Delaunay graph $\uG_\epsilon$
obtained from $\uG_0$ by a Delaunay deformation.
The construction of  $\uG_\epszero$ can be seen in two stages: 
\begin{enumerate}
\item 
First $\uG_\epsilon$ is the end point of the continuous family of Delaunay deformations 
$$\uG_0\to\uG_\varepsilon\to \uG_\epsilon\ :\ \ 0\to\varepsilon\to\epsilon$$
obtained by continuously deforming the embedding 
$\zedzero \mapsto z_\varepsilon = \zedzero  + \varepsilon \,F(\zedzero)$ 
of the initial graph $\uG_0$ over the range $0\le\varepsilon\le \epsilon$
while maintaining the Delaunay condition (and performing edge flips as required)
at each stage of the deformation.

\item 
Then, starting with $\uG_\epsilon$, reverse the deformation $z_\varepsilon$ by letting $\varepsilon$ move from $\epsilon$ to $0$
$$
\uG_\epsilon\to\uG_{\epsilon:\varepsilon}\to \uG_\epszero\ :\ \ \epsilon\to\varepsilon\to 0
$$ 
but \textbf{without performing any edge flips}. In general
$\uG_{\epsilon:\varepsilon}$ will denote the graph whose vertex, edge, and face sets coincide with $\uG_\epsilon$ but whose embedding is $z_\varepsilon$. 
\end{enumerate}
More schematically 
\begin{equation}
\label{Teps0eps}
\begin{array}{ccccc}
    0  &  \stackrel{\varepsilon}{\longrightarrow} & \epsilon & \stackrel{\varepsilon}{\longrightarrow} & 0  \\
  \uG_0 &  \stackrel{\text{Delaunay}}{\longrightarrow}  &  \uG_\epsilon & \stackrel{\text{rigid}}{\longrightarrow} & \uG_\epszero
\end{array}
\end{equation}
It is clear that $\uG_\epszero$ is a graph (and in general a triangulation) with the vertex set as the original Delaunay graph $\uG_0$, but is generically \emph{not a Delaunay graph}. 

\subsubsection{An illustrative example}\ \\
Let us give a simple but illustrative example of such deformations of a triangulation $\uT_0\to\uT_\epsilon\to\uT_\epszero$.
The original triangulation $\uT_0$ is a biperiodic lattice. Vertices are labelled by $(m,n)\in\mathbb{Z}^2$ with coordinates
$$ \zedzero(m,n)= b(m+n/2)+ \mathrm{i}\, n\quad,\qquad 0<b\ll 1\ \text{a small parameter}$$
Hence the Delaunay triangulation $\uT_0$ is made of ``thin'' up and down triangles such that
$$\mathtt{heigth}\ = \ 1\quad,\qquad \mathtt{basis}\ = \ b$$
We choose as deformation function a simple shear parallel to the real axis, so that the deformed coordinates of vertices are
$$ z_\epsilon(m,n)= b(m+n/2)+ \mathrm{i}\, n\,+\ \epsilon\,n $$
The effect of a Delaunay deformation $\uT_0\to\uT_\epsilon$ is depicted on Fig.~\ref{animationshearA}, on the special case of $b=1/10$, and for $0\le\epsilon\le \epsilon_0=1/10$. 
Note that a flip occurs for every
$$\epsilon= {2k+1\over 2} b\ ,\quad k\in\mathbb{Z}$$
and that the Delaunay deformation $\uT_0\to\uT_\epsilon$ is then periodic
$$\uT_{\epsilon+ k b}=\uT_\epsilon\ ,\quad k\in\mathbb{Z}$$
Note also that if 
$$0\ll b\ll \epsilon\ll 1$$
a large number of flips $N_\mathtt{flip}(\epsilon)\simeq \epsilon/b$ occur, even when $\epsilon$ is small.
\begin{figure}[h!]
\begin{center}
\includegraphics[width=.75in]{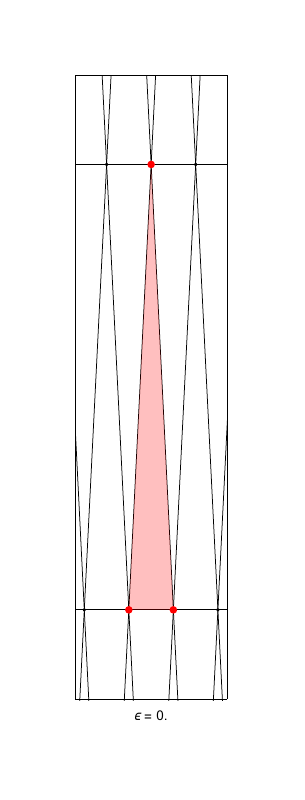}
\includegraphics[width=.75in]{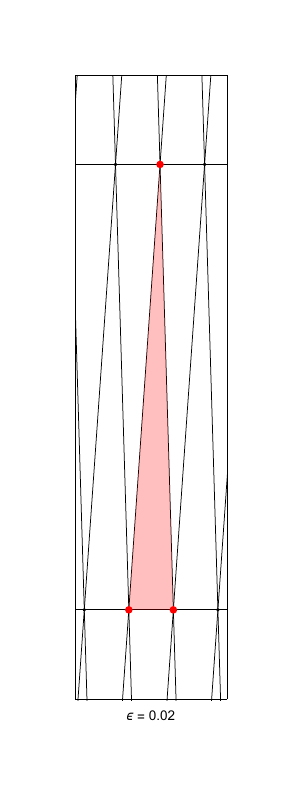}
\includegraphics[width=.75in]{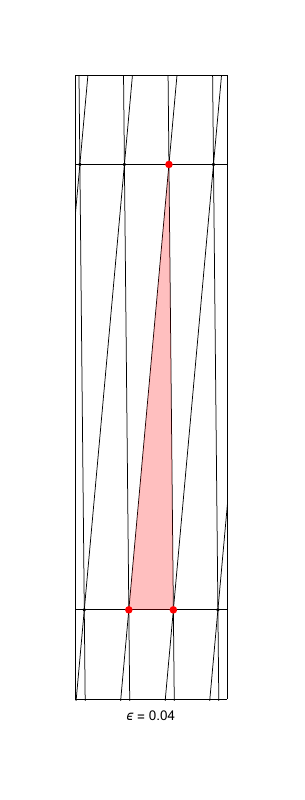}
\includegraphics[width=.75in]{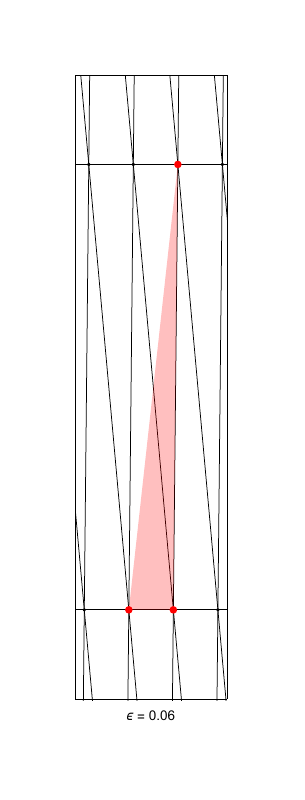}
\includegraphics[width=.75in]{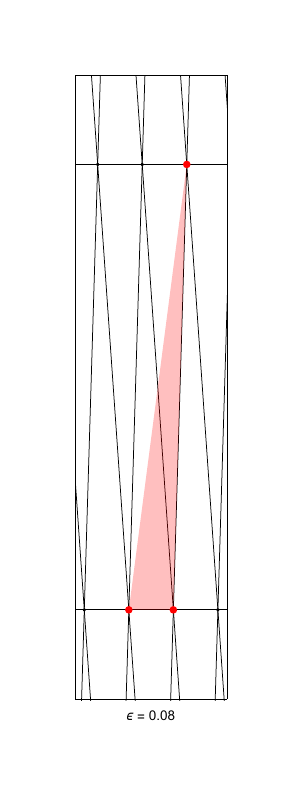}
\includegraphics[width=.75in]{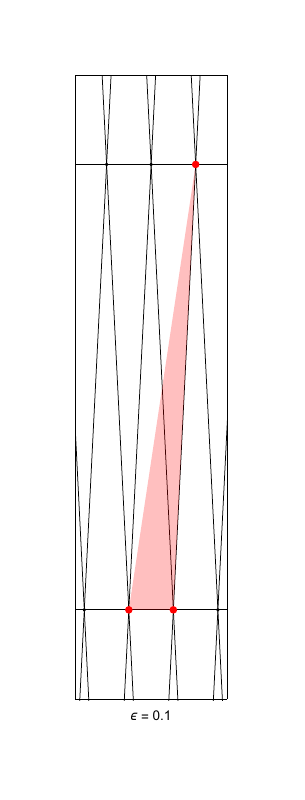}
\caption{Deformation of a periodic isoradial Delaunay triangulation $\uT_{\scriptscriptstyle{0}}\to \uT_\epsilon$ by a global shear $z\to z+\epsilon\, \mathfrak{Im} z$, keeping it Delaunay. 
On this example, the base and the height of the triangles are respectively  $b=1/10$ and $h=1$, so that a flip occur for $\varepsilon=b/2=1/20$, and we choose $\epsilon=b=1/10$. Since a flip occurs at $\epsilon_\mathtt{f}=b/2$, a triangle such has the one depicted in red, which is an original face of $\uT_{\scriptscriptstyle{0}}$, stays a face of $\uT_\epsilon$ for $0<\epsilon<\epsilon_\mathtt{f}$, but us not a face after the flip for $\epsilon>\epsilon_\mathtt{f}$.}
\label{animationshearA}
\end{center}
\end{figure}
The corresponding no-flip back-deformation $\epsilon:\ \epsilon_0\to 0$ which send back $\uT_{\epsilon_0}\to\uT_\epszero$ is depicted on Fig.~\ref{animationshearB}. It is clear on this figure that no back-flip occurs at $\epsilon=.05$, so that an original face of $\uT_{\epsilon_0}$ stays a face of $\uT_\epszero$. However $\uT_\epszero$ is a triangulation which is not Delaunay anymore.
\begin{figure}[h!]
\begin{center}
\includegraphics[width=.75in]{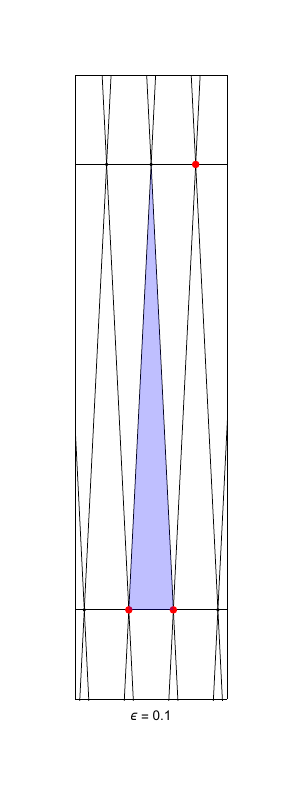}
\includegraphics[width=.75in]{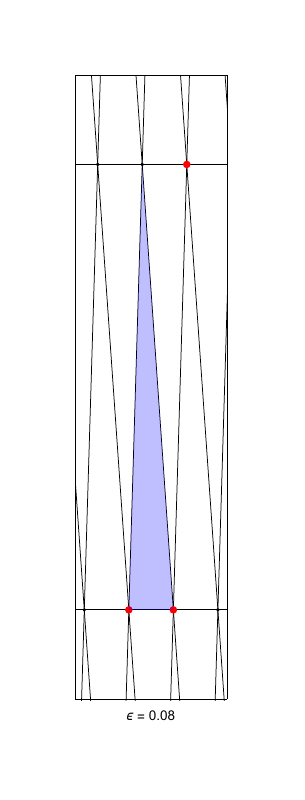}
\includegraphics[width=.75in]{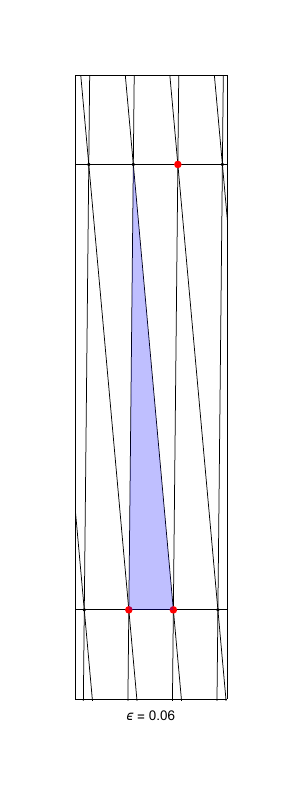}
\includegraphics[width=.75in]{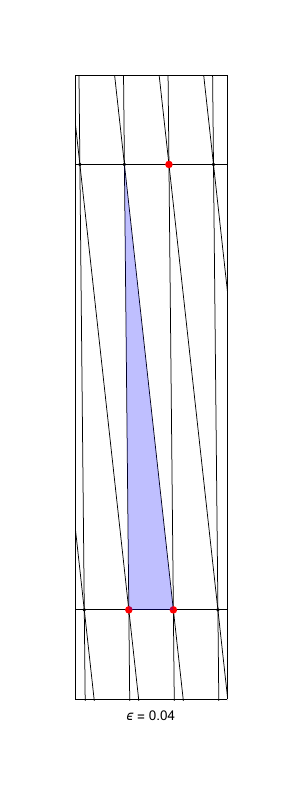}
\includegraphics[width=.75in]{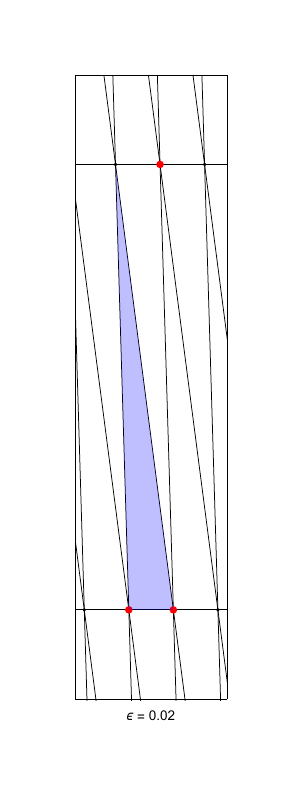}
\includegraphics[width=.75in]{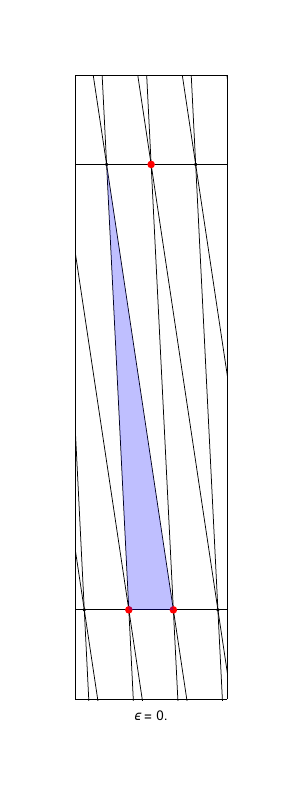}
\caption{The back-deformation of the triangulation of Fig.~\ref{animationshearA} $\uT_\epsilon\to\uT_\epszero$, keeping the edges and faces of the triangulations fixed (no-flips). An original face of $\uT_\epsilon$ (in blue)  stays a face of $\uT_\epszero$. However $\uT_\epszero$ is not Delaunay.}
\label{animationshearB}
\end{center}
\end{figure}
\color{black}

\subsection{Full variation of operators without flips}
\label{ssFullVar}

\subsubsection{Variation of the area}
Consider the variation of the triangulation $\uT \to \uT_\epsilon$ given by deforming the embedding $z(\uu) \to z_\epsilon(\uu)=z(\uu)+\epsilon F(\uu)$ \emph{without flips} (so that in fact $\uT_\epsilon$ should be denoted $\uT_{0:\epsilon}$ with the notations of the previous section).
For a triangle $\uf$ the full variation of its area is from \ref{Aform} and \ref{diff-formula-nablas}
\begin{equation}
\label{AtoAeps}
A\to A_\epsilon=A\left(1+\epsilon (\nabla F+\bar\nabla\bar F)+\epsilon^2 ( \nabla F\,\bar\nabla\bar F- \bar\nabla F\,\nabla\bar F)\right)
\end{equation}
For brevity $D(\epsilon;F)$ will denote the scaling factor 

 \begin{equation}
\label{DepsF}
D(\epsilon;F)=1+\epsilon (\nabla F+\bar\nabla\bar F)+\epsilon^2 ( \nabla F\ \bar\nabla\bar F- \bar\nabla F\ \nabla\bar F)
\end{equation}

\subsubsection{Variation of the discrete derivatives}

{The vertex sets $\mathrm{V}(\uT)$ and $\mathrm{V}(\uT_\epsilon)$ are, by definition, identical
and the face sets $\mathrm{F}(\uT)$ and $\mathrm{F}(\uT_\epsilon)$ agree so long as no flips occur in
the deformation $\uT \to \uT_\epsilon$. Consequently the nabla operators $\nabla$ and $\nabla_\epsilon$ 
(and their conjugates $\bar \nabla$ and $\bar{\nabla}_\epsilon$) share a common range and domain. Accordingly we have:}

\begin{equation}
\label{full-variation-nabla-epsilon}
\begin{split}
\nabla\to\nabla_\epsilon&= {1+\epsilon\bar\nabla\bar F\over 
D(\epsilon;F)
}\,\nabla - {\epsilon\nabla \bar F\over 
D(\epsilon;F)
}\,\bar\nabla
\\
\bar\nabla\to\bar\nabla_\epsilon&= {1+\epsilon\nabla F\over D(\epsilon;F)}\,\bar\nabla -{\epsilon \bar\nabla F\over D(\epsilon;F)}\,\nabla
\end{split}
\end{equation}

\subsubsection{A word of caution: deformations of functions}
{Recall that we may {\it restrict}
a smooth, complex-valued function $G: \Bbb{C} \longrightarrow \Bbb{C}$ 
to the vertex set of the triangulation $\uT$
using its graph embedding $z: \mathrm{V}(\uT) \longrightarrow \Bbb{C}$. 
Bearing some abuse of notation, we define and denote this restriction by $G(\uv) := G( z(\uv))$
for vertices $\uv \in \mathrm{V}(\uT)$. Some care is needed when restricting a smooth function $G$ to the deformed
triangulation $\uT_\epsilon$. The vertex sets of $\uT$ and $\uT_\epsilon$ are identical but, of course, their respective embeddings $z$ and $z_\epsilon$ are not,
and consequently the functions $\uv \mapsto G(z(\uv))$ and $\uv \mapsto G(z_\epsilon(\uv))$ do not
agree. In order to side step this discrepancy we introduce a deformed, smooth function $G_\epsilon: \Bbb{C} \longrightarrow
\Bbb{C}$ defined implicitly by 
\begin{equation}
\label{Gtransport}
G_\epsilon \big(w + \epsilon F(w) \big)=G(w)
\end{equation}
\noindent 
for all $w \in \Bbb{C}$, where $\epsilon \geq 0$ is fixed and sufficiently small. By construction,
\begin{equation}
G_\epsilon(z_\epsilon(\uv)) = G(z(\uv)) =: G(\uv)
\end{equation}
\noindent
To stress the role of the deformed embedding $z_\epsilon$ we shall define and denote $G_\epsilon(\uv) := G_\epsilon( z_\epsilon(\uv) )$
for $\uv \in \mathrm{V}(T_\epsilon)$. When $G = F$ this allows us to write
\begin{equation}
\label{Ftransport}
z_{\epsilon+\epsilon'}(\uv) = z(\uv)+(\epsilon +\epsilon') F(z(\uv)) = z_{\epsilon}+\epsilon' F_{\epsilon}(z_{\epsilon}(\uv))
\end{equation}

\subsubsection{Variation of the circumradii}
The full variation of the circumradius $R(\uf)$ of a face is more complicated. For a face with vertices labelled $1,2,3$ i.e. $\uf=(123)$ using \ref{Rcform} we get
\begin{equation}
\label{ }
R^2 \to R^2_\epsilon = R^2\, {N_{12}(\epsilon;F)N_{23}(\epsilon;F)N_{31}(\epsilon;F)\over D(\epsilon;F)^2}
\end{equation}
with
\begin{equation}
\label{ }
\begin{split}
N_{\uu\uv}(\epsilon;F)= 1&+\epsilon \left( \nabla F+\bar\nabla \bar F+ \bar C_{\uu\uv}\nabla\bar F + C_{\uu\uv}\bar\nabla F\right)
\\
&+ \epsilon^2\left(\nabla F\ \bar\nabla\bar F + \bar\nabla F\  \nabla \bar F + \bar C_{\uu\uv} \nabla F\  \nabla \bar F +  C_{\uu\uv} \bar \nabla F \ \bar \nabla \bar F
\right)
\end{split}
\end{equation}
where $C_{\uu\uv}$ for an (unoriented) edge $\overline{\uu\uv}$ denotes
\begin{equation}
\label{Cuv}
C_{\uu\uv}= {\bar{z}(\uu)-\bar{z}(\uv) \over {z(\uu) - z(\uv)}}
\end{equation}

\subsubsection{Variation of the operators}
Thus we get the variation of the Laplacian operators from
\begin{equation}
\label{ }
\Delta\to\Delta(\epsilon) = 2 \left( \nabla^\top_\epsilon\, A_\epsilon\, \nabla_\epsilon + \bar\nabla^\top_\epsilon\, A_\epsilon\, \bar \nabla_\epsilon\right)
\end{equation}
\begin{equation}
\label{ }
\mathcal{D}\to\mathcal{D}(\epsilon) =  4 \bar \nabla^\top_\epsilon\, {A_\epsilon\over R^2_\epsilon}\, \nabla_\epsilon
\end{equation}
that we do not write explicitly.
Note that all the expression we got are rational functions in $\epsilon$, and that when keeping only the first order in $\epsilon$ in a series expansion, we recover the results of Sect.~\ref{ssVarOp}.

\subsection{Full variation of operators under Delaunay deformations (with flips)}
\label{ssFullVarFlip}
Here we address the case of a critical triangulation $\uT=\uT_\mathrm{cr}$ with isoradius $R_0>0$ 
whose embedding undergoes a 
deformation 
$$z \to z_\epsilon := z + \epsilon F $$
\noindent
where flips are allowed, so that the deformed graph $\uT_\epsilon$ remains Delaunay. As before
the displacement function $F$ is the (restriction) of a smooth complex-valued function 
on the plane with compact support.
We consider the full variation of the operators
associated to the deformation $\uT_\mathrm{cr} \to T_\epsilon$, namely
\begin{equation}
\label{deltaOfull}
\delta\Delta(\epsilon)=\Delta(\epsilon)-\Delta_\mathrm{cr}
\quad,\qquad
\delta\mathcal{D}(\epsilon)=\mathcal{D}(\epsilon)-\mathcal{D}_\mathrm{cr}
\end{equation}
instead of the instantaneous, first order terms $\deltae\Delta$ and $\deltae\mathcal{D}$ in the respective $\epsilon$-expansions
as done in Sect.~\ref{ssVarOp} and \ref{sCalculations}.
We shall need uniform estimates for the $\epsilon \to 0$ limit of terms related to the variations $\delta\Delta(\epsilon)$ and $\delta\mathcal{D}(\epsilon)$
which are independent of the initial critical lattice $\uT_{\mathrm{cr}}$. Furthermore 
uniform estimates for the $R_0 \to 0$ limit will be needed, as this is synonymous with the $\ell \to \infty$ scaling limit.

Unfortunately, the exact results of the previous section \ref{ssFullVar} cannot be directly applied, since flips generically occur within
the continuous family of Delaunay graphs $\uT_\varepsilon$ as
the deformation parameter $\varepsilon$ moves 
from zero to $\epsilon>0$.
Nevertheless, we may write each variation as the integral of a derivative, and then try to get 
uniform bounds on the derivatives. This is what we discuss in the remaining part of this Appendix.

Let us first consider the simpler case of the Laplace-Beltrami operator $\Delta$. We can write
\begin{equation}
\label{delVarDint}
\delta\Delta(\epsilon)=\int_0^\epsilon d\varepsilon\, \Delta'(\varepsilon)\quad\text{with}\qquad \Delta'(\varepsilon)={d\over d\varepsilon}\Delta(\varepsilon)=\thedelta_\varepsilon\Delta(\varepsilon)
\end{equation}
Indeed, since $F$ is smooth with compact support, there is a finite (possibly large) number of flips as $\varepsilon$ increases, and we know that 
$\Delta(\varepsilon)$ is a continuous function of $\varepsilon$, and its derivative exists and is continuous in the interval between the flips.
Therefore the derivative $\Delta'(\varepsilon)$ is bounded and piecewise continuous, so that the integral \ref{delVarDint} makes sense.
For a given value $\varepsilon \geq 0$, the first order term in formula \ref{VarDelta} extends to the case of 
$\Delta(\varepsilon)$ defined on $\uT_\varepsilon$ and w.r.t. the transported displacement function $F_\varepsilon$ in the plane.

\begin{equation}
\label{DeltaPrime}
\Delta'(\varepsilon)\ =\ \nabla_{\!\varepsilon}^{\!\top}{\cdot} A_\varepsilon{\cdot}\textswab{D}_\varepsilon{\cdot}\nabla_{\!\varepsilon}+
   \overline\nabla_{\!\varepsilon}^{\!\top}{\cdot} A_\varepsilon{\cdot}\overline{\textswab{D}}_\varepsilon{\cdot}\overline\nabla_{\!\varepsilon}
\end{equation} 
with
\begin{equation}
\label{swabE}
\textswab{D}_\varepsilon=-4\,\overline\nabla_{\!\varepsilon} F_{\varepsilon}\quad,\qquad
\overline{\textswab{D}}_\varepsilon= -4\,\nabla_{\!\varepsilon} \bar F_{\!\varepsilon}
\end{equation}

Similarly, we can write the variation of the K\"ahler operator as
\begin{equation}
\label{ }
\delta\mathcal{D}(\epsilon)=\int_0^\epsilon d\varepsilon\, \mathcal{D}'(\varepsilon)\quad,\qquad \mathcal{D}'(\varepsilon)={d\over d\varepsilon}\mathcal{D}(\varepsilon) = \thedelta_\varepsilon\mathcal{D}(\varepsilon)
\end{equation}
The results of Section~\ref{ssVarOp} give for the derivative of $\mathcal{D}$
\begin{equation}
\label{KahlerPrime}
\mathcal{D}'\!(\varepsilon)= 
  \overline\nabla_{\!\varepsilon}^{\!\top} {A_{\!\varepsilon}} 
\textswab{K}_\varepsilon
  \nabla_{\!\varepsilon} 
+ \nabla_{\!\varepsilon}^{\!\top} {A_{\!\varepsilon}} 
\textswab{H}_\varepsilon
\nabla_{\!\varepsilon}+\overline\nabla_{\!\varepsilon}^{\!\top} {A_{\!\varepsilon}} 
\overline{\textswab{H}}_\varepsilon
\overline\nabla_{\!\varepsilon}\   
\end{equation}
with 
\begin{equation}
\label{swabf1}
\begin{split}
\textswab{K}_\varepsilon&= -{4\over R_{\!\varepsilon}^2} \left(\nabla_{\!\varepsilon} F_{\!\varepsilon} +\overline\nabla_{\!\varepsilon} \bar F_{\!\varepsilon} + C_{\!\varepsilon}\,\overline\nabla_{\!\varepsilon} F_{\!\varepsilon} + \bar C_{\!\varepsilon}\,\nabla_{\!\varepsilon}\bar F_{\!\varepsilon}\right)\\
\textswab{H}_\varepsilon&= -{4\over R_{\!\varepsilon}^2} \,\overline\nabla_{\!\varepsilon} F_{\varepsilon} \quad ,\qquad 
\overline{\textswab{H}}_\varepsilon = -{4\over R_{\!\varepsilon}^2}\nabla_{\!\varepsilon} \bar F_{\!\varepsilon}
\end{split}
\end{equation}
and with the $C_{\!\varepsilon}$ and $\bar C_{\!\varepsilon}$ defined by \ref{CDef} for faces the triangulation $T_{\!\varepsilon}$, namely for a face $\uf=(123)$,
\begin{equation}
\label{C123def}
C(\uf)=C_{123}={\bar z_1-\bar z_2\over z_1-z_2}+{\bar z_2-\bar z_3\over z_2-z_3}+{\bar z_3-\bar z_1\over z_3-z_1}
\end{equation} 
Note that we can decompose $C(\uf)$ as a sum of the terms $C_{\uu\uv}$ defined in \ref{Cuv} 
for edges $\overline{\uu \uv}$ of $\uf$. Specifically $C(123)=C_{12}+C_{23}+C_{31}$ where $\uf=(123)$.

\subsection{Uniform bounds under Delaunay deformations (with flips)}
\label{BoundNablaF}
\subsubsection{Bounds on continuous derivatives}
Now we study wether it is possible to give uniform bounds w.r.t. $\varepsilon$  and $\uT_\epsilon$ on the various coefficients $A_\epsilon$, $R_\epsilon$, $\textswab{D}_\epsilon$,$\textswab{R}_\epsilon$ and $\textswab{H}_\epsilon$ of the previous section~\ref{ssFullVarFlip}, and on the operators $\nabla_\epsilon$ and $\overline\nabla_\epsilon$.
From now on, let $F:\,\mathbb{C}\to\mathbb{C}$ be a given smooth deformation function with compact support.
Let
\begin{equation}
\label{M1M2Def}
\begin{split}
M_1&=\sup_{z\in\mathbb{C}} \max\left\{|\partial F(z)|,|\bar\partial F(z)|\right\}
\\ 
M_2&=\sup_{z\in\mathbb{C}} \max\left\{|\partial^2 F(z)|,|\partial\bar\partial F(z)|,|\bar\partial^2 F(z)|\right\}
\end{split}
\end{equation}
We will consider the transported function $F_\epsilon$ defined by \ref{Ftransport}, and the transported version of \ref{M1M2Def}
\begin{equation}
\label{M1M2EpsBnd}
\begin{split}
M_1(\epsilon)&=\sup_{z\in\mathbb{C}} \max\left[|\partial F_\epsilon(z)|,|\bar\partial F_\epsilon(z)|\right]
\\ 
M_2(\epsilon)&=\sup_{z\in\mathbb{C}} \max\left[|\partial^2 F_\epsilon(z)|,|\partial\bar\partial F_\epsilon(z)|,|\bar\partial^2 F_\epsilon(z)|\right]
\end{split}
\end{equation}
By differentiating the functional relation \ref{Ftransport} between $F$ and $F_\epsilon$, one gets the general inequalities
\begin{equation}
\label{MepsIneq}
M_1(\epsilon)\le \overline M_1(\epsilon) = {M_1\over 1-2\epsilon\,M_1}
\ ,\quad
M_2(\epsilon)\le \overline M_2(\epsilon) = {M_2\over (1-2\epsilon\,M_1)^3}
\end{equation}
valid as long as $\epsilon$ is small enough, namely 
\begin{equation}
\label{epsbound1}
0\le \epsilon < \check\epsilon_F=1/(2 M_1)
\end{equation}
which ensures that $F_\epsilon$ is not multivalued (and stays smooth with compact support).

\subsubsection{Bounds on discrete derivatives}
\label{sssBndDiscrDer}

Let $\uT_\mathrm{cr}$ be a critical (Delaunay isoradial) triangulation with isoradius $R_0$, and $\uT_\epsilon$ be the Delaunay triangulation $\uT_\epsilon$ obtained by the $\epsilon$-deformation $z\to z+\epsilon F$.
We shall establish bounds on the norm of the discrete derivatives of $F_\epsilon$ on the triangulation $\uT_\epsilon$, as well as inequalities on the radii $R(\uf)$ of the faces of $\uT_\epsilon$.

First we define for a generic triangulation $\uT$ and a generic smooth function $G$ with compact support
\begin{equation}
\label{ }
B_{G}(\uT)=\sup_{\mathrm{faces}\,\uf\in \uT}\max\left(|\nabla G(\uf)|,|\overline\nabla G(\uf)|\right)
\end{equation}
We use Lemma~\ref{lemmabound}, which gives a bound on the difference between the discrete derivative $\nabla G(\uf)$ and the continuous derivate $\partial G$ of $G$ at the circumcenter of $\uf$. This bound involves the circumradius of $\uf$ and the $\mathrm{max}$ of the second derivative of $G$ inside the circumcircle. 
Denote the max of the circumradii of the faces $\uf$ of a triangulation $\uT$
\begin{equation}
\label{RmaxDef}
R_{\mathtt{max}}(\uT)=\max_{\uf\in\uT}\ R(\uf)
\end{equation}
For the initial critical triangulation $\uT_{\mathrm{cr}}$ Lemma~\ref{lemmabound} implies
\begin{equation}
\label{BFTcritbound}
B_F(\uT_{\mathrm{cr}})\le M_1+4\,M_2\, R_{0}
\end{equation}
but for $\uT_\epsilon$ it becomes
\begin{equation}
\label{BFTebound}
B_F(\uT_\epsilon)\le M_1(\epsilon)+4\,M_2(\epsilon)\, R_{\mathtt{max}}(\uT_\epsilon)
\end{equation}
and we need an estimate of $R_{\mathtt{max}}(\uT_\epsilon)$.

\subsection{Inequalities for general variations of circumradii (with or without flips)}
\label{ssVArRGen}
\subsubsection{The problem}
In order to get a bound on $R_{\mathtt{max}}(\uT_\epsilon)$, we now derive a bound on the variation of the circumradius of the faces, of a triangulation under a deformation $z\to z+\epsilon F$.

Let us consider the following general deformation scheme.
We start with an initial Delaunay triangulation $\uT_0$ which need not be isoradial. 
We deform the embedding $z \to z+ \varepsilon\, F(z)$ of $\uT_0$ within the range $0 \leq \varepsilon \leq \epsilon$ (with $\epsilon<\dot\epsilon_F$ defined by \ref{epsforF}).
If at any stage of the deformation the circumradii $R(\uf_1)$ and $R(\uf_2)$ of two  neighboring faces $\uf_1$ and $\uf_2$ 
agree, we may either (i) \emph{perform an edge flip}, so that $\uf_1$, $\uf_2$ are replaced by two new faces $\uf'_1, \uf'_2$ or (2) \emph{not perform the flip}.
Thus we get a family of triangulations $\{\uT_\varepsilon:\,\varepsilon\in[0,\epsilon]$, in general not Delaunay, which share the same vertex set and have vertex embeddings $z_\epsilon=z_0+\epsilon\, F(z_0)$.

Now consider an initial face (triangle) $\uf_0$ of $\uT_0$, with initial circumradius $R(0)=R_0(\uf_0)$. When deforming $\uT_\varepsilon$ from $0$ to $\epsilon$, we can  continuously follow the face $\uf_0$, and when it sustains a flip, we \emph{choose one of the two faces} created by the flip. In this way we get a ``continuous'' family of 
faces $\{\uf_\varepsilon\in\uT_\varepsilon:\,\varepsilon\in[0,\epsilon]\}$, so that $\varepsilon \mapsto R(\uf_\varepsilon)$ is a continuous, piecewise differentiable function  
(this is the crucial point for the following argument).

\subsubsection{Bounds on the derivative of $R$ and consequences}
Now, in between the flips, from \ref{var1A}, \ref{var1AR2} the derivative of the circumradius $R(\uf_\varepsilon)$ of  this face $\uf_\varepsilon$ is
\begin{equation}
\label{varRone}
R'(\uf_\varepsilon)=
{d\over d\varepsilon}R(\uf_\varepsilon)= {R(\uf_\varepsilon) \over 2}
\left( \nabla_\varepsilon F_\varepsilon(\uf_\varepsilon)+\overline\nabla_\varepsilon \bar F_\varepsilon(\uf_\varepsilon)+ C_\varepsilon(\uf_\varepsilon)\overline\nabla_\epsilon F_\varepsilon(\uf_\varepsilon) + \bar C_\varepsilon(\uf_\varepsilon) \nabla_\varepsilon \bar F_\varepsilon(\uf_\varepsilon)\right)
\end{equation}
Using Lemma~\ref{lemmabound} again, for this face $\uf_\varepsilon$ of the triangulation $\uT_\varepsilon$
we get the bound
\begin{equation}
\label{dRdepsbound}
\left|\nabla_\varepsilon F_\varepsilon(\uf_\varepsilon)\right|\ \text{and}\ \left|\bar\nabla_\varepsilon F_\varepsilon(\uf_\varepsilon)\right| \ \le\ M_1(\varepsilon) + 4\, R(\uf_\varepsilon)\,M_2(\varepsilon)
\end{equation}
and from the definition of $C$ \ref{C123def} we have
\begin{equation}
\label{ }
| C_\varepsilon(\uf_\varepsilon)|\le 3
\end{equation}
We thus get the bound
\begin{equation}
\label{boundDerivR}
\left|{d\over d\varepsilon}R(\uf_\varepsilon)\right|\  \le\   4\, \overline M_1(\varepsilon)\, R(\uf_\varepsilon) + 16\, \overline M_2(\varepsilon)\, R(\uf_\varepsilon)^2
\end{equation}
Remember that the functions $\overline M_1(\varepsilon)$ and $\overline M_2(\varepsilon)$ are explicitly known functions of $\varepsilon$ and the constants $M_1$ and $M_2$ associated to the displacement function $F$.
$$\overline M_1(\varepsilon)={M_1\over 1-2\,\varepsilon\,M_1}
\ ,\quad
\overline M_1(\varepsilon)={M_2\over (1-2\,\varepsilon\,M_1)^3}
$$

\subsubsection{Bounds on the circumradii $R(\uf_\epsilon)$.}
\label{sssBndR}
Using the inequality \ref{boundDerivR} we get uniform bounds on the variation of the circumradius of faces $R(\uf_\epsilon)$ under deformations $z\to z_\epsilon=z+\epsilon F(z)$.

\begin{Prop}
\label{PropReps}
The radius of the face $\uf_\epsilon$ satisfy the inequalities
\begin{equation}
\label{RepsIneq}
\bar R_-(\epsilon,R(\uf_0)) \ \le\  R(\uf_\epsilon)\ \le  \bar R_+(\epsilon,R(\uf_0))
\end{equation}
with the functions of the radius variable $R$
\begin{equation}
\label{barRplus}
\bar R_+(\epsilon,R ) = {R\over \left(1+{M_2 R\over M_1}\right)\left(1-2 M_1\,\epsilon\right)^2 -{M_2 R\over M_1}\left(1-2 M_1\,\epsilon\right)^{-2}}
\end{equation}
and 
\begin{equation}
\label{barRminus}
\bar R_-(\epsilon,R)={ R  \left(1-2 M_1 \epsilon\right)^2 \over  1 + {8 M_2 R\over M_1} \log\left({1\over 1-2 M_1 \epsilon}\right)}
\end{equation}
The inequality \ref{RepsIneq} is satisfied at least if 
\begin{equation}
\label{EpsDomain}
0\le \epsilon < \epsilon_{\mathtt{max}}(R(\uf_0))
\ \text{with}\quad
\epsilon_{\mathtt{max}}(R) := {1\over 2 M_1}\left( 1-\left(1+{M_1\over R M_2}\right)^{-1/4}\right)
\end{equation}
the value of $\epsilon$ where $\bar R_+(\epsilon,R)$ diverges.
Note that 
$\epsilon_{\mathtt{max}}(R)<{1/(2 M_1)}$
 and that $\bar R_-(\epsilon,R)$ is positive and well defined for 
 $\epsilon_{\mathtt{max}}(R)<{1/(2 M_1)}$.
 \end{Prop}

\begin{proof}
Let us for simplicity change of variables and consider instead of $\epsilon$ the variable $y$ 
\begin{equation}
y= -\log(1-2 M_1\,\epsilon)
\end{equation}
and the function $V(y)$ defined as 
\begin{equation}
\label{ }
V(y)\ =\ {(1-2 M_1\,\epsilon)^{2}\over R(\uf_\epsilon)}
\end{equation}
and denote\begin{equation}
\label{RtoV}
V_0=V(0)={1\over R(\uf_0)}
\end{equation}
After some algebra the inequality \ref{boundDerivR} becomes a simple linear inequality
\begin{equation}
\label{BoundDerivV}
- A-4\,V(y)\ \le \ {d V(y)\over dy}\ \le\  A
\quad,\qquad A={8 M_2\over M_1}
\end{equation}
The rightmost inequality implies obviously
\begin{equation}
\label{not-sure}
V(y)\le \overline V_{\!-}(y)=V_0+A y
\end{equation}
The leftmost inequality gives for the function
\begin{equation}
T(y)= V(y)\ \mathrm{e}^{-4 y}
\end{equation}
which is such that $T(0)=V_0$, the inequality
\begin{equation}
\label{ }
{d T(y)\over dy} \ge -A\ \mathrm{e}^{4 y}
\end{equation}
which implies
\begin{equation}
\label{ }
T(y)\ge V_0 - {A\over 4} \left(\mathrm{e}^{4 y}-1\right)
\end{equation}
hence
\begin{equation}
\label{ }
V(y)\ge \overline V_{\!+}(y)=\left(V_0+{A\over 4}\right)\mathrm{e}^{-4 y} -{A\over 4}=\overline V_{\!+}(y)
\end{equation}
Note that the functions $\overline V_{\!-}(y)$ and $\overline V_{\!+}(y)$ are the functions which saturate the inequalities \ref{BoundDerivV} for $V$ with the same initial condition
$\overline V_{\!-}(0)=\overline V_{\!+}(0)=V(0)=V_0$.
Going back from $V$ to $R(\uf)$ through \ref{RtoV}, and defining $\overline R_{+}$ and $\overline R_{-}$ through
\begin{equation}
\label{ }
V_{\!+}(y)\ =\ {(1-2 M_1\,\epsilon)^{2}\over \overline R_{+}(\uf_\epsilon)}
\quad\text{and}\quad
V_{\!-}(y)\ =\ {(1-2 M_1\,\epsilon)^{2}\over \overline R_{-}(\uf_\epsilon)}
\end{equation}
we get the results of Prop.~\ref{PropReps}

\end{proof}

Proposition~\ref{PropReps} is the main result of this section.
Note that it does not require the initial triangulation to be Delaunay or isoradial.
It is also completely independent of whether we perform flips or do not perform flips during the deformation.
It depends only on the deformation function $F$ and on the initial radius of the initial face we start from.

Notice that when the initial radius of the initial face becomes very small 
\ref{RepsIneq} implies that
\begin{equation}
\label{ }
(1-2\,\epsilon\, M_1)^{2}\ \le\  \lim_{R(\uf_0) \rightarrow 0} \,
{R(\uf_\epsilon)\over R(\uf_0)}\ \le(1-2\,\epsilon\, M_1)^{-2}
\end{equation}

\subsubsection{Final estimates}
With Prop.~\ref{PropReps} we can complete the estimates of the previous sections \ref{BoundNablaF}.
We start from an initial critical triangulation $\uT_\mathrm{cr}$ with initial radius $R_0$, and deform it into the Delaunay triangulation $\uT_\epsilon$.
The inequality \ref{RepsIneq}  implies that
\begin{equation}
\label{RmaxDef}
R_{\mathtt{max}}(\uT_\epsilon)=\max_{\uf\in\uT_\epsilon}\ R(\uf)\ \le\ \bar R_+(\epsilon,R_0)
\end{equation}
hence
\begin{equation}
\label{BFTebound}
B_F(\uT_\epsilon)=\max_{\uf\in\uT_\epsilon}\left(|\nabla_\epsilon F_\epsilon|,|\bar\nabla_\epsilon F_\epsilon|\right) \le\  \bar M_1(\epsilon)+4\,\bar M_2(\epsilon)\, \bar R_+(\epsilon,R_0)
\end{equation}
We can  bound the coefficients in the derivative w.r.t. $\epsilon$ of the Laplace-Beltrami operator $\Delta(\epsilon)$ (in \ref{DeltaPrime}), and of the 
K\"ahler operator $\mathcal{D}(\epsilon)$ (in \ref{KahlerPrime}).
\begin{equation}
\label{ }
\begin{split}
|\textswab{D}_\epsilon|&\ \le\  4\,\bar M_1(\epsilon)+16\,\bar M_2(\epsilon)\, \bar R_+(\epsilon,R_0)\\ 
|\textswab{K}_\epsilon|&\ \le \ {16 \bar M_1(\epsilon)+64\,\bar M_2(\epsilon)\, \bar R_+(\epsilon,R_0)\over \bar R_-(\epsilon,R_0)^2}\\
|\textswab{H}_\epsilon|&\ \le \ {4 \bar M_1(\epsilon)+16\,\bar M_2(\epsilon)\, \bar R_+(\epsilon,R_0)\over \bar R_-(\epsilon,R_0)^2}
\end{split}
\end{equation}
Using the explicit forms of $\bar M_1(\epsilon)$ and $\bar M_2(\epsilon)$ given by \ref{MepsIneq}, and of $R_+(\epsilon,R_0)$ and $\bar R_-(\epsilon,R_0)$ given by \ref{barRplus} and \ref{barRminus}, one deduces that  $|\textswab{D}_\epsilon|$,  $|\textswab{K}_\epsilon|$ and $|\textswab{H}_\epsilon|$ are uniformly bounded. More precisely we can summarize the estimates we obtained into the following proposition.

\begin{Prop}\label{propboundswab}
Let us choose a smooth displacement function $F$ with bounds $M_1$ and $M_2$ associated to its first and second derivatives.
Let us also choose $\epsilon_\mathrm{b}$ strictly smaller than $\epsilon_{\mathtt{max}}(R_0=1)$ given by
\begin{equation}
\label{epsmax}
0<\epsilon_\mathrm{b}<\epsilon_{\mathtt{max}}(1) =  {1\over 2 M_1}\left( 1-\left(1+{M_1\over  M_2}\right)^{-1/4}\right)
\end{equation}
for instance $\epsilon_\mathrm{b}=\epsilon_{\mathtt{max}}(R_0=1)/2$; see formula (\ref{EpsDomain}) for
a definition of $\epsilon_{\mathtt{max}}(R_0)$.
Then consider an arbitrary initial critical triangulations (isoradial and Delaunay)  $\uT_0$ with circumradius $R_0$, some $\epsilon>0$,  the deformed Delaunay lattice $\uT_\epsilon$ obtained from $\uT_0$ by the deformation $z\to z+\epsilon\, F(z)$, and an arbitrary face $\uf$ of $\uT_\epsilon$.

Then the factors $ \textswab{D}_\epsilon(\uf)$ (given by \ref{swabE}), $ \textswab{K}_\epsilon(\uf)$ and $ \textswab{H}_\epsilon(\uf)$ (given by \ref{swabf1}) for the face $\uf$ are uniformly bounded over the sets of: (i) initial triangulation $\uT_0$ with isoradius $R_0$ less or equal to one, (ii) deformation parameter $\epsilon$ smaller or equal to $\epsilon_\mathrm{b}$, (iii) and faces $\uf$ of $\uT_\epsilon$.
Namely, there exist  constants  $D_0$, $K_0$ and $H_0$ which depend only of $F$ and on the choice of $\epsilon_\mathrm{b}$ such that
\begin{equation}
\label{ }
 |\textswab{D}_\epsilon(\uf)|\le D_0\ ,\quad  |\textswab{K}_\epsilon(\uf)|\le K_0 \ ,\quad  |\textswab{H}_\epsilon(\uf)|\le H_0
\end{equation}
Similarily, there exists a constant $P_0(F;\epsilon_\mathrm{b})$, which depends only of $F$ and on $\epsilon_\mathrm{b}$, which uniformly bounds 
the variation of the radius of the faces 
\begin{equation}
\label{ }
\left | (R(\uf_\epsilon)-R_0)/R_0 \right | \le \epsilon\, P_0(F;\epsilon_\mathrm{b})
\end{equation}
\end{Prop}

\subsection{Consequence for the control of the scaling limit of $\Delta$} 
\label{ssScalLimD}
\subsubsection{The Laplace-Beltrami operator $\Delta$}
To simplify, we use a $2\times 2$ block matrix notation. The $\Delta$ operator and its $\epsilon$-derivative $\Delta'$ on the deformed lattice $\uT_\epsilon$ reads
\begin{equation}
\label{LBmatrixNot}
\Delta(\epsilon) = 
2 \begin{pmatrix}
      \nabla_{\!\epsilon}    \\
      \overline\nabla_{\!\epsilon}  
\end{pmatrix}^\dagger
\begin{pmatrix}
     A_\epsilon &   0  \\
   0   &  A_\epsilon
\end{pmatrix}
\begin{pmatrix}
      \nabla_{\!\epsilon}    \\
      \overline\nabla_{\!\epsilon}  
\end{pmatrix}
\ ,\quad
\Delta'(\epsilon) = 
-4  \begin{pmatrix}
      \nabla_{\!\epsilon}    \\
      \overline\nabla_{\!\epsilon}  
\end{pmatrix}^\dagger
\begin{pmatrix}
     0 &   \!\!A_\epsilon\,\nabla_{\!\epsilon}\overline F_\epsilon  \\
   A_\epsilon\,\overline\nabla_{\!\epsilon} F_\epsilon   &  0
\end{pmatrix}
\begin{pmatrix}
      \nabla_{\!\epsilon}    \\
      \overline\nabla_{\!\epsilon}  
\end{pmatrix}
\end{equation}
Remember that $A_\epsilon$, $\nabla_{\!\epsilon}\overline F_\epsilon$ and $\overline\nabla_{\!\epsilon} F_\epsilon$ are defined for the faces of the deformed triangulation $\uT_\epsilon$, whose vertices have positions $z_\epsilon= z +\epsilon F(z)$, while $\Delta(\epsilon)$ and $\Delta'(\epsilon)$ acts on the functions defined on the vertices of $\uT_\epsilon$. Since $\uT_\epsilon$ is obtained by deforming an initial critical lattice $\uT_0=\uT_{\mathrm{cr}}$, let us rewrite them in terms on objects defined for the ``back-deformed'' lattice $\uT_\epszero$  defined by the procedure introduced in sect.~\ref{ssFlipsNoFlips}
(see \ref{Teps0eps} and the example illustrated in Figs.~\ref{animationshearA} and \ref{animationshearB}).
\begin{equation*}
\label{Teps0eps2}
\underset{}{\uT_{\mathrm{cr}}=\uT_0}\ \overset{\text{Delaunay}}{\longrightarrow}\ \underset{}{\uT_{\epsilon}}\ \stackrel{\text{no flip}}{\longrightarrow}\ \uT_\epszero
\end{equation*}

Again, $\uT_\epszero$ has the same vertices as $\uT_{0}$, but the edges and faces of $\uT_{\epsilon}$. In other words, $\uT_{\epsilon}$ is obtained from $\uT_\epszero$ by the deformation $z\to z_\epsilon= z +\epsilon F(z)$, but without flips. We can therefore express the objects relative to the faces of $\uT_{\epsilon}$ in terms of those relative to the faces of $\uT_\epszero$.
The area $A_\epsilon$ of a face $\uf_\epsilon$ of $\uT_\epsilon$ is related to the area $A$ of the corresponding face $\uf=\uf_\epszero$ of $\uT_\epszero$ by \ref{AtoAeps}, namely 
\begin{equation}
\label{ }
A_\epsilon= D(\epsilon;F)\, A
\end{equation}
with from \ref{DepsF}
\begin{equation}
\label{ }
D(\epsilon;F)=1+\epsilon (\nabla F+\overline\nabla \bar F)+\epsilon^2 ( \nabla F\ \overline\nabla\bar F- \overline\nabla F\ \nabla\bar F)
\end{equation}
Note that the operators $\nabla$ and $\bar\nabla$ refer now to faces of $\uT_\epszero$. In a strict sense they should be denoted $\nabla_\epszero$ and $\bar\nabla_\epszero$. We omit the subscript to simplify notation.
The discrete derivative operators on $\uT_\epsilon$ are expressed in terms of those on $\uT_\epszero$ by \ref{full-variation-nabla-epsilon}, which 
can be expressed in the block matrix notation as
\begin{equation}
\label{ }
\begin{pmatrix}
      \nabla_{\!\epsilon}    \\
      \overline\nabla_{\!\epsilon}  
\end{pmatrix}
= 
{1\over D(\epsilon;F)}
\begin{pmatrix}
  {1+\epsilon\, \overline\nabla\,\overline F}    &  {-\epsilon\, \nabla\overline F} \\
    {-\epsilon\, \overline\nabla F}  &  {1+\epsilon\, \nabla F}
\end{pmatrix}
\begin{pmatrix}
      \nabla    \\
      \overline\nabla  
\end{pmatrix}
\end{equation}
In particular
\begin{equation}
\label{ }
\begin{pmatrix}
      \nabla_{\!\epsilon}  F_{\!\epsilon}  \\
      \overline\nabla_{\!\epsilon}  F_{\!\epsilon} 
\end{pmatrix}
= 
{1\over D(\epsilon;F)}
\begin{pmatrix}
  {1+\epsilon\, \overline\nabla\,\overline F}    &  {-\epsilon\, \nabla\overline F} \\
    {-\epsilon\, \overline\nabla F}  &  {1+\epsilon\, \nabla F}
\end{pmatrix}
\begin{pmatrix}
      \nabla   F  \\
      \overline\nabla   F
\end{pmatrix}
\end{equation}
Again the discrete $\nabla$ and $\overline\nabla$ refer now to faces of $\uT_\epszero$.
Including this in \ref{LBmatrixNot} one gets
\begin{equation}
\label{LBprimeEps0}
\Delta'(\epsilon)= \begin{pmatrix}
      \nabla    \\
      \overline\nabla  
\end{pmatrix}^\dagger
A\ 
\raisebox{-.8ex}{\text{\Huge{$\mathbb{D}$}}}(\epsilon;F)
\begin{pmatrix}
      \nabla     \\
      \overline\nabla  
\end{pmatrix}
\end{equation}
with {\Large{$\mathbb{D}$}} the $2\times 2$ block matrix
\begin{equation}
\label{TheBigDMatrix}
\begin{split}
&\begin{matrix}\text{\Huge{$\mathbb{D}$}}\end{matrix}(\epsilon;F)=\\
&{(-4) \over {D(\epsilon;F)}^2}
\begin{pmatrix}
 -\epsilon \nabla \overline F\,\overline\nabla F \left( 2+\epsilon (\nabla F +\overline\nabla\,\overline F)\right)     & 
 \nabla\overline F\left({(1+\epsilon\,\nabla F)}^2-\epsilon^2 \overline\nabla F\,\nabla\overline F\right)   \\
 \overline\nabla F\left({(1+\epsilon\,\overline\nabla\, \overline F)}^2-\epsilon^2 \nabla \overline F\,\overline\nabla F\right)     &
 -\epsilon \nabla \overline F\,\overline\nabla F \left( 2+\epsilon (\nabla F +\overline\nabla\,\overline F)\right)  
\end{pmatrix}
\end{split}
\end{equation}

\subsubsection{Scaling limit for $\Delta(\epsilon)$}
We can now study the scaling limit of the deformed operator $\Delta(\epsilon)$. We proceed as follows.
As before, we choose a smooth displacement function $F$ with compact support $F:\mathbb{C}\to\mathbb{C}$.
For each $r\in(0,1]$ (or simply a decreasing sequence of $(r_n)_{n \in \Bbb{N}}$ converging to $0$), we associate an arbitrary critical triangulation 
of the plane $\uT^{r}_{\mathrm{cr}}=\uT^{r}_{0}$ with isoradius $r$.
Finally we choose a finite bound $\epsilon_\mathrm{b}'$ such that
\begin{equation}
\label{epsbprimedef}
0<\epsilon_\mathrm{b}'<{1\over 2}\,\epsilon_{\mathtt{max}}(1)
\end{equation}
for the deformation parameter $\epsilon$ where $\epsilon_{\mathtt{max}}(1)$ is given by \ref{epsmax} above.
The calculations leading to the bounds of Prop.~\ref{propboundswab} for the deformation $\uT_0\to\uT_\epsilon$ can be easily repeated for the double deformations $\uT^r_0\to\uT^r_\epsilon\to\uT^r_{\epsilon,0}$.
In particular, the circumradius of each face $\uf$ of $\uT^{r}_{\epsilon,0}$ is bounded uniformly by
\begin{equation}
\label{RR0epsbound}
\epsilon\le \epsilon_\mathrm{b}\ ,\ r\le1\ \implies\ | R(\uf)-r | \le \epsilon\,r\, P_0(F;2\,\epsilon'_\mathrm{b})
\end{equation}
with $P_0$ defined in Prop.~\ref{propboundswab}.
This allows us to uniformly control the $r\to 0$ limit of the discrete derivatives $\nabla$ and $\overline\nabla$ by using Lemma~\ref{lemmabound} combined with the previous ingredients. 

\begin{Prop}
\label{PropLimDpr}
Let $F$ be a smooth displacement function with compact support, fix $\epsilon$, and let $\mathcal{F}=\{\uT^{r}_0\}$ be a  family of critical triangulations as above. 
To each point $z\in\Bbb{C}$ and to each $r$ we associate the face $\uf^r_\epszero(z)$ of the deformed triangulation $\uT^{r}_\epszero$ which contains $z$. 
Note that the set of $z$ which are either vertices or else belong to an edge of the triangulation is a set of measure zero and can be ignored.
Then in the $r\to 0$ limit, the discrete derivative operators $\nabla$ and $\overline\nabla$ for the face  $\uf^r_\epszero(z)$ converge uniformly towards the continuum partial derivative $\partial$ and $\bar\partial$ at the point $z$.
More precisely let $\phi$ be a \emph{smooth function} (or at least of class $C^2$) with compact support $\Omega$ of the plane. Then 
\begin{equation}
\label{ }
\lim_{r\to 0}\nabla\phi(\uf^r_\epszero(z))\ =\ \partial\phi(z)
\ ,\quad
\lim_{r\to 0}\overline\nabla\phi(\uf^r_\epszero(z))\ =\ \bar\partial\phi(z)
\end{equation} 
Moreover, the limit is uniform: Namely there is a constant $\mathtt{C}$ independent of $z\in \Omega$, the choice of the family $\mathcal{F}$ of triangulations, and the value of $\epsilon\in[0,\epsilon'_\mathrm{b}]$ (but still depending on $F$, on $\epsilon'_\mathrm{b}$ and on $\phi$), such that
\begin{equation}
\label{boundProp13}
\left | \nabla\phi(\uf^r_\epszero(z))-\partial\phi(z) \right | \ \text{and}\ \ \left | \overline\nabla\phi(\uf^r_\epszero(z))-\bar \partial\phi(z) \right |\ \le\ \mathtt{C}\ r 
\end{equation}
\end{Prop}
\begin{proof}
Let us apply the bound \ref{eqRobustLemma} obtained in Rem.~\ref{robust-lemma} in the proof of Lemma \ref{lemmabound} in Appendix~\ref{prooflemmabound} to the face $\uf^r_\epszero(z)$, to get
\begin{equation}
\Big| \nabla \phi (\uf^r_\epszero(z) )   - \partial \phi (z) \Big| \leq R(\uf^r_\epszero(z)) \, 
\Bigg(  \, {5 \over 2} \, \sup_{z \in \Omega} \big| \partial^2 \phi  \big| \, + \, 3 \,  \sup_{z \in \Omega} \big| \partial \overline{\partial} \phi \big| \, + \,  {1 \over 2} \, \sup_{z \in \Omega}
\big| \overline{\partial}^2  \phi \big| \, \Bigg)
\end{equation}
We then use \ref{RR0epsbound} to bound uniformly the circumradius of $\uf^r_\epszero(z)$  by
\begin{equation}
\label{ }
R(\uf^r_\epszero(z))\le r \left(1+ \epsilon\, P_0(F; 2\,\epsilon'_\mathrm{b})\right)
\end{equation}
This leads to the bound \ref{boundProp13}. The same argument applies to $\overline\nabla\phi$.
\end{proof}
\color{black}
It follows that the full variation of the discrete Laplace-Beltrami operator $\delta\Delta(\epsilon)=\Delta(\epsilon)-\Delta$ converges \emph{uniformly} towards a local Laplace-like operator which depend on $\epsilon$ and $F$, in the following sense.
\begin{Prop}
\label{propboundDelta}
Let $F$, $\epsilon$  and $\mathcal{F}=\{\uT^{r}_0\}$ as in Prop.~\ref{PropLimDpr} and $\phi$ be  a \emph{smooth function} (or at least of class $C^2$) with compact support $\Omega$ of the plane. 
Then
\begin{equation}
\label{ }
\phi\cdot\delta\Delta(\epsilon)\cdot\phi=\sum_{\uu,\uv\in{\uT}^{r}_0} \bar\phi(\uu)\, \left(\delta\Delta(\epsilon)\right)_{\uu\uv}\, \phi(\uv)
\end{equation}
converges uniformly when $r\to 0$ towards the local quadratic form
\begin{equation}
\label{continuousintegral}
\int_\Omega d^2 z\ 
\begin{pmatrix}
      \partial\phi    \\
      \bar\partial\phi  
\end{pmatrix}^\dagger\,
\text{\raisebox{-.75ex}{\Huge{$\mathbb{E}$}}}(\epsilon;F)\,
\begin{pmatrix}
      \partial\phi    \\
      \bar\partial\phi  
\end{pmatrix}
\end{equation}
with {\Large{$\mathbb{E}$}}$(\epsilon;F)$ the $2\times 2$ matrix
\begin{equation}
\label{bigEbigEp}
\begin{split}
&\text{\raisebox{-.5ex}{\huge{$\mathbb{E}$}}}(\epsilon;F)\ =\ \int_0^\epsilon \!d\varepsilon\  \text{\raisebox{-.5ex}{\huge{$\mathbb{E}$}}}'(\varepsilon;F)\quad\text{with}\\
&\text{\raisebox{-.5ex}{\huge{$\mathbb{E}$}}}'(\varepsilon;F)\ =\ {-4\over ((1+\varepsilon\,\partial F)(1+\varepsilon\,\bar\partial \bar F) -\varepsilon^2\, \bar\partial F\,\partial\bar F)^2}\ \times \\
&\ \begin{pmatrix}
 -\varepsilon \partial \bar F\,\bar\partial F \left( 2+\varepsilon (\partial F +\bar\partial\,\bar F)\right)     & 
 \partial\bar F\left({(1+\varepsilon\,\partial F)}^2-\varepsilon^2 \bar\partial F\,\partial\bar F\right)   \\
 \bar\partial F\left({(1+\varepsilon\,\bar\partial\, \bar F)}^2-\varepsilon^2 \partial \bar F\,\bar\partial F\right)     &
 -\varepsilon \partial \bar F\,\bar\partial F \left( 2+\varepsilon (\partial F +\bar\partial\,\bar F)\right)  
\end{pmatrix}
\end{split}
\end{equation}
\end{Prop}
\begin{proof}
One just writes $\delta\Delta(\epsilon)$ as
$$\delta\Delta(\epsilon)=\delta\Delta(\epsilon)=\int_0^\epsilon \!d\varepsilon\, \Delta'(\varepsilon)$$
and use the explicit representation \ref{LBprimeEps0} \ref{TheBigDMatrix} for $\Delta'(\varepsilon)$ to write
\begin{equation}
\label{ }
\phi\cdot\Delta'(\varepsilon)\cdot\phi= \sum_{\uf\in\uT^r_{\varepsilon: {\scriptscriptstyle 0}}} A(\uf)\, \begin{pmatrix}
      \nabla\phi (\uf)   \\
      \overline\nabla\phi (\uf) 
\end{pmatrix}^\dagger \cdot
\left[\raisebox{-.75ex}{\text{{\Huge{$\mathbb{D}$}}}}
(\varepsilon;F)\right](\uf)\cdot
\begin{pmatrix}
      \nabla\phi (\uf)   \\
      \overline\nabla\phi (\uf) 
\end{pmatrix}
\end{equation}
which is a Riemann sum. 
Then \ref{RR0epsbound} and Prop.~\ref{PropLimDpr} ensures that in the $r\to 0$ limit  this converges uniformly towards an ordinary integral involving continuous derivatives of $\phi$ and $F$ ($\mathbb{D}$ becoming $\mathbb{E}$). One thus recover \ref{continuousintegral}.
\end{proof}

\subsection{Scaling limit for the bi-local  deformation term for $\Delta$}\ 
\label{ssSLbilocDelta}

These arguments can be repeated for studying the scaling limit $\ell\to\infty$ of the bi-local term
\begin{equation}
\label{ }
\tr\left[\delta_1\Delta(\epsilon_1)\cdot\Delta_\mathrm{cr}^{-1}\cdot\delta_2\Delta(\epsilon_2)\cdot\Delta_\mathrm{cr}^{-1}\right]
\end{equation} 
for finite deformation parameters $\epsilon_{1}$ and $\epsilon_{2}$.
Again we consider two smooth deformation functions $F_1$ and $F_2$ with disjoint compact supports $\Omega_1$ and $\Omega_2$.
$\delta_1\Delta(\epsilon_1)=\Delta(\epsilon_1)-\Delta_\mathrm{cr}$ (resp. $\delta_2\Delta(\epsilon_2)=\Delta(\epsilon_2)-\Delta_\mathrm{cr}$) is the variation of the Laplace-Beltrami operator under the deformation 
$z\to z+\epsilon_1\,F_1(z)$ (resp. $z\to z+\epsilon_2\,F_2(z)$).
As above, instead of considering a fixed initial critical lattice $\uT_\mathrm{cr}$ with isoradius $R_0=1$, and rescaled deformation functions $F_\ell(z)=\ell\,F(z/\ell)$, with $\ell\to\infty$ a rescaling parameter, we consider a family $\mathcal{F}=\{\uT^r\}$ of critical lattices with isodradii $r$, fixed deformation functions $F$'s, and study the limit $r\to 0$. This is equivalent since by a change of variable  $r\sim 1/\ell$.

For a finite $0<r\le 1$, deforming the initial $\uT^r_\mathrm{cr}$ critical lattice, the bi-local deformation term reads as a double sum over the faces of the two non-isoradial lattices $\uT_{\epsilon_1: {\scriptscriptstyle 0}}^r$ and $\uT_{\epsilon_2:{\scriptscriptstyle 0}}^r$, which share the same vertices, but not the same faces, 
with $\uT^r_\mathrm{cr}$, of the explicit form
\begin{equation}
\label{bilocalepsfinite}
\begin{split}
&\Tr\left[\Delta'(\epsilon_1)\cdot\Delta_\mathrm{cr}^{-1}\cdot\Delta'(\epsilon_2)\cdot\Delta_\mathrm{cr}^{-1}\right] \ = 
\sum_{\uf_1\in\uT_{\epsilon_1: {\scriptscriptstyle 0}}^r}\  \sum_{\uf_2\in\uT_{\epsilon_2: {\scriptscriptstyle 0}}^r}
A(\uf_1)
\ A(\uf_2)\\ 
&\tr\left(
{
\left[\raisebox{-.6ex}{\text{\huge{$\mathbb{D}$}}}(\epsilon_1;F_1)\right](\uf_1)
\cdot
\left[
\begin{pmatrix}\nabla\\  \overline\nabla \end{pmatrix}
\Delta_\mathrm{cr}^{-1}
\begin{pmatrix}\nabla\\  \overline\nabla \end{pmatrix}^{\!\dagger}
\right]
_{\uf_1 \uf_2}
\hskip -1.em\cdot
\left[\raisebox{-.6ex}{\text{\huge{$\mathbb{D}$}}}(\epsilon_2;F_2)\right](\uf_2)
\cdot
\left[
\begin{pmatrix}\nabla\\  \overline\nabla \end{pmatrix}
\Delta_\mathrm{cr}^{-1}
\begin{pmatrix}\nabla\\  \overline\nabla \end{pmatrix}^{\!\dagger}
\right]
_{\uf_2 \uf_1}
}\right)
\end{split}
\end{equation}
The trace $\Tr\,[\ ]$ in the l.h.s. of \ref{bilocalepsfinite} is the ``big trace'' over the infinite set of vertices of the critical lattice.
The trace $\tr\,(\ )$ in the r.h.s of \ref{bilocalepsfinite} is a finite trace over a product of $2\times 2$ matrices.
This appears again as a double Riemann discrete sum over the faces of the triangulations $\uT_{\epsilon_1: {\scriptscriptstyle 0}}^r$ and 
$\uT_{\epsilon_2: {\scriptscriptstyle 0}}^r$. 

Studying the scaling limit $r\to 0$ might seem similar to what was done above for $\Delta$.
There is, however, a delicate point.
$\Delta^{-1}_\mathrm{cr}$ is the critical propagator on the critical lattice $\uT^r_\mathrm{cr}$, given by the explicit Kenyon integral representation. But its elements $\left[\Delta^{-1}_\mathrm{cr}\right]_{\uu,\uv}$ are not given by the restriction of a smooth function of the 
positions of the vertices $G(z(\uu),z(\uv))$.

Indeed, the large distance asymptotics of $\Delta^{-1}_\mathrm{cr}$ on a critical lattice with isoradius $R_0=1$ given by Prop.~\ref{pGasymptotics} implies that the propagator $\Delta^{-1}_\mathrm{cr}$ on a lattice $\uT^r_\mathrm{cr}$ can be separated in a dominant smooth part $G_\mathrm{D}$ and a subdominant non-smooth part $G_\mathrm{SD}$.
\begin{equation}
\label{ }
\left[\Delta^{-1}_\mathrm{cr}\right]_{\uu,\uv}=G_\mathrm{D}(\uu,\uv)+G_\mathrm{SD}(\uu,\uv)
\end{equation}
The dominant smooth part is the continuum propagator (note now the $r$ dependence)
\begin{equation}
\label{Gsmooth}
G_\mathrm{D}(\uu,\uv)=-{1\over 2 \pi}\Big( \log \big( 2 \left | z(\uu)-z(\uv) \right|/ r \big)  +\gamma_\mathrm{euler}\Big)
\end{equation}
The subdominant non-smooth part is 

\begin{equation}
\label{Gnonsmooth}
G_\mathrm{SD}(\uu,\uv)={1\over 2 \pi}\left(
\sum_{m\ge d\ge 1} (-1)^d (2m+d-1)!\, \mathfrak{Re}\left( c_{m,d}(\uu,\uv) \left({r/2\over z(\uv)-z(\uu)}\right)^{2 m}\right)
\right)
\end{equation}
with the coefficients $c_{m,d}(\uu,\uv)$ defined by \ref{def-c-md}.
Note that now $p_1(\uu,\uv)=(z(\uv)-z(\uu))/r$.
From Lemma~\ref{lemma-p-estimate} the $c_{m,d}$'s are of order O(1) irrespective of $(\uu,\uv)$, so the sum of the terms given by a fixed $m>0$ is bounded by a $\mathrm{O}(r^{2m})$ in the scaling $r\to 0$ limit, and is indeed subdominant.

In the scaling limit $r\to 0$ the sum over triangles in equation \ref{bilocalepsfinite} becomes a Riemann integral.

\begin{equation}
\label{limMeasure}
\sum_{\uf_1\in\uT_{\epsilon_1: {\scriptscriptstyle 0}}^r}\  \sum_{\uf_2\in\uT_{\epsilon_2: {\scriptscriptstyle 0}}^r} A(\uf_1)\ A(\uf_2)
\quad \longrightarrow\quad 
\int_{\Omega_1}\! d^2 z_1\  \int_{\Omega_2}\! d^2 z_2
\end{equation}
The $\mathbb{D}(\epsilon_a;F_a)(\uf_a)$ ($a=1,2$) in the r.h.s. of Eqn.~\ref{bilocalepsfinite} are easy to control since they converge uniformly to  $\mathbb{E}'(\epsilon_a;F_a)(z_a)$ given by \ref{bigEbigEp}.
Controling the scaling limit of the discrete derivatives of the smooth part of the propagator is also easy by means of Lemma~\ref{lemmabound}. 
We get the uniform limit
\begin{equation}
\label{DderivSmooth}
\left[
\begin{pmatrix}\nabla\\  \overline\nabla \end{pmatrix}
G_\mathrm{s}
\begin{pmatrix}\nabla\\  \overline\nabla \end{pmatrix}^{\!\dagger}
\right]
_{\uf_1 \uf_2}
\quad \xrightarrow[\ r\to 0\ ]{}  \quad 
-{1\over 4 \pi} 
\begin{pmatrix}
   0   &  (z_1-z_2)^{-2}  \\
   (\bar z_1-\bar z_2)^{-2}   & 0 
\end{pmatrix}
\end{equation}
The non-trivial point is to get a uniform bound on the scaling limit of the left+right discrete derivatives of the non-smooth part of the propagator, and to show that it is subdominant. 
This issue has been discussed in detail in Sect.~\ref{subsection-second-order} through Lemmas~\ref{lnablabound} and \ref{lNGN}.
However Lemma~\ref{lnablabound} relies on the fact that the discrete derivatives $\nabla$ and $\overline\nabla$ are relative to the faces $\uf$ of an isoradial triangulation $\uT_0$. This is not the case anymore here, since  the discrete derivatives are relative to the faces of a non-isoradial triangulation $\uT^r_\epszero$ derived from an isoradial one $\uT^r_{0}$ by flips of edges, without moving the position of the vertices.

We can repeat the analysis of Sect.~\ref{subsection-second-order} for this more general case. 
The dangerous contribution which could give a term of order $|z_1-z_2|^{-2}$ is the $m=1$ term in \ref{Gnonsmooth}, which is explicitly proportional to the real part of
\begin{equation*}
\label{ }
{p_3(\uu,\uv)\  r^3\over (z(\uu)-z(\uv))^3}
\end{equation*}
The most dangerous contribution comes from applying left+right discrete derivatives to $p_3(\uu,\uv))$. Generically a naive dimensional analysis shows that each discrete derivative applied on $p_3$ will bring a term of order $r^{-1}$, so that we will get for a pair of triangles
$\uf_1\in\uT^r_{\epsilon_1: {\scriptscriptstyle 0}}\cap\Omega_1$, $\uf_2\in\uT^r_{\epsilon_2: {\scriptscriptstyle 0}}\cap\Omega_2$
\begin{equation*}
\label{ }
\sum_{\uu_1\in\uf_1}\sum_{\uu_2\in\uf_2}
\begin{pmatrix}\nabla\\  \overline\nabla \end{pmatrix}_{\uf_1,\uu_1}
p_3(\uu_1,\uu_2)
\begin{pmatrix}\nabla\\  \overline\nabla \end{pmatrix}^{\!\dagger}_{\uu_2,\uf_1}
\quad\sim\quad 
\mathtt{cst.}\ {r^{-2}}
\end{equation*}
However, we shall see that this estimate is generically \emph{not uniform}. Namely, the $\mathtt{cst.}$ in this estimate can be arbitrarily large !
One should remember that from Lemma~\ref{lnablabound} if $\uf_1$ and $\uf_2$ are faces of the original isoradial triangulation $\uT_0$ then this $\mathtt{cst.}$ is bounded by $\mathtt{cst.}\le 9$.

This is a technical point which comes from the fact that generically, if we start from an isoradial Delaunay triangulation $\uT_0$ with isoradius $r$, and consider an arbitrary triangle $\ut=(\uu_1,\uu_2,\uu_3)$ which is not a face $\uf$ of $\uT_0$, this triangle may have a circumradius $R(\ut)$ very large ($R(\ut)\gg r$), and an area $A(\ut)$ arbitrarily small ($A(\ut)\ll r^{2}$).  
``Experimental mathematics'' studies of such singular cases and some analytical estimates lead us to the following conjecture.

\begin{Conjecture}
\label{Conf3bound}
Let $\uT_0^r$
be an isoradial Delaunay triangulation of the plane with isoradius $r$ and let $p_3(\uu,\uv)$ be the function defined by 
$$p_3(\uu,\uv)=\sum_{j=1}^{2n} \mathrm{e}^{3 \imath \theta_j}\quad,\qquad \theta_j=\arg \big( z(\uv_j) - z(\uv_{j-1} ) \big)$$
for any pair of vertices $(\uu,\uv)$ of $\uT_0^r$ 
where $\mathbbm{v}=(\uv_0, \cdots, \uv_k)$ is a path in the rhombic lattice ${\uT_0^r}^\lozenge$ going from 
$\uv_0= \uu$ to $\uv_k = \uv $ (see Def.~\ref{def-p_n} and \ref{stuff2}).

For any non-degenerate triangle $\ut=(\uu_1,\uu_2,\uu_3)$ in $\uT_0^r$ (not necessarily a face, as illustrated in Fig.~\ref{FigConjecture1}), let $\nabla p_3(\ut)$ and $\overline\nabla p_3(\ut)$ be the discrete derivatives 
of the function $\uu \mapsto p_3(\uu, \uv)$ evaluated at the triangle $\ut$, where the vertex $\uv$ is fixed, according to Defs.~\ref{nablaDef} and \ref{barnablaDef}.

Then there is a uniform bound
\begin{equation}
\label{ }
|\nabla p_3(\ut)| \quad\text{and}\quad |\overline\nabla p_3(\ut)| \ \le\ \mathtt{cst.}\, R(\ut)/r^2
\end{equation}
where the circumradius $R(\ut)$ of the triangle $\ut$ given by formula \ref{Rcform},
and $\mathtt{cst.}$ is a number of order $\mathrm{O}(1)$ independent on the choice of critical triangulation $\uT_0^r$ and of the triangle $\ut$.
Among the examples we have studied, we found $\mathtt{cst.}=6$.
\end{Conjecture}
\begin{figure}[h!]
\begin{center}
\includegraphics[width=3.5in]{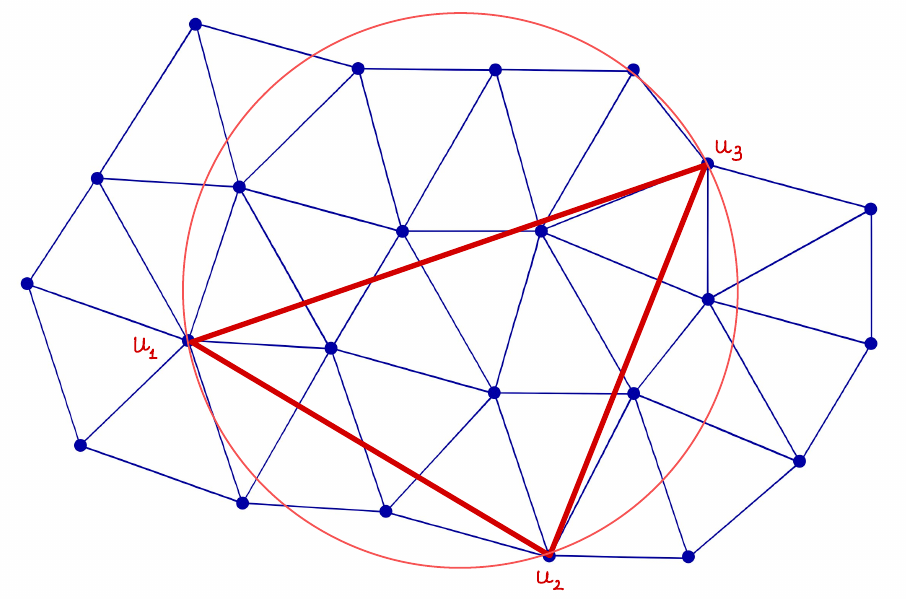}
\caption{Example of a triangle $\ut=(\uu_1,\uu_2,\uu_3)$ in an isoradial graph, with its circumcircle, as considered in Conjecture \ref{Conf3bound}.}
\label{FigConjecture1}
\end{center}
\end{figure}

Assuming the validity of the conjecture, it is easy to adapt the arguments of Sect.~\ref{subsection-second-order}, and to use that fact that the circumradii of the faces $\uf_1$ and $\uf_2$ of the deformed-back-deformed non-isoradial triangulations $\uT_{\epsilon_1: {\scriptscriptstyle 0}}^r$ and $\uT_{\epsilon_2: {\scriptscriptstyle 0}}^r$ are uniformly bounded for $\epsilon_1$ and $\epsilon_2$ small enough by \ref{RR0epsbound}. This leads to

\begin{lemma}
\label{ }
Assuming Conjecture~\ref{Conf3bound}, the left-right discrete derivative of the non-smooth part of the propagator is \emph{uniformly bounded} in the scaling limit $r\to 0$ by
\begin{equation}
\label{ }
\left |
\left[
\begin{pmatrix}\nabla\\  \overline\nabla \end{pmatrix}
G_\mathrm{ns}
\begin{pmatrix}\nabla\\  \overline\nabla \end{pmatrix}^{\!\dagger}
\right]
_{\uf_1 \uf_2}
\right |
\quad\le\quad
\mathtt{cst.}\ {r\over {\left| z(\uf_1)-z(\uf_2) \right|^3}}
\end{equation}
It is therefore subdominant when compared to the contribution of the smooth part of the propagator given by \ref{DderivSmooth}.
\end{lemma}

Combining the previous results, we can state the following proposition about the existence of the scaling limit of the bi-local term
\begin{Prop}
\label{PropLimTr2}
Assuming Conjecture~\ref{Conf3bound}, the bi-local term 
$\Tr\left[\Delta'(\epsilon_1)\cdot\Delta_\mathrm{cr}^{-1}\cdot\Delta'(\epsilon_2)\cdot\Delta_\mathrm{cr}^{-1}\right]$ defined on critical triangulations $\uT_0^r$ converges uniformly in the scaling limit $r\to 0$ towards the bi-local term
\begin{equation}
\label{ }
\begin{split}
\int_{\Omega_1} \! d^2 z_1 \int_{\Omega_2} \! d^2 z_2\  
&\tr\left[\mathbb{E}'(\epsilon_1;F_1)(z_1)\cdot 
\begin{pmatrix}
   0   &  (z_1-z_2)^{-2}  \\
   (\bar z_1-\bar z_2)^{-2}   & 0 
\end{pmatrix}\cdot
\right.
\\
&\left.\hskip 2.em
{\mathbb{E}'(\epsilon_2;F_2)(z_2)\cdot 
\begin{pmatrix}
   0   &  (z_1-z_2)^{-2}  \\
   (\bar z_1-\bar z_2)^{-2}   & 0 
\end{pmatrix}
}\right]
\end{split}
\end{equation}
\end{Prop}
Note that this term depends on the fours derivatives $\partial F_1$, $\bar\partial F_1$,  $\partial F_2$, $\bar\partial F_2$ and their c.c., and contains both the analytic term $(z_1-z_2)^{-4}$, the anti-analytic term $(\bar z_1-\bar z_2)^{-4}$, and the mixed term $(z_1-z_2)^{-2}(\bar z_1-\bar z_2)^{-2}$.

Finally, from the explicit expression \ref{bigEbigEp}, the limit $\epsilon\to 0$ of $\mathbb{E}'(\epsilon; F)$ exists and is uniform.
\begin{equation}
\label{ }
\lim_{\epsilon\to 0} \mathbb{E}'(\epsilon; F)\ =\ \begin{pmatrix}
    0  &  -4\,\partial \bar F  \\
   -4\,\bar\partial F   &  0
\end{pmatrix}
\end{equation}
Together with Prop.~\ref{PropLimTr2}, this leads to the commutation of limits result for $\Delta$.
\begin{Prop}
\label{prvsepsComm}
Assuming Conjecture~\ref{Conf3bound},
the limit $\epsilon\to 0$ and the scaling limit $r\to 0$ for the bi-local term exist, are uniform, and commute. One recovers the result obtained previously for the scaling limit of the OPE on the lattice for $\Delta$.
\begin{equation}
\label{ }
\begin{split}
&\lim_{\epsilon\to 0}\lim_{r\to 0}\Tr\left[\Delta'(\epsilon_1)\cdot\Delta_\mathrm{cr}^{-1}\cdot\Delta'(\epsilon_2)\cdot\Delta_\mathrm{cr}^{-1}\right]
\\
&\qquad=\ \lim_{r\to 0}\lim_{\epsilon\to 0}\Tr\left[\Delta'(\epsilon_1)\cdot\Delta_\mathrm{cr}^{-1}\cdot\Delta'(\epsilon_2)\cdot\Delta_\mathrm{cr}^{-1}\right]
\\
&\qquad\qquad
=\ {1\over \pi^2} \int_{\Omega_1} d^2 z_1 \int_{\Omega_2} d^2 z_2\ 
{\bar\partial F_1(z_1)\ \bar\partial F_2(z_2)\over (z_1-z_2)^4}
\ +\ 
{\partial \bar F_1(z_1)\ \partial \bar F_2(z_2)\over (\bar z_1-\bar z_2)^4}
\end{split}
\end{equation}
\end{Prop}

The Conjecture~\ref{ConjectureLimTr2D} in the introduction is a special case of Proposition~\ref{prvsepsComm}.
We simply repeat the arguments given at the beginning of Section~\ref{ssSLbilocDelta}, which show that one can equivalently define the scaling limit by choosing a given deformation function $F$, and letting the isoradius $R_\mathrm{cr}$ of the critical graphs $\uG_\mathrm{cr}$ go to zero $R_\mathrm{cr}\to 0$, or fixing the isoradius $R_\mathrm{cr}$ of the critical graph $\uG_\mathrm{cr}$, but then instead introducing the rescaled displacement function $F_\ell$ and letting $\ell$ go to infinity $\ell\to\infty$. 

\subsection{About the scaling limit of the K\"ahler operator $\mathcal{D}$}
\label{ssScalingKahler}
\newcommand{\bigbbF}{\raisebox{-.6ex}{\text{\huge{$\mathbb{F}$}}}}
We now discuss briefly the deformations of the K\"ahler operator, without giving details of the calculations.
In the block matrix representation, the K\"ahler operator $\mathcal{D}$ and its $\epsilon$-derivative read
\begin{equation}
\label{KahlerBlock}
\mathcal{D}(\epsilon) = 
4 \begin{pmatrix}
      \nabla_{\!\epsilon}   \\
      \overline\nabla_{\!\epsilon}  
\end{pmatrix}^\dagger
\begin{pmatrix}
     A_\epsilon/R_\epsilon^2 &   0  \\
   0   &  0
\end{pmatrix}
\begin{pmatrix}
      \nabla_{\!\epsilon}    \\
      \overline\nabla_{\!\epsilon}  
\end{pmatrix}
\ ,\quad
\mathcal{D}'(\epsilon) = 
 \begin{pmatrix}
      \nabla_{\!\epsilon}    \\
      \overline\nabla_{\!\epsilon}  
\end{pmatrix}^\dagger
\begin{pmatrix}
     A_\epsilon\, \textswab{R}_\epsilon &   A_\epsilon\,\overline{\textswab{H}}_\epsilon  \\
   A_\epsilon\,\textswab{H}_\epsilon   &  0
\end{pmatrix}
\begin{pmatrix}
      \nabla_{\!\epsilon}    \\
      \overline\nabla_{\!\epsilon}  
\end{pmatrix}
\end{equation}
with $A_\epsilon$ and $R_\epsilon$ the areas and circumradii of the faces of the deformed lattice $\uT_\epsilon$, while $\textswab{R}_\epsilon$ and $\textswab{H}_\epsilon$ are given by  \ref{swabf1} and \ref{C123def}.
In order to study $\mathcal{D}(\epsilon)$ at finite epsilon and to compare it to  $\mathcal{D}(0)=\mathcal{D}_\mathrm{cr}$, and its scaling limit, one can try to repeat the argument for $\Delta$ presented in the previous section.
It is enough to consider $\mathcal{D}'(\epsilon)$.
We start from a critical lattice $\uT_{0}^r$ with isoradius $r$, perform the deformation $z\to z+\epsilon F(z)$, and reexpress $\mathcal{D}'(\epsilon)$, defined on the deformed Delaunay lattice $\uT_\epsilon^r$, on the back-deformed lattice $\uT_\epszero^r$.
We can thus rewrite $\mathcal{D}'(\epsilon)$ under a block form similar to \ref{LBprimeEps0}

\begin{equation}
\label{KprimeEps0}
\mathcal{D}'(\epsilon)= \begin{pmatrix}
      \nabla    \\
      \overline\nabla  
\end{pmatrix}^\dagger
A\cdot
\bigbbF'\! (\epsilon;F)
\begin{pmatrix}
      \nabla     \\
      \overline\nabla  
\end{pmatrix}
\end{equation}
with $\mathbb{F}'(\epsilon;F)$ a $2\times 2$ block matrix made of diagonal matrices relative to the faces $\uf$ of $\uT^r_\epszero$, defined implicitly by \ref{KprimeEps0}.
The $2\times 2$ matrix extracted of $\mathbb{F}'$ relative to a face $\uf$, $\left[\mathbb{F}'(\epsilon;F)\right](\uf)$, can be computed explicitly out of the $\nabla F(\uf)$ and $\overline\nabla F(\uf)$, and of the geometry of the face $\uf$, but the result will be quite long and not very illuminating at this stage.
The difference with the previous case of $\Delta$ is that for a face $\uf$ (let us denote its  vertices $(123)$) $\mathbb{F}'$  will depend explicitly of the circumradius $R(\uf)$ of the face, and of the phases $C_\ue$ associated to the unoriented edges $\ue=(12)$, $(23)$ and $(31)$ of $\uf$, defined by \ref{Cuv}. Indeed the coefficient $\textswab{H}(\uf)$ depends explicitly of $R(\uf)$, and the coefficient $\textswab{R}(\uf)$ depends also of the coefficients $C(\uf)=\sum_{\ue\in\uf}C_\ue$.
Moreover the variation of these coefficients under the back-deformation $\uT_\epszero\leftrightarrow\uT_\epsilon$ depends also of these $C_\ue$.

We can now use  Pro.~\ref{propboundswab} which bound the $\textswab{H}(\uf)$ and $\textswab{R}(\uf)$ and $R(\uf)$, and the fact that since the $C_\ue$ are phases so that $|C_\ue|=1$, to bound uniformly the coefficients of the matrices $\left[\mathbb{F}'(\epsilon;F)\right](\uf)$'s w.r.t the deformation parameter $\epsilon$ (small enough) and the triangulations $\uT^r_0$. More precisely

\begin{Prop}
\label{BoundOnDprime}
Let $F$ be a displacement function, $\mathcal{F}=\{\uT^{r}_0:r\in(0,1]\}$ a family of critical triangulations labelled by their isoradius $r$, and $\epsilon\in(0,\epsilon_\mathrm{b}']$ with $\epsilon_\mathrm{b}'$ defined by \ref{epsbprimedef}.
There is a constant $\mathtt{Cst.}$ which depends only on $F$ and the choice of $\epsilon_\mathrm{b}'$ such that there is a uniform bound for the matrix elements of the $\left[\mathbb{F}'(\epsilon;F)\right](\uf)$ matrices
\begin{equation}
\label{ }
\parallel  \left[\mathbb{F}'(\epsilon;F) \right] (\uf) \parallel \ \le\ \mathtt{Cst.} \ r^{-2}
\end{equation}
with $\parallel\cdot\parallel$ the standard operator norm on matrices (for instance).
\end{Prop}
\begin{proof}
The proof relies on writing explicitly the matrix $\mathbb{F}'$. 
This is lengthy but not difficult. 
Note that the $r^{-2}$ factor, where $r$ is the isoradius of the initial lattice $\uT_0^r$, comes from the $A/R_0^2$ in the initial definition of $\mathcal{D}$ \ref{KahlerBlock}.
\end{proof}

If we look now at the limit $r\to 0$, keeping $\epsilon$ fixed, denoting as in Prop.~\ref{PropLimDpr} the face of $\uT_\epszero^r$ which contains the point $z$ by $\uf^r_\epszero(z)$, there is for a generic family $\mathcal{F}=\{\uT_\epszero^r\}$ no reason that the ratio
$\bar R(\uf^r_\epszero(z))=R(\uf^r_\epszero(z))/r$ and the coefficients $C_\ue(\uf^r_\epszero(z))$ and $C(\uf^r_\epszero(z))$ converge towards fixed values $\bar R(z;\epsilon)$, $C(z;\epsilon)$, $C_\ue(z;\epsilon)$ in the scaling limit $r\to 0$.
Indeed, these quantities depend explicitly on the detailed local geometrical structure of the lattices $\uT_0^r$ in the neighborhood of the point $z$, for each value of $r$. Only for some \emph{very specific sequences} of $\uT_0^r$, for instance iterative isoradial refinements of the initial lattice for $r=1$, can we expect strong correlations leading to the existence of a $r\to 0$ limit for these quantities.
We can therefore state:

\begin{Prop}\label{NoLimitDprime}
Under the hypothesis of Prop.~\ref{BoundOnDprime} the matrix $\mathbb{F}'(\epsilon;F)/r^2$ has generically no local scaling limit for $\epsilon$ finite when $r\to 0$.
\begin{equation}
\label{ }
\lim_{r\to 0} \left[\mathbb{F}'(\epsilon;F)\right](\uf^r_\epszero(z))/r^2\qquad\text{does not exist}
\end{equation}
Of course one must have $z\in\Omega=\mathtt{supp}(F)$, since otherwise this limit exists and is zero.
The same is obviously true for the non-existence of the $r\to 0$ limit of the bi-local term at finite $\epsilon_1$, $\epsilon_2$
\begin{equation}
\label{BilocKaehler}
\lim_{r\to 0} \Tr\left[\mathcal{D}'(\epsilon_1)\cdot\mathcal{D}_\mathrm{cr}^{-1}\cdot\mathcal{D}'(\epsilon_2)\cdot\mathcal{D}_\mathrm{cr}^{-1}\right]
\qquad\text{does not exist}
\end{equation}
\end{Prop}
Therefore, the existence of a scaling limit for $\mathcal{D}$ could make sense in a much more limited setting than for $\Delta$.
Remember that we want to compare (i) the limit $\epsilon\to 0$, which, for $\mathcal{D}'$ as well as for $\Delta'$, has the effect of keeping only the terms linear in $\nabla F$, $\overline \nabla F$ and their c.c.; (ii) the scaling limit $r\to 0$; which allows replacing the discrete derivatives $\nabla$, $\overline \nabla$ by continuous derivatives $\partial$ and $\bar\partial$, and in particular \ref{DderivSmooth}.
In fact the best result we obtain so far concern  the ``simultaneous limit" when $\epsilon$ and $r$ go to zero, and is stated in the following proposition.

\begin{Prop}
\label{PropLimTr2D}
Let $F$ be a displacement function, $\mathcal{F}=\{\uT^{r}_0:r\in(0,1]\}$ a family of critical triangulations labelled by their isoradius $r$, and $\epsilon_\mathrm{b}'$ defined by \ref{epsbprimedef}. We consider the ``simultaneous limit" where
\begin{equation}
\label{ }
r\to 0\quad,\qquad \epsilon_a=\epsilon(r)=r\,\mathtt{c}_a\,\quad\text{with}\quad 0\le\mathtt{c}_a\le \epsilon_\mathrm{b}'\quad\text{for}\ a=1,2
\end{equation}
Assuming the validity of Conjecture~\ref{Conf3bound}, the bi-local term of \ref{BilocKaehler} converges uniformly towards its continuum limit given in Th.~\ref{TfinalOPElike1}.
\begin{equation}
\label{ }
\begin{split}
\lim_{\substack{r\to 0\\ \epsilon_1/r=\mathtt{c}_1 \\ \epsilon_2/r=\mathtt{c}_2}} 
&\Tr\left[\mathcal{D}'(\epsilon_1)\cdot\mathcal{D}_\mathrm{cr}^{-1}\cdot\mathcal{D}'(\epsilon_2)\cdot\mathcal{D}_\mathrm{cr}^{-1}\right]\ =\ \\
& {1\over \pi^2}\int_{\Omega_1}\! d^2 z_1 \int_{\Omega_2}\! d^2 z_2\ \left(
{\bar\partial F_1(z_1)\,\bar\partial F_2(z_2)\over (z_1-z_2)^4} + {\partial \bar F_(z_1)\,\partial \bar F_2(z_2)\over (\bar z_1-\bar z_2)^4} \right)
\end{split}
\end{equation}

\end{Prop}


\section{Discussion and Perspectives}
\label{sDiscussion}

\subsection{The aim of the study}
\label{sDiscusAim}\ \\
In this work we study properties of the measure on planar graphs introduced by \cite{DavidEynard2014} in order to better understand the relationship between this discrete model and continuum models of random geometries on the plane arising from Conformal Field Theories (CFT), in particular Quantum Liouville Theory. 
The model is defined as an integral over the space of all Delaunay graphs of the plane. 
We do not study as a whole the global properties of this integral and its associated measure. 
Rather, we study the measure in the neighborhood of very specific graphs, namely isoradial Delaunay graphs. 
Our motivation is twofold 
: (i) isoradial graphs can be viewed as a discretization of flat geometry, so this should amount to some ``semiclassical limit''; (ii) deforming the geometry is a way to introduce a stress-energy tensor into the statistical model, whose properties are crucial for conformal theories.

The measure of the model is a K\"ahler measure (in fact equivalent to the Weil-Petersson measure) and its density can be written as the determinant of a Laplacian-like K\"ahler operator $\mathcal{D}$ (defined on the Delaunay graphs), with specific global conformal invariance properties under PSL($2,\mathbb{C}$) transformations. In order to compare our result with other cases, we study in parallel the K\"ahler operator $\mathcal{D}$, the ordinary discrete Laplace-Beltrami operator $\Delta$ (which is not PSL($2,\mathbb{C}$) invariant), and the conformal Laplacian $\Deltaconf$
which, like $\mathcal{D}$, also enjoys a global PSL($2,\mathbb{C}$) invariance property.
 
\subsection{The first-order variations and discretized CFT}
\label{sDiscus1stO}
\subsubsection{The Laplace-Beltrami operator $\Delta$}\ \\ 
The calculation for the first order variation for the discretized Laplace-Beltrami $\Delta$ is easy to discuss in the framework of discretized CFT on the lattice. We refer to Appendix~\ref{aCFTreminder} for a reminder of the definitions and properties of CFT which are needed in this discussion.
Our  result \ref{varTrLlLB2} in Prop.~\ref{T1stVar} states that

\begin{equation}
\label{varTrLlLB3}
\deltae \log\det(\Delta)=
- \!\!\sum_{\substack{\scriptstyle \mathrm{faces} \\ \ \ \ \uf \, \in \, \widehat{\uG}_{0^+}}}\!\! 
4\,A(\uf)\Big( \overline\nabla\! F(\uf) \,  Q(\uf) + \nabla\! \overline{F}(\uf) \,  \overline{Q}(\uf)    \Big) 
\end{equation} 
with
\begin{equation}
\label{Kexplicit}
Q(\uf) := \, [\nabla\Delta^{-1} \nabla^\top]_{\uf \uf} \, = \, \sum_{\uu, \uv} \, 
\nabla_{\! \uf \uu }\nabla_{\! \uf \uv}\,\-{[\Delta_\mathrm{cr}^{-1} ]}_{\uu \uv}
\end{equation}
\noindent
and where the operators $\nabla, \overline{\nabla}$ are defined in
formulae (\ref{nablaDef}) and (\ref{barnablaDef}) respectively.
Equation (\ref{varTrLlLB3}) can be read as the discretized version of the first order variation of the partition function 
under a diffeomorphism for a CFT (see \ref{1stVarLZcomplexA}) given by 
\begin{equation*}
\label{1stVarLZcomplex}
\begin{split}
\deltae \log(Z)=
&-{1 \over \pi}\int d^2x \, \Big( \, \overline{\partial}  F(x)\, \langle T(x)\rangle
 + \, \partial \overline{F}(x)\,  \langle \overline{T}(x) \rangle \Big)
\end{split} 
\end{equation*}
where the sum over faces discretizes the 
integral over the plane, and the derivatives $\nabla\bar F$ and $\overline{\nabla} F$ serve as discrete versions of $\partial\bar F$ and $\bar\partial F$. 

\begin{equation}
\label{Discr2Cont}
\sum_\uf 
 A(\uf)\ \leftrightarrow\ \int d^2x
 \ ,\quad
\nabla\bar F\ \leftrightarrow\ \partial \bar F
 \ ,\quad
\overline{\nabla} F\ \leftrightarrow\ \bar\partial F
\end{equation}
{The term $Q(\uf)$ is given by the vacuum expectation value (v.e.v.) 
\begin{equation}
\label{KtoTLB}
4\pi\,Q(\uf)=\langle T_{_{\!\Delta}}(\uf)\rangle
\end{equation}
of a discretized stress-energy tensor $T_{_{\!\Delta}}$
for a theory with Grassmann fields $(\Phi,\bar\Phi)$ attached to the vertices of the triangulation 
$\uG_\mathrm{cr}$ with discretized action $S$
\begin{equation}
\label{SLBDiscr}
S[\Phi,\bar\Phi]=\Phi{\cdot}\Delta\bar\Phi=\sum_{\stackrel{\scriptstyle \mathrm{vertices}}{\ \, \uu, \uv \, \in \, 
\uG_\mathrm{cr}}} \Phi_{\uu} \Delta_{\uu \uv} \bar\Phi_{\uv}
\end{equation}
and where
\begin{equation}
\label{TLBDiscr} 
T_{_{\!\Delta}}(\uf)=-4\pi \nabla\Phi(\uf) \nabla\bar\Phi(\uf)= - 4 \pi \sum_{\uu, \uv \, \in \, \uf} \nabla_{\! \uf \uu}\Phi_{\uu} \, \nabla_{\! \uf \uv}  \bar\Phi_{\uv}
\end{equation}
for a face (triangle) $\uf$ of the triangulation $\uG_\mathrm{cr}$.} 

Note that this definition \ref{TLBDiscr} for the discrete stress-energy tensor follows directly from \ref{SLBDiscr} and the variation of the discrete Laplace-Beltrami operator $\Delta$ given by Prop.~\ref{PVarDelta} and eq.~\ref{VarDelta}.

The above discussion is valid regardless of whether 
we consider the variation of the Laplace-Bletrami operator defined on an isoradial Delaunay graph $\uG_\mathrm{cr}$ 
or instead on a general Delaunay graph $\uG$. 
Indeed, \ref{TLBDiscr} follows from the general equation \ref{VarDelta} for the variation of $\Delta$ on generic triangulations. 
Note also that the absence of a $\nabla F+\overline\nabla\bar F$ term in the variation of $\Delta$ means $\mathrm{Tr}(\mathbf{T})=T^{z\bar z}=T^{\bar z z}$ is zero, and  that the discrete Laplace-Beltrami operator $\Delta$ has a discrete conformal invariance property.

The interesting result, relevant for the discussion here, is that for an isoradial Delaunay graph $\uG_\mathrm{cr}$ the term $Q(\uf)$, i.e. the v.e.v. of the discretized stress-energy tensor $T$, depends only on the local geometry of the graph,  i.e. on the shape of the triangle $\uf$, as stated in prop.~\ref{T1stVar}. This is not true when $\uG$ is not isoradial; in that case, $\langle T(\uf) \rangle$ will depend on the full geometry of the lattice.

\subsubsection{The K\"ahler operator  $\mathcal{D}$}
\ \\
The first order variation for the K\"ahler operator $\mathcal{D}$ is given by \ref{varTrLlLK1} in Prop.~\ref{T1stVarD}. 
The first term in \ref{varTrLlLK1} is the same as the first order variation for $\Delta$ in \ref{varTrLlLB1}, which is rewritten in \ref{varTrLlLB3} as a sum over the triangles of the lattice involving the discrete derivatives of the deformation $\overline\nabla F$ and $\nabla\bar F$.
The second term in  \ref{varTrLlLK1} involves the first order variation $\deltae R(\uf)$ of the circumradii $R(\uf,\epsilon)$ of a face, which can be obtained from \ref{var1A} and \ref{var1AR2}. 
The final result is
\begin{equation}
\label{var1D2}
\deltae \log \det (\mathcal{D}) =
- \sum_{\substack{\scriptstyle \mathrm{faces} \\ \ \ \uf \in \widehat{\uG}_{0^+}}}
\left(\Big(4\, A(\uf) Q(\uf)+{1\over 2} C(\uf)\Big) \overline\nabla\! F(\uf) +{1\over 2} \nabla\! F(\uf)\right) + \mathrm{c.c.}
\end{equation}
with the geometrical factor $C(\uf)$  for a triangle $\uf$ given by \ref{CDef}, 
while $Q(\uf)$ is given by \ref{Kexplicit}, and corresponds to the v.e.v. of the discretized stress energy tensor $T_{_{\!\Delta}}(\uf)$ defined by \ref{TLBDiscr} for the Laplace-Beltrami theory.

Like for the Laplace-Beltrami operator, the variation \ref{var1D2} can be written in term of a discretized stress-energy tensor $\mathbf{T}_{\scriptscriptstyle{\mathcal{D}}}$ for a theory with discretized action
\begin{equation}
\label{SKDiscr}
S_{\scriptscriptstyle{\mathcal{D}}}[\Phi,\bar\Phi]=\Phi{\cdot}\mathcal{D}\bar\Phi
\end{equation}

\begin{equation}
\label{var1LDKdiscr}
\begin{split}
\deltae \log\det(\mathcal{D})=\tr\left[\deltae \mathcal{D}\cdot\mathcal{D}^{-1}\right]=&-{1 \over\pi} \sum_{\uf} A(\uf) \left(\overline\nabla F(\uf)\, \langle T_{\!\scriptscriptstyle{\mathcal{D}}}(\uf)\rangle+ \nabla \bar F(\uf)\, \langle \bar T_{\!\scriptscriptstyle{\mathcal{D}}}(\uf)\rangle\right)\\
&+{1 \over 2}\sum_{\uf} A(\uf) \left(\nabla F(\uf)+\overline\nabla\bar F(\uf)\right)  \langle \tr \mathbf{T}_{\!\scriptscriptstyle{\mathcal{D}}}(\uf) \rangle 
\end{split}
\end{equation}
where the components of the discretized stress-energy tensor are
\begin{equation}
\label{TKdiscr}
\begin{split}
 T_{\!\scriptscriptstyle{\mathcal{D}}}&= -\, 4\pi\, {1\over R^2}\,\left( \nabla\Phi\,\nabla\bar\Phi+ C\, \overline\nabla \Phi\ \nabla\bar\Phi\right) \\
 \bar T_{\!\scriptscriptstyle{\mathcal{D}}}&= -\, 4\pi\, {1\over R^2}\,\left( \overline\nabla\Phi\,\overline\nabla\bar\Phi+ \bar C\, \overline\nabla \Phi\ \nabla\bar\Phi\right) \\
 \tr \mathbf{T}_{\!\scriptscriptstyle{\mathcal{D}}}&=\ \ \ \ \, 8 \,   {1\over R^2}\,\left( \overline\nabla\Phi\,\nabla\bar\Phi\right)
\end{split}
\end{equation}
One should note the nonzero term $(\overline\nabla F+\nabla\bar F)/2$  in \ref{var1D2} and the non-vanishing of the v.e.v. of the trace of a discrete stress-energy tensor $\tr\mathbf{T}_{_\mathcal{D}}$.
This follows from the fact that the dimension of the matrix elements of $\mathcal{D}$ is $\mathtt{length}^{-2}$.

The definition \ref{TKdiscr} and the variation formula \ref{var1LDKdiscr} remain valid if we replace the isoradial Delaunay graph $\uG_\mathrm{cr}$ by a generic Delaunay graph $\uG$. The additional term $C(\uf)$ in \ref{TKdiscr}, which depends explicitly on the local geometry of the graph in the neighborhood of the triangle $\uf$. 
This term cannot be written simply in the continuum limit $\ell\to\infty$ in terms of continuous derivatives $\partial$ and $\bar\partial$ of a ``smooth'' complex Grassmann field $\Phi(x)$ in the flat continuum plane $\mathbb{R}^2$. This implies that $\mathbf{T}_{_{\!\mathcal{D}}}$ has no direct interpretation in a continuum field theory setting, in contrast with $\mathbf{T}_{\!\scriptscriptstyle{\Delta}}$.
\color{black}

Again, the interesting explicit local form given in  Prop.~\ref{T1stVarD} and in Remark~\ref{rLDKvar} are only valid  for the variation of an isoradial Delaunay graph
$\uG_\mathrm{cr}$.

\subsubsection{The conformal Laplacian $\Deltaconf$}
\label{sssConfDeltaT}\ \\
The result given by Prop.~\ref{T1stVarC} for $\Deltaconf$ admits a similar interpretation. Again the absence of a 
$\nabla F+\bar\nabla\bar F$ term signals the conformal invariance of $\Deltaconf$, which in this case is ensured from start, before one takes the scaling limit. 
The first order variation can still be written as a sum over triangles, of the form
\begin{equation}
\label{var1DeltaConf2}
\deltae \log \det(\Deltaconf) =
- \sum_{\stackrel{\scriptstyle \mathrm{faces}}{\ \ \, \uf \in \widehat{\uG}_{0^+}}}  4\,A(\uf)
\Big( \overline\nabla\! F(\uf) Q_{_\mathrm{conf}}(\uf) + \mathrm{c.c.}\Big) 
\end{equation}
but now the local face term $Q_{_\mathrm{conf}}(\uf)$ differs from $Q(\uf)$ when one or several of the edges of 
the triangle $\uf$ are chords, owing to the additional terms in \ref{varTrLlC1}.
More precisely, the contribution for a chord can be separated into equal contributions for its adjacent ``north'' and ``south'' triangles, so that one writes
\begin{equation}
\label{ }
A(\uf)Q_{_\mathrm{conf}}(\uf)= A(\uf)Q(\uf)+H_{_{\mathrm{anom}}}(\uf)
\end{equation}
with the anomalous term $H_{_{\mathrm{anom}}}(\uf)$ for a (counter-clockwise oriented) face $\uf$ expressed as a sum over 
its (oriented) edges $\vec{\ue}$ which are chords
\begin{equation}
\label{ }
H_\mathrm{\scriptscriptstyle anom}(\uf):= \ \sum_{\stackrel{\scriptstyle \text{chords}}{ \vec{\ue} \, \in \, \partial \uf }}  H( \vec{\ue} , \uf)
\quad \text{with} \quad 
H(\vec{\ue},\uf) := \ {1\over 8 \pi \imath } \, \theta_\mathrm{n} (\vec{\ue} \, ) \cot \theta_\mathrm{n} (\vec{\ue} \,) \,
\mathcal{E}_\mathrm{n}(\vec{\ue} \,)
\end{equation}

\noindent
and where $\mathcal{E}_\mathrm{n}(\vec{\ue} \, )$ is defined 
in \ref{mathcalE-terms}. These explicit results are valid when deforming an isoradial Delaunay graph $\uG_\mathrm{cr}$.

Again, for a deformation of a generic triangulation $\uG_\mathrm{cr}$, the variation \ref{var1DeltaConf2} can be written in terms of a discretized stress-energy tensor $\mathbf{T}_{\!\scriptscriptstyle{\Deltaconf}}$ a theory for a Grassmann field $(\Phi,\bar\Phi)$ with action $S_{_\mathrm{conf}}=\Phi{\cdot}\Deltaconf \bar\Phi$
\begin{equation}
\label{var1ConfDiscr}
\begin{split}
\deltae \log\det(\Deltaconf)=&-{1 \over\pi} \sum_{\uf} A(\uf) \left(\overline\nabla F(\uf) \langle T_{\!\scriptscriptstyle{\Deltaconf}}(\uf)\rangle+ \nabla \bar F(\uf) \langle \bar T_{\!\scriptscriptstyle{\Deltaconf}}(\uf)\rangle\right)\\
&+{1 \over 2}\sum_{\uf} A(\uf) \, \left(\nabla F(\uf)+\overline\nabla\bar F(\uf)\right) \, \langle \tr(\mathbf{T}_{\!\scriptscriptstyle{\Deltaconf}}(\uf))\rangle 
\end{split}
\end{equation}
One has generically $\tr(\mathbf{T}_{\!\scriptscriptstyle{\Deltaconf}})=0$ (conformal invariance).
The discretized analytic and anti-analytic components $T_{\!\scriptscriptstyle{\Deltaconf}}$ and $\bar T_{\!\scriptscriptstyle{\Deltaconf}}$ can be written explicitly, using Section \ref{ssVarOp} and in particular \ref{VarThetaNS} in Remark~\ref{rVArConf}.
We get a generic form for $T_{\!\scriptscriptstyle{\Deltaconf}}$ involving all possible binomials of discrete derivatives of the fields
\begin{equation}
\label{TconfDiscrSchem}
T_{\!\scriptscriptstyle{\Deltaconf}}= \text{\cminfamily{A}}\,\nabla\Phi\,\nabla\bar\Phi+\text{\cminfamily{B}}\,\nabla\Phi\,\overline\nabla\bar\Phi+\text{\cminfamily{C}}\,\overline\nabla\Phi\nabla\bar\Phi+\text{\cminfamily{D}}\,\overline\nabla\Phi\,\overline\nabla\bar\Phi
\end{equation}
The coefficients $\text{\cminfamily{A}}(\uf)$, $\text{\cminfamily{B}}(\uf)$, $\text{\cminfamily{C}}(\uf)$ and $\text{\cminfamily{D}}(\uf)$ depend not only on the geometry of the triangle $\uf$, but also of its three neighbouring triangles $\uf'$, $\uf''$ and $\uf'''$, since they depend explicitly of the conformal angles of the three edges $\ue'$, $\ue''$ and $\ue'''$ of $\uf$. See Fig.~\ref{fFace&Neighbours}.
Like $\mathbf{T}_{_{\mathrm{\!\Delta}}}$, the discrete stress-energy tensor $\mathbf{T}_{\!\scriptscriptstyle{\Deltaconf}}$ 
is quadratic in the local derivatives of the fields $(\Phi,\bar\Phi)$. However, it involves not only the term $\nabla\Phi\,\nabla\bar\Phi$ but three other terms.
Furthermore the coefficient $\text{\cminfamily{A}}(\uf)$ of the $\nabla\Phi\,\nabla\bar\Phi$ term is nonconstant and depends on the geometry $\uf$ and its neighbours.

\begin{figure}[h]
\begin{center}
\includegraphics[width=2in]{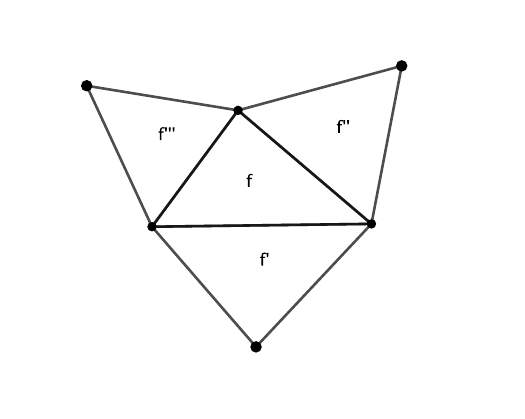}
\caption{A face $f$ (triangle) and its three neighbours}
\label{fFace&Neighbours}
\end{center}
\end{figure}

\subsection{The second order variations and discretized CFT}\ \\
\label{ssDi2ndOr}
We now discuss along the same line our result for the second order variation and its scaling limit.

\color{blue}
\subsubsection{The Laplace-Beltrami operator $\Delta$}
\label{Disc2ndDelta}
\color{black}
Here we consider the Beltrami-Laplace operator of
the Delaunay graph $\uG_\epsilon$ obtained through a bi-local deformation 
$z_{\underline{\epsilon}}(\uv) := z_\mathrm{cr}(\uv) + \epsilon_1 F_1(\uv) + \epsilon_2 F_2(\uv)$
of the critical embedding of an isoradial Delaunay graph $\uG_\mathrm{cr}$. 
The $\epsilon_1 \epsilon_2$ cross-term of $\log \det \Delta(\underline{\epsilon})$ can be calculated exactly using Proposition \ref{VarDelta}
and expressed using the limit graph $\uG_{0^+}$ and any weak Delaunay triangulation $\widehat{\uG}_{0^+}$ 
which completes it. This gives
\begin{equation}
\label{2nd-variation-log-det}
\begin{array}{l}
\deltaeo \, \deltaet \, \log \det \Delta 
\D = \ - \, \tr \Big[ \deltaeo \Delta \cdot \Delta_\mathrm{cr}^{-1} \cdot \deltaet \Delta \cdot \Delta_\mathrm{cr}^{-1} \Big] \\ 
\D = \ - 64\, \tr \Bigg[ \frak{Re} \Big[ \overline{\nabla}^\top (\nabla \bar{F}_1) A \overline{\nabla} \Big] \cdot \Delta_\mathrm{cr}^{-1}
\cdot \frak{Re} \Big[ \overline{\nabla}^\top (\nabla \bar{F}_2) A \overline{\nabla} \Big] \cdot \Delta_\mathrm{cr}^{-1} \Bigg] \\ 
\D = \ \left\{
\begin{array}{c}
\D \underset{\substack{\text{triangles} \\
\ \ \ \, \ux_1, \ux_2 \, \in \, \widehat{\uG}_{0^+} }}{- \sum}
32 \, A(\ux_1) A(\ux_2) \ \frak{Re} \Bigg[ \overline{\nabla}F_1(\ux_1) \overline{\nabla}F_2(\ux_2) 
\left(\big[ \nabla \Delta_\mathrm{cr}^{-1} \nabla^\top  \big]_{\ux_1 \ux_2}\right)^2 \Bigg] \\
  + \\ 
\D \underset{\substack{\text{triangles} \\
\ \ \ \, \ux_1, \ux_2 \, \in \, \widehat{\uG}_{0^+} }}{- \sum}
32 \, A(\ux_1) A(\ux_2) \ \frak{Re} \Bigg[ \overline{\nabla}F_1(\ux_1) \nabla \bar{F}_2(\ux_2) 
\left(\big[ \nabla \Delta_\mathrm{cr}^{-1} \overline{\nabla}^\top  \big]_{\ux_1 \ux_2}\right)^2 \Bigg] 
\end{array} \right.
\end{array}
\end{equation}
\noindent
Using formula \ref{TLBDiscr} for the discrete stress-energy tensor $T_\Delta$ and 
applying Wick's theorem we can express the two-point v.e.v.'s  
\begin{equation}
\label{discrete-two-point-correlators}
\begin{split}
\D {1 \over {32 \pi^2}} \, \Big\langle T_\Delta(\ux_1) \, T_\Delta(\ux_2) \Big\rangle_\mathrm{conn.}
&= \ \left(\big[ \nabla \Delta_\mathrm{cr}^{-1} \nabla^\top  \big]_{\ux_1 \ux_2}\right)^2 \\
\D {1 \over {32 \pi^2}} \, \Big\langle T_\Delta(\ux_1) \, \overline{T}_\Delta(\ux_2) \Big\rangle_\mathrm{conn.}
&= \ \left(\big[ \nabla \Delta_\mathrm{cr}^{-1} \overline{\nabla}^\top  \big]_{\ux_1 \ux_2}\right)^2
\end{split}
\end{equation}
\noindent
and the c.c.
So far we do not require the initial graph to be isoradial: 
We may in fact replace the critical graph $\uG_\mathrm{cr}$ with any Delaunay graph $\uG_0$ 
equipped with its corresponding Beltrami-Laplace operator $\Delta_0$ and Green's function $\Delta_0^{-1}$, 
and the variational formula \ref{2nd-variation-log-det} and the double correlator identity \ref{discrete-two-point-correlators}
remain valid. If, however, we incorporate a scaling parameter $\ell > 0$ and 
consider the bi-local smoothly deformed embedding
$z_{\underline{\epsilon}, \ell}(\uv) := z_\mathrm{cr}(\uv) + \epsilon_1\, \ell\, F_{1;\ell}(\uv) + \epsilon_2\, \ell \,F_{2; \ell}(\uv)$
then the isoradial property (as manifest in the asymptotic expansion \ref{SharpAsymptotics1} for the critical Green's function $\Delta_\mathrm{cr}^{-1}$) 
is sufficient to establish the convergence of the scaling limit of formula \ref{2nd-variation-log-det},
which is consistent with the OPE of a CFT with the expected central charge $c=-2$, namely
\begin{equation}
\label{scaling-limit-2nd-variation-log-det}
\lim_{\ell \rightarrow \infty} \, 
\deltaeo \, \deltaet \, \log \det \Delta(\ell)  \ = \  
{c\over {\pi^2}}\iint_{\Omega_1\times\Omega_2} dx^2_1 \, dx^2_2 \ \, \frak{Re} \Bigg[ {\bar\partial F_1(x_1)\,\bar\partial F_2(x_2)\over (x_1-x_2)^4}  \Bigg]
\end{equation}
As we have seen $\nabla \Delta_0^{-1} \nabla^\top$ and $\nabla \Delta_0^{-1} \overline{\nabla}^\top $
(and their complex conjugates) must decay in accordance with Lemma \ref{lNGN} in order
for \ref{scaling-limit-2nd-variation-log-det} to hold.
Our result is, of course, not surprising, and should be viewed as a check of the validity of our approach.

\subsubsection{The K\"ahler operator  $\mathcal{D}$}
\label{disc2ndKahler}
Prop.~\ref{The3} and its scaling limit given in Section~\ref{sScalLimit} are the novel results of the paper.
They state that the scaling limit of the bi-local second order variation for $\log \det  \mathcal{D}(\underline{\epsilon} \, , \ell)$
and  $\log \det  \Delta (\underline{\epsilon} \, , \ell)$ are identical. 
\begin{equation}
\label{ }
-{1\over \pi^2}\iint_{\Omega_1\times\Omega_2} \!\!\!\! dx_1\,dx_2 \left( {\bar\partial F_1(x_1)\,\bar\partial F_2(x_2)\over (x_1-x_2)^4} + {\partial \bar F_1(x_1)\,\partial \bar F_2(x_2)\over (\bar x_1-\bar x_2)^4}\right) 
\end{equation}
This result is interesting for two reasons.

The $\mathcal{D}$ operator has a different form and even a different scaling dimension than $\Delta$. Its variation \ref{VarD} and the associated stress-energy tensor \ref{TKdiscr} are different. However the second order variation has exactly the same OPE as the second order variation for $\Delta$, and it corresponds to a CFT with the same central charge $$c=-2\ .$$ 

This value for the central charge is in our opinion somehow unexpected, and this is interesting per se. 
Indeed it was suggested by one of us (F.D.) in the original paper \cite{DavidEynard2014} that the measure over triangulations given by $\det(\mathcal{D})$ (later shown in \cite{CharbonnierDavidEynard2019} to coincide with the Weil-Petersson metric over marked complex curve), had a direct relation with the gauge fixing Fadeev-Popov determinant in two-dimensional quantum gravity. If true, it should be related to the so-called \texttt{b-c} ghosts system in Polyakov's formulation as  Liouville theory of 2D gravity and non-critical strings (see \cite{Friedan1982}). 
Then one could have expected a different value for the central charge, since the central charge for the \texttt{b-c} system is $c=-26$, and the central charge for the corresponding Liouville quantum gravity (at $Q=5/\sqrt{6}$ i.e. $\gamma=\sqrt{8/3}$) is $c=26$.

\subsubsection{The conformal laplacian $\Deltaconf$}\ \\
\label{disc2ndConf}
For the conformal Laplacian operator $\Deltaconf$, we do not have such a simple result, and the corresponding OPE cannot be interpreted as coming from a CFT. 
There are additional contributions that come from the chords, which have been studied in section~\ref{ssConfLapl2nd}, and are the \texttt{chord-chord} term given by \ref{chord2chord} and the \texttt{chord-edge} term given by \ref{chord2edge}.
The later \texttt{chord-edge} term has the expected harmonic form (depending only on $(x-x')^{-4}$ and its c.c.), but with  a local geometry dependent coefficient involving both $\overline\nabla\! F_1 \overline\nabla\! F_2$ and $\nabla\! \bar F_1 \overline\nabla\! F_2$ terms. 
The \texttt{chord-chord} term is even more involved and contains a non-harmonic term, proportional to $|x-x'|^{-4}$, with a more complicated geometrical dependence in the geometry of the faces and the chords.
In Appendix~\ref{tiling-by-cocyclic-quad} we give an explicit example of a critical lattice with a finite density of chords where these additional ``anomalous'' terms give a macroscopic anomalous contribution to the second order variation, which precludes an interpretation in terms of conformal field theory in the scaling limit.
Of course this comes from the anomalous terms in the expression of the discretized stress-energy tensor $T_{\scriptscriptstyle{\mathrm{\!conf}}}$ (of general schematic form given in \ref{TconfDiscrSchem}), which does not have a simple universal field theoretical interpretation in the scaling limit. This is also a new - although somehow negative - result.

\subsection{Relations and differences with other discrete models}
\label{ssRelDiff}\ \\
The operators that we study here are defined on planar isoradial Delaunay graphs. 
Isoradial graph embeddings play a very important role in the study of two-dimensional models of statistical mechanics in theoretical physics and in mathematics. 
In particular they are an essential tool in the proof of the conformal invariance of the Ising model at its critical point, and in the study of the conformal invariance of other critical models. They are very important in our study too, since they afford control of the large distance properties of the respective Green’s functions.

However, we stress that there is an important difference in terms of perspective. 
In studies of critical statistical models on such graphs, the underlying graph is fixed
and the proofs of the existence of a scaling limit and of its conformal invariance are undertaken for a fixed lattice. 
The random triangulation model of \cite{DavidEynard2014} is a statistical model \emph{of} planar graphs, rather than a statistical model \emph{on a} planar graph. 
The planar isoradial graphs that we consider here are just some special ``semi-classical'' configurations, which minimize a ``local curvature functional'', as discussed in the Sect.~\ref{s3RhombicG}, Formula~\ref{scalar-curvature}.

There are nevertheless relations between our work and some recent works, especially in regard to defining a notion of a discrete stress-energy tensor. 
Let us briefly discuss two of them.

\subsubsection{Discrete stress-energy tensor in the loop model of Chelkak et al.}
\label{sssChelkak}

In \cite{Chelkak-Glazman-Smirnov2018} Chelkak, Glazman and Smirnov study the famous critical O($n$) loop model \cite{Domanyetal1981} \cite{NienhuisDG} \cite{Kostov1989} on abstract discrete surfaces with boundaries (denotes $G_\delta$) made by gluing together equilateral triangles $\triangle$ and rhombs $\lozenge(\theta)$ 
of unit length $\delta$ where each rhomb has an independent acute angle $\theta$ selected in the range $0 < \theta \leq  {\pi \over 2}$, as depicted in Fig.~\ref{Triangle+Rhomb}.
The surface has in general conical singularities at all of its vertices.
A discrete surface may admit more than one tessellation into triangles and rhombs if some vertices are flat (no conical defect). Two tessellations are equivalent (i.e. they describe the same surface) if one can be transformed into the other by applying a sequence of the following three kinds of local operations: 
(i) \emph{Yang-Baxter transformations} which flips a flat hexagon made up of three rhombs sharing a common vertex, (ii) \emph{pentagonal transformations}
which interchange a triangle and a rhomb which form a flat pentagon with a triangle and two rhombs, (iii) \emph{split transformations} which dissect a rhomb 
$\lozenge({\pi \over 3})$
into a pair of equilateral triangles sharing a common edge; this is depicted in Fig.~\ref{YG+Pent}.

\begin{figure}[h]
\begin{center}
\includegraphics[height=.8in]{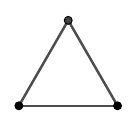}
\qquad
\includegraphics[height=.8in]{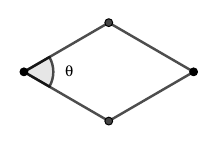}
\caption{The triangles and rhombs of \cite{Chelkak-Glazman-Smirnov2018}}
\label{Triangle+Rhomb}
\end{center}
\end{figure}

\begin{figure}[h]
\begin{center}
\includegraphics[height=1.in]{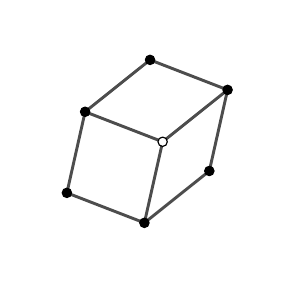}
\hskip -1.em
\raisebox{.4 in}{$\longleftrightarrow$}
\hskip -1.em
\includegraphics[height=1in]{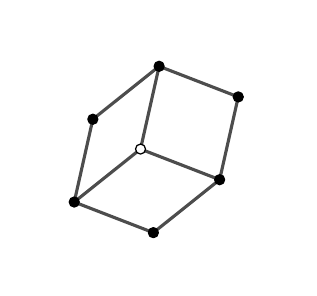}
\raisebox{.4 in}{\ ,\ }
\includegraphics[height=1.2in]{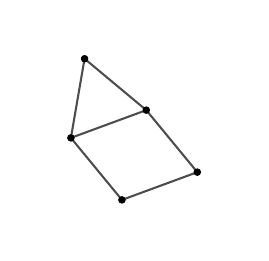}
\hskip -1.em
\raisebox{.4 in}{$\longleftrightarrow$}
\hskip -1.em
\includegraphics[height=1.2in]{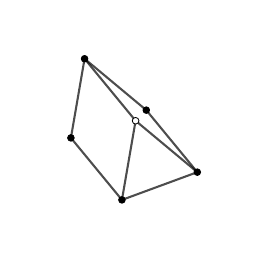}
\includegraphics[height=.55in]{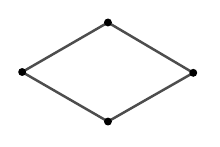}
\hskip -.1em
\raisebox{.25 in}{$\longleftrightarrow$}
\hskip -.1em
\includegraphics[height=.55in]{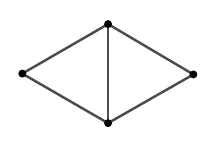}
\caption{Yang-Baxter, pentagonal and split moves of \cite{Chelkak-Glazman-Smirnov2018}; white vertices {\tiny{$\ocircle$}} have to be flat (no conical singularity)}
\label{YG+Pent}
\end{center}
\end{figure}

The states of the $\mathrm{O}(n)$ loop model for a tessellated surface $G_\delta$ are configurations $\gamma$ consisting of non-crossing loops and strands (joining boundary components, if present) drawn on the surface $G_\delta$ which can be obtained by concatenating local arrangements of arcs, one  for each triangle and rhomb in $G_\delta$.  
A local weight $\mathrm{w}_\gamma(\uf)$ is associated to each face $\uf$ of $G_\delta$ which depends on the configuration of the loops on $\uf$, the geometry of the face (hence of angle $\theta$ if $\uf=\lozenge(\theta)$ is a rhomb), and on a parameter $s$ (related to the temperature). 
A factor $n$ (loop fugacity) is associated to each closed loop. 
The  local weight $\mathrm{w}_\gamma(\uf)$ (that we do not discuss here) are taken to have a very specific form in order to satisfy  the Yang-Baxter and pentagonal relations ensuring that the model is the same for equivalent tessellations of the surface. 

The partition function $Z^\frak{b}(G_\delta)$ for the $O(n)$ loop model on a fixed surface $G_\delta$ equipped with a boundary condition ${\frak{b}}$ (specifying 
which boundary edges are joined by arcs), is given by the sum over states (loops+arcs configurations $\gamma$)  by
\begin{equation}
\label{O(n)-partition-function}
Z^\frak{b}(G_\delta)
\,  := \ 
\sum_{\text{$\frak{b}$-configurations} \, \gamma} \, n^{\# \text{loops}(\gamma)} \, \prod_{\substack{\text{faces} \\ \uf \, \in \, G_\delta}} \mathrm{w}_\gamma(\uf)
\end{equation}
In addition, when the specific relation between $n$ (the loop fugacity) and $s$ (the temperature parameter) $$n=-\cos(4\pi\,s /3)$$ holds, then the loop model is critical.

In \cite{Chelkak-Glazman-Smirnov2018} Chelkak et al. consider a \emph{planar version} without conical defects where all rhombs have angle $\theta = {\pi \over 3}$, and such that the discrete surface $G_\delta$ is a compact, connected domain $\Omega$ of the triangular lattice.
In this planar case, they define a \emph{discrete stress-energy tensor} as the response of the model to an infinitesimal $\epsilon$-deformation of the original planar surface into a non-planar surface with conical defects. More precisely, two deformations are considered: 
(i) replacing two adjacent equilateral triangles (forming a rhomb $\lozenge(\pi/3)$) by a rhomb $\lozenge(\theta)$ with angle $\theta={\pi \over 3}+\epsilon$,
(ii) replacing two aligned edges by a ``almost flat'' rhomb $\lozenge(\epsilon)$ (see Fig.~\ref{RhombDeform}).
\begin{figure}[h!]
\begin{center}
\includegraphics[width=1.66in]{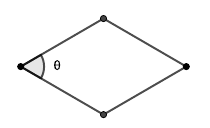}
\raisebox{.5in}{$\longrightarrow$}
\includegraphics[width=1.66in]{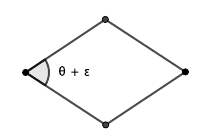}
\\ 
\includegraphics[width=1.66in]{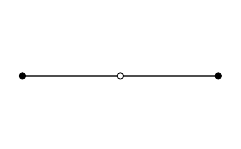}
\raisebox{.52in}{$\longrightarrow$}
\raisebox{-.06in}{\includegraphics[width=1.66in]{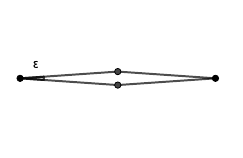}}
\caption{The $\epsilon$-deformations of rhombs in \cite{Chelkak-Glazman-Smirnov2018}}
\label{RhombDeform}
\end{center}
\end{figure}
The variation of the logarithm of the partition function under such $\epsilon$-deformations defines the e.v. of discrete stress-energy tensor $\mathcal{T}_{\ue | \um}$ associated to edges $\ue$ or to midlines $\um$ (of the honeycomb lattice built from the original triangular lattice), and out of these related real objects, a discrete complex stress-energy tensor $\mathcal{T}$ can be associated to the vertices and the faces of the lattice (with relations). In \cite{Chelkak-Glazman-Smirnov2018} it is conjectured that this object is approximately discrete-holomorphic and converges to the stress-energy tensor of the corresponding CFT in the scaling limit.

\subsubsection{Similarities and differences}
There are similarities but also important differences with the approach and results of our study. The discrete conformal Laplacian $\Deltaconf$ defined in \ref{ConfDelta} is also defined with respect to a rhombic tessellated surface $\uS_\uG^\lozenge$ naturally associated to a Delaunay graph $\uG$ in the plane (see Sect.~\ref{ssBasicPlanar} and especially Def.~\ref{dRhomSurf}). 
However $\uS_\uG^\lozenge$ is constructed only out of rhombs $\lozenge(\ue)$ associated to  edges $\ue$ of $\uG$, and contains no equilateral triangles.
Moreover, the rhombic surface $\uS_\uG^\lozenge$ is bipartite, with black and white vertices corresponding to vertices and faces of $\uG$ respectively. 
Finally, and most importantly, the black vertices of $\uS_\uG^\lozenge$ must be flat (they do not carry a conical singularity), while the white vertices may carry 
a conical singularity (corresponding to a non-zero Ricci curvature given by \ref{scalar-curvature}), see Fig.~\ref{rhomb-us-marked}.
\begin{figure}[h!]
\begin{center}
\includegraphics[width=2in]{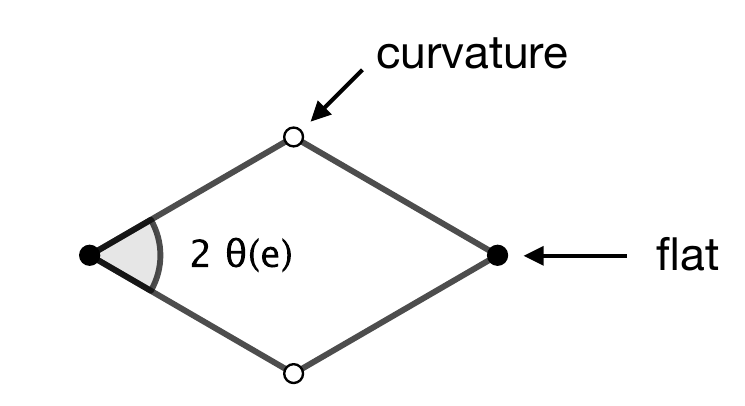}
\caption{The rhombs which build the tessellated surface $\uS_\uG^\lozenge$ in this paper}
\label{rhomb-us-marked}
\end{center}
\end{figure}
Thus our model considers only a subspace of the space of tessellated surfaces of \cite{Chelkak-Glazman-Smirnov2018}.

Like \cite{Chelkak-Glazman-Smirnov2018}, the stress-energy tensor in our study is defined in terms of deformations. However an important difference is that we consider deformations of $\uS_\uG^\lozenge$ which are induced from deformations of the underlying Delaunay graph $\uG$ in the plane. This space of deformations differs from
those considered in \cite{Chelkak-Glazman-Smirnov2018} in two respects. First, our deformations preserve the flatness of the black vertices of $\uS_\uG^\lozenge$. Second, and this is essential, our discrete stress-energy tensor has a specific invariance properties under \emph{global continuous analytic transformations} of the plane, i.e. M\"obius transformations. This holds a priori, independent of the specific geometry of the Delaunay graph $\uG$.

In \cite{Chelkak-Glazman-Smirnov2018} as well as in other studies, the framework is different. One looks for a discrete stress-energy tensor on an isoradial critical graph $\uG$ which has some specific invariance properties under the \emph{discrete analytic and anti-analytic transformations} of $\uG$. Discrete analyticity is a very special and powerful property, but it depends explicitly on the critical graph considered. It is only in the scaling limit that discrete analyticity can be shown to ``converge'' (this is a crude presentation of beautiful and precise results) towards the usual analyticity in the continuum (i.e. in the complex plane $\mathbb{C}$).

Another difference is that our setting includes deformations of ``flat rhombs'' (corresponding to chords) which are not deformations of aligned edges, as 
considered in \cite{Chelkak-Glazman-Smirnov2018} and depicted in Fig.~\ref{RhombDeform}.
These deformations induce the appearance of the ``curvature dipoles'' discussed in Sect.~\ref{subsubsection-dipole}, which complicate the analysis of 
$\Deltaconf$.

The overlap between our work and the results of  \cite{Chelkak-Glazman-Smirnov2018} is restricted to the case of the $\Deltaconf$ operator, which is related to the GFF. Strictly speaking the authors of \cite{Chelkak-Glazman-Smirnov2018} consider the critical O($n$) loop model for $n\in[-2,2]$, but it is known that the GFF can be related to the $n=2$ model, and that there is some relation between the Laplace-Beltrami operator on a graph and the $n=-2$ model.

On the other hand, the Laplace-Beltrami operator $\Delta$ and the K\"ahler operator $\mathcal{D}$, which we would like to study on a general Delaunay graph $\uG$, are not defined in terms of the abstract rhombic surface $\uS_\uG^\lozenge$. 
We do not know how to relate precisely, and in general, their corresponding discrete stress-energy tensors to the construction of the stress-energy tensor given by \cite{Chelkak-Glazman-Smirnov2018}.

\subsubsection{Stress-energy tensor constructions through lattice representations of the Virasoro algebra}
In an approach taken by Hongler et al. in \cite{Hongler2019}, a stress-energy tensor for some lattice models is defined
implicitly by identifying its modes through an action of the Virasoro algebra on an appropriately defined
vector space $\frak{F} := \frak{F}^\mathrm{loc}/\frak{F}^\mathrm{null}$ of {\it lattice local fields} (modulo {\it null fields}) 
supported on the graph. This construction avoids interpreting
the stress-energy tensor as a response to a deformation of the graph embedding. Instead
an intermediate action of the Heisenberg algebra is introduced using a discrete holomorphic
current along with a technique of discrete contour integration and 
a notion of discrete half-integer power functions.
Only the special cases 
of the discrete GFF and of the Ising model on the square lattice $\uG = \Bbb{Z}^2_\delta$ with mesh size $\delta$ 
are handled in \cite{Hongler2019}. 
However, we expect that most of their technology (e.g. the notions of medial and corner graphs, discrete power functions,  and discrete contour integration) is readily adaptable to  arbitrary isoradial graphs (and their rhombic graphs
where the theory of discrete holomorphicity is well behaved).
The space of lattice local fields $\frak{F}^\mathrm{loc}$ of \cite{Hongler2019} depends on the translation properties of $\uG = \Bbb{Z}^2_\delta$. 
Specifically $\frak{F}^\mathrm{loc}$ consists of fields which can be constructed as polynomial expressions of elementary fields $\phi_\delta(z)$  together with their translates $\phi_\delta(z + x \delta)$ for  $x$ in some fixed, finite set  $V \subset \Bbb{Z}^2$ of admissible displacements.
For a general isoradial graph
one would need to specify an adequate vector space of lattice local fields $\frak{F}^\mathrm{loc}$ 
on which a representation of the Virasoro algebra could be supported. 
Bearing this, it would be natural
to examine whether the stress-energy tensor(s) for the operator(s) considered in our paper
can be realized by such putative Virasoro algebra action(s).
For older references of representations of Virasoro algebra in lattice models, see the references in \cite{Hongler2019}.

\subsection{Open questions and possible extensions}
\label{ssOQandPE}\ \\

\noindent{\textbf{1:}} We would like to reiterate the problem of settling Conjecture \ref{Conf3bound} of Sect.~\ref{ssVArRGen}, or in lieu of that, finding another adequate bound on $\displaystyle {{R(\uf)}^{-1}} \nabla p_3(\uf)$ uniform in the faces $\uf$ of $\uT^{(r)}_0$ and the scaling parameter $\ell$ (or $r=1/\ell$), in order to complete the proof of props.~\ref{PropLimTr2} and \ref{prvsepsComm}  as well as \ref{PropLimTr2D}.

\medskip\noindent{\textbf{2:}} Instead of using an isoradial Delaunay graph, we could instead begin with a Delaunay graph which is ``smoothly non-isoradial'', in the sense that the circumradii of the faces $R(\uf)$ vary slowly with the position of the faces in the plane. Studying the Laplace-like operators $\Delta$, $\Deltaconf$ and $\mathcal{D}$ and their deformations on such a graph is an interesting problem which might entail finding asymptotic expansions of the corresponding Green functions.

\medskip\noindent{\textbf{3:}} The properties that make a general isoradial graph $\uG$ so useful as a starting point in our analysis are a reflection of the underlying notion of discrete analyticity supported on the lozenge graph $\uG^\lozenge$. Chelkak, Smirnov and others \cite{ChelkakTalk2018} have introduced the concept of s-holomorphicity and s-embeddings of graphs, and one can try to develop a theory of deformations for such graphs and their associated operators. 

\medskip\noindent{\textbf{4:}} In the scaling limit, random planar graphs are known to be related to Liouville conformal field theory. Finding a notion of discrete Liouville local field, with good properties in the scaling limit, for the model of random Delaunay triangulations is still an open problem. A solution could lead to an alternative discrete stress-energy tensor on a Delaunay graph, different from the one considered here, and with different properties under geometrical deformations of the graphs; in particular having a
discrete central charge different from $c=-2$ (possibly $c=-26$).

\medskip\noindent{\textbf{5:}} It should also be interesting to study the existence and description of a stress-energy tensor for other discrete models on Delaunay graphs, such as Dirac Fermions, the Ising model, the O(N) model, etc. using the approach of our work. It would be fruitful to compare the results with the approaches taken in
\cite{Hongler2019} and \cite{Chelkak-Glazman-Smirnov2018} (see section~\ref{ssRelDiff}).


\newpage
\section*{Acknowledgements}
The authors thank for hospitality Universidad de los Andes, Bogot\'a, Colombia, where this work was started, as well as the Chebyshev Laboratory at St.~Petersburg State University, Russia, the Perimeter Institute, Waterloo, Canada, and the Institute of the Mathematical Sciences of the Americas, U. of Miami, USA, where part of this work was done. In addition to these institutions, J. S. also thanks Institut Henri Poincar\'e in Paris, 
Institut des Hautes \'Etudes Scientifiques in
Bures-sur-Yvette, and Institute de Physique Th\'eorique of Saclay 
for their hospitality; the later for several visits during the course of this project.

The authors thank Dmitry Chelkak for discussions at IHES hosted by Hugo Duminil-Copin.
F. D. thanks Bertrand Eynard and Philippe di Francesco for discussions.
J. S. would also like to thank Peter Zograf and Ga\"etan Borot: 
Peter Zograf for the opportunity to present a short lecture course 
addressing preliminary results of this paper 
at the Chebyshev Laboratory, and 
Ga\"etan Borot for inviting and hosting J. S. at the 2017 IHP trimester program ``Combinatorics and interactions'' 
and later at the Max-Planck-Institut, Bonn for a short visit to discuss aspects of the project.  
{{We} are very grateful to the reviewer of this paper, for {their} perceptive remarks and questions, which led to many improvements.}

This paper is partly a result of the ERC-SyG project, Recursive and Exact New Quantum Theory (ReNewQuantum) which received funding from the European Research Council (ERC) under the European Union's Horizon 2020 research and innovation programme under grant agreement No 810573.

\appendix
\section{Reminders: the stress-energy tensor in QFT and the central charge in 2D CFT}
\label{aCFTreminder}
\subsection{The stress-energy tensor}
For completeness, we recall some textbook material of {quantum field theories} (QFT) and {conformal field theories} (CFT), which can be found for instance in \cite{francesco1997conformal}. A central concept in field theory is the \emph{stress-energy tensor} $\mathbf{T}=(T^{\mu\nu})$ (also denoted the \emph{energy-momentum tensor} in the literature). Firstly, $\mathbf{T}$ can be viewed (in flat space) as the conserved current $\mathbf{J}^\nu=(T_{\mu}^{\  \nu})$ associated to space-time translation invariance, and is defined through Noether's theorem by the action of an infinitesimal local change of coordinates
\begin{equation}
\label{xdiff}
x^\nu \to x^\nu+\xi^\nu(x)
\end{equation} 
on the action $\mathcal{S}$ (classical or quantum) of the theory.
Secondly  $\mathbf{T}$ can be viewed (in a general curved space) as the ``response of the theory'' to an infinitesimal variation of the classical ``background metric'' $\mathbf{g}=(g_{\mu\nu})$
\begin{equation}
\label{varGdG}
g_{\mu\nu}\to g_{\mu\nu}+\delta g_{\mu\nu}
\end{equation} 
of the space-time $M$ where the theory ``lives''.  More precisely $\mathbf{T}$ is defined classically by the functional derivative
of the action $\mathcal S$
\begin{equation}
\label{TasVarG}
T^{\mu\nu}(x)=
- \, {2\over\sqrt{g(x)}}{\delta\mathcal{S}\ \over\delta g_{\mu\nu}(x)}
\end{equation}
For a quantum theory (i.e. a local QFT), $\mathbf{T}$ is now a quantum operator. Its vacuum expectation value (the vacuum-vacuum matrix element) is given by the \emph{first order} variation of the logarithm of the partition function $Z$ of the QFT under an infinitesimal variation of the metric $\delta g_{\mu\nu}$
\begin{equation}
\label{1stVarLZ}
\delta\log Z = {1\over 2} \int_M \!\!dx\,\sqrt{g(x)}\,\delta g_{\mu\nu}(x)\,\langle T^{\mu\nu}(x)\rangle
\ + \ \cdots
\end{equation}
Similarly the first order variation of the vacuum expectation of an observable $\mathcal{O}$, for instance a product of local operators $\mathcal{O}_1(x_1)\cdots\mathcal{O}_n(x_n)$, gives by the connected correlator of $\mathbf{T}$ times $\mathcal{O}$
\begin{equation}
\label{StVarO}
\delta\langle\mathcal{O}\rangle = {1\over 2} \int_M \!\!dx\,\sqrt{g(x)}\ \delta g_{\mu\nu}\ 
\Big({\langle T^{\mu\nu}(x)\,\mathcal{O}\rangle}_{\mathrm{conn.}}
\ +\ \textrm{contact terms}\Big) \,+\,\cdots
\end{equation}
where the so-called ``contact terms'' are present in \ref{StVarO} when the position $x$ of $\mathbf{T}$ coincides with that of some local operators in $\mathcal{O}$.

These two definitions of the stress-energy tensor $\mathbf{T}$ are closely related, and in fact equivalent (with the proper definitions of $\mathbf{T}$), since a diffeomorphism \ref{xdiff} induces a change of metric \begin{equation}
\label{vargdiff}
\delta g_{\mu\nu}=D_{\!\mu}\xi_\nu+D_{\!\nu}\xi_\mu
\end{equation}
with $D_\mu$ the covariant derivative and $\xi_\nu=g_{\nu\rho}\xi^\rho$.

These definitions extend to {the} higher order terms {in $\delta g_{\mu\nu}$} 
{and give}  expectation values of products of $\mathbf{T}$ (correlators). For instance, the second-order term in the variation of $\log Z$ gives the two-point connected correlator
\begin{equation}
\label{2Tcorrel}
{1\over 8} \int_M \!\!\!\!dx\,\sqrt{g(x)}\,\delta g_{\mu\nu}(x)\,
\int_M \!\!\!\!dy\,\sqrt{g(y)}\,\delta g_{\rho\sigma}(y)\,
{\langle T^{\mu\nu}(x)\,T^{\rho\sigma}(y)\rangle}
_{\mathrm{conn.}}
\ +\ \text{contact terms}
\end{equation}
and so on.

\subsection{The stress-energy tensor in two-dimensional CFT}
\label{ssT2DCFT}
In two dimensions, it is standard to work in complex coordinates $z
=x^1+\imath\, x^2$, $\bar z
=x^1-\imath \,x^2$, so that the flat metric is $$g_{zz}=g_{\bar z\bar z}=0\,, \ g_{z\bar z}=g_{\bar zz}=1/2\ .$$
An infinitesimal diffeomorphism
$ z \mapsto z+\epsilon\, F(z,\bar z)$
thus amounts to a variation of the metric 
\begin{equation}
\label{dgdD}
\delta g_{zz}=\epsilon\, \partial \bar F\,,\quad\delta g_{\bar z\bar z}=\epsilon\, \bar \partial F\,,\quad\delta g_{z\bar z}=\delta g_{\bar z z}=\epsilon\, (\partial F+\bar\partial\bar F)/2
\end{equation}

For QFT's in two dimensions (in particular for CFT's), especially important are the holomorphic and antiholomorphic components of the stress-energy tensor $\mathbf{T}$, which are denoted $T$ and $\bar T$ in the literature (see e.g. \cite{francesco1997conformal}). In the flat metric they are
\begin{equation}
\label{TbarTdef}
T=-{\pi\over 2}\,T^{\bar z\bar z}=-2\pi\, T_{zz}
\ ,\quad
\bar T=-{\pi\over 2}\,T^{zz}=-2\pi\, T_{\bar z\bar z}
\end{equation}
The variation of $\log Z$ \ref{1stVarLZ} reads
\begin{equation}
\label{1stVarLZcomplexA}
\begin{split}
\delta\log(Z)=
&-{\epsilon\over \pi}\int d^2x\ \Big(\partial\bar F(x)\, \langle \bar T(x)\rangle
 + \bar\partial F(x)\,  \langle T(x) \rangle \Big)
\\
&+ {\epsilon\over 2}\int d^2x\ \big( \partial F(x)+\bar\partial\bar F(x) \big) \,  \langle \tr \mathbf{T}(x) \rangle 
\ +\ \cdots
\end{split} 
\end{equation}
where $\tr\mathbf{T}= T^\mu_{\ \mu}=T^{\mu\nu} g_{\nu\mu}=T^{z\bar z}=T^{\bar z z}$.

Conformal invariance in 2D implies that  $T^{z\bar z}=T^{\bar z z}=\tr\mathbf{T}=0$ identically vanishes.
For a quantum theory (a CFT) this requires a proper definition of the renormalized stress-energy tensor, and this identity is valid up to very specific contact terms.
The law of conservation for the current $\partial_\mu T^{\mu\nu}=0$ reduces to $\bar\partial T=0$, $\partial \bar T=0$, hence the terminology holomorphic and anti-holomorphic components. 
This is valid for a CFT in a flat metric. 

For a 2D CFT defined on a general surface with a non-flat metric $\mathbf{g}$,  one can still use (local) conformal coordinates where the metric reads $ds^2=\rho(z,\bar z)\, dz d\bar z$, so that the analyticity property of $T$ and $\bar T$ are preserved. $\rho$ is the conformal factor of the metric.
A most important property is that the trace of the stress-energy tensor does not vanish anymore.  
Its expectation value is given by the \emph{trace anomaly}
\begin{equation}
\label{TraceAnomaly}
\langle\tr\mathbf{T}(x)\rangle  = g_{\mu\nu}(x)\langle T^{\mu\nu}(x)\rangle = {c\over 24\,\pi}\, R_{\mathrm{scal}}(x)
\end{equation}
with $c$ the \emph{central charge} of the theory, and $R_{\mathrm{scal}}(x)$ the scalar curvature of the metric $\mathbf{g}$.
The trace anomaly is a quantum anomaly, caused by short distance quantum fluctuations and renormalization effects. See e.g. \cite{Friedan1982} for details.
It can be derived from the \emph{short distance operator product expansion} (OPE) for the stress-energy tensor, which takes the form (for the holomorphic component $T$)
\begin{equation}
\label{TTOPE}
T(z)\,T(z')\ \mathop{=}_{z\to z'}\ {c\over 2} {1\over (z-z')^4}\ +\, \text{subdominant terms}
\end{equation}
\ref{TraceAnomaly} can be obtained from \ref{TTOPE} e.g. by writing $\langle T(z) T(z')\rangle $ as  the functional derivative $\delta\langle T(z)\rangle/\delta g_{\bar z\bar z}(z')$ and comparing with the classical $\delta R_{_\mathrm{scal}}(z)/\delta g_{\bar z\bar z}(z')$ (see e.g. \cite{Friedan1982} or \cite{francesco1997conformal}).

For a discrete statistical model, corresponding to a lattice regularized QFT, conformal invariance is expected {to hold} only at a critical point and in the large distance scaling limit (a famous example is the Ising model). The scaling limit of the model corresponds to a CFT. 
The discretized stress-energy tensor $\mathbf{T}_{\mathrm{reg.}}$ can be defined, but it contains in general short distance UV divergent terms, proportional to negative powers and logarithms of the short distance regulator $a$ (the lattice mesh) or powers the high momentum/energy cut-off $\Lambda\sim 1/a$. By dimensional analysis 
\begin{equation}
\label{ }
\mathbf{T}_{\mathrm{reg.}}\ \propto\,\Lambda^2\sim a^{-2}
\end{equation}
The definition of the continuum limit $a\to 0$ ($\Lambda\to\infty$) requires a renormalization prescription in order to define a renormalized stress-energy tensor $\mathbf{T}$ with the correct properties for conformal invariance (OPE, trace anomaly).

\newcommand{\Sboson}{S_{\scriptscriptstyle{\mathrm{boson}}}}
\newcommand{\Zboson}{Z_{\scriptscriptstyle{\mathrm{boson}}}}
\newcommand{\Tboson}{T_{\scriptscriptstyle{\mathrm{boson}}}}
\newcommand{\Sdelta}{S_{\scriptscriptstyle{\Delta}}}
\newcommand{\Zdelta}{Z_{\scriptscriptstyle{\Delta}}}
\newcommand{\Tdelta}{T_{\scriptscriptstyle{\Delta}}}

\subsection{The two-dimensional boson and the $\Delta$ theory}
\label{AppGFF}
Finally we recall that our results for the Laplace-Beltrami operator $\Delta$ can be interpreted in the framework of the standard free boson CFT {(which has central charge $c=1$)}.
Indeed, {for the classical free boson}, the action $\Sboson$ and the stress-energy tensor are (on a closed Riemannian manifold {$\mathcal{M}$})
\begin{equation}
\label{BosonAction}
\Sboson[\phi]={1\over 2}\int_\mathcal{M} d^2x\,\sqrt{g}\, \partial_\mu\phi\, g^{\mu\nu}\,\partial_\nu\phi={1 \over 2} \, \int_\mathcal{M} d^2x \, \sqrt{g} \, \phi(x)  \Delta_g \phi(x),
\end{equation}
with stress-energy tensor
\begin{equation}
\label{BosonT}
T^{\mu \nu} =
\Big(- {1 \over 2} \, g^{\mu \nu} g^{\rho \sigma} +  g^{\rho \mu} g^{\sigma \nu}  
\Big) \partial_\rho \phi \, \partial_\sigma \phi
\end{equation}
In two-dimensional flat space, using complex coordinates, $\Delta_g= -4\,\partial \bar\partial$. The action and the components of the stress-energy tensor are
\begin{equation}
\label{BosonActionZ}
\Sboson[\phi]= 2\int d^2x\ \partial\phi\,\bar\partial\phi
\end{equation}
\begin{equation}
\label{TBosonZ}
\quad T=- {2\pi} (\partial\phi)^2\ ,\quad\bar T=-{2\pi} (\bar\partial\phi)^2
\ ,\quad\tr \mathbf{T} =T^{z\bar z}=T^{\bar z z}=0
\end{equation}
The last identity shows that the two-dimensional free boson is indeed conformally invariant.
The partition function for the boson is related to the determinant of $\Delta_g$ by the functional integral
\begin{equation}
\label{BosonZ}
\Zboson=\int \frak{D}[\phi]\, \mathrm{e}^{-S[\phi]}=\det (\Delta_g)^{-1/2}
\end{equation}
with $\det(\Delta)$ the properly defined functional determinant of $\Delta$, taking into account renormalization and the zero mode. 

Formally $\det (\Delta_g)=\Zboson^{-2}$ is the partition function of the ``$n=\!-2$ components'' free boson CFT, with $c=-2$. 
Equivalently, a standard trick is to write $\det (\Delta_g)$ as the partition function of a theory for a \emph{scalar complex Grassmann field}:  A spin zero field obeying Fermi-Dirac statistics, described by a pair of conjugate Grassmann (anti-commuting) fields $(\Phi,\bar\Phi)$, where the $\Phi(x)$'s and $\bar\Phi(x)$'s are the generators of an infinite dimensional Grassmann (or exterior) algebra. 
The partition function $\Zdelta$ is given by {a} Berezin functional integral (see e.g. \cite{francesco1997conformal}, \cite{QFandSforM1999} and the original reference \cite{Berezin1966}). It reads, using the Berezin integration notation
\begin{equation}
\label{LBZ}
\Zdelta= \det\Delta_g = \int \frak{D}[\Phi,\!\bar\Phi]\, \mathrm{e}^{-S[\Phi,\bar\Phi]}
\ ,\qquad \frak{D}[\Phi{,}\bar\Phi]=\prod_x d\Phi(x) d\bar\Phi(x)
\end{equation}
with the action $\Sdelta$ (here a degree 2 element of the Grassmann algebra) 
which is simply the Grassmann version of the action for a complex bosonic scalar field
\begin{equation}
\label{LBaction}
\Sdelta[\Phi,\bar\Phi]=4 \int\! d^2x\ \partial\Phi\,\bar\partial \bar\Phi = \int\! d^2x\ \Phi{\cdot}\Delta_g\bar\Phi
\end{equation}
Of course, unlike the bosonic case, the Berezin functional integral cannot be thought in terms of probabilistic averages over random real or complex 
fields ``living'' on a space-time manifold, but as an algebraic construction. 
In the fermionic theory, the two-point functions (the propagator) are (note the anti-commutivity)
\begin{equation}
\label{LBprop}
\langle\bar\Phi(x)\Phi(y)\rangle=-\langle\Phi(x)\bar\Phi(y)\rangle=\left[\Delta_g^{-1} \right]_{xy}
\,,\ \ 
\langle\Phi(x)\Phi(y)\rangle=\langle\bar\Phi(x)\bar\Phi(y)\rangle=0
\end{equation}
The stress-energy tensor components are
\begin{equation}
\label{TLBzzb}
T_{\!_\Delta}=
- {4\pi}\, \partial\Phi\,\partial\bar\Phi\ ,\quad \bar T_{\!_\Delta}=
-{4\pi}\, \bar\partial\Phi\,\bar\partial\bar\Phi
\ ,\quad \tr \mathbf{T}_{\!_\Delta} =0
\end{equation}
As explained in the discussion section \ref{sDiscussion}, our results for the variations of the discretized Laplacians $\Delta$, $\Deltaconf$ and the K\"ahler operator $\mathcal{D}$ (defined on  a triangulation $\uT$) can be easily formulated in terms of  discretized stress-energy tensors attached to the faces of $\uT$. However, only for the Laplace-Beltrami operator $\Delta$ can the discretized stress energy tensor be given a simple continuum limit  formulation as the stress-energy tensor of a continuum QFT.

\subsection{The conformal ghost-antighosts theory}
\label{AppBCsystem}
For completeness we recall what is the ghost-antighost CFT theory for two-dimensional gravity. Two dimensional gravity is a quantum theory for the Riemannian 2d metric tensor $\boldsymbol{g}=(g_{\mu\nu})$ on a Riemann surface (e.g. the sphere). It must be invariant under local diffeomorphisms
\begin{equation}
\label{DiffXi}
x^\mu\to x'^\mu=x^\mu + \epsilon\, \xi^\mu\ ,\quad g_{\mu\nu}\to g_{\mu\nu} -\epsilon\, ( D_\mu \xi_\nu+ D_\nu\xi_\mu)
\end{equation}
with $\boldsymbol {\xi}=(\xi^\mu)$ a vector field, and $\boldsymbol{D}=(D_\mu)$ the covariant derivative in the metric $\boldsymbol{g}$.
In Polyakov's formulation (see the original article by Polyakov on the bosonic string  \cite{Polyakov81}, and the Les Houches lecture notes by Friedan \cite{Friedan1982} for details), the diffeomorphism local invariance is fixed by the conformal gauge. A background classical metric $\bar{\boldsymbol{g}}=(\bar g_{\mu\nu})$ is chosen and the metric are fixed to be conformal w.r.t.  $\bar{\boldsymbol{g}}$, i.e. of the form
\begin{equation}
\label{ConfMetric}
g_{\mu\nu}(x)=\Lambda(x)\,\bar g_{\mu\nu}(x)\ ,\ \Lambda(x)\ \text{conformal factor}
\end{equation}
This amounts to enforcing the local gauge fixing condition
\begin{equation}
\label{ }
\bar K^{\mu\nu}=\bar g^{\mu\rho} g_{\rho\sigma}\bar g^{\sigma\nu} - {1\over 2} \bar g^{\rho\sigma}g_{\sigma\rho}\ \bar g^{\mu\nu} = 0
\end{equation}
The variation of the gauge fixing term $\bar K$ under a general diffeomorphism \ref{DiffXi} is, when deforming a conformal metric of the form \ref{ConfMetric}
\begin{equation}
\label{ }
\bar K^{\mu\nu}=0 \ \to\ \bar K^{\mu\nu}=-\epsilon\, \Lambda\,\left(  \bar D^\mu \xi^\nu + \bar D^\nu \xi^\mu - \bar g^{\mu\nu}\, \bar D_\tau \xi^\tau\right) 
\end{equation}
with $\bar{\boldsymbol{D}} =(\bar D_\mu)$ the covariant derivative w.r.t. the background metric $\boldsymbol{g}$. It can be written as
\begin{equation}
\label{ }
-\epsilon\,(\mathbf{J}\cdot \xi)^{\mu\nu}
\end{equation}
where $\mathbf{J}$ is a differential operator which maps a vector field $\boldsymbol{\xi}$ onto a symmetric traceless tensor  (w.r.t. the background metric $\bar{\boldsymbol{g}}$).
Quantizing the metric $\boldsymbol{g}$ in the conformal gauge  gives in the functional integral a Fadeev-Popov determinant, which can be written as a Grassmann functional integral in terms of two anticommuting ghost fields,  $\boldsymbol{c}$ and $\boldsymbol{b}$, where
\begin{equation}
\label{ }
\boldsymbol{c}=(c^\mu)\quad\text{is a type $(1,0)$ tensor}
\end{equation}
and 
\begin{equation}
\label{ }
\boldsymbol{b}=(b_{\mu\nu})\quad\text{is a type $(0,2)$ symmetric traceless tensor}
\end{equation}
such that $b_{\mu\nu}=b_{\nu\mu}$  and $\bar g^{\nu\mu}b_{\mu\nu}=0$.
The Fadeev-Popov determinant reads
\begin{equation}
\label{ }
\det[\mathbf{J}] \ =\ \int \frak{D}[\boldsymbol{b},\boldsymbol{c}]\ \mathrm{e}^{\boldsymbol{b}\cdot\mathbf{J}\cdot\boldsymbol{c}}
\end{equation}
The action for the $\boldsymbol{b}\text{-}\boldsymbol{c}$ system is (here in the background metric $\bar{\boldsymbol{g}}$)
\begin{equation}
\label{ }
S_{\texttt{ghost}}[\boldsymbol{b},\boldsymbol{c}]\ =\ \boldsymbol{b}\cdot\mathbf{J}\cdot\boldsymbol{c}= \int d^2x\, \sqrt{\bar g}\,\ b_{\mu\nu} (\bar D^\mu c^\nu + \bar D^\nu c^\mu - \bar g^{\mu\nu} \bar D_\tau c^\tau)
\end{equation}
The symmetric stress-energy tensor for this ghost action is
\begin{equation}
\label{ }
T^{\mu\nu}_{\texttt{ghost}}= b^{\mu\tau}\bar D^\nu c_\tau + b^{\nu\tau} \bar D^\mu c_\tau + \bar D_\tau b^{\mu\nu} c^\tau - \bar g^{\mu\nu} b^{\tau\rho} \bar D_\tau c_\rho
\end{equation}
As shown by Polyakov in \cite{Polyakov81}, this $\boldsymbol{b}\text{-}\boldsymbol{c}$ system is a conformal theory (CFT) with central charge $c=-26$.
As a consequence, when fixing the conformal gauge \ref{ConfMetric} in the functional integral for 2D gravity, which is
\begin{equation}
\label{ }
Z\ =\ \int \frak{D}[{g}]\ \mathrm{e}^{-\int_M d^2x\,\mu_0\,\sqrt{g}}
\end{equation} 
the resulting effective action for the remaining conformal factor
\begin{equation}
\label{ }
\Lambda(x)=\exp(\varphi(x))
\qquad\text{with}\quad\varphi(x)\quad\text{the Liouville field}
\end{equation}
is the Liouville action, which defines the Liouville 2D gravity model.


\section{Proof of Lemma \ref{lemmabound}}
\label{prooflemmabound}
\begin{proof}
For $j=2, 3$ introduce interpolations $z_j(t) := tz_j + (1-t)z_1$ between $z_j$ and $z_1$. In addition
set $z(s,t) := sz_3(t) + (1-s)z_2(t)$. We start from the definition of $\nabla$
\begin{equation}
\label{DefNabla2}
\D \nabla \phi(\uf) \ = \  { \big[ \phi(z_2) - \phi(z_1) \big] \big[ \overline{z}_3 - \overline{z}_1 \big]  \, - \,
\big[ \phi(z_3) - \phi(z_1) \big] \big[ \overline{z}_2 - \overline{z}_1 \big] 
\over { -4 \imath {A(\uf)} }}
\end{equation}
where by formula \ref{Rcform} we have for the area of the triangle $\uf$
\begin{equation}
\label{Rcformbis}
4\, {A(\uf)} = | z_1 -z_2 | \, |z_2- z_3  | \, | z_3 -z_1|\, /\, R(\uf)
\end{equation}
The numerator can be expressed by
\begin{equation}
\label{numeratornabla}
\begin{array}{c}
\D \big[ \phi(z_2) - \phi(z_1) \big] \big[ \overline{z}_3 - \overline{z}_1 \big]  \, - \,
\big[ \phi(z_3) - \phi(z_1) \big] \big[ \overline{z}_2 - \overline{z}_1 \big]  \\ \\
\D = \int_0^1 \, dt \, {d \over {dt}} \Bigg[ \phi\big( z_2(t) \big) \big[ \overline{z}_3 - \overline{z}_1 \big] \, - \, \phi \big( z_3(t) \big) 
\big[ \overline{z}_2 - \overline{z}_2 \big] \Bigg] \\ \\
\D = 
\left\{
\begin{array}{cl}
\D \int_0^1 \, dt \, 
\Bigg[ \big[z_2 - z_1 \big] \big[ \overline{z}_3 - \overline{z}_1 \big] \partial \phi \big(z_2(t) \big) \, - \, 
\big[z_3 - z_1 \big] \big[ \overline{z}_2 - \overline{z}_1 \big] \partial \phi \big(z_3(t) \big)  \Bigg] 
&\D \text{$*$-integral} \\ 
\D + \\ 
\D \int_0^1 \, dt \, 
\Bigg[ \big[\overline{z}_2 - \overline{z}_1 \big] \big[ \overline{z}_3 - \overline{z}_1 \big] \Big[ \partial \phi \big(z_2(t) \big) \, - \, \partial \phi \big(z_3(t) \big) \Big] \Bigg] 
&\D \text{$**$-integral}
\end{array}
\right.
\end{array}
\end{equation}

\noindent
Apply the fundamental theorem of calculus once again, the $*$-integral in \ref{numeratornabla} can be expressed  as a double integral
\begin{equation}
\begin{array}{l}
\D -\int_0^1 \int_0^1 \, dt \, ds \ {d \over {ds}} 
\Bigg[ \partial \phi \big( z(s,t) \big) \Big( s [ z_3 - z_1 ] \big[ \overline{z}_2 - \overline{z}_1 \big]  + 
(1-s) [ z_2 - z_1] \big[ \overline{z}_3 - \overline{z}_1 \big] \Big)\Bigg] \\ \\
\D= \left\{
\begin{array}{c}
\D \int_0^1 \int_0^1 \, dt \, ds \ \partial \phi \big( z(s,t) \big) 
\underbrace{ \Big(
[ z_2 - z_1] \big[ \overline{z}_3 - \overline{z}_1 \big]  
- [ z_3 - z_1] \big[ \overline{z}_2 - \overline{z}_1 \big] \Big) }_{= \ - 4 \imath {A(\uf)} } 
\\ 
\D + \\ 
\D  \int_0^1 \int_0^1 \, t\, dt \, ds \ \partial \partial \phi \big( z(s,t) \big)  [ z_2 - z_3]
\Big( s [ z_3 - z_1] \big[ \overline{z}_2 - \overline{z}_1 \big]  + 
(1-s)  [ z_2 - z_1] \big[ \overline{z}_3 - \overline{z}_1 \big] \Big) \\ \\
\D + \\ 
\D  \int_0^1 \int_0^1 \, t\, dt \, ds \ \partial \overline{\partial} \phi \big( z(s,t) \big) \big[ \overline{z}_2 - \overline{z}_3 \big] 
\Big( s [ z_3 - z_1] \big[ \overline{z}_2 - \overline{z}_1 \big]  + 
(1-s) [ z_2 - z_1] \big[ \overline{z}_3 - \overline{z}_1 \big] \Big) \\ \\
\end{array}
\right.
\end{array}
\end{equation}

Dividing  the $\text{$*$-integral}$ in \ref{numeratornabla} by $(- 4 \frak{Im} {A(\uf)})$ we obtain a first contribution to $\nabla\phi(\uf)$, namely
\begin{equation}
\begin{array}{l}
\D
 \left\{
\begin{array}{c}
\D \int_0^1 \int_0^1 \, dt \, ds \ \partial \phi \big( z(s,t) \big)  \\ 
\D + \\ 
\D \imath R(\uf) \,  \int_0^1 \int_0^1 \, t dt \, ds \ \partial \partial \phi \big( z(s,t) \big)  
{z_2 - z_3  \over {|z_2 - z_3|}}
\Bigg( s {z_3 - z_1 \over {|z_3 - z_1|}} 
{\overline{z}_2 - \overline{z}_1 \over {|z_2 - z_1 |}}  + 
(1-s)  {z_2 - z_1 \over {|z_2 - z_1|}} 
{\overline{z}_3 - \overline{z}_1 \over {|z_3 - z_1|} } 
\Bigg) \\ 
\D + \\ 
\D  \imath R(\uf) \, \int_0^1 \int_0^1 \, t dt \, ds \ \partial \overline{\partial} \phi \big( z(s,t) \big) 
{\overline{z}_2 - \overline{z}_3 \over {|z_2 - z_3|}}
\Bigg( s {z_3 - z_1 \over {|z_3 -z_1|}} 
{\overline{z}_2 - \overline{z}_1 \over {|z_2 -z_1|}} + 
(1-s) {z_2 - z_1 \over {|z_2 -z_1|}}
{\overline{z}_3 - \overline{z}_1 \over {|z_3 - z_1|}}
\Bigg) 
\end{array}
\right.
\end{array}
\end{equation}
Again, by the fundamental theorem of calculus, we can transform the 
$**$-integral \ref{numeratornabla} and obtain
\begin{equation}
\begin{array}{l}
\D \int_0^1 \, dt \, 
\Bigg( \big[\overline{z}_2 - \overline{z}_1 \big] \big[ \overline{z}_3 - \overline{z}_1 \big] \Big[  \partial \phi \big(z_2(t) \big) \, - \, \partial \phi \big(z_3(t) \big) \Big] \Bigg) \\ 
\D= - \big[\overline{z}_2 - \overline{z}_1 \big] \big[ \overline{z}_3 - \overline{z}_1 \big] \, 
\int_0^1 \int_0^1 \, dt \,ds \, {d \over {ds}}  
\Bigg( \partial \phi \big(z(s,t) \big) \Bigg) \\ 
\D=  \big[\overline{z}_2 - \overline{z}_1 \big] \big[ \overline{z}_3 - \overline{z}_1 \big] \, 
\int_0^1 \int_0^1 \, t\, dt \, ds \, \Bigg( [z_2 - z_3] \, \partial \overline{\partial} \phi \big( z(s,t) \big) + [\overline{z}_2 - \overline{z}_3] \, \overline{\partial} \overline{\partial}
\phi \big(  z(s,t) \big) \Bigg) 
\end{array}
\end{equation}
Dividing  the $\text{$**$-integral}$ in \ref{numeratornabla} by $(- 4 \frak{Im} {A(\uf)})$ we obtain a second contribution to $\nabla\phi(\uf)$, namely

\begin{equation}
\begin{array}{l}
\D 
 \imath R(\uf) \, 
{\overline{z}_2 - \overline{z}_1 \over {|z_2 - z_1|}}
{\overline{z}_3 - \overline{z}_1 \over {|z_3 - z_1|}} \, 
\int_0^1 \int_0^1 \, t\,dt \, ds \, \Bigg(
{z_2 - z_3 \over {|z_2- z_3|}}  
\, \partial \overline{\partial} \phi \big( z(s,t) \big) + 
{\overline{z}_2 - \overline{z}_3 \over {|z_2 - z_3|}}
\, \overline{\partial} \overline{\partial}
\phi \big(  z(s,t) \big) \Bigg)
\end{array}
\end{equation}

\noindent
So we end up with
\begin{equation}
\label{thediff1}
\begin{split}
&\nabla\phi(\uf)-\int_0^1\int_0^1 dt\,ds\,\partial \phi \big( z(s,t) \big)\\
&
=\left\{
\begin{array}{c}
\D \imath R(\uf) \,  \int_0^1 \int_0^1 \, t\, dt \, ds \ \partial \partial \phi \big( z(s,t) \big)  
{z_2 - z_3  \over {|z_2 - z_3|}}
\Bigg( s {z_3 - z_1 \over {|z_3 - z_1|}} 
{\overline{z}_2 - \overline{z}_1 \over {|z_2 - z_1 |}}  + 
(1-s)  {z_2 - z_1 \over {|z_2 - z_1|}} 
{\overline{z}_3 - \overline{z}_1 \over {|z_3 - z_1|} } 
\Bigg) \\ 
\D + \\ 
\D  \imath R(\uf) \, \int_0^1 \int_0^1 \, t\, dt \, ds \ \partial \overline{\partial} \phi \big( z(s,t) \big) 
{\overline{z}_2 - \overline{z}_3 \over {|z_2 - z_3|}}
\Bigg( s {z_3 - z_1 \over {|z_3 -z_1|}} 
{\overline{z}_2 - \overline{z}_1 \over {|z_2 -z_1|}} + 
(1-s) {z_2 - z_1 \over {|z_2 -z_1|}}
{\overline{z}_3 - \overline{z}_1 \over {|z_3 - z_1|}}
\Bigg) \\ 
\D + \\ 
\D \imath R(\uf) \, 
{\overline{z}_2 - \overline{z}_1 \over {|z_2 - z_1|}}
{\overline{z}_3 - \overline{z}_1 \over {|z_3 - z_1|}} \, 
\int_0^1 \int_0^1 \, t\,dt \, ds \, \Bigg(
{z_2 - z_3 \over {|z_2- z_3|}}  
\, \partial \overline{\partial} \phi \big( z(s,t) \big) + 
{\overline{z}_2 - \overline{z}_3 \over {|z_2 - z_3|}}
\, \overline{\partial} \overline{\partial}
\phi \big(  z(s,t) \big) \Bigg)
\end{array}
\right.
\end{split}
\end{equation}

\noindent
Thus we can bound the norm of the r.h.s. of \ref{thediff1} by
\begin{equation}
\label{ }
R(\uf)\int_0^1\int_0^1 t\,dt\, ds \left(|\partial\partial\phi(z(s,t)|+2|\partial\bar\partial\phi(z(s,t)|+|\bar\partial\bar\partial\phi(z(s,t)|\right)
\end{equation}
Thus we have
\begin{equation}
\label{bound1}
\left|\nabla\phi(\uf)-\int_0^1\int_0^1 dt\,ds\,\partial \phi \big( z(s,t) \big)\right|
\le
R(\uf) \Bigg( {1 \over 2} \, \sup_{z \in \uf} \Big| \partial \partial \phi (z) \Big| 
\, + \, \sup_{z \in \uf} \Big| \partial \overline{\partial} \phi (z) \Big| \, + \, 
{1 \over 2} \, \sup_{z \in \uf} \Big| \overline{\partial} \overline{\partial} \phi (z) \Big| \Bigg)
\end{equation}
Finally we come to bound the difference 
between the $\partial \phi \big( z(s,t) \big)$ and $\partial \phi(z_\uf)$ where $z_\uf$ is the circumcenter of $\uf$.
Again, by the fundamental theorem of calculus, defining
$$z(p,s,t)=p\, z(s,t) + (1-p) z_\uf$$
we write
\begin{equation}
\label{ }
\begin{split}
 \partial \phi &\big( z(s,t) \big)-\partial \phi(z_\uf)=
\int_0^1 dp \, {d \over {dp}} \partial \phi \big( z(p,s,t) \big)
\\
=& \int_0^1 dp\, \big((z(s,t) - z_\uf)\partial \partial \phi(z(p,s,t)) + ({\bar z}(s,t) - {\bar z}_\uf)\partial \bar\partial \phi(z(p,s,t))\big)
\end{split}
\end{equation}
Since $z(s,t)$ is inside the triangle $\uf$, it is also in the disk $B_\uf$ of radius $R(\uf)$ with center $z_\uf$, 
hence $|z(s,t) - z_\uf|\le R(\uf)$
and we get the bound

\begin{equation}
|\partial \phi \big( z(s,t) \big)-\partial \phi(z_\uf)|\le R(\uf)\left(
\sup_{z \in B_\uf} \Big| \partial \partial \phi (z) \Big| \, + \,  \, \sup_{z \in B_\uf} \Big| \partial \overline{\partial} \phi (z) \Big|
\right) 
\end{equation}

\noindent
which when averaged becomes

\begin{equation}
\label{bound2}
\Bigg| \int_0^1\int_0^1 dt\,ds\,\partial \phi \big( z(s,t) \big)
- \partial \phi(z_\uf) \Bigg| \le R(\uf)\left(
\sup_{z \in B_\uf} \Big| \partial \partial \phi (z) \Big| \, + \,  \, \sup_{z \in B_\uf} \Big| \partial \overline{\partial} \phi (z) \Big|
\right) 
\end{equation}

\noindent
Combining the bounds \ref{bound1} and \ref{bound2} we get the final result of lemma~\ref{lemmabound}
\begin{equation}
\Big| \nabla \phi (\uf )   - \partial \phi (z_\uf )\Big| \leq R(\uf) \, \Bigg(  \, {3 \over 2} \, 
\sup_{z \in B_\uf}
\big| \partial^2 \phi  \big| \, + \, 
2 \,  \sup_{z \in B_\uf} \big| \partial \overline{\partial} \phi \big| \, + \,  {1 \over 2} \, \sup_{z \in B_\uf}
\big| \overline{\partial}^2  \phi \big| \, \Bigg)
\end{equation}
\end{proof}

\begin{Rem}
\label{robust-lemma}
For a general point $w \in B_\uf$ we have $|z(s,t) - w |\le 2R(\uf)$ and after modifying our estimates by
a factor of $2$ we obtain
\begin{equation}
\label{eqRobustLemma}
\Big| \nabla \phi (\uf )   - \partial \phi (w) \Big| \leq R(\uf) \, \Bigg(  \, {5 \over 2} \, 
\sup_{z \in B_\uf}
\big| \partial^2 \phi  \big| \, + \, 
3 \,  \sup_{z \in B_\uf} \big| \partial \overline{\partial} \phi \big| \, + \,  {1 \over 2} \, \sup_{z \in B_\uf}
\big| \overline{\partial}^2  \phi \big| \, \Bigg)
\end{equation}
\end{Rem}

\newpage
\section{Continuum limits of curvature anomalies: an example}
\label{tiling-by-cocyclic-quad}

In this appendix we present an example of an isoradial Delaunay graph $\uG_\mathrm{cr}$
for which the
anomalous terms
of the associated conformal Laplacian $\Deltaconf$
(as defined in formulae \ref{var-anomalous}
and examined in equations \ref{chord2chord}, and \ref{chord2edge} of Section \ref{ssConfLApDEl})
have well-defined, non-trivial
$\ell \rightarrow \infty$ scaling limits. 
Unlike the continuum limits addressed in Corollary \ref{finalOPElike1},
the limit values of the anomalous edge-to-chord, chord-to-edge, and chord-to-chord terms
computed in Proposition \ref{continuum-limits-anomalies} 
of this section
reflect features of the underlying geometry of the 
initial critical graph $\uG_\mathrm{cr}$, specifically
the choice of fundamental quadrilateral $\mathcal{Q}$
used to construct $\uG_\mathrm{cr}$. See Figure \ref{CyclicQuadTiling}.
We emphasize that this is a very specific example; for ``generic'' isoradial Delaunay graph $\uG_\mathrm{cr}$, no such continuum limit exist.

Begin with four angles $\alpha_1 < \alpha_2 < \alpha_3 < \alpha_4$ in the interval $[0, 2\pi)$ and construct the cyclic quadrilateral $\mathcal{Q}$
whose vertices are the unit complex numbers $\frak{z}_k := \exp (\mathrm{i} \alpha_k)$
with $k \in \{1,2,3,4 \}$. We will require that the origin is contained in the interior of $\mathcal{Q}$, which is
achieved whenever $\alpha_3- \alpha_1 > \pi$ or $\alpha_4 -\alpha_2 > \pi$. This constraint insures that
the tiling we are about to construct is Delaunay. 
Let $\mathcal{Q}^\mathrm{op}$ denote the quadrilateral obtained by rotating $\mathcal{Q}$ by $\mathrm{180}$ degrees.
A cyclic quadrilateral 
with associated angles $\alpha_1 = \pi/3$, $\alpha_2 = 5\pi/7$, $\alpha_3 = 13\pi/9$,
and $\alpha_4 = 21\pi / 11$ is illustrated in Figure \ref{FundamentalQuad}.

\begin{figure}[h]
\begin{center}
\raisebox{-.8in}{\includegraphics[width=2.3in]{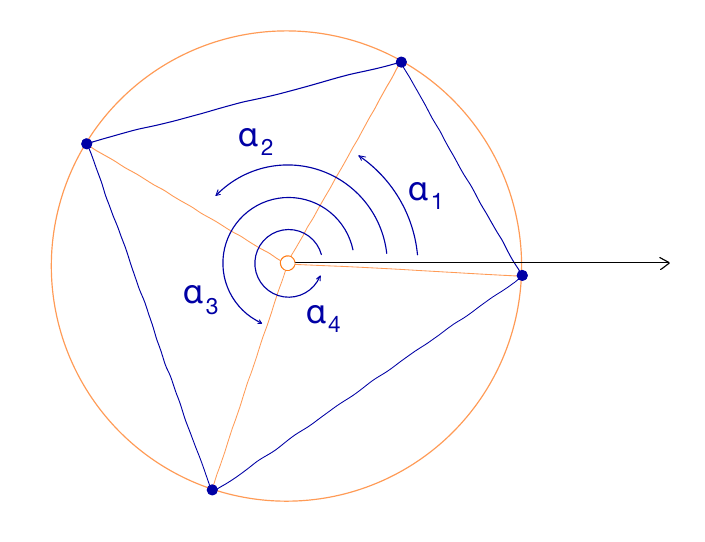}}
\caption{The fundamental quadrilateral $\mathcal{Q}$ considered in the example}
\label{FundamentalQuad}
\end{center}
\end{figure}

Construct a doubly periodic, quadrilateral tiling $\uG_\mathrm{cr}$ of the plane using translations of $\mathcal{Q}$ and $\mathcal{Q}^\mathrm{op}$. 
Clearly $\uG_\mathrm{cr}$ will be isoradial and Delaunay in the sense 
of Section \ref{sssNotations}; by construction each face of $\uG_\mathrm{cr}$ is a cyclic quadrilateral.
Figure \ref{CyclicQuadTiling} depicts such a tiling.

\begin{figure}[h]
\begin{center}
\raisebox{-.8in}{\includegraphics[width=2.3in]{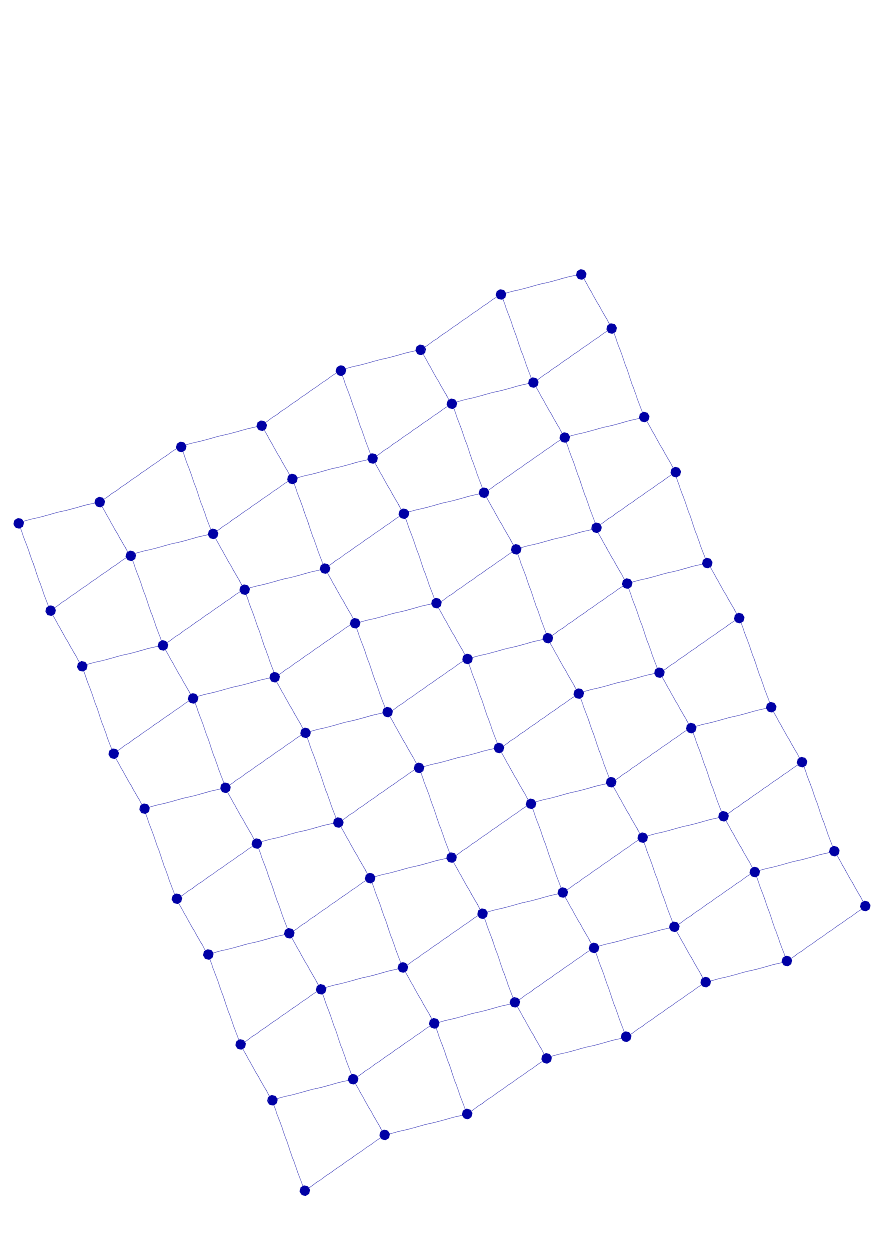}}
\caption{Fragment of a tiling $\uG_\mathrm{cr}$ by a cyclic quadrilateral $\uq$}
\label{CyclicQuadTiling}
\end{center}
\end{figure}

For each quadrilateral face $\uq$ of $\uG_\mathrm{cr}$ let $z_\uq$ denote the complex coordinate of its center;
with respect to this center, the four vertices $\mathrm{v}_\uq(k)$ of $\uq$, with $k \in \{1, 2, 3, 4 \}$, 
have complex coordinates $z(\mathrm{v}_\uq(k)) = z_\uq \, \pm \, \frak{z}_k$
where the sign is $+$ if $\uq$ is a translation of $\mathcal{Q}$ and $-$ if $\uq$ is a translation of $\mathcal{Q}^\mathrm{op}$.
Let $\mathrm{e}^+_\uq$ denote the chord
of the quadrilateral $\uq$ joining vertices $\mathrm{v}_\uq(2)$ and $\mathrm{v}_\uq(4)$
while $\mathrm{e}^-_\uq$ will denote the chord 
joining $\mathrm{v}_\uq(1)$ and $\mathrm{v}_\uq(3)$. Up to a sign, the corresponding north angles are given by
$\vartheta_+ := {\alpha_2 - \alpha_4}$ and $\vartheta_-:=  \alpha_1 - \alpha_3$ respectively. Define $z_+:= \frak{z}_2 - \frak{z}_4$
and $z_- :=  \frak{z_1} - \frak{z}_3$. Let $A_\mathcal{Q}$ denote the area of $\mathcal{Q}$.

Let $F(z)$ be a smooth complex-valued function 
with compact support together 
with deformation and scaling parameter
values $\epsilon > 0$ and $\ell > 0$.
Let $\uG_{\epsilon,\ell}$ denote
the graph obtained by 
deforming the embedding of $\uG_\mathrm{cr}$
by $z \mapsto z + \epsilon \, \ell F(z/\ell)$
and then adjoining edges
$\mathrm{e}^+_\uq$ or $\mathrm{e}^-_\uq$
to those quadrilateral faces $\uq$ of $\uG_\mathrm{cr}$
according to whether 
$\theta_{\epsilon,\ell}( \mathrm{e}^+_\uq ) > 0$ or
$\theta_{\epsilon,\ell}( \mathrm{e}^-_\uq ) > 0$
respectively;
these conditions are mutually exclusive, as the
signs of 
$\theta_{\epsilon,\ell}( \mathrm{e}^+_\uq )$
and 
$\theta_{\epsilon,\ell}( \mathrm{e}^-_\uq ) $
are opposite.
Neither edge is selected if both conformal angles are zero. 
As long as $\epsilon > 0$ lies within the 
range $0 < \epsilon < \tilde{\epsilon}_{F , \ell}$ 
prescribed by Prop.~\ref{epsilontildeF} 
the graph $\uG_{\epsilon,\ell}$ will remain Delaunay.

\bigskip
\noindent
As an example consider the following "mollified" shear of $\uG_\mathrm{cr}$. For simplicity we 
consider the case where the support of $F$ has {\bf one} connected component (in particular, it is a disk $\Bbb{D}$
with unit radius):

\[ F(z) \ := \ \left\{ 
\begin{array}{cl}
\D \exp \Bigg( i\phi \, + \, {|z|^2 \over {|z|^2 - 1}} \Bigg) \, \frak{Im}[z] 
&\D \text{if $|z| \leq 1$} \\ \\
\D 0 
&\D \text{otherwise} 
\end{array}
\right. \]

\bigskip
\noindent
Figure \ref{mollified-shear} depicts the effect of 
of the corresponding deformation $z \mapsto z + \epsilon \ell F(z/\ell)$.
The reader will notice
that the support of $F_\ell: z \mapsto \ell F(z/\ell)$ is partitioned roughly into three "unidirectional" zones 
consisting of deformed quadrilaterals whose diagonals share the same alignment.
In general, for any smooth compactly supported perturbation $z \mapsto z + \epsilon \, \ell F(z/\ell)$,
the support of $F_\ell$ will be partitioned into 
such zones of constant alignment.
If we ignore the 
quadrilaterals $\uq$ for 
which $\theta'_{0,\ell}(\mathrm{e}^+_\uq )$ vanishes
then the remaining set of quadrilaterals can be partitioned
into zones over which the sign of $\theta'_{0,\ell}(\mathrm{e}^+_\uq)$ is constant.
For $\ell >> 0$ large, the interfaces between these zones
approximate the level curves of
$\frak{Im} \big[ \, \overline{\partial} F_\ell \, \mathcal{E} \big] = 0$
within the disk $ \Bbb{D}_\ell$ of radius $\ell$ where

\[ \mathcal{E} := \ \frak{e}_{12} \, - \, \frak{e}_{23} \, + \, \frak{e}_{34} \, - \, \frak{e}_{14}
\ \  \text{and} \ \
\D \frak{e}_{mn} := {{ \overline{\frak{z}}_m - \overline{\frak{z}}_n }  \over { \frak{z}_m - \frak{z}_n }}
\ \ \text{for}  \ m, n \in \{1,2,3,4\}. \]

\noindent
The appearance of continuous interfaces is
a prodigy (of the existence) of the scaling limit for the anomaly, as
formalized in Lemma \ref{uniform-convergence} and Proposition \ref{continuum-limits-anomalies}.
In the case of the mollified-shear example, the corresponding level curves  are depicted in red by Figure \ref{mollified-shear}. 


\begin{figure}[h]
\begin{center}
\raisebox{-.8in}{\includegraphics[width=2.3in]{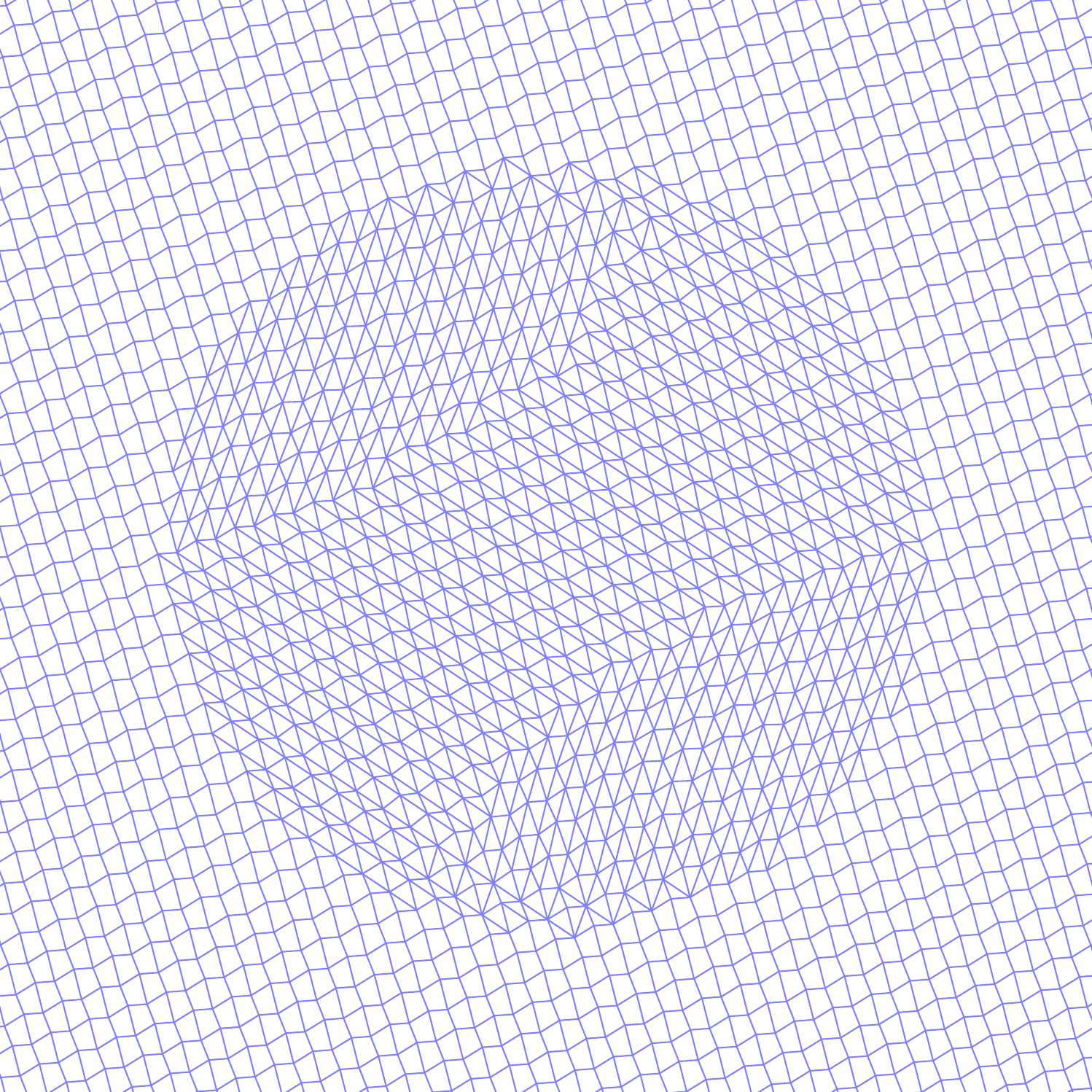}}
\raisebox{-.8in}{\includegraphics[width=2.3in]{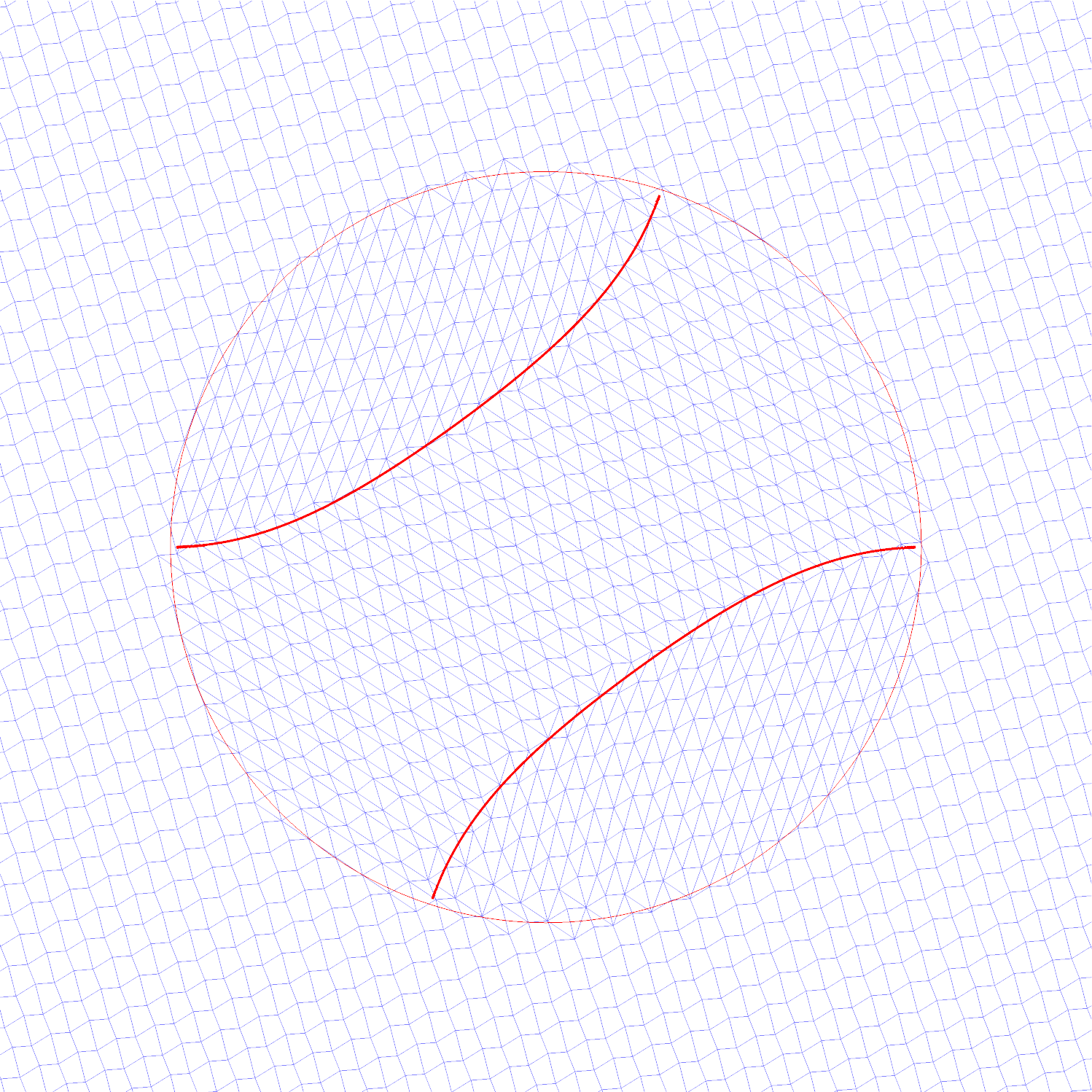}}
\caption{mollified-shear with angle value $\phi = -{\pi \over 5}$, 
deformation parameter value $\epsilon =0.1 $, and scaling parameter value $\ell = 22$}
\label{mollified-shear}
\end{center}
\end{figure}

In order to analyze the anomalous terms arising in the 
second order variation of the conformal Laplacian 
we return to using a scaled, bi-local perturbation as 
prescribed in Section \ref{sScalingLimit}.
As before $F_1(z)$ and $F_2(z)$ are complex-valued functions whose
supports $\Omega_1$ and $\Omega_2$ are compact 
and whose lattice closures $\overline{\Omega}_1$ and
$\overline{\Omega}_2$ are disjoint. In addition
$\underline{\epsilon} =(\epsilon_1, \epsilon_2 )$
is a pair of independent deformation parameters 
and $\ell > 0$ is a scaling parameter. 
Let $\uG_{\underline{\epsilon}, \ell}$ denote
the Delaunay graph associated to the
deformed embedding

$$ z_{\underline{\epsilon}, \ell} (\uv) :=  z(\uv) + \epsilon_1 F_{1; \ell}( z(\uv) ) + \epsilon_2 F_{2; \ell} ( z (\uv)) $$

\noindent
where the deformation parameters
are constrained with the range $0 \leq \epsilon_1 ,\epsilon_2 < \min( \tilde{\epsilon}_{F_1}, \tilde{\epsilon}_{F_2})$ 
whose bounds $\tilde{\epsilon}_{F_1}, \tilde{\epsilon}_{F_2}$
are specified in Prop.~\ref{epsilontildeF}.



Given $p \in \Bbb{C}$ and a value of the scaling parameter $\ell > 0$ 
center a copy of the fundamental quadrilateral $\mathcal{Q}$
about the dilated point $\ell p \in \Bbb{C}$. The 
coordinates of its vertices are $q_\ell(p;k) = \ell p + \frak{z}_k$ for $k \in \{1, 2, 3, 4 \}$.
The perturbation will displace these vertices by
$q_\ell(p;k) \mapsto q_{\underline{\epsilon},\ell}(p;k)$
where 

$$q_{\underline{\epsilon},\ell}(p;k) := q_\ell(p;k) + \epsilon_1 F_{1;\ell} \big( q_\ell(p;k)  \big) + \epsilon_2 F_{2; \ell} \big( q_\ell(p;k) \big)$$

\noindent
The conformal angle $\kappa_{\underline{\epsilon},\ell}(p)$ and its $\epsilon_i$-derivatives
$\thedelta_{\epsilon_i} \kappa_\ell(p)$ are accordingly defined by:  

\begin{equation}
\begin{array}{ll}
\D \kappa_{\underline{\epsilon},\ell}(p) 
&\D = \ \frak{Im} \log \left[
{{ \big( q_{\underline{\epsilon},\ell}(p;4) - q_{\underline{\epsilon},\ell}(p;3) \big) \, \big( q_{\underline{\epsilon},\ell}(p;2) - q_{\underline{\epsilon},\ell}(p;1) \big) } 
\over {\big( q_{\underline{\epsilon},\ell}(p;4) - q_{\underline{\epsilon},\ell}(p;1) \big) \, \big(q_{\underline{\epsilon},\ell}(p;2) - q_{\underline{\epsilon},\ell}(p;3) \big)} } \right]  
\end{array}
\end{equation}

\begin{equation}
\label{derivative-conformal-angle-at-zero}
\begin{array}{ll}
\D  \thedelta_{\epsilon_i} \kappa_\ell(p)
&\D = \ {\partial \over {\partial \epsilon_i} } \, \Bigg|_{\epsilon_i = 0} \kappa_{\underline{\epsilon},\ell}(p) \\ \\
&\D =  \
\left\{
\begin{array}{c}
\D \frak{Im} \left[ \overline{\nabla} F_i \Big( p + \frak{z}_1/\ell, \, p + \frak{z}_2/\ell, \, p + \frak{z}_4/\ell \Big) \big( \frak{e}_{12} - \frak{e}_{14} \big) \right] \\ 
\D + \\ 
\D \frak{Im} \left[ \overline{\nabla} F_i \Big(p + \frak{z}_2/\ell, \, p + \frak{z}_3/\ell, \, p + \frak{z}_4/\ell \Big) \big( \frak{e}_{34} - \frak{e}_{23} \big) \right] 
\end{array} 
\right. \\ \\
&\D = \ \frak{Im} \Big[ \, \overline{\partial} F_i(p) \, \mathcal{E} \Big] \ + \ \mathrm{O}\big( 1/ \ell \big)
\end{array}
\end{equation}

\begin{lemma}
\label{uniform-convergence}
Fix a value of the scaling parameter $\ell > 0$, then for any pair of points $p, z \in \mathrm{supp} F_i$ with $|z-p| < 1/\ell$

\begin{equation}
\Big| \thedelta_{\epsilon_i} \kappa_\ell(z) - \frak{Im} \Big[ \overline{\partial}F_i(p)  \, \mathcal{E} \Big] \Big| \leq 4/\ell \, M_i(z,\ell)
\quad \mathrm{where}
\end{equation}

\begin{equation}
M_i(z,\ell) := \, \max_{|w-z| \, < \, 1/\ell} \big| \partial^2 F_i(w) \big| \, 
+ \, 2 \max_{|w-z| \, < \, 1/\ell} \big| \partial \overline{\partial} F_i(w) \big| \,
+ \,  \max_{|w-z| \, < \, 1/\ell} \big| \overline{\partial}^2 F_i(w) \big|
\end{equation}
\end{lemma}

\begin{proof}
For brevity we'll simply write $F$ instead of either $F_1$ or $F_2$
and $\thedelta_\epsilon \kappa_\ell(z)$ instead of 
$\thedelta_{\epsilon_1} \kappa_\ell(z)$ or $\thedelta_{\epsilon_2} \kappa_\ell(z)$.
For indices  $i, j ,k \in \{1,2,3,4 \}$ we'll use the provisional notation

$$ A_{ijk} := \  \overline{\nabla} F \Big( z + \frak{z}_i /\ell, \, z + \frak{z}_j /\ell, \, z + \frak{z}_k/\ell \Big) 
 \, -  \, \overline{\partial}F(p)  $$

\noindent
By Remark \ref{robust-lemma} if $|z - p| < 1/\ell$ we have

$$ \big| A_{ijk} \big| \, \leq \,  
R \, \Bigg(  \, {5 \over 2} \, 
\max_{z \in B}
\big| \partial^2 F  \big| \, + \, 
3 \,  \max_{z \in B} \big| \partial \overline{\partial} F \big| \, + \,  {1 \over 2} \, \max_{z \in B}
\big| \overline{\partial}^2  F \big| \, \Bigg)
$$

\noindent
where $B$ is the disk of radius $R= 1/\ell$ centered at $z$. By formula
(\ref{derivative-conformal-angle-at-zero}) we have

\begin{equation}
\begin{array}{ll}
\Big| \thedelta_\epsilon \kappa_\ell(z) \, -  \, \frak{Im} \Big[ \overline{\partial}F(p) \mathcal{E} \Big] \Big|
& = \,
\Big|  \frak{Im} \Big[  A_{124} \big( \frak{e}_{12} - \frak{e}_{14} \big)  \, + \,
 A_{234}  \big( \frak{e}_{34} - \frak{e}_{23} \big)  \Big]  \Big| \\ \\
& \leq \, 
\big| A_{124} \big|  \cdot \big| \frak{e}_{12} - \frak{e}_{14} \big|
\, + \, 
\big| A_{234} \big| \cdot  \big| \frak{e}_{34} - \frak{e}_{23} \big| \\ \\
& \leq \, 
2 \Big( \big| A_{124} \big| \, + \, 
\big| A_{234} \big| \Big) \\ \\
\end{array}
\end{equation}
Accordingly we have
$$
\Big| \thedelta_\epsilon \kappa_\ell(z)   \, -  \, \frak{Im} \Big[ \overline{\partial}F(p) \mathcal{E} \Big] \Big|
\, \leq \,
4/ \ell
\Bigg(  \, {5 \over 2} \, 
\max_{z \in B}
\big| \partial^2 F  \big| \, + \, 
3 \,  \max_{z \in B} \big| \partial \overline{\partial} F \big| \, + \,  {1 \over 2} \, \max_{z \in B}
\big| \overline{\partial}^2  F \big| \, \Bigg)
$$
\end{proof}

\begin{Def}
For a fixed value of the scaling parameter $\ell >0$ and any (continuous) function $\phi: \Bbb{C} \longrightarrow \Bbb{C}$
let us introduce the following piecewise abridgment 

\begin{equation}
\big\langle \phi \big\rangle_\ell (p) := \ 
\left\{ 
\begin{array}{ll}
\D \phi \big(z_\uq/\ell \big)
&\D
\begin{array}{l}
\D \text{whenever $\ell p \in \mathrm{int} \big( \uq \big)$} \\
\D \text{for a quadrilateral $\uq$}
\end{array} \\ \\
\D {1 \over 2}  \sum_{k=1}^2 \phi \big(z_{\uq_k}/\ell \big)
&\D 
\begin{array}{l}
\D \text{whenever $\ell p \in \mathrm{int} \big( \partial \uq_1 \cap \partial \uq_2 \big)$ for} \\
\D \text{a pair of quadrilaterals $\uq_1$ and $\uq_2$}
\end{array} \\ \\
\D {1 \over 4}  \sum_{k=1}^4 \phi \big(z_{\uq_k}/\ell \big)
&\D 
\begin{array}{l}
\D \text{whenever $\ell p \in \partial \uq_1 \cap \partial \uq_2  \cap \partial \uq_3 \cap \partial \uq_4$} \\
\D \text{for quadrilaterals $\uq_1$, $\uq_2$, $\uq_3$, and $\uq_4$}
\end{array}

\end{array}
\right.
\end{equation}
\end{Def}

\begin{Rem}
Let $\chi_{i ; \ell} := \big\langle \thedelta_{\epsilon_i} \kappa_\ell  \big\rangle_\ell $
then
$\chi_{i; \ell}  \longrightarrow   \frak{Im} \big[ \, \overline{\partial} F_i \, \mathcal{E}  \big]$ uniformly
in the $\ell \rightarrow \infty$ limit. Furthermore $\chi^{\pm}_{i ; \ell} \longrightarrow 
 \frak{Im}^\pm \big[ \, \overline{\partial} F_i \, \mathcal{E}  \big]$ uniformly as $\ell \rightarrow \infty$
 where $g^+(p) := \mathrm{max} \, (g(p), 0)$ and $g^-(p) := - \,  \mathrm{min} \, (g(p),0)$
 for any real-valued function $g: \Bbb{C} \longrightarrow \Bbb{R}$.
\end{Rem}

\begin{Prop}
\label{continuum-limits-anomalies}
For signs  $\sigma, \tau \in \{+ , - \}$ define

\begin{equation}
\begin{array}{l}
\D J^{(\sigma, \tau)} := \
{  \tan^2 \vartheta_\sigma \tan^2 \vartheta_\tau \over {16 \pi^2A^2_\mathcal{Q}}}  \iint\limits_{\Omega_1 \times \Omega_2} d^2x \, d^2y \
\frak{Im}^\sigma \Big[ \, \overline{\partial}F_1(x) \, \mathcal{E} \Big] 
\Bigg[ \frak{Re} \, { z_\sigma \, z_\tau \ \over {\big(x - y \big)^2}} \Bigg]^2
\frak{Im}^\tau \Big[ \, \overline{\partial}F_2(y) \,  \mathcal{E} \Big] \\ \\
\D J_\sigma^{(1)} := \
{  \tan^2 \vartheta_\sigma \over {8 \pi^2A_\mathcal{Q}}}  \iint\limits_{\Omega_1 \times \Omega_2} d^2x \, d^2y \
\frak{Im}^\sigma \Big[ \, \overline{\partial}F_1(x) \, \mathcal{E} \Big] \,  \frak{Re}
\Bigg[ { z_\sigma^2 \, \overline{\partial}F_2(y) \ \over {\big(x - y \big)^4}} \Bigg] \\ \\
\D J_\sigma^{(2)} := \
{  \tan^2 \vartheta_\sigma \over {8 \pi^2A_\mathcal{Q}}}  \iint\limits_{\Omega_1 \times \Omega_2} d^2x \, d^2y \
\frak{Re}
\Bigg[ { \overline{\partial}F_1(x) \, z_\sigma^2 \ \over {\big(x - y \big)^4}} \Bigg] \,
\frak{Im}^\sigma \Big[ \, \overline{\partial}F_2(y) \, \mathcal{E} \Big]
\end{array}
\end{equation}

\noindent
The continuum limits of the edge-to-chord $\Bbb{A}_\ell^{\mathrm{ed}\times\mathrm{ch}}$, chord-to-edge $\Bbb{A}_\ell^{\mathrm{ch}\times\mathrm{ed}}$, and chord-to-chord  $\Bbb{A}_\ell^{\mathrm{ch}\times\mathrm{ch}}$ anomalies exist and their values are: 

\begin{equation}
\begin{array}{ll}
\D \lim_{\ell \rightarrow \infty} \Bbb{A}^{\mathrm{ed} \times \mathrm{ch}}_\ell 
&\D = \
J^{(2)}_+ \, + \, J^{(2)}_- \\ \\
\D \lim_{\ell \rightarrow \infty} \Bbb{A}^{\mathrm{ch} \times \mathrm{ed}}_\ell 
&\D = \
J^{(1)}_+ \, + \, J^{(1)}_- \\ \\
\D \lim_{\ell \rightarrow \infty} \Bbb{A}^{\mathrm{ch} \times \mathrm{ch}}_\ell 
&\D = \
J^{\scriptscriptstyle (+,+)} \, + \, J^{\scriptscriptstyle (+,-)} \, + \, J^{\scriptscriptstyle (-,+)} \, + \, J^{\scriptscriptstyle (-,-)}
\end{array}
\end{equation}
\end{Prop}

\begin{proof}
We'll verify the claim in the case of the chord-to-chord anomaly  $\Bbb{A}_\ell^{\mathrm{ch}\times\mathrm{ch}}$ and leave the remaining
cases to the reader. Begin with a pair of signs $\sigma, \tau \in \{\pm \}$.
For $(x,y) \in \Omega_1 \times \Omega_2$ let's introduce the following step-function

\begin{equation}
\D \Phi_\ell^{\sigma,\tau}(x,y) \, :=
\left\{
\begin{array}{ll}
\D \Big[  \thedelta_{\epsilon_1} \kappa_\ell (z_\ux/\ell) \Big]^\sigma
\cdot
\Bigg[ \frak{Re} \,  {z_\sigma \, z_\tau
\over {\big(z_\ux - z_\uy \big)^2}} \Bigg]^2 
\cdot \Big[  \thedelta_{\epsilon_2} \kappa_\ell (z_\uy/\ell) \Big]^\tau
&\D 
\begin{array}{l} \text{$\ell x \in \mathrm{int}(\ux)$} \\ 
\text{$\ell y \in \mathrm{int}(\uy)$} \\
\text{$\ux,\uy \in \mathrm{F}(\uG_\mathrm{cr})$}
\end{array} \\ \\
\D \text{bounded noise}
&\D 
\begin{array}{l}
\text{otherwise} 
\end{array}
\end{array}
\right.
\end{equation}

\bigskip
\noindent
Note that $\Bbb{A}_\ell^{\mathrm{ch}\times\mathrm{cr}} = \Bbb{J}_\ell^{\scriptscriptstyle (+,+)} +  \Bbb{J}_\ell^{\scriptscriptstyle (+,-)}
+  \Bbb{J}_\ell^{\scriptscriptstyle (-,+)} + \Bbb{J}_\ell^{\scriptscriptstyle (-,-)}$ where

\begin{equation}
\Bbb{J}_\ell^{( \sigma,\tau )} \ = \
{ \tan^2 \vartheta_\sigma \tan^2 \vartheta_\tau \over {16 \pi^2}} 
\sum_{\stackrel { \scriptstyle \ux \, \in \, \mathrm{F}(\uG_\mathrm{cr} )}{\ux \, \cap \, \Omega_1(\ell) \ne \emptyset }} \,
\sum_{\stackrel{ \scriptstyle \uy \in \mathrm{F}(\uG_\mathrm{cr})}
{\scriptstyle \uy \, \cap \, \Omega_2(\ell) \ne \emptyset} }
\Phi_\ell^{\sigma,\tau}\Big(z_\ux/\ell, z_\uy/\ell \Big)
\end{equation}

\bigskip
\noindent
It follows from Lemma \ref{uniform-convergence} that $\Phi_\ell^{\sigma,\tau}(x,y) \rightarrow \Phi^{\sigma,\tau}(x,y)$ converges uniformly 
on $\Omega_1 \times \Omega_2$ as $\ell \rightarrow \infty$ where

\[ 
\Phi^{\sigma,\tau}(x,y) \, := \
\frak{Im}^\sigma \Big[ \, \overline{\partial} F_1(x) \, \mathcal{E} \Big] \cdot
\Bigg[ \frak{Re} \,
{ z_\sigma \, z_\tau \ \over {\big(x - y \big)^2}} \Bigg]^2
\cdot \frak{Im}^\tau \Big[ \, \overline{\partial}F_2(y) \, \mathcal{E} \Big]
\]

\begin{equation}
\begin{array}{ll}
\D J^{(\sigma,\tau)}
&\D = \ 
{ \tan^2 \vartheta_\sigma \tan^2 \vartheta_\tau \over {16 \pi^2 A_\mathcal{Q}^2}} 
\, \iint\limits_{\Omega_1 \times \Omega_2} d^2x \, d^2y \ \Phi^{\sigma,\tau}(x,y) \\ \\
&\D = \ 
{ \tan^2 \vartheta_\sigma \tan^2 \vartheta_\tau \over {16 \pi^2 A_\mathcal{Q}^2}} 
\, \lim_{\ell \rightarrow \infty}
\, \iint\limits_{\Omega_1 \times \Omega_2} d^2x \, d^2y \ \Phi_\ell^{\sigma,\tau}(x,y) \\ \\
&\D = \
{ \tan^2 \vartheta_\sigma \tan^2 \vartheta_\tau \over {16 \pi^2}} 
\, \lim_{\ell \rightarrow \infty} 
\sum_{\stackrel { \scriptstyle \ux \, \in \, \mathrm{F}(\uG_\mathrm{cr} )}{\ux \, \cap \, \Omega_1(\ell) \ne \emptyset }} \,
\sum_{\stackrel{ \scriptstyle \uy \in \mathrm{F}(\uG_\mathrm{cr})}
{\scriptstyle \uy \, \cap \, \Omega_2(\ell) \ne \emptyset} }
\Phi_\ell^{\sigma,\tau}\Big(z_\ux/\ell, z_\uy/\ell \Big) \\ \\
&\D = \
\lim_{\ell \rightarrow \infty} \,
\Bbb{J}_\ell^{(\sigma,\tau) } 
\end{array} 
\end{equation}
\end{proof}

\bibliographystyle{alpha}
\bibliography{perturbing}

\end{document}